\newcommand\pfloc{https://github.com/plclub/cbpv-effects-coeffects/}

\newif\ifanonymous
\newif\ifextended
\newif\ifartifact
\newif\ifproducts

\anonymousfalse
\extendedtrue
\artifacttrue  
\productsfalse

\PassOptionsToPackage{dvipsnames,svgnames,x11names}{xcolor}
\PassOptionsToPackage{pdfusetitle}{hyperref}
\documentclass[nonacm,acmsmall,screen=true,review=false,\ifanonymous anonymous=true\else anonymous=false\fi]{acmart}

\usepackage{xcolor}
\usepackage{url}
\usepackage{stmaryrd}
\usepackage[T1]{fontenc}

\usepackage{draft}
\usepackage{ottalt}

\usepackage{mathpartir}
\usepackage{supertabular}

\newnote{scw}{blue} 
\newnote{cht}{purple} 
\newnote{ssa}{orange} 
\newnote{esa}{olive} 
\newnote{jvg}{green} 

\draftfalse

\usepackage{bbold}
\usepackage{caption}
\usepackage{subcaption}

\usepackage[para]{footmisc}

\title{Effects and Coeffects in Call-By-Push-Value \ifextended (Extended Version)
 \fi}
\acmJournal{PACMPL}

\author{Cassia Torczon}
\orcid{0009-0003-6717-9586}
\affiliation{
\institution{University of Pennsylvania}
\city{Philadelphia}
\country{USA}
}
\email{ctorczon@seas.upenn.edu}

\author{Emmanuel {Su\'arez Acevedo}}
\orcid{0009-0002-5515-6099}
\affiliation{
\institution{University of Pennsylvania}
\city{Philadelphia}
\country{USA}
}
\email{emsu@seas.upenn.edu}

\author{Shubh Agrawal}
\orcid{0009-0006-1844-3856}
\affiliation{
\institution{University of Michigan}
\city{Ann Arbor}
\country{USA}
}
\email{shbhgrwl@umich.edu}

\author{Joey Velez-Ginorio}
\orcid{0009-0004-6451-5107}
\affiliation{
\institution{University of Pennsylvania}
\city{Philadelphia}
\country{USA}
}
\email{joeyv@seas.upenn.edu}

\author{Stephanie Weirich}
\orcid{0000-0002-6756-9168}
\affiliation{
\institution{University of Pennsylvania}
\city{Philadelphia}
\country{USA}
}
\email{sweirich@seas.upenn.edu}

\makeatletter
\hypersetup{pdftitle={\@title},pdfauthor={\@author}}
\makeatother

\usepackage{bbding}
\usepackage{hyperref}
\usepackage{xspace}

\newcommand\coeffectcolor[0]{ACMBlue}
\newcommand\effectcolor[0]{ACMRed}

\begin{document}

%

\ifartifact
\newcommand*\link[3]{%
  {#3}\protect\footnote{\texttt{{#1}:{#2}}\!}%
}
\fi


\newcommand{\ottdrule}[4][]{{\displaystyle\frac{\begin{array}{l}#2\end{array}}{#3}\quad\ottdrulename{#4}}}
\newcommand{\ottusedrule}[1]{\[#1\]}
\newcommand{\ottpremise}[1]{ #1 \\}
\newenvironment{ottdefnblock}[3][]{ \framebox{\mbox{#2}} \quad #3 \\[0pt]}{}
\newenvironment{ottfundefnblock}[3][]{ \framebox{\mbox{#2}} \quad #3 \\[0pt]\begin{displaymath}\begin{array}{l}}{\end{array}\end{displaymath}}
\newcommand{\ottfunclause}[2]{ #1 \equiv #2 \\}
\newcommand{\ottnt}[1]{\mathit{#1}}
\newcommand{\ottmv}[1]{\mathit{#1}}
\newcommand{\ottkw}[1]{\mathbf{#1}}
\newcommand{\ottsym}[1]{#1}
\newcommand{\ottcom}[1]{\text{#1}}
\newcommand{\ottdrulename}[1]{\textsc{#1}}
\newcommand{\ottcomplu}[5]{\overline{#1}^{\,#2\in #3 #4 #5}}
\newcommand{\ottcompu}[3]{\overline{#1}^{\,#2<#3}}
\newcommand{\ottcomp}[2]{\overline{#1}^{\,#2}}
\newcommand{\ottgrammartabular}[1]{\begin{supertabular}{llcllllll}#1\end{supertabular}}
\newcommand{\ottmetavartabular}[1]{\begin{supertabular}{ll}#1\end{supertabular}}
\newcommand{\ottrulehead}[3]{$#1$ & & $#2$ & & & \multicolumn{2}{l}{#3}}
\newcommand{\ottprodline}[6]{& & $#1$ & $#2$ & $#3 #4$ & $#5$ & $#6$}
\newcommand{\ottfirstprodline}[6]{\ottprodline{#1}{#2}{#3}{#4}{#5}{#6}}
\newcommand{\ottlongprodline}[2]{& & $#1$ & \multicolumn{4}{l}{$#2$}}
\newcommand{\ottfirstlongprodline}[2]{\ottlongprodline{#1}{#2}}
\newcommand{\ottbindspecprodline}[6]{\ottprodline{#1}{#2}{#3}{#4}{#5}{#6}}
\newcommand{\ottprodnewline}{\\}
\newcommand{\ottinterrule}{\\[5.0mm]}
\newcommand{\ottafterlastrule}{\\}
\newcommand{\ottmetavars}{
\ottmetavartabular{
 $ \ottmv{var} ,\, \ottmv{x} ,\, \ottmv{y} ,\, \ottmv{z} $ & \ottcom{variables} \\
 $ \ottmv{s} $ & \ottcom{strings} \\
 $ \ottmv{index} ,\, \ottmv{i} ,\, \ottmv{k} $ &  \\
}}

\newcommand{\ottusage}{
\ottrulehead{\ottnt{usage}  ,\ \textcolor{\coeffectcolor}{q}}{::=}{}}

\newcommand{\otteff}{
\ottrulehead{\ottnt{eff}  ,\ \textcolor{\effectcolor}{\phi}  ,\ \psi}{::=}{}}

\newcommand{\ottvaltype}{
\ottrulehead{\ottnt{valtype}  ,\ \ottnt{A}}{::=}{\ottcom{value types}}\ottprodnewline
\ottfirstprodline{|}{\ottkw{unit}}{}{}{}{\ottcom{unit for pai\effectcolor values}}\ottprodnewline
\ottprodline{|}{\ottnt{A_{{\mathrm{1}}}}  \ottsym{+}  \ottnt{A_{{\mathrm{2}}}}}{}{}{}{\ottcom{sum}}\ottprodnewline
\ottprodline{|}{ \ottnt{A_{{\mathrm{1}}}} \times \ottnt{A_{{\mathrm{2}}}} }{}{}{}{\ottcom{pai\effectcolor values}}\ottprodnewline
\ottprodline{|}{ \ottnt{A_{{\mathrm{1}}}} \mathop{\&} \ottnt{A_{{\mathrm{2}}}} }{}{}{}{\ottcom{pai\effectcolor values}}\ottprodnewline
\ottprodline{|}{\ottkw{U} \, \ottnt{B}}{}{}{}{\ottcom{thunked computation}}\ottprodnewline
\ottprodline{|}{\ottsym{(}  \ottnt{A}  \ottsym{)}} {\textsf{S}}{}{}{}}

\newcommand{\ottcomptype}{
\ottrulehead{\ottnt{comptype}  ,\ \ottnt{B}}{::=}{\ottcom{computation types}}\ottprodnewline
\ottfirstprodline{|}{\ottnt{A}  \to  \ottnt{B}}{}{}{}{\ottcom{parameterized computation}}\ottprodnewline
\ottprodline{|}{ \top }{}{}{}{\ottcom{unit for pai\effectcolor computation}}\ottprodnewline
\ottprodline{|}{ \ottnt{B_{{\mathrm{1}}}}   \mathop{\&}   \ottnt{B_{{\mathrm{2}}}} }{}{}{}{\ottcom{paired computation}}\ottprodnewline
\ottprodline{|}{ \ottnt{B_{{\mathrm{1}}}}  \times  \ottnt{B_{{\mathrm{2}}}} }{}{}{}{\ottcom{paired computation}}\ottprodnewline
\ottprodline{|}{\ottkw{F} \, \ottnt{A}}{}{}{}{\ottcom{computation monad}}\ottprodnewline
\ottprodline{|}{\ottsym{(}  \ottnt{B}  \ottsym{)}} {\textsf{S}}{}{}{}}

\newcommand{\ottcomputation}{
\ottrulehead{\ottnt{computation}  ,\ \ottnt{M}  ,\ \ottnt{N}  ,\ \colorbox{\effectcolor}{m}  ,\ \colorbox{\effectcolor}{n}}{::=}{\ottcom{computation terms}}\ottprodnewline
\ottfirstprodline{|}{ \lambda  \ottmv{x} . \ottnt{M} }{}{}{}{\ottcom{value-abstraction}}\ottprodnewline
\ottprodline{|}{\ottnt{M} \, \ottnt{V}}{}{}{}{\ottcom{hole filling}}\ottprodnewline
\ottprodline{|}{\textcolor{\coeffectcolor}{q} \, \ottnt{M} \, \ottnt{V}}{}{}{}{\ottcom{hole filling}}\ottprodnewline
\ottprodline{|}{ \ottnt{V}  ;  \ottnt{N} }{}{}{}{\ottcom{elimination for value unit}}\ottprodnewline
\ottprodline{|}{ \ottkw{let}\; ( \ottmv{x_{{\mathrm{1}}}} ,  \ottmv{x_{{\mathrm{2}}}} ) =  \ottnt{V} \; \ottkw{in}\;  \ottnt{N} }{}{}{}{\ottcom{elimination for value pairs}}\ottprodnewline
\ottprodline{|}{ \ottkw{case}_{\textcolor{\coeffectcolor}{ \textcolor{\coeffectcolor}{q} } } \;  \ottnt{V} \; \ottkw{of}\;( \ottmv{x_{{\mathrm{1}}}} , \ottmv{x_{{\mathrm{2}}}} )\; \rightarrow\;  \ottnt{N} }{}{}{}{\ottcom{elimination for value pairs with coeffects}}\ottprodnewline
\ottprodline{|}{\ottkw{case} \, \ottnt{V} \, \ottkw{of} \, \ottkw{inl} \, \ottmv{x_{{\mathrm{1}}}}  \to  \ottnt{M_{{\mathrm{1}}}}  \mathsf{;} \, \ottkw{inr} \, \ottmv{x_{{\mathrm{2}}}}  \to  \ottnt{M_{{\mathrm{2}}}}}{}{}{}{\ottcom{elim for value sums}}\ottprodnewline
\ottprodline{|}{ \ottkw{case}_{\textcolor{\coeffectcolor}{ \textcolor{\coeffectcolor}{q} } }\;  \ottnt{V} \; \ottkw{of}\; \ottkw{inl} \; \ottmv{x_{{\mathrm{1}}}}  \rightarrow\;  \ottnt{M_{{\mathrm{1}}}}  ;  \ottkw{inr} \; \ottmv{x_{{\mathrm{2}}}}  \rightarrow\;  \ottnt{M_{{\mathrm{2}}}} }{}{}{}{\ottcom{elim for value sums with coeffects}}\ottprodnewline
\ottprodline{|}{ \langle\rangle }{}{}{}{\ottcom{intro for comp unit}}\ottprodnewline
\ottprodline{|}{ \langle  \ottnt{M_{{\mathrm{1}}}} , \ottnt{M_{{\mathrm{2}}}}  \rangle }{}{}{}{\ottcom{intro for comp pairs}}\ottprodnewline
\ottprodline{|}{ \ottnt{M}  . 1 }{}{}{}{\ottcom{elim for comp pair}}\ottprodnewline
\ottprodline{|}{ \ottnt{M}  . 2 }{}{}{}{\ottcom{elim for comp pair}}\ottprodnewline
\ottprodline{|}{\colorbox{\effectcolor}{V}  \ottsym{!}}{}{}{}{\ottcom{thunk forcing}}\ottprodnewline
\ottprodline{|}{\ottkw{return} \, \ottnt{V}}{}{}{}{\ottcom{returned value}}\ottprodnewline
\ottprodline{|}{\ottmv{x}  \leftarrow  \ottnt{M} \, \ottkw{in} \, \ottnt{N}}{}{}{}{\ottcom{monadic bind}}\ottprodnewline
\ottprodline{|}{\ottsym{(}  \ottnt{M_{{\mathrm{1}}}}  \ottsym{,}  \ottnt{M_{{\mathrm{2}}}}  \ottsym{)}}{}{}{}{\ottcom{intro for comp tensor}}\ottprodnewline
\ottprodline{|}{ \ottkw{case}_{\textcolor{\coeffectcolor}{ \textcolor{\coeffectcolor}{q} } }\;  \ottnt{M} \; \ottkw{of}\;( \ottmv{x_{{\mathrm{1}}}} , \ottmv{x_{{\mathrm{2}}}} )\; \rightarrow\;  \ottnt{N} }{}{}{}{\ottcom{elim for comp tensor}}\ottprodnewline
\ottprodline{|}{\ottsym{(}  \ottnt{M}  \ottsym{)}} {\textsf{S}}{}{}{}}

\newcommand{\ottvalue}{
\ottrulehead{\ottnt{value}  ,\ \ottnt{V}  ,\ \colorbox{\effectcolor}{V}}{::=}{\ottcom{value terms}}\ottprodnewline
\ottfirstprodline{|}{\ottmv{x}}{}{}{}{}\ottprodnewline
\ottprodline{|}{\ottsym{()}}{}{}{}{\ottcom{value unit}}\ottprodnewline
\ottprodline{|}{\ottsym{\{}  \ottnt{M}  \ottsym{\}}}{}{}{}{\ottcom{thunked computation}}\ottprodnewline
\ottprodline{|}{\ottsym{(}  \ottnt{V_{{\mathrm{1}}}}  \ottsym{,}  \ottnt{V_{{\mathrm{2}}}}  \ottsym{)}}{}{}{}{\ottcom{intro for value pair}}\ottprodnewline
\ottprodline{|}{ \langle  \ottnt{V_{{\mathrm{1}}}}  ,  \ottnt{V_{{\mathrm{2}}}}  \rangle }{}{}{}{\ottcom{intro for value with-pair}}\ottprodnewline
\ottprodline{|}{\ottkw{inl} \, \ottnt{V}}{}{}{}{\ottcom{intro for value sum}}\ottprodnewline
\ottprodline{|}{\ottkw{inr} \, \ottnt{V}}{}{}{}{\ottcom{intro for value sum}}\ottprodnewline
\ottprodline{|}{\ottnt{V}  \ottsym{.}  \ottsym{1}}{}{}{}{\ottcom{elim for val with}}\ottprodnewline
\ottprodline{|}{\ottnt{V}  \ottsym{.}  \ottsym{2}}{}{}{}{\ottcom{elim for val with}}\ottprodnewline
\ottprodline{|}{\ottsym{(}  \ottnt{V}  \ottsym{)}} {\textsf{S}}{}{}{}}

\newcommand{\ottcombinding}{
\ottrulehead{\ottnt{combinding}  ,\ \ottnt{b}}{::=}{}\ottprodnewline
\ottfirstprodline{|}{ \ottmv{x} ^{ \textcolor{\coeffectcolor}{ \textcolor{\coeffectcolor}{q} } } =  \ottnt{e} }{}{}{}{}\ottprodnewline
\ottprodline{|}{\ottnt{b_{{\mathrm{1}}}}  \ottsym{,} \, ... \, \ottsym{,}  \ottnt{b_{\ottmv{k}}}}{}{}{}{}}

\newcommand{\otttyp}{
\ottrulehead{\ottnt{typ}  ,\ \tau}{::=}{\ottcom{lambda-calculus types}}\ottprodnewline
\ottfirstprodline{|}{\tau_{{\mathrm{1}}}  \to  \tau_{{\mathrm{2}}}}{}{}{}{\ottcom{function type}}\ottprodnewline
\ottprodline{|}{\ottkw{unit}}{}{}{}{\ottcom{multiplicative unit}}\ottprodnewline
\ottprodline{|}{ \mathbf{unit} }{}{}{}{}\ottprodnewline
\ottprodline{|}{ \mathbb{1} }{}{}{}{\ottcom{multiplicative unit}}\ottprodnewline
\ottprodline{|}{ \tau_{{\mathrm{1}}}  \otimes  \tau_{{\mathrm{2}}} }{}{}{}{\ottcom{multiplicative conjunction (tensor, strict)}}\ottprodnewline
\ottprodline{|}{ \tau_{{\mathrm{1}}}  \mathop{\&}   \tau_{{\mathrm{2}}} }{}{}{}{\ottcom{additive conjunction, with (nonstrict)}}\ottprodnewline
\ottprodline{|}{\tau_{{\mathrm{1}}}  \ottsym{+}  \tau_{{\mathrm{2}}}}{}{}{}{\ottcom{sum type}}}

\newcommand{\ottexp}{
\ottrulehead{\ottnt{exp}  ,\ \ottnt{e}}{::=}{\ottcom{lambda-calculus expressions}}\ottprodnewline
\ottfirstprodline{|}{\ottmv{x}}{}{}{}{\ottcom{variable}}\ottprodnewline
\ottprodline{|}{ \lambda  \ottmv{x} . \ottnt{e} }{}{}{}{\ottcom{abstraction}}\ottprodnewline
\ottprodline{|}{\ottnt{e_{{\mathrm{1}}}} \, \ottnt{e_{{\mathrm{2}}}}}{}{}{}{\ottcom{application}}\ottprodnewline
\ottprodline{|}{ \ottnt{e_{{\mathrm{1}}}} ^{ \textcolor{\coeffectcolor}{q} }  \ottnt{e_{{\mathrm{2}}}} }{}{}{}{\ottcom{application}}\ottprodnewline
\ottprodline{|}{\ottsym{()}}{}{}{}{\ottcom{unit value}}\ottprodnewline
\ottprodline{|}{ \texttt{unit} }{}{}{}{}\ottprodnewline
\ottprodline{|}{ \ottnt{e_{{\mathrm{1}}}}  ;  \ottnt{e_{{\mathrm{2}}}} }{}{}{}{\ottcom{sequencing}}\ottprodnewline
\ottprodline{|}{ \ottnt{e_{{\mathrm{1}}}} ; \ottnt{e_{{\mathrm{2}}}} }{}{}{}{\ottcom{alt-sequencing}}\ottprodnewline
\ottprodline{|}{\ottsym{(}  \ottnt{e_{{\mathrm{1}}}}  \ottsym{,}  \ottnt{e_{{\mathrm{2}}}}  \ottsym{)}}{}{}{}{\ottcom{strict product}}\ottprodnewline
\ottprodline{|}{ \ottkw{let}\; ( \ottmv{x_{{\mathrm{1}}}} ,  \ottmv{x_{{\mathrm{2}}}} ) =  \ottnt{e_{{\mathrm{1}}}} \; \ottkw{in}\;  \ottnt{e_{{\mathrm{2}}}} }{}{}{}{\ottcom{pattern matching}}\ottprodnewline
\ottprodline{|}{ \ottkw{let}_{ \textcolor{\coeffectcolor}{q} }\; ( \ottmv{x_{{\mathrm{1}}}} ,  \ottmv{x_{{\mathrm{2}}}} ) =  \ottnt{e_{{\mathrm{1}}}} \; \ottkw{in}\;  \ottnt{e_{{\mathrm{2}}}} }{}{}{}{\ottcom{graded pattern matching}}\ottprodnewline
\ottprodline{|}{ \langle  \ottnt{e_{{\mathrm{1}}}} , \ottnt{e_{{\mathrm{2}}}}  \rangle }{}{}{}{\ottcom{nonstrict product}}\ottprodnewline
\ottprodline{|}{\ottnt{e}  \ottsym{.}  \ottsym{1}}{}{}{}{\ottcom{first projection}}\ottprodnewline
\ottprodline{|}{\ottnt{e}  \ottsym{.}  \ottsym{2}}{}{}{}{\ottcom{second projection}}\ottprodnewline
\ottprodline{|}{\ottkw{inl} \, \ottnt{e}}{}{}{}{}\ottprodnewline
\ottprodline{|}{\ottkw{inr} \, \ottnt{e}}{}{}{}{}\ottprodnewline
\ottprodline{|}{ \ottkw{case}\;  \ottnt{e} \; \ottkw{of}\; \ottkw{inl}\;  \ottmv{x_{{\mathrm{1}}}}  \rightarrow  \ottnt{e_{{\mathrm{1}}}}  \ottkw{;} \ottkw{inr}\;  \ottmv{x_{{\mathrm{2}}}}  \rightarrow  \ottnt{e_{{\mathrm{2}}}} }{}{}{}{}\ottprodnewline
\ottprodline{|}{ \ottkw{case}_{\textcolor{\coeffectcolor}{ \textcolor{\coeffectcolor}{q} } }\;  \ottnt{e} \; \ottkw{of}\;\ottkw{inl}\; \ottmv{x_{{\mathrm{1}}}}  \rightarrow\;  \ottnt{e_{{\mathrm{1}}}}  ; \ottkw{inr}\; \ottmv{x_{{\mathrm{2}}}}  \rightarrow\;  \ottnt{e_{{\mathrm{2}}}} }{}{}{}{}}

\newcommand{\ottval}{
\ottrulehead{\ottnt{val}}{::=}{\ottcom{lambda calculus values}}\ottprodnewline
\ottfirstprodline{|}{ \lambda  \ottmv{x} . \ottnt{e} }{}{}{}{\ottcom{abstraction}}\ottprodnewline
\ottprodline{|}{\ottsym{()}}{}{}{}{\ottcom{unit value}}\ottprodnewline
\ottprodline{|}{\ottsym{(}  \ottnt{val_{{\mathrm{1}}}}  \ottsym{,}  \ottnt{val_{{\mathrm{2}}}}  \ottsym{)}}{}{}{}{\ottcom{strict product}}\ottprodnewline
\ottprodline{|}{ \langle  \ottnt{e_{{\mathrm{1}}}} , \ottnt{e_{{\mathrm{2}}}}  \rangle }{}{}{}{\ottcom{nonstrict product}}\ottprodnewline
\ottprodline{|}{\ottkw{inl} \, \ottnt{val}}{}{}{}{}\ottprodnewline
\ottprodline{|}{\ottkw{inr} \, \ottnt{val}}{}{}{}{}}

\newcommand{\ottcontext}{
\ottrulehead{\ottnt{context}  ,\ \Gamma}{::=}{\ottcom{typing context}}\ottprodnewline
\ottfirstprodline{|}{\varnothing}{}{}{}{}\ottprodnewline
\ottprodline{|}{ \cdot }{}{}{}{}\ottprodnewline
\ottprodline{|}{\ottmv{x}  \ottsym{:}  \ottnt{A}}{}{}{}{}\ottprodnewline
\ottprodline{|}{\ottmv{x}  \ottsym{:}  \tau}{}{}{}{}\ottprodnewline
\ottprodline{|}{ \ottmv{x}  : [  \tau  ]_{\textcolor{\coeffectcolor}{ \textcolor{\coeffectcolor}{q} } } }{}{}{}{}\ottprodnewline
\ottprodline{|}{ \ottmv{x}  :^{\textcolor{\effectcolor}{ \textcolor{\effectcolor}{\phi} } }  \tau }{}{}{}{}\ottprodnewline
\ottprodline{|}{ \Gamma_{{\mathrm{1}}}  ,  \Gamma_{{\mathrm{2}}} }{}{}{}{}}

\newcommand{\ottgamma}{
\ottrulehead{\textcolor{\coeffectcolor}{\gamma}}{::=}{\ottcom{usage list}}\ottprodnewline
\ottfirstprodline{|}{\emptyset}{}{}{}{}\ottprodnewline
\ottprodline{|}{ \textcolor{\coeffectcolor}{ \textcolor{\coeffectcolor}{\gamma}_{{\mathrm{1}}} \mathop{,} \textcolor{\coeffectcolor}{q} } }{}{}{}{}\ottprodnewline
\ottprodline{|}{ \textcolor{\coeffectcolor}{ \textcolor{\coeffectcolor}{\gamma}_{{\mathrm{1}}} \mathop{,} \textcolor{\coeffectcolor}{\gamma}_{{\mathrm{2}}} } }{}{}{}{}\ottprodnewline
\ottprodline{|}{ \textcolor{\coeffectcolor}{\overline{0}_1} }{}{}{}{}\ottprodnewline
\ottprodline{|}{ \textcolor{\coeffectcolor}{\overline{0}_2} }{}{}{}{}\ottprodnewline
\ottprodline{|}{ \textcolor{\coeffectcolor}{\overline{0} } }{}{}{}{}\ottprodnewline
\ottprodline{|}{ \overline{1} }{}{}{}{}\ottprodnewline
\ottprodline{|}{ \textcolor{\coeffectcolor}{ \textcolor{\coeffectcolor}{q} \cdot \textcolor{\coeffectcolor}{\gamma} } }{}{}{}{}\ottprodnewline
\ottprodline{|}{ \textcolor{\coeffectcolor}{ \textcolor{\coeffectcolor}{\gamma}_{{\mathrm{1}}} \ottsym{+} \textcolor{\coeffectcolor}{\gamma}_{{\mathrm{2}}} } }{}{}{}{}\ottprodnewline
\ottprodline{|}{ \textcolor{\coeffectcolor}{\gamma} }{}{}{}{}}

\newcommand{\ottPsi}{
\ottrulehead{\Psi}{::=}{\ottcom{annotated contexts}}\ottprodnewline
\ottfirstprodline{|}{ \textcolor{\coeffectcolor}{ \textcolor{\coeffectcolor}{\gamma} }\! \cdot \! \Gamma }{}{}{}{}\ottprodnewline
\ottprodline{|}{ \llbracket { \Gamma } \rrbracket_{\textsc{n} } }{}{}{}{\ottcom{translation from comonadic lang}}\ottprodnewline
\ottprodline{|}{ \llbracket { \Gamma } \rrbracket_{\textsc{v} } }{}{}{}{\ottcom{translation from comonadic lang}}\ottprodnewline
\ottprodline{|}{ \Psi_{{\mathrm{1}}}   \mathop{,}   \Psi_{{\mathrm{2}}} }{}{}{}{}\ottprodnewline
\ottprodline{|}{\ottsym{(}  \Psi  \ottsym{)}}{}{}{}{}\ottprodnewline
\ottprodline{|}{ \ottmv{x}  :^{\textcolor{\coeffectcolor}{ \textcolor{\coeffectcolor}{q} } }  \tau }{}{}{}{}\ottprodnewline
\ottprodline{|}{ \ottmv{x}  :^{\textcolor{\coeffectcolor}{ \textcolor{\coeffectcolor}{q} } }  \ottnt{A} }{}{}{}{}\ottprodnewline
\ottprodline{|}{\emptyset}{}{}{}{}\ottprodnewline
\ottprodline{|}{ \Psi }{}{}{}{}}

\newcommand{\ottrho}{
\ottrulehead{\rho}{::=}{\ottcom{environments}}\ottprodnewline
\ottfirstprodline{|}{\emptyset}{}{}{}{}\ottprodnewline
\ottprodline{|}{ \cdot }{}{}{}{}\ottprodnewline
\ottprodline{|}{\ottmv{x}  \mapsto  \ottnt{W}}{}{}{}{}\ottprodnewline
\ottprodline{|}{\ottmv{x}  \mapsto  \ottnt{e}}{}{}{}{}\ottprodnewline
\ottprodline{|}{ \rho   \mathop{,}   \rho' }{}{}{}{}}

\newcommand{\ottmu}{
\ottrulehead{\mu}{::=}{\ottcom{weighted environments}}\ottprodnewline
\ottfirstprodline{|}{ \textcolor{\coeffectcolor}{ \textcolor{\coeffectcolor}{\gamma} }\! \cdot \! \rho }{}{}{}{}\ottprodnewline
\ottprodline{|}{ \ottmv{x}   \mapsto ^{\textcolor{\coeffectcolor}{ \textcolor{\coeffectcolor}{q} } }  \ottnt{V} }{}{}{}{}\ottprodnewline
\ottprodline{|}{ \ottmv{x}   \mapsto ^{\textcolor{\coeffectcolor}{ \textcolor{\coeffectcolor}{q} } }  \ottnt{W} }{}{}{}{}\ottprodnewline
\ottprodline{|}{ \ottmv{x}   \mapsto ^{\textcolor{\coeffectcolor}{ \textcolor{\coeffectcolor}{q} } }  \ottnt{e} }{}{}{}{}\ottprodnewline
\ottprodline{|}{ \ottmv{x}   \mapsto ^{\textcolor{\coeffectcolor}{ \textcolor{\coeffectcolor}{q} } }  \ottnt{val} }{}{}{}{}\ottprodnewline
\ottprodline{|}{ \mu_{{\mathrm{1}}}   \mathop{,} \;  \mu_{{\mathrm{2}}} }{}{}{}{}\ottprodnewline
\ottprodline{|}{\emptyset}{}{}{}{}\ottprodnewline
\ottprodline{|}{\ottsym{(}  \mu  \ottsym{)}}{}{}{}{}\ottprodnewline
\ottprodline{|}{ \mu }{}{}{}{}}

\newcommand{\ottW}{
\ottrulehead{\ottnt{W}}{::=}{\ottcom{closed CBPV values}}\ottprodnewline
\ottfirstprodline{|}{\ottsym{()}}{}{}{}{}\ottprodnewline
\ottprodline{|}{\ottsym{(}  \ottnt{W_{{\mathrm{1}}}}  \ottsym{,}  \ottnt{W_{{\mathrm{2}}}}  \ottsym{)}}{}{}{}{}\ottprodnewline
\ottprodline{|}{ \mathbf{clo}( \rho ,  \ottnt{M} ) }{}{}{}{\ottcom{closure}}\ottprodnewline
\ottprodline{|}{ \mathbf{clo}( \textcolor{\coeffectcolor}{\gamma} ,  \rho ,  \ottnt{M} ) }{}{}{}{\ottcom{closure}}\ottprodnewline
\ottprodline{|}{ \mathbf{dclo}( \textcolor{\coeffectcolor}{\gamma} ,  \rho ,  \textcolor{\coeffectcolor}{q} ,  \ottnt{M} ) }{}{}{}{\ottcom{dist closure}}\ottprodnewline
\ottprodline{|}{ \mathbf{clo}( \mu ,  \ottnt{M} ) }{}{}{}{\ottcom{closure}}\ottprodnewline
\ottprodline{|}{ \mathbf{clo}( \rho , \{  \ottnt{M}  \} ) }{}{}{}{\ottcom{closure}}\ottprodnewline
\ottprodline{|}{ \mathbf{clo}( \mu , \{  \ottnt{M}  \} ) }{}{}{}{\ottcom{closure}}\ottprodnewline
\ottprodline{|}{ \mathbf{clo}( \mu , \langle  \ottnt{V_{{\mathrm{1}}}}  ,  \ottnt{V_{{\mathrm{2}}}}  \rangle ) }{}{}{}{}\ottprodnewline
\ottprodline{|}{\ottkw{inl} \, \ottnt{W}}{}{}{}{}\ottprodnewline
\ottprodline{|}{\ottkw{inr} \, \ottnt{W}}{}{}{}{}}

\newcommand{\ottT}{
\ottrulehead{\ottnt{T}}{::=}{\ottcom{closed and terminal CBPV computations}}\ottprodnewline
\ottfirstprodline{|}{\ottkw{return} \, \ottnt{W}}{}{}{}{}\ottprodnewline
\ottprodline{|}{\ottsym{<>}}{}{}{}{}\ottprodnewline
\ottprodline{|}{ \mathbf{clo}(  \rho ,  \ottnt{M}  ) }{}{}{}{\ottcom{$\ottnt{M}$ must be an abstraction, with-pair, or unit}}\ottprodnewline
\ottprodline{|}{ \mathbf{clo}(  \mu ,  \ottnt{M}  ) }{}{}{}{\ottcom{$\ottnt{M}$ must be an abstraction, with-pair, or unit}}\ottprodnewline
\ottprodline{|}{\ottsym{(}  \ottnt{W_{{\mathrm{1}}}}  \ottsym{,}  \ottnt{W_{{\mathrm{2}}}}  \ottsym{)}}{}{}{}{}}

\newcommand{\ottvalset}{
\ottrulehead{\ottnt{valset}}{::=}{\ottcom{sets of values}}}

\newcommand{\ottcompset}{
\ottrulehead{\ottnt{compset}}{::=}{\ottcom{sets of terminal computations}}}

\newcommand{\ottvalclosure}{
\ottrulehead{\ottnt{valclosure}}{::=}{\ottcom{sets of values}}}

\newcommand{\ottcompclosure}{
\ottrulehead{\ottnt{compclosure}}{::=}{\ottcom{sets of computations}}}

\newcommand{\ottgrammar}{\ottgrammartabular{
\ottusage\ottinterrule
\otteff\ottinterrule
\ottvaltype\ottinterrule
\ottcomptype\ottinterrule
\ottcomputation\ottinterrule
\ottvalue\ottinterrule
\ottcombinding\ottinterrule
\otttyp\ottinterrule
\ottexp\ottinterrule
\ottval\ottinterrule
\ottcontext\ottinterrule
\ottgamma\ottinterrule
\ottPsi\ottinterrule
\ottrho\ottinterrule
\ottmu\ottinterrule
\ottW\ottinterrule
\ottT\ottinterrule
\ottvalset\ottinterrule
\ottcompset\ottinterrule
\ottvalclosure\ottinterrule
\ottcompclosure\ottafterlastrule
}}

\newcommand{\ottdrulecbvXXstepXXappXXfun}[1]{\ottdrule[#1]{%
\ottpremise{\ottnt{e}  \rightarrow_{\mathit{v} }  \ottnt{e'}}%
}{
\ottnt{e} \, \ottnt{e_{{\mathrm{2}}}}  \rightarrow_{\mathit{v} }  \ottnt{e'} \, \ottnt{e_{{\mathrm{2}}}}}{%
{\ottdrulename{cbv\_step\_app\_fun}}{}%
}}

\newcommand{\ottdrulecbvXXstepXXappXXarg}[1]{\ottdrule[#1]{%
\ottpremise{\ottnt{e}  \rightarrow_{\mathit{v} }  \ottnt{e'}}%
}{
\ottnt{val} \, \ottnt{e}  \rightarrow_{\mathit{v} }  \ottnt{val} \, \ottnt{e'}}{%
{\ottdrulename{cbv\_step\_app\_arg}}{}%
}}

\newcommand{\ottdrulecbvXXstepXXappXXabs}[1]{\ottdrule[#1]{%
}{
\ottsym{(}   \lambda  \ottmv{x} . \ottnt{e}   \ottsym{)} \, \ottnt{val}  \rightarrow_{\mathit{v} }  \ottnt{e}  \ottsym{[}  \ottmv{x}  \leadsto  \ottnt{val}  \ottsym{]}}{%
{\ottdrulename{cbv\_step\_app\_abs}}{}%
}}

\newcommand{\ottdrulecbvXXstepXXprodl}[1]{\ottdrule[#1]{%
\ottpremise{\ottnt{e}  \rightarrow_{\mathit{v} }  \ottnt{e'}}%
}{
\ottsym{(}  \ottnt{e}  \ottsym{,}  \ottnt{e_{{\mathrm{2}}}}  \ottsym{)}  \rightarrow_{\mathit{v} }  \ottsym{(}  \ottnt{e'}  \ottsym{,}  \ottnt{e_{{\mathrm{2}}}}  \ottsym{)}}{%
{\ottdrulename{cbv\_step\_prodl}}{}%
}}

\newcommand{\ottdrulecbvXXstepXXprodr}[1]{\ottdrule[#1]{%
\ottpremise{\ottnt{e}  \rightarrow_{\mathit{v} }  \ottnt{e'}}%
}{
\ottsym{(}  \ottnt{val}  \ottsym{,}  \ottnt{e}  \ottsym{)}  \rightarrow_{\mathit{v} }  \ottsym{(}  \ottnt{val}  \ottsym{,}  \ottnt{e'}  \ottsym{)}}{%
{\ottdrulename{cbv\_step\_prodr}}{}%
}}

\newcommand{\ottdrulecbvXXstepXXsplit}[1]{\ottdrule[#1]{%
}{
 \ottkw{let}\; ( \ottmv{x_{{\mathrm{1}}}} ,  \ottmv{x_{{\mathrm{2}}}} ) =  \ottsym{(}  \ottnt{val_{{\mathrm{1}}}}  \ottsym{,}  \ottnt{val_{{\mathrm{2}}}}  \ottsym{)} \; \ottkw{in}\;  \ottnt{e}   \rightarrow_{\mathit{v} }  \ottsym{(}  \ottnt{e}  \ottsym{[}  \ottmv{x_{{\mathrm{1}}}}  \leadsto  \ottnt{val_{{\mathrm{1}}}}  \ottsym{]}  \ottsym{[}  \ottmv{x_{{\mathrm{2}}}}  \leadsto  \ottnt{val_{{\mathrm{2}}}}  \ottsym{]}  \ottsym{)}}{%
{\ottdrulename{cbv\_step\_split}}{}%
}}

\newcommand{\ottdrulecbvXXstepXXsequence}[1]{\ottdrule[#1]{%
}{
 \ottsym{()}  ;  \ottnt{e}   \rightarrow_{\mathit{v} }  \ottnt{e}}{%
{\ottdrulename{cbv\_step\_sequence}}{}%
}}

\newcommand{\ottdrulecbvXXstepXXfst}[1]{\ottdrule[#1]{%
\ottpremise{\ottnt{e}  \rightarrow_{\mathit{v} }  \ottnt{e'}}%
}{
\ottnt{e}  \ottsym{.}  \ottsym{1}  \rightarrow_{\mathit{v} }  \ottnt{e'}  \ottsym{.}  \ottsym{1}}{%
{\ottdrulename{cbv\_step\_fst}}{}%
}}

\newcommand{\ottdrulecbvXXstepXXsnd}[1]{\ottdrule[#1]{%
\ottpremise{\ottnt{e}  \rightarrow_{\mathit{v} }  \ottnt{e'}}%
}{
\ottnt{e}  \ottsym{.}  \ottsym{2}  \rightarrow_{\mathit{v} }  \ottnt{e'}  \ottsym{.}  \ottsym{2}}{%
{\ottdrulename{cbv\_step\_snd}}{}%
}}

\newcommand{\ottdrulecbvXXstepXXpairXXfst}[1]{\ottdrule[#1]{%
}{
 \langle  \ottnt{e_{{\mathrm{1}}}} , \ottnt{e_{{\mathrm{2}}}}  \rangle   \ottsym{.}  \ottsym{1}  \rightarrow_{\mathit{v} }  \ottnt{e_{{\mathrm{1}}}}}{%
{\ottdrulename{cbv\_step\_pair\_fst}}{}%
}}

\newcommand{\ottdrulecbvXXstepXXpairXXsnd}[1]{\ottdrule[#1]{%
}{
 \langle  \ottnt{e_{{\mathrm{1}}}} , \ottnt{e_{{\mathrm{2}}}}  \rangle   \ottsym{.}  \ottsym{2}  \rightarrow_{\mathit{v} }  \ottnt{e_{{\mathrm{2}}}}}{%
{\ottdrulename{cbv\_step\_pair\_snd}}{}%
}}

\newcommand{\ottdrulecbvXXstepXXinl}[1]{\ottdrule[#1]{%
\ottpremise{\ottnt{e}  \rightarrow_{\mathit{v} }  \ottnt{e'}}%
}{
\ottkw{inl} \, \ottnt{e}  \rightarrow_{\mathit{v} }  \ottkw{inl} \, \ottnt{e'}}{%
{\ottdrulename{cbv\_step\_inl}}{}%
}}

\newcommand{\ottdrulecbvXXstepXXinr}[1]{\ottdrule[#1]{%
\ottpremise{\ottnt{e}  \rightarrow_{\mathit{v} }  \ottnt{e'}}%
}{
\ottkw{inr} \, \ottnt{e}  \rightarrow_{\mathit{v} }  \ottkw{inr} \, \ottnt{e'}}{%
{\ottdrulename{cbv\_step\_inr}}{}%
}}

\newcommand{\ottdrulecbvXXstepXXcasel}[1]{\ottdrule[#1]{%
}{
 \ottkw{case}\;  \ottsym{(}  \ottkw{inl} \, \ottnt{val}  \ottsym{)} \; \ottkw{of}\; \ottkw{inl}\;  \ottmv{x_{{\mathrm{1}}}}  \rightarrow  \ottnt{e_{{\mathrm{1}}}}  \ottkw{;} \ottkw{inr}\;  \ottmv{x_{{\mathrm{2}}}}  \rightarrow  \ottnt{e_{{\mathrm{2}}}}   \rightarrow_{\mathit{v} }  \ottnt{e_{{\mathrm{1}}}}  \ottsym{[}  \ottmv{x_{{\mathrm{1}}}}  \leadsto  \ottnt{val}  \ottsym{]}}{%
{\ottdrulename{cbv\_step\_casel}}{}%
}}

\newcommand{\ottdrulecbvXXstepXXcaser}[1]{\ottdrule[#1]{%
}{
 \ottkw{case}\;  \ottsym{(}  \ottkw{inr} \, \ottnt{val}  \ottsym{)} \; \ottkw{of}\; \ottkw{inl}\;  \ottmv{x_{{\mathrm{1}}}}  \rightarrow  \ottnt{e_{{\mathrm{1}}}}  \ottkw{;} \ottkw{inr}\;  \ottmv{x_{{\mathrm{2}}}}  \rightarrow  \ottnt{e_{{\mathrm{2}}}}   \rightarrow_{\mathit{v} }  \ottnt{e_{{\mathrm{2}}}}  \ottsym{[}  \ottmv{x_{{\mathrm{2}}}}  \leadsto  \ottnt{val}  \ottsym{]}}{%
{\ottdrulename{cbv\_step\_caser}}{}%
}}

\newcommand{\ottdrulecbvXXstepXXtick}[1]{\ottdrule[#1]{%
}{
 \textcolor{\effectcolor}{ \ottkw{tick} }   \rightarrow_{\mathit{v} }  \ottsym{()}}{%
{\ottdrulename{cbv\_step\_tick}}{}%
}}

\newcommand{\ottdrulecbvXXstepXXreturn}[1]{\ottdrule[#1]{%
\ottpremise{\ottnt{e}  \rightarrow_{\mathit{v} }  \ottnt{e'}}%
}{
\ottkw{return} \, \ottnt{e}  \rightarrow_{\mathit{v} }  \ottkw{return} \, \ottnt{e'}}{%
{\ottdrulename{cbv\_step\_return}}{}%
}}

\newcommand{\ottdrulecbvXXstepXXbind}[1]{\ottdrule[#1]{%
\ottpremise{\ottnt{e}  \rightarrow_{\mathit{v} }  \ottnt{e'}}%
}{
\ottkw{bind} \, \ottmv{x}  \ottsym{=}  \ottnt{e} \, \ottkw{in} \, \ottnt{e_{{\mathrm{1}}}}  \rightarrow_{\mathit{v} }  \ottkw{bind} \, \ottmv{x}  \ottsym{=}  \ottnt{e'} \, \ottkw{in} \, \ottnt{e_{{\mathrm{1}}}}}{%
{\ottdrulename{cbv\_step\_bind}}{}%
}}

\newcommand{\ottdrulecbvXXstepXXbindXXret}[1]{\ottdrule[#1]{%
}{
\ottkw{bind} \, \ottmv{x}  \ottsym{=}  \ottkw{return} \, \ottnt{val} \, \ottkw{in} \, \ottnt{e}  \rightarrow_{\mathit{v} }  \ottnt{e}  \ottsym{[}  \ottmv{x}  \leadsto  \ottnt{val}  \ottsym{]}}{%
{\ottdrulename{cbv\_step\_bind\_ret}}{}%
}}

\newcommand{\ottdefnstepXXcbv}[1]{\begin{ottdefnblock}[#1]{$\ottnt{e}  \rightarrow_{\mathit{v} }  \ottnt{e'}$}{}
\ottusedrule{\ottdrulecbvXXstepXXappXXfun{}}
\ottusedrule{\ottdrulecbvXXstepXXappXXarg{}}
\ottusedrule{\ottdrulecbvXXstepXXappXXabs{}}
\ottusedrule{\ottdrulecbvXXstepXXprodl{}}
\ottusedrule{\ottdrulecbvXXstepXXprodr{}}
\ottusedrule{\ottdrulecbvXXstepXXsplit{}}
\ottusedrule{\ottdrulecbvXXstepXXsequence{}}
\ottusedrule{\ottdrulecbvXXstepXXfst{}}
\ottusedrule{\ottdrulecbvXXstepXXsnd{}}
\ottusedrule{\ottdrulecbvXXstepXXpairXXfst{}}
\ottusedrule{\ottdrulecbvXXstepXXpairXXsnd{}}
\ottusedrule{\ottdrulecbvXXstepXXinl{}}
\ottusedrule{\ottdrulecbvXXstepXXinr{}}
\ottusedrule{\ottdrulecbvXXstepXXcasel{}}
\ottusedrule{\ottdrulecbvXXstepXXcaser{}}
\ottusedrule{\ottdrulecbvXXstepXXtick{}}
\ottusedrule{\ottdrulecbvXXstepXXreturn{}}
\ottusedrule{\ottdrulecbvXXstepXXbind{}}
\ottusedrule{\ottdrulecbvXXstepXXbindXXret{}}
\end{ottdefnblock}}

\newcommand{\ottdefnsJSmallStepCBV}{
\ottdefnstepXXcbv{}}

\newcommand{\ottdrulestepXXapp}[1]{\ottdrule[#1]{%
\ottpremise{\ottnt{M}  \longrightarrow  \ottnt{M'}}%
}{
\ottnt{M} \, \ottnt{V}  \longrightarrow  \ottnt{M'} \, \ottnt{V}}{%
{\ottdrulename{step\_app}}{}%
}}

\newcommand{\ottdrulestepXXappXXabs}[1]{\ottdrule[#1]{%
}{
\ottsym{(}   \lambda  \ottmv{x} . \ottnt{M}   \ottsym{)} \, \ottnt{V}  \longrightarrow  \ottnt{M}  \ottsym{[}  \ottmv{x}  \leadsto  \ottnt{V}  \ottsym{]}}{%
{\ottdrulename{step\_app\_abs}}{}%
}}

\newcommand{\ottdrulestepXXforceXXthunk}[1]{\ottdrule[#1]{%
}{
\ottsym{\{}  \ottnt{M}  \ottsym{\}}  \ottsym{!}  \longrightarrow  \ottnt{M}}{%
{\ottdrulename{step\_force\_thunk}}{}%
}}

\newcommand{\ottdrulestepXXletin}[1]{\ottdrule[#1]{%
\ottpremise{\ottnt{M}  \longrightarrow  \ottnt{M'}}%
}{
\ottmv{x}  \leftarrow  \ottnt{M} \, \ottkw{in} \, \ottnt{N}  \longrightarrow  \ottmv{x}  \leftarrow  \ottnt{M'} \, \ottkw{in} \, \ottnt{N}}{%
{\ottdrulename{step\_letin}}{}%
}}

\newcommand{\ottdrulestepXXletinXXret}[1]{\ottdrule[#1]{%
}{
\ottmv{x}  \leftarrow  \ottkw{return} \, \ottnt{V} \, \ottkw{in} \, \ottnt{N}  \longrightarrow  \colorbox{\effectcolor}{n}  \ottsym{[}  \ottmv{x}  \leadsto  \ottnt{V}  \ottsym{]}}{%
{\ottdrulename{step\_letin\_ret}}{}%
}}

\newcommand{\ottdrulestepXXsplit}[1]{\ottdrule[#1]{%
}{
 \ottkw{let}\; ( \ottmv{x_{{\mathrm{1}}}} ,  \ottmv{x_{{\mathrm{2}}}} ) =  \ottsym{(}  \ottnt{V_{{\mathrm{1}}}}  \ottsym{,}  \ottnt{V_{{\mathrm{2}}}}  \ottsym{)} \; \ottkw{in}\;  \ottnt{N}   \longrightarrow  \ottsym{(}  \colorbox{\effectcolor}{n}  \ottsym{[}  \ottmv{x_{{\mathrm{1}}}}  \leadsto  \ottnt{V_{{\mathrm{1}}}}  \ottsym{]}  \ottsym{[}  \ottmv{x_{{\mathrm{2}}}}  \leadsto  \ottnt{V_{{\mathrm{2}}}}  \ottsym{]}  \ottsym{)}}{%
{\ottdrulename{step\_split}}{}%
}}

\newcommand{\ottdrulestepXXsequence}[1]{\ottdrule[#1]{%
}{
 \ottsym{()}  ;  \ottnt{N}   \longrightarrow  \colorbox{\effectcolor}{n}}{%
{\ottdrulename{step\_sequence}}{}%
}}

\newcommand{\ottdrulestepXXfst}[1]{\ottdrule[#1]{%
\ottpremise{\ottnt{M}  \longrightarrow  \ottnt{M'}}%
}{
 \ottnt{M}  . 1   \longrightarrow   \ottnt{M'}  . 1 }{%
{\ottdrulename{step\_fst}}{}%
}}

\newcommand{\ottdrulestepXXsnd}[1]{\ottdrule[#1]{%
\ottpremise{\ottnt{M}  \longrightarrow  \ottnt{M'}}%
}{
 \ottnt{M}  . 2   \longrightarrow   \ottnt{M'}  . 2 }{%
{\ottdrulename{step\_snd}}{}%
}}

\newcommand{\ottdrulestepXXcpairXXfst}[1]{\ottdrule[#1]{%
}{
  \langle  \ottnt{M_{{\mathrm{1}}}} , \ottnt{M_{{\mathrm{2}}}}  \rangle   . 1   \longrightarrow  \ottnt{M_{{\mathrm{1}}}}}{%
{\ottdrulename{step\_cpair\_fst}}{}%
}}

\newcommand{\ottdrulestepXXcpairXXsnd}[1]{\ottdrule[#1]{%
}{
  \langle  \ottnt{M_{{\mathrm{1}}}} , \ottnt{M_{{\mathrm{2}}}}  \rangle   . 2   \longrightarrow  \ottnt{M_{{\mathrm{2}}}}}{%
{\ottdrulename{step\_cpair\_snd}}{}%
}}

\newcommand{\ottdrulestepXXcasel}[1]{\ottdrule[#1]{%
}{
\ottkw{case} \, \ottsym{(}  \ottkw{inl} \, \ottnt{V}  \ottsym{)} \, \ottkw{of} \, \ottkw{inl} \, \ottmv{x_{{\mathrm{1}}}}  \to  \ottnt{M_{{\mathrm{1}}}}  \mathsf{;} \, \ottkw{inr} \, \ottmv{x_{{\mathrm{2}}}}  \to  \ottnt{M_{{\mathrm{2}}}}  \longrightarrow  \ottnt{M_{{\mathrm{1}}}}  \ottsym{[}  \ottmv{x_{{\mathrm{1}}}}  \leadsto  \ottnt{V}  \ottsym{]}}{%
{\ottdrulename{step\_casel}}{}%
}}

\newcommand{\ottdrulestepXXcaser}[1]{\ottdrule[#1]{%
}{
\ottkw{case} \, \ottsym{(}  \ottkw{inr} \, \ottnt{V}  \ottsym{)} \, \ottkw{of} \, \ottkw{inl} \, \ottmv{x_{{\mathrm{1}}}}  \to  \ottnt{M_{{\mathrm{1}}}}  \mathsf{;} \, \ottkw{inr} \, \ottmv{x_{{\mathrm{2}}}}  \to  \ottnt{M_{{\mathrm{2}}}}  \longrightarrow  \ottnt{M_{{\mathrm{2}}}}  \ottsym{[}  \ottmv{x_{{\mathrm{2}}}}  \leadsto  \ottnt{V}  \ottsym{]}}{%
{\ottdrulename{step\_caser}}{}%
}}

\newcommand{\ottdrulestepXXletinXXretq}[1]{\ottdrule[#1]{%
}{
\ottmv{x}  \leftarrow   \ottkw{return} _{\textcolor{\coeffectcolor}{ \textcolor{\coeffectcolor}{q} } }\;  \ottnt{V}  \, \ottkw{in} \, \ottnt{N}  \longrightarrow  \colorbox{\effectcolor}{n}  \ottsym{[}  \ottmv{x}  \leadsto  \ottnt{V}  \ottsym{]}}{%
{\ottdrulename{step\_letin\_retq}}{}%
}}

\newcommand{\ottdrulestepXXtick}[1]{\ottdrule[#1]{%
}{
 \textcolor{\effectcolor}{ \ottkw{tick} }   \longrightarrow  \ottkw{return} \, \ottsym{()}}{%
{\ottdrulename{step\_tick}}{}%
}}

\newcommand{\ottdefnstep}[1]{\begin{ottdefnblock}[#1]{$\ottnt{M}  \longrightarrow  \ottnt{M'}$}{\ottcom{small-step operational semantics (substitution)}}
\ottusedrule{\ottdrulestepXXapp{}}
\ottusedrule{\ottdrulestepXXappXXabs{}}
\ottusedrule{\ottdrulestepXXforceXXthunk{}}
\ottusedrule{\ottdrulestepXXletin{}}
\ottusedrule{\ottdrulestepXXletinXXret{}}
\ottusedrule{\ottdrulestepXXsplit{}}
\ottusedrule{\ottdrulestepXXsequence{}}
\ottusedrule{\ottdrulestepXXfst{}}
\ottusedrule{\ottdrulestepXXsnd{}}
\ottusedrule{\ottdrulestepXXcpairXXfst{}}
\ottusedrule{\ottdrulestepXXcpairXXsnd{}}
\ottusedrule{\ottdrulestepXXcasel{}}
\ottusedrule{\ottdrulestepXXcaser{}}
\ottusedrule{\ottdrulestepXXletinXXretq{}}
\ottusedrule{\ottdrulestepXXtick{}}
\end{ottdefnblock}}

\newcommand{\ottdefnsJSmallStepCBPV}{
\ottdefnstep{}}

\newcommand{\ottdruleevalXXvalXXvar}[1]{\ottdrule[#1]{%
\ottpremise{\ottmv{x}  \mapsto  \ottnt{W} \, \in \, \rho}%
}{
 \rho  \vdash  \ottmv{x} \ \Downarrow\  \ottnt{W} }{%
{\ottdrulename{eval\_val\_var}}{}%
}}

\newcommand{\ottdruleevalXXvalXXunit}[1]{\ottdrule[#1]{%
}{
 \rho  \vdash  \ottsym{()} \ \Downarrow\  \ottsym{()} }{%
{\ottdrulename{eval\_val\_unit}}{}%
}}

\newcommand{\ottdruleevalXXvalXXthunk}[1]{\ottdrule[#1]{%
}{
 \rho  \vdash  \ottsym{\{}  \ottnt{M}  \ottsym{\}} \ \Downarrow\   \mathbf{clo}( \rho , \{  \ottnt{M}  \} )  }{%
{\ottdrulename{eval\_val\_thunk}}{}%
}}

\newcommand{\ottdruleevalXXvalXXvpair}[1]{\ottdrule[#1]{%
\ottpremise{ \rho  \vdash  \ottnt{V_{{\mathrm{1}}}} \ \Downarrow\  \ottnt{W_{{\mathrm{1}}}} }%
\ottpremise{ \rho  \vdash  \ottnt{V_{{\mathrm{2}}}} \ \Downarrow\  \ottnt{W_{{\mathrm{2}}}} }%
}{
 \rho  \vdash  \ottsym{(}  \ottnt{V_{{\mathrm{1}}}}  \ottsym{,}  \ottnt{V_{{\mathrm{2}}}}  \ottsym{)} \ \Downarrow\  \ottsym{(}  \ottnt{W_{{\mathrm{1}}}}  \ottsym{,}  \ottnt{W_{{\mathrm{2}}}}  \ottsym{)} }{%
{\ottdrulename{eval\_val\_vpair}}{}%
}}

\newcommand{\ottdruleevalXXvalXXinl}[1]{\ottdrule[#1]{%
\ottpremise{ \rho  \vdash  \ottnt{V} \ \Downarrow\  \ottnt{W} }%
}{
 \rho  \vdash  \ottkw{inl} \, \ottnt{V} \ \Downarrow\  \ottkw{inl} \, \ottnt{W} }{%
{\ottdrulename{eval\_val\_inl}}{}%
}}

\newcommand{\ottdruleevalXXvalXXinr}[1]{\ottdrule[#1]{%
\ottpremise{ \rho  \vdash  \ottnt{V} \ \Downarrow\  \ottnt{W} }%
}{
 \rho  \vdash  \ottkw{inr} \, \ottnt{V} \ \Downarrow\  \ottkw{inr} \, \ottnt{W} }{%
{\ottdrulename{eval\_val\_inr}}{}%
}}

\newcommand{\ottdefnevalXXval}[1]{\begin{ottdefnblock}[#1]{$ \rho  \vdash  \ottnt{V} \ \Downarrow\  \ottnt{W} $}{\ottcom{environment-based semantics for CBPV (large-step)}}
\ottusedrule{\ottdruleevalXXvalXXvar{}}
\ottusedrule{\ottdruleevalXXvalXXunit{}}
\ottusedrule{\ottdruleevalXXvalXXthunk{}}
\ottusedrule{\ottdruleevalXXvalXXvpair{}}
\ottusedrule{\ottdruleevalXXvalXXinl{}}
\ottusedrule{\ottdruleevalXXvalXXinr{}}
\end{ottdefnblock}}

\newcommand{\ottdruleevalXXcompXXabs}[1]{\ottdrule[#1]{%
}{
 \rho  \vdash   \lambda  \ottmv{x} . \ottnt{M}  \ \Downarrow\   \mathbf{clo}(  \rho ,   \lambda  \ottmv{x} . \ottnt{M}   )  }{%
{\ottdrulename{eval\_comp\_abs}}{}%
}}

\newcommand{\ottdruleevalXXcompXXcpair}[1]{\ottdrule[#1]{%
}{
 \rho  \vdash   \langle  \ottnt{M_{{\mathrm{1}}}} , \ottnt{M_{{\mathrm{2}}}}  \rangle  \ \Downarrow\   \mathbf{clo}(  \rho ,   \langle  \ottnt{M_{{\mathrm{1}}}} , \ottnt{M_{{\mathrm{2}}}}  \rangle   )  }{%
{\ottdrulename{eval\_comp\_cpair}}{}%
}}

\newcommand{\ottdruleevalXXcompXXunit}[1]{\ottdrule[#1]{%
}{
 \rho  \vdash   \langle\rangle  \ \Downarrow\   \mathbf{clo}(  \rho ,   \langle\rangle   )  }{%
{\ottdrulename{eval\_comp\_unit}}{}%
}}

\newcommand{\ottdruleevalXXcompXXappXXabs}[1]{\ottdrule[#1]{%
\ottpremise{ \rho  \vdash  \ottnt{M} \ \Downarrow\   \mathbf{clo}(  \rho' ,   \lambda  \ottmv{x} . \ottnt{M'}   )  }%
\ottpremise{ \rho  \vdash  \ottnt{V} \ \Downarrow\  \ottnt{W} }%
\ottpremise{  \rho'   \mathop{,}   \ottmv{x}  \mapsto  \ottnt{W}   \vdash  \ottnt{M'} \ \Downarrow\  \ottnt{T} }%
}{
 \rho  \vdash  \ottnt{M} \, \ottnt{V} \ \Downarrow\  \ottnt{T} }{%
{\ottdrulename{eval\_comp\_app\_abs}}{}%
}}

\newcommand{\ottdruleevalXXcompXXforceXXthunk}[1]{\ottdrule[#1]{%
\ottpremise{ \rho  \vdash  \ottnt{V} \ \Downarrow\   \mathbf{clo}( \rho' ,  \ottnt{M} )  }%
\ottpremise{ \rho'  \vdash  \ottnt{M} \ \Downarrow\  \ottnt{T} }%
}{
 \rho  \vdash  \ottnt{V}  \ottsym{!} \ \Downarrow\  \ottnt{T} }{%
{\ottdrulename{eval\_comp\_force\_thunk}}{}%
}}

\newcommand{\ottdruleevalXXcompXXreturn}[1]{\ottdrule[#1]{%
\ottpremise{ \rho  \vdash  \ottnt{V} \ \Downarrow\  \ottnt{W} }%
}{
 \rho  \vdash  \ottkw{return} \, \ottnt{V} \ \Downarrow\  \ottkw{return} \, \ottnt{W} }{%
{\ottdrulename{eval\_comp\_return}}{}%
}}

\newcommand{\ottdruleevalXXcompXXletinXXret}[1]{\ottdrule[#1]{%
\ottpremise{ \rho  \vdash  \ottnt{M} \ \Downarrow\  \ottkw{return} \, \ottnt{W} }%
\ottpremise{  \rho   \mathop{,}   \ottmv{x}  \mapsto  \ottnt{W}   \vdash  \ottnt{N} \ \Downarrow\  \ottnt{T} }%
}{
 \rho  \vdash  \ottmv{x}  \leftarrow  \ottnt{M} \, \ottkw{in} \, \ottnt{N} \ \Downarrow\  \ottnt{T} }{%
{\ottdrulename{eval\_comp\_letin\_ret}}{}%
}}

\newcommand{\ottdruleevalXXcompXXsplit}[1]{\ottdrule[#1]{%
\ottpremise{ \rho  \vdash  \ottnt{V} \ \Downarrow\  \ottsym{(}  \ottnt{W_{{\mathrm{1}}}}  \ottsym{,}  \ottnt{W_{{\mathrm{2}}}}  \ottsym{)} }%
\ottpremise{   \rho   \mathop{,}   \ottmv{x_{{\mathrm{1}}}}  \mapsto  \ottnt{W_{{\mathrm{1}}}}    \mathop{,}   \ottmv{x_{{\mathrm{2}}}}  \mapsto  \ottnt{W_{{\mathrm{2}}}}   \vdash  \ottnt{N} \ \Downarrow\  \ottnt{T} }%
}{
 \rho  \vdash   \ottkw{let}\; ( \ottmv{x_{{\mathrm{1}}}} ,  \ottmv{x_{{\mathrm{2}}}} ) =  \ottnt{V} \; \ottkw{in}\;  \ottnt{N}  \ \Downarrow\  \ottnt{T} }{%
{\ottdrulename{eval\_comp\_split}}{}%
}}

\newcommand{\ottdruleevalXXcompXXsequence}[1]{\ottdrule[#1]{%
\ottpremise{ \rho  \vdash  \ottnt{V} \ \Downarrow\  \ottsym{()} }%
\ottpremise{ \rho  \vdash  \ottnt{M} \ \Downarrow\  \ottnt{T} }%
}{
 \rho  \vdash   \ottnt{V}  ;  \ottnt{M}  \ \Downarrow\  \ottnt{T} }{%
{\ottdrulename{eval\_comp\_sequence}}{}%
}}

\newcommand{\ottdruleevalXXcompXXfst}[1]{\ottdrule[#1]{%
\ottpremise{ \rho  \vdash  \ottnt{M} \ \Downarrow\   \mathbf{clo}(  \rho' ,   \langle  \ottnt{M_{{\mathrm{1}}}} , \ottnt{M_{{\mathrm{2}}}}  \rangle   )  }%
\ottpremise{ \rho'  \vdash  \ottnt{M_{{\mathrm{1}}}} \ \Downarrow\  \ottnt{T} }%
}{
 \rho  \vdash   \ottnt{M}  . 1  \ \Downarrow\  \ottnt{T} }{%
{\ottdrulename{eval\_comp\_fst}}{}%
}}

\newcommand{\ottdruleevalXXcompXXsnd}[1]{\ottdrule[#1]{%
\ottpremise{ \rho  \vdash  \ottnt{M} \ \Downarrow\   \mathbf{clo}(  \rho' ,   \langle  \ottnt{M_{{\mathrm{1}}}} , \ottnt{M_{{\mathrm{2}}}}  \rangle   )  }%
\ottpremise{ \rho'  \vdash  \ottnt{M_{{\mathrm{2}}}} \ \Downarrow\  \ottnt{T} }%
}{
 \rho  \vdash   \ottnt{M}  . 2  \ \Downarrow\  \ottnt{T} }{%
{\ottdrulename{eval\_comp\_snd}}{}%
}}

\newcommand{\ottdruleevalXXcompXXcaseXXinl}[1]{\ottdrule[#1]{%
\ottpremise{ \rho  \vdash  \ottnt{V} \ \Downarrow\  \ottkw{inl} \, \ottnt{W} }%
\ottpremise{  \rho   \mathop{,}   \ottmv{x_{{\mathrm{1}}}}  \mapsto  \ottnt{W}   \vdash  \ottnt{M_{{\mathrm{1}}}} \ \Downarrow\  \ottnt{T} }%
}{
 \rho  \vdash  \ottkw{case} \, \ottnt{V} \, \ottkw{of} \, \ottkw{inl} \, \ottmv{x_{{\mathrm{1}}}}  \to  \ottnt{M_{{\mathrm{1}}}}  \mathsf{;} \, \ottkw{inr} \, \ottmv{x_{{\mathrm{2}}}}  \to  \ottnt{M_{{\mathrm{2}}}} \ \Downarrow\  \ottnt{T} }{%
{\ottdrulename{eval\_comp\_case\_inl}}{}%
}}

\newcommand{\ottdruleevalXXcompXXcaseXXinr}[1]{\ottdrule[#1]{%
\ottpremise{ \rho  \vdash  \ottnt{V} \ \Downarrow\  \ottkw{inr} \, \ottnt{W} }%
\ottpremise{  \rho   \mathop{,}   \ottmv{x_{{\mathrm{2}}}}  \mapsto  \ottnt{W}   \vdash  \ottnt{M_{{\mathrm{2}}}} \ \Downarrow\  \ottnt{T} }%
}{
 \rho  \vdash  \ottkw{case} \, \ottnt{V} \, \ottkw{of} \, \ottkw{inl} \, \ottmv{x_{{\mathrm{1}}}}  \to  \ottnt{M_{{\mathrm{1}}}}  \mathsf{;} \, \ottkw{inr} \, \ottmv{x_{{\mathrm{2}}}}  \to  \ottnt{M_{{\mathrm{2}}}} \ \Downarrow\  \ottnt{T} }{%
{\ottdrulename{eval\_comp\_case\_inr}}{}%
}}

\newcommand{\ottdefnevalXXcomp}[1]{\begin{ottdefnblock}[#1]{$ \rho  \vdash  \ottnt{M} \ \Downarrow\  \ottnt{T} $}{\ottcom{environment-based semantics for CBPV (large-step)}}
\ottusedrule{\ottdruleevalXXcompXXabs{}}
\ottusedrule{\ottdruleevalXXcompXXcpair{}}
\ottusedrule{\ottdruleevalXXcompXXunit{}}
\ottusedrule{\ottdruleevalXXcompXXappXXabs{}}
\ottusedrule{\ottdruleevalXXcompXXforceXXthunk{}}
\ottusedrule{\ottdruleevalXXcompXXreturn{}}
\ottusedrule{\ottdruleevalXXcompXXletinXXret{}}
\ottusedrule{\ottdruleevalXXcompXXsplit{}}
\ottusedrule{\ottdruleevalXXcompXXsequence{}}
\ottusedrule{\ottdruleevalXXcompXXfst{}}
\ottusedrule{\ottdruleevalXXcompXXsnd{}}
\ottusedrule{\ottdruleevalXXcompXXcaseXXinl{}}
\ottusedrule{\ottdruleevalXXcompXXcaseXXinr{}}
\end{ottdefnblock}}

\newcommand{\ottdefnsJEnv}{
\ottdefnevalXXval{}\ottdefnevalXXcomp{}}

\newcommand{\ottdruleevalXXeffXXcompXXabs}[1]{\ottdrule[#1]{%
}{
 \rho  \vdash_{\mathit{eff} }   \lambda  \ottmv{x} . \ottnt{M}   \Downarrow   \mathbf{clo}(  \rho ,   \lambda  \ottmv{x} . \ottnt{M}   )   \mathop{\#} \textcolor{\effectcolor}{  \textcolor{\effectcolor}{\varepsilon}  } }{%
{\ottdrulename{eval\_eff\_comp\_abs}}{}%
}}

\newcommand{\ottdruleevalXXeffXXcompXXreturn}[1]{\ottdrule[#1]{%
\ottpremise{ \rho  \vdash  \ottnt{V} \ \Downarrow\  \ottnt{W} }%
}{
 \rho  \vdash_{\mathit{eff} }  \ottkw{return} \, \ottnt{V}  \Downarrow  \ottkw{return} \, \ottnt{W}  \mathop{\#} \textcolor{\effectcolor}{  \textcolor{\effectcolor}{\varepsilon}  } }{%
{\ottdrulename{eval\_eff\_comp\_return}}{}%
}}

\newcommand{\ottdruleevalXXeffXXcompXXcpair}[1]{\ottdrule[#1]{%
}{
 \rho  \vdash_{\mathit{eff} }   \langle  \ottnt{M_{{\mathrm{1}}}} , \ottnt{M_{{\mathrm{2}}}}  \rangle   \Downarrow   \mathbf{clo}(  \rho ,   \langle  \ottnt{M_{{\mathrm{1}}}} , \ottnt{M_{{\mathrm{2}}}}  \rangle   )   \mathop{\#} \textcolor{\effectcolor}{  \textcolor{\effectcolor}{\varepsilon}  } }{%
{\ottdrulename{eval\_eff\_comp\_cpair}}{}%
}}

\newcommand{\ottdruleevalXXeffXXcompXXunit}[1]{\ottdrule[#1]{%
}{
 \rho  \vdash_{\mathit{eff} }   \langle\rangle   \Downarrow   \mathbf{clo}(  \rho ,   \langle\rangle   )   \mathop{\#} \textcolor{\effectcolor}{  \textcolor{\effectcolor}{\varepsilon}  } }{%
{\ottdrulename{eval\_eff\_comp\_unit}}{}%
}}

\newcommand{\ottdruleevalXXeffXXcompXXappXXabs}[1]{\ottdrule[#1]{%
\ottpremise{ \rho  \vdash_{\mathit{eff} }  \ottnt{M}  \Downarrow   \mathbf{clo}(  \rho' ,   \lambda  \ottmv{x} . \ottnt{M'}   )   \mathop{\#} \textcolor{\effectcolor}{ \textcolor{\effectcolor}{\phi}_{{\mathrm{1}}} } }%
\ottpremise{ \rho  \vdash  \ottnt{V} \ \Downarrow\  \ottnt{W} }%
\ottpremise{  \rho'   \mathop{,}   \ottmv{x}  \mapsto  \ottnt{W}   \vdash_{\mathit{eff} }  \ottnt{M'}  \Downarrow  \ottnt{T}  \mathop{\#} \textcolor{\effectcolor}{ \textcolor{\effectcolor}{\phi}_{{\mathrm{2}}} } }%
}{
 \rho  \vdash_{\mathit{eff} }  \ottnt{M} \, \ottnt{V}  \Downarrow  \ottnt{T}  \mathop{\#} \textcolor{\effectcolor}{  \textcolor{\effectcolor}{ \textcolor{\effectcolor}{\phi}_{{\mathrm{1}}}  \cdot  \textcolor{\effectcolor}{\phi}_{{\mathrm{2}}} }  } }{%
{\ottdrulename{eval\_eff\_comp\_app\_abs}}{}%
}}

\newcommand{\ottdruleevalXXeffXXcompXXforceXXthunk}[1]{\ottdrule[#1]{%
\ottpremise{ \rho  \vdash  \ottnt{V} \ \Downarrow\   \mathbf{clo}( \rho' , \{  \ottnt{M}  \} )  }%
\ottpremise{ \rho'  \vdash_{\mathit{eff} }  \ottnt{M}  \Downarrow  \ottnt{T}  \mathop{\#} \textcolor{\effectcolor}{ \textcolor{\effectcolor}{\phi} } }%
}{
 \rho  \vdash_{\mathit{eff} }  \ottnt{V}  \ottsym{!}  \Downarrow  \ottnt{T}  \mathop{\#} \textcolor{\effectcolor}{ \textcolor{\effectcolor}{\phi} } }{%
{\ottdrulename{eval\_eff\_comp\_force\_thunk}}{}%
}}

\newcommand{\ottdruleevalXXeffXXcompXXletinXXret}[1]{\ottdrule[#1]{%
\ottpremise{ \rho  \vdash_{\mathit{eff} }  \ottnt{M}  \Downarrow  \ottkw{return} \, \ottnt{W}  \mathop{\#} \textcolor{\effectcolor}{ \textcolor{\effectcolor}{\phi}_{{\mathrm{1}}} } }%
\ottpremise{  \rho   \mathop{,}   \ottmv{x}  \mapsto  \ottnt{W}   \vdash_{\mathit{eff} }  \ottnt{N}  \Downarrow  \ottnt{T}  \mathop{\#} \textcolor{\effectcolor}{ \textcolor{\effectcolor}{\phi}_{{\mathrm{2}}} } }%
}{
 \rho  \vdash_{\mathit{eff} }  \ottmv{x}  \leftarrow  \ottnt{M} \, \ottkw{in} \, \ottnt{N}  \Downarrow  \ottnt{T}  \mathop{\#} \textcolor{\effectcolor}{  \textcolor{\effectcolor}{ \textcolor{\effectcolor}{\phi}_{{\mathrm{1}}}  \cdot  \textcolor{\effectcolor}{\phi}_{{\mathrm{2}}} }  } }{%
{\ottdrulename{eval\_eff\_comp\_letin\_ret}}{}%
}}

\newcommand{\ottdruleevalXXeffXXcompXXsplit}[1]{\ottdrule[#1]{%
\ottpremise{ \rho  \vdash  \ottnt{V} \ \Downarrow\  \ottsym{(}  \ottnt{W_{{\mathrm{1}}}}  \ottsym{,}  \ottnt{W_{{\mathrm{2}}}}  \ottsym{)} }%
\ottpremise{   \rho   \mathop{,}   \ottmv{x_{{\mathrm{1}}}}  \mapsto  \ottnt{W_{{\mathrm{1}}}}    \mathop{,}   \ottmv{x_{{\mathrm{2}}}}  \mapsto  \ottnt{W_{{\mathrm{2}}}}   \vdash_{\mathit{eff} }  \ottnt{N}  \Downarrow  \ottnt{T}  \mathop{\#} \textcolor{\effectcolor}{ \textcolor{\effectcolor}{\phi} } }%
}{
 \rho  \vdash_{\mathit{eff} }   \ottkw{let}\; ( \ottmv{x_{{\mathrm{1}}}} ,  \ottmv{x_{{\mathrm{2}}}} ) =  \ottnt{V} \; \ottkw{in}\;  \ottnt{N}   \Downarrow  \ottnt{T}  \mathop{\#} \textcolor{\effectcolor}{ \textcolor{\effectcolor}{\phi} } }{%
{\ottdrulename{eval\_eff\_comp\_split}}{}%
}}

\newcommand{\ottdruleevalXXeffXXcompXXsequence}[1]{\ottdrule[#1]{%
\ottpremise{ \rho  \vdash  \ottnt{V} \ \Downarrow\  \ottsym{()} }%
\ottpremise{ \rho  \vdash_{\mathit{eff} }  \ottnt{N}  \Downarrow  \ottnt{T}  \mathop{\#} \textcolor{\effectcolor}{ \textcolor{\effectcolor}{\phi} } }%
}{
 \rho  \vdash_{\mathit{eff} }   \ottnt{V}  ;  \ottnt{N}   \Downarrow  \ottnt{T}  \mathop{\#} \textcolor{\effectcolor}{ \textcolor{\effectcolor}{\phi} } }{%
{\ottdrulename{eval\_eff\_comp\_sequence}}{}%
}}

\newcommand{\ottdruleevalXXeffXXcompXXfst}[1]{\ottdrule[#1]{%
\ottpremise{ \rho  \vdash_{\mathit{eff} }  \ottnt{M}  \Downarrow   \mathbf{clo}(  \rho' ,   \langle  \ottnt{M_{{\mathrm{1}}}} , \ottnt{M_{{\mathrm{2}}}}  \rangle   )   \mathop{\#} \textcolor{\effectcolor}{ \textcolor{\effectcolor}{\phi}_{{\mathrm{1}}} } }%
\ottpremise{ \rho'  \vdash_{\mathit{eff} }  \ottnt{M_{{\mathrm{1}}}}  \Downarrow  \ottnt{T}  \mathop{\#} \textcolor{\effectcolor}{ \textcolor{\effectcolor}{\phi}_{{\mathrm{2}}} } }%
}{
 \rho  \vdash_{\mathit{eff} }   \ottnt{M}  . 1   \Downarrow  \ottnt{T}  \mathop{\#} \textcolor{\effectcolor}{  \textcolor{\effectcolor}{ \textcolor{\effectcolor}{\phi}_{{\mathrm{1}}}  \cdot  \textcolor{\effectcolor}{\phi}_{{\mathrm{2}}} }  } }{%
{\ottdrulename{eval\_eff\_comp\_fst}}{}%
}}

\newcommand{\ottdruleevalXXeffXXcompXXsnd}[1]{\ottdrule[#1]{%
\ottpremise{ \rho  \vdash_{\mathit{eff} }  \ottnt{M}  \Downarrow   \mathbf{clo}(  \rho' ,   \langle  \ottnt{M_{{\mathrm{1}}}} , \ottnt{M_{{\mathrm{2}}}}  \rangle   )   \mathop{\#} \textcolor{\effectcolor}{ \textcolor{\effectcolor}{\phi}_{{\mathrm{1}}} } }%
\ottpremise{ \rho'  \vdash_{\mathit{eff} }  \ottnt{M_{{\mathrm{2}}}}  \Downarrow  \ottnt{T}  \mathop{\#} \textcolor{\effectcolor}{ \textcolor{\effectcolor}{\phi}_{{\mathrm{2}}} } }%
}{
 \rho  \vdash_{\mathit{eff} }   \ottnt{M}  . 2   \Downarrow  \ottnt{T}  \mathop{\#} \textcolor{\effectcolor}{  \textcolor{\effectcolor}{ \textcolor{\effectcolor}{\phi}_{{\mathrm{1}}}  \cdot  \textcolor{\effectcolor}{\phi}_{{\mathrm{2}}} }  } }{%
{\ottdrulename{eval\_eff\_comp\_snd}}{}%
}}

\newcommand{\ottdruleevalXXeffXXcompXXcaseXXinl}[1]{\ottdrule[#1]{%
\ottpremise{ \rho  \vdash  \ottnt{V} \ \Downarrow\  \ottkw{inl} \, \ottnt{W} }%
\ottpremise{  \rho   \mathop{,}   \ottmv{x_{{\mathrm{1}}}}  \mapsto  \ottnt{W}   \vdash_{\mathit{eff} }  \ottnt{M_{{\mathrm{1}}}}  \Downarrow  \ottnt{T}  \mathop{\#} \textcolor{\effectcolor}{ \textcolor{\effectcolor}{\phi} } }%
}{
 \rho  \vdash_{\mathit{eff} }  \ottkw{case} \, \ottnt{V} \, \ottkw{of} \, \ottkw{inl} \, \ottmv{x_{{\mathrm{1}}}}  \to  \ottnt{M_{{\mathrm{1}}}}  \mathsf{;} \, \ottkw{inr} \, \ottmv{x_{{\mathrm{2}}}}  \to  \ottnt{M_{{\mathrm{2}}}}  \Downarrow  \ottnt{T}  \mathop{\#} \textcolor{\effectcolor}{ \textcolor{\effectcolor}{\phi} } }{%
{\ottdrulename{eval\_eff\_comp\_case\_inl}}{}%
}}

\newcommand{\ottdruleevalXXeffXXcompXXcaseXXinr}[1]{\ottdrule[#1]{%
\ottpremise{ \rho  \vdash  \ottnt{V} \ \Downarrow\  \ottkw{inr} \, \ottnt{W} }%
\ottpremise{  \rho   \mathop{,}   \ottmv{x_{{\mathrm{2}}}}  \mapsto  \ottnt{W}   \vdash_{\mathit{eff} }  \ottnt{M_{{\mathrm{2}}}}  \Downarrow  \ottnt{T}  \mathop{\#} \textcolor{\effectcolor}{ \textcolor{\effectcolor}{\phi} } }%
}{
 \rho  \vdash_{\mathit{eff} }  \ottkw{case} \, \ottnt{V} \, \ottkw{of} \, \ottkw{inl} \, \ottmv{x_{{\mathrm{1}}}}  \to  \ottnt{M_{{\mathrm{1}}}}  \mathsf{;} \, \ottkw{inr} \, \ottmv{x_{{\mathrm{2}}}}  \to  \ottnt{M_{{\mathrm{2}}}}  \Downarrow  \ottnt{T}  \mathop{\#} \textcolor{\effectcolor}{ \textcolor{\effectcolor}{\phi} } }{%
{\ottdrulename{eval\_eff\_comp\_case\_inr}}{}%
}}

\newcommand{\ottdruleevalXXeffXXcompXXtick}[1]{\ottdrule[#1]{%
}{
 \rho  \vdash_{\mathit{eff} }   \textcolor{\effectcolor}{ \ottkw{tick} }   \Downarrow  \ottkw{return} \, \ottsym{()}  \mathop{\#} \textcolor{\effectcolor}{  \textcolor{\effectcolor}{\ottkw{Tick} }  } }{%
{\ottdrulename{eval\_eff\_comp\_tick}}{}%
}}

\newcommand{\ottdefnevalXXeffXXcomp}[1]{\begin{ottdefnblock}[#1]{$ \rho  \vdash_{\mathit{eff} }  \ottnt{M}  \Downarrow  \ottnt{T}  \mathop{\#} \textcolor{\effectcolor}{ \textcolor{\effectcolor}{\phi} } $}{\ottcom{environment-based semantics for CBPV (large-step)}}
\ottusedrule{\ottdruleevalXXeffXXcompXXabs{}}
\ottusedrule{\ottdruleevalXXeffXXcompXXreturn{}}
\ottusedrule{\ottdruleevalXXeffXXcompXXcpair{}}
\ottusedrule{\ottdruleevalXXeffXXcompXXunit{}}
\ottusedrule{\ottdruleevalXXeffXXcompXXappXXabs{}}
\ottusedrule{\ottdruleevalXXeffXXcompXXforceXXthunk{}}
\ottusedrule{\ottdruleevalXXeffXXcompXXletinXXret{}}
\ottusedrule{\ottdruleevalXXeffXXcompXXsplit{}}
\ottusedrule{\ottdruleevalXXeffXXcompXXsequence{}}
\ottusedrule{\ottdruleevalXXeffXXcompXXfst{}}
\ottusedrule{\ottdruleevalXXeffXXcompXXsnd{}}
\ottusedrule{\ottdruleevalXXeffXXcompXXcaseXXinl{}}
\ottusedrule{\ottdruleevalXXeffXXcompXXcaseXXinr{}}
\ottusedrule{\ottdruleevalXXeffXXcompXXtick{}}
\end{ottdefnblock}}

\newcommand{\ottdefnsJEffEnv}{
\ottdefnevalXXeffXXcomp{}}

\newcommand{\ottdruleevalXXcoeffXXvalXXvar}[1]{\ottdrule[#1]{%
}{
     \textcolor{\coeffectcolor}{  \textcolor{\coeffectcolor}{\overline{0}_1}  }\! \cdot \! \rho_{{\mathrm{1}}}    \mathop{,} \;   \ottmv{x}   \mapsto ^{\textcolor{\coeffectcolor}{  \textcolor{\coeffectcolor}{1}  } }  \ottnt{W}      \mathop{,} \;   \textcolor{\coeffectcolor}{  \textcolor{\coeffectcolor}{\overline{0}_2}  }\! \cdot \! \rho_{{\mathrm{2}}}    \vdash_{\mathit{coeff} }  \ottmv{x}  \Downarrow  \ottnt{W} }{%
{\ottdrulename{eval\_coeff\_val\_var}}{}%
}}

\newcommand{\ottdruleevalXXcoeffXXvalXXunit}[1]{\ottdrule[#1]{%
}{
  \textcolor{\coeffectcolor}{  \textcolor{\coeffectcolor}{\overline{0} }  }\! \cdot \! \rho   \vdash_{\mathit{coeff} }  \ottsym{()}  \Downarrow  \ottsym{()} }{%
{\ottdrulename{eval\_coeff\_val\_unit}}{}%
}}

\newcommand{\ottdruleevalXXcoeffXXvalXXthunk}[1]{\ottdrule[#1]{%
}{
  \textcolor{\coeffectcolor}{ \textcolor{\coeffectcolor}{\gamma} }\! \cdot \! \rho   \vdash_{\mathit{coeff} }  \ottsym{\{}  \ottnt{M}  \ottsym{\}}  \Downarrow   \mathbf{clo}(  \textcolor{\coeffectcolor}{ \textcolor{\coeffectcolor}{\gamma} }\! \cdot \! \rho  , \{  \ottnt{M}  \} )  }{%
{\ottdrulename{eval\_coeff\_val\_thunk}}{}%
}}

\newcommand{\ottdruleevalXXcoeffXXvalXXvpair}[1]{\ottdrule[#1]{%
\ottpremise{  \textcolor{\coeffectcolor}{ \textcolor{\coeffectcolor}{\gamma}_{{\mathrm{1}}} }\! \cdot \! \rho   \vdash_{\mathit{coeff} }  \ottnt{V_{{\mathrm{1}}}}  \Downarrow  \ottnt{W_{{\mathrm{1}}}} }%
\ottpremise{  \textcolor{\coeffectcolor}{ \textcolor{\coeffectcolor}{\gamma}_{{\mathrm{2}}} }\! \cdot \! \rho   \vdash_{\mathit{coeff} }  \ottnt{V_{{\mathrm{2}}}}  \Downarrow  \ottnt{W_{{\mathrm{2}}}} }%
}{
  \textcolor{\coeffectcolor}{  \textcolor{\coeffectcolor}{ \textcolor{\coeffectcolor}{\gamma}_{{\mathrm{1}}} \ottsym{+} \textcolor{\coeffectcolor}{\gamma}_{{\mathrm{2}}} }  }\! \cdot \! \rho   \vdash_{\mathit{coeff} }  \ottsym{(}  \ottnt{V_{{\mathrm{1}}}}  \ottsym{,}  \ottnt{V_{{\mathrm{2}}}}  \ottsym{)}  \Downarrow  \ottsym{(}  \ottnt{W_{{\mathrm{1}}}}  \ottsym{,}  \ottnt{W_{{\mathrm{2}}}}  \ottsym{)} }{%
{\ottdrulename{eval\_coeff\_val\_vpair}}{}%
}}

\newcommand{\ottdruleevalXXcoeffXXvalXXinl}[1]{\ottdrule[#1]{%
\ottpremise{  \textcolor{\coeffectcolor}{ \textcolor{\coeffectcolor}{\gamma} }\! \cdot \! \rho   \vdash_{\mathit{coeff} }  \ottnt{V}  \Downarrow  \ottnt{W} }%
}{
  \textcolor{\coeffectcolor}{ \textcolor{\coeffectcolor}{\gamma} }\! \cdot \! \rho   \vdash_{\mathit{coeff} }  \ottkw{inl} \, \ottnt{V}  \Downarrow  \ottkw{inl} \, \ottnt{W} }{%
{\ottdrulename{eval\_coeff\_val\_inl}}{}%
}}

\newcommand{\ottdruleevalXXcoeffXXvalXXinr}[1]{\ottdrule[#1]{%
\ottpremise{  \textcolor{\coeffectcolor}{ \textcolor{\coeffectcolor}{\gamma} }\! \cdot \! \rho   \vdash_{\mathit{coeff} }  \ottnt{V}  \Downarrow  \ottnt{W} }%
}{
  \textcolor{\coeffectcolor}{ \textcolor{\coeffectcolor}{\gamma} }\! \cdot \! \rho   \vdash_{\mathit{coeff} }  \ottkw{inr} \, \ottnt{V}  \Downarrow  \ottkw{inr} \, \ottnt{W} }{%
{\ottdrulename{eval\_coeff\_val\_inr}}{}%
}}

\newcommand{\ottdruleevalXXcoeffXXvalXXvwith}[1]{\ottdrule[#1]{%
}{
  \textcolor{\coeffectcolor}{ \textcolor{\coeffectcolor}{\gamma} }\! \cdot \! \rho   \vdash_{\mathit{coeff} }   \langle  \ottnt{V_{{\mathrm{1}}}}  ,  \ottnt{V_{{\mathrm{2}}}}  \rangle   \Downarrow   \mathbf{clo}(  \textcolor{\coeffectcolor}{ \textcolor{\coeffectcolor}{\gamma} }\! \cdot \! \rho  , \langle  \ottnt{V_{{\mathrm{1}}}}  ,  \ottnt{V_{{\mathrm{2}}}}  \rangle )  }{%
{\ottdrulename{eval\_coeff\_val\_vwith}}{}%
}}

\newcommand{\ottdruleevalXXcoeffXXvalXXvfst}[1]{\ottdrule[#1]{%
\ottpremise{  \textcolor{\coeffectcolor}{ \textcolor{\coeffectcolor}{\gamma} }\! \cdot \! \rho   \vdash_{\mathit{coeff} }  \ottnt{V}  \Downarrow   \mathbf{clo}(  \textcolor{\coeffectcolor}{ \textcolor{\coeffectcolor}{\gamma}' }\! \cdot \! \rho'  , \langle  \ottnt{V_{{\mathrm{1}}}}  ,  \ottnt{V_{{\mathrm{2}}}}  \rangle )  }%
\ottpremise{  \textcolor{\coeffectcolor}{ \textcolor{\coeffectcolor}{\gamma}' }\! \cdot \! \rho'   \vdash_{\mathit{coeff} }  \ottnt{V_{{\mathrm{1}}}}  \Downarrow  \ottnt{W} }%
}{
  \textcolor{\coeffectcolor}{ \textcolor{\coeffectcolor}{\gamma} }\! \cdot \! \rho   \vdash_{\mathit{coeff} }  \ottnt{V}  \ottsym{.}  \ottsym{1}  \Downarrow  \ottnt{W} }{%
{\ottdrulename{eval\_coeff\_val\_vfst}}{}%
}}

\newcommand{\ottdruleevalXXcoeffXXvalXXvsnd}[1]{\ottdrule[#1]{%
\ottpremise{  \textcolor{\coeffectcolor}{ \textcolor{\coeffectcolor}{\gamma} }\! \cdot \! \rho   \vdash_{\mathit{coeff} }  \ottnt{V}  \Downarrow   \mathbf{clo}(  \textcolor{\coeffectcolor}{ \textcolor{\coeffectcolor}{\gamma}' }\! \cdot \! \rho'  , \langle  \ottnt{V_{{\mathrm{1}}}}  ,  \ottnt{V_{{\mathrm{2}}}}  \rangle )  }%
\ottpremise{  \textcolor{\coeffectcolor}{ \textcolor{\coeffectcolor}{\gamma}' }\! \cdot \! \rho'   \vdash_{\mathit{coeff} }  \ottnt{V_{{\mathrm{2}}}}  \Downarrow  \ottnt{W} }%
}{
  \textcolor{\coeffectcolor}{ \textcolor{\coeffectcolor}{\gamma} }\! \cdot \! \rho   \vdash_{\mathit{coeff} }  \ottnt{V}  \ottsym{.}  \ottsym{2}  \Downarrow  \ottnt{W} }{%
{\ottdrulename{eval\_coeff\_val\_vsnd}}{}%
}}

\newcommand{\ottdruleevalXXcoeffXXvalXXvsub}[1]{\ottdrule[#1]{%
\ottpremise{  \textcolor{\coeffectcolor}{ \textcolor{\coeffectcolor}{\gamma}_{{\mathrm{1}}} }\! \cdot \! \rho   \vdash_{\mathit{coeff} }  \ottnt{V}  \Downarrow  \ottnt{W} }%
\ottpremise{ \textcolor{\coeffectcolor}{ \textcolor{\coeffectcolor}{\gamma}_{{\mathrm{2}}} }\; \textcolor{\coeffectcolor}{\mathop{\leq_{\mathit{co} } } } \; \textcolor{\coeffectcolor}{ \textcolor{\coeffectcolor}{\gamma}_{{\mathrm{1}}} } }%
}{
  \textcolor{\coeffectcolor}{ \textcolor{\coeffectcolor}{\gamma}_{{\mathrm{2}}} }\! \cdot \! \rho   \vdash_{\mathit{coeff} }  \ottnt{V}  \Downarrow  \ottnt{W} }{%
{\ottdrulename{eval\_coeff\_val\_vsub}}{}%
}}

\newcommand{\ottdefnevalXXcoeffXXval}[1]{\begin{ottdefnblock}[#1]{$ \mu  \vdash_{\mathit{coeff} }  \ottnt{V}  \Downarrow  \ottnt{W} $}{\ottcom{environment-based coeffect semantics for CBPV (large-step)}}
\ottusedrule{\ottdruleevalXXcoeffXXvalXXvar{}}
\ottusedrule{\ottdruleevalXXcoeffXXvalXXunit{}}
\ottusedrule{\ottdruleevalXXcoeffXXvalXXthunk{}}
\ottusedrule{\ottdruleevalXXcoeffXXvalXXvpair{}}
\ottusedrule{\ottdruleevalXXcoeffXXvalXXinl{}}
\ottusedrule{\ottdruleevalXXcoeffXXvalXXinr{}}
\ottusedrule{\ottdruleevalXXcoeffXXvalXXvwith{}}
\ottusedrule{\ottdruleevalXXcoeffXXvalXXvfst{}}
\ottusedrule{\ottdruleevalXXcoeffXXvalXXvsnd{}}
\ottusedrule{\ottdruleevalXXcoeffXXvalXXvsub{}}
\end{ottdefnblock}}

\newcommand{\ottdruleevalXXcoeffXXcompXXabs}[1]{\ottdrule[#1]{%
\ottpremise{ \textcolor{\coeffectcolor}{ \textcolor{\coeffectcolor}{q}' }\;  \textcolor{\coeffectcolor}{\mathop{\leq_{\mathit{co} } } } \; \textcolor{\coeffectcolor}{ \textcolor{\coeffectcolor}{q} } }%
}{
  \textcolor{\coeffectcolor}{ \textcolor{\coeffectcolor}{\gamma} }\! \cdot \! \rho   \vdash_{\mathit{coeff} }   \lambda  \ottmv{x} ^{\textcolor{\coeffectcolor}{ \textcolor{\coeffectcolor}{q} } }. \ottnt{M}   \Downarrow   \mathbf{clo}(   \textcolor{\coeffectcolor}{ \textcolor{\coeffectcolor}{\gamma} }\! \cdot \! \rho  ,   \lambda  \ottmv{x} ^{\textcolor{\coeffectcolor}{ \textcolor{\coeffectcolor}{q}' } }. \ottnt{M}   )  }{%
{\ottdrulename{eval\_coeff\_comp\_abs}}{}%
}}

\newcommand{\ottdruleevalXXcoeffXXcompXXcpair}[1]{\ottdrule[#1]{%
}{
  \textcolor{\coeffectcolor}{ \textcolor{\coeffectcolor}{\gamma} }\! \cdot \! \rho   \vdash_{\mathit{coeff} }   \langle  \ottnt{M_{{\mathrm{1}}}} , \ottnt{M_{{\mathrm{2}}}}  \rangle   \Downarrow   \mathbf{clo}(   \textcolor{\coeffectcolor}{ \textcolor{\coeffectcolor}{\gamma} }\! \cdot \! \rho  ,   \langle  \ottnt{M_{{\mathrm{1}}}} , \ottnt{M_{{\mathrm{2}}}}  \rangle   )  }{%
{\ottdrulename{eval\_coeff\_comp\_cpair}}{}%
}}

\newcommand{\ottdruleevalXXcoeffXXcompXXcunit}[1]{\ottdrule[#1]{%
}{
  \textcolor{\coeffectcolor}{  \textcolor{\coeffectcolor}{\overline{0} }  }\! \cdot \! \rho   \vdash_{\mathit{coeff} }   \langle\rangle   \Downarrow  \ottsym{<>} }{%
{\ottdrulename{eval\_coeff\_comp\_cunit}}{}%
}}

\newcommand{\ottdruleevalXXcoeffXXcompXXappXXabs}[1]{\ottdrule[#1]{%
\ottpremise{  \textcolor{\coeffectcolor}{ \textcolor{\coeffectcolor}{\gamma}_{{\mathrm{1}}} }\! \cdot \! \rho   \vdash_{\mathit{coeff} }  \ottnt{M}  \Downarrow   \mathbf{clo}(   \textcolor{\coeffectcolor}{ \textcolor{\coeffectcolor}{\gamma}' }\! \cdot \! \rho'  ,   \lambda  \ottmv{x} ^{\textcolor{\coeffectcolor}{ \textcolor{\coeffectcolor}{q} } }. \ottnt{M'}   )  }%
\ottpremise{  \textcolor{\coeffectcolor}{ \textcolor{\coeffectcolor}{\gamma}_{{\mathrm{2}}} }\! \cdot \! \rho   \vdash_{\mathit{coeff} }  \ottnt{V}  \Downarrow  \ottnt{W} }%
\ottpremise{    \textcolor{\coeffectcolor}{ \textcolor{\coeffectcolor}{\gamma}' }\! \cdot \! \rho'     \mathop{,} \;    \ottmv{x}   \mapsto ^{\textcolor{\coeffectcolor}{ \textcolor{\coeffectcolor}{q} } }  \ottnt{W}     \vdash_{\mathit{coeff} }  \ottnt{M'}  \Downarrow  \ottnt{T} }%
\ottpremise{ \textcolor{\coeffectcolor}{ \textcolor{\coeffectcolor}{\gamma} } \equiv \textcolor{\coeffectcolor}{  \textcolor{\coeffectcolor}{ \textcolor{\coeffectcolor}{\gamma}_{{\mathrm{1}}} \ottsym{+}  \textcolor{\coeffectcolor}{ \textcolor{\coeffectcolor}{q} \cdot \textcolor{\coeffectcolor}{\gamma}_{{\mathrm{2}}} }  }  } }%
}{
  \textcolor{\coeffectcolor}{ \textcolor{\coeffectcolor}{\gamma} }\! \cdot \! \rho   \vdash_{\mathit{coeff} }  \ottnt{M} \, \ottnt{V}  \Downarrow  \ottnt{T} }{%
{\ottdrulename{eval\_coeff\_comp\_app\_abs}}{}%
}}

\newcommand{\ottdruleevalXXcoeffXXcompXXsplit}[1]{\ottdrule[#1]{%
\ottpremise{  \textcolor{\coeffectcolor}{ \textcolor{\coeffectcolor}{\gamma}_{{\mathrm{1}}} }\! \cdot \! \rho   \vdash_{\mathit{coeff} }  \ottnt{V}  \Downarrow  \ottsym{(}  \ottnt{W_{{\mathrm{1}}}}  \ottsym{,}  \ottnt{W_{{\mathrm{2}}}}  \ottsym{)} }%
\ottpremise{     \textcolor{\coeffectcolor}{ \textcolor{\coeffectcolor}{\gamma}_{{\mathrm{2}}} }\! \cdot \! \rho    \mathop{,} \;   \ottmv{x_{{\mathrm{1}}}}   \mapsto ^{\textcolor{\coeffectcolor}{ \textcolor{\coeffectcolor}{q} } }  \ottnt{W_{{\mathrm{1}}}}      \mathop{,} \;   \ottmv{x_{{\mathrm{2}}}}   \mapsto ^{\textcolor{\coeffectcolor}{ \textcolor{\coeffectcolor}{q} } }  \ottnt{W_{{\mathrm{2}}}}    \vdash_{\mathit{coeff} }  \ottnt{N}  \Downarrow  \ottnt{T} }%
\ottpremise{ \textcolor{\coeffectcolor}{ \textcolor{\coeffectcolor}{\gamma} } \equiv \textcolor{\coeffectcolor}{  \textcolor{\coeffectcolor}{  \textcolor{\coeffectcolor}{ \textcolor{\coeffectcolor}{q} \cdot \textcolor{\coeffectcolor}{\gamma}_{{\mathrm{1}}} }  \ottsym{+} \textcolor{\coeffectcolor}{\gamma}_{{\mathrm{2}}} }  } }%
}{
  \textcolor{\coeffectcolor}{ \textcolor{\coeffectcolor}{\gamma} }\! \cdot \! \rho   \vdash_{\mathit{coeff} }   \ottkw{case}_{\textcolor{\coeffectcolor}{ \textcolor{\coeffectcolor}{q} } } \;  \ottnt{V} \; \ottkw{of}\;( \ottmv{x_{{\mathrm{1}}}} , \ottmv{x_{{\mathrm{2}}}} )\; \rightarrow\;  \ottnt{N}   \Downarrow  \ottnt{T} }{%
{\ottdrulename{eval\_coeff\_comp\_split}}{}%
}}

\newcommand{\ottdruleevalXXcoeffXXcompXXreturn}[1]{\ottdrule[#1]{%
\ottpremise{  \textcolor{\coeffectcolor}{ \textcolor{\coeffectcolor}{\gamma} }\! \cdot \! \rho   \vdash_{\mathit{coeff} }  \ottnt{V}  \Downarrow  \ottnt{W} }%
}{
  \textcolor{\coeffectcolor}{  \textcolor{\coeffectcolor}{ \textcolor{\coeffectcolor}{q} \cdot \textcolor{\coeffectcolor}{\gamma} }  }\! \cdot \! \rho   \vdash_{\mathit{coeff} }   \ottkw{return} _{\textcolor{\coeffectcolor}{ \textcolor{\coeffectcolor}{q} } }\;  \ottnt{V}   \Downarrow   \ottkw{return} _{\textcolor{\coeffectcolor}{ \textcolor{\coeffectcolor}{q} } }  \ottnt{W}  }{%
{\ottdrulename{eval\_coeff\_comp\_return}}{}%
}}

\newcommand{\ottdruleevalXXcoeffXXcompXXletinXXret}[1]{\ottdrule[#1]{%
\ottpremise{  \textcolor{\coeffectcolor}{ \textcolor{\coeffectcolor}{\gamma}_{{\mathrm{1}}} }\! \cdot \! \rho   \vdash_{\mathit{coeff} }  \ottnt{M}  \Downarrow   \ottkw{return} _{\textcolor{\coeffectcolor}{ \textcolor{\coeffectcolor}{q}_{{\mathrm{1}}} } }  \ottnt{W}  }%
\ottpremise{   \textcolor{\coeffectcolor}{ \textcolor{\coeffectcolor}{\gamma}_{{\mathrm{2}}} }\! \cdot \! \rho    \mathop{,} \;   \ottmv{x}   \mapsto ^{\textcolor{\coeffectcolor}{  \textcolor{\coeffectcolor}{ \textcolor{\coeffectcolor}{q}_{{\mathrm{1}}}   \cdot   \textcolor{\coeffectcolor}{q}_{{\mathrm{2}}} }  } }  \ottnt{W}    \vdash_{\mathit{coeff} }  \ottnt{N}  \Downarrow  \ottnt{T} }%
}{
  \textcolor{\coeffectcolor}{  \textcolor{\coeffectcolor}{  \textcolor{\coeffectcolor}{ \textcolor{\coeffectcolor}{q}_{{\mathrm{2}}} \cdot \textcolor{\coeffectcolor}{\gamma}_{{\mathrm{1}}} }  \ottsym{+} \textcolor{\coeffectcolor}{\gamma}_{{\mathrm{2}}} }  }\! \cdot \! \rho   \vdash_{\mathit{coeff} }   \ottmv{x}  \leftarrow^{\textcolor{\coeffectcolor}{ \textcolor{\coeffectcolor}{q}_{{\mathrm{2}}} } }  \ottnt{M} \ \ottkw{in}\  \ottnt{N}   \Downarrow  \ottnt{T} }{%
{\ottdrulename{eval\_coeff\_comp\_letin\_ret}}{}%
}}

\newcommand{\ottdruleevalXXcoeffXXcompXXforceXXthunk}[1]{\ottdrule[#1]{%
\ottpremise{  \textcolor{\coeffectcolor}{ \textcolor{\coeffectcolor}{\gamma} }\! \cdot \! \rho   \vdash_{\mathit{coeff} }  \ottnt{V}  \Downarrow   \mathbf{clo}(  \textcolor{\coeffectcolor}{ \textcolor{\coeffectcolor}{\gamma}' }\! \cdot \! \rho'  , \{  \ottnt{M}  \} )  }%
\ottpremise{  \textcolor{\coeffectcolor}{ \textcolor{\coeffectcolor}{\gamma}' }\! \cdot \! \rho'   \vdash_{\mathit{coeff} }  \ottnt{M}  \Downarrow  \ottnt{T} }%
}{
  \textcolor{\coeffectcolor}{ \textcolor{\coeffectcolor}{\gamma} }\! \cdot \! \rho   \vdash_{\mathit{coeff} }  \ottnt{V}  \ottsym{!}  \Downarrow  \ottnt{T} }{%
{\ottdrulename{eval\_coeff\_comp\_force\_thunk}}{}%
}}

\newcommand{\ottdruleevalXXcoeffXXcompXXfst}[1]{\ottdrule[#1]{%
\ottpremise{  \textcolor{\coeffectcolor}{ \textcolor{\coeffectcolor}{\gamma} }\! \cdot \! \rho   \vdash_{\mathit{coeff} }  \ottnt{M}  \Downarrow   \mathbf{clo}(   \textcolor{\coeffectcolor}{ \textcolor{\coeffectcolor}{\gamma}' }\! \cdot \! \rho'  ,   \langle  \ottnt{M_{{\mathrm{1}}}} , \ottnt{M_{{\mathrm{2}}}}  \rangle   )  }%
\ottpremise{  \textcolor{\coeffectcolor}{ \textcolor{\coeffectcolor}{\gamma}' }\! \cdot \! \rho'   \vdash_{\mathit{coeff} }  \ottnt{M_{{\mathrm{1}}}}  \Downarrow  \ottnt{T} }%
}{
  \textcolor{\coeffectcolor}{ \textcolor{\coeffectcolor}{\gamma} }\! \cdot \! \rho   \vdash_{\mathit{coeff} }   \ottnt{M}  . 1   \Downarrow  \ottnt{T} }{%
{\ottdrulename{eval\_coeff\_comp\_fst}}{}%
}}

\newcommand{\ottdruleevalXXcoeffXXcompXXsnd}[1]{\ottdrule[#1]{%
\ottpremise{  \textcolor{\coeffectcolor}{ \textcolor{\coeffectcolor}{\gamma} }\! \cdot \! \rho   \vdash_{\mathit{coeff} }  \ottnt{M}  \Downarrow   \mathbf{clo}(   \textcolor{\coeffectcolor}{ \textcolor{\coeffectcolor}{\gamma}' }\! \cdot \! \rho'  ,   \langle  \ottnt{M_{{\mathrm{1}}}} , \ottnt{M_{{\mathrm{2}}}}  \rangle   )  }%
\ottpremise{  \textcolor{\coeffectcolor}{ \textcolor{\coeffectcolor}{\gamma}' }\! \cdot \! \rho'   \vdash_{\mathit{coeff} }  \ottnt{M_{{\mathrm{2}}}}  \Downarrow  \ottnt{T} }%
}{
  \textcolor{\coeffectcolor}{ \textcolor{\coeffectcolor}{\gamma} }\! \cdot \! \rho   \vdash_{\mathit{coeff} }   \ottnt{M}  . 2   \Downarrow  \ottnt{T} }{%
{\ottdrulename{eval\_coeff\_comp\_snd}}{}%
}}

\newcommand{\ottdruleevalXXcoeffXXcompXXsequence}[1]{\ottdrule[#1]{%
\ottpremise{  \textcolor{\coeffectcolor}{ \textcolor{\coeffectcolor}{\gamma}_{{\mathrm{1}}} }\! \cdot \! \rho   \vdash_{\mathit{coeff} }  \ottnt{V}  \Downarrow  \ottsym{()} }%
\ottpremise{  \textcolor{\coeffectcolor}{ \textcolor{\coeffectcolor}{\gamma}_{{\mathrm{2}}} }\! \cdot \! \rho   \vdash_{\mathit{coeff} }  \ottnt{N}  \Downarrow  \ottnt{T} }%
\ottpremise{ \textcolor{\coeffectcolor}{ \textcolor{\coeffectcolor}{\gamma} } \equiv \textcolor{\coeffectcolor}{  \textcolor{\coeffectcolor}{ \textcolor{\coeffectcolor}{\gamma}_{{\mathrm{1}}} \ottsym{+} \textcolor{\coeffectcolor}{\gamma}_{{\mathrm{2}}} }  } }%
}{
  \textcolor{\coeffectcolor}{ \textcolor{\coeffectcolor}{\gamma} }\! \cdot \! \rho   \vdash_{\mathit{coeff} }   \ottnt{V}  ;  \ottnt{N}   \Downarrow  \ottnt{T} }{%
{\ottdrulename{eval\_coeff\_comp\_sequence}}{}%
}}

\newcommand{\ottdruleevalXXcoeffXXcompXXcaseXXinl}[1]{\ottdrule[#1]{%
\ottpremise{  \textcolor{\coeffectcolor}{ \textcolor{\coeffectcolor}{\gamma}_{{\mathrm{1}}} }\! \cdot \! \rho   \vdash_{\mathit{coeff} }  \ottnt{V}  \Downarrow  \ottkw{inl} \, \ottnt{W} }%
\ottpremise{   \textcolor{\coeffectcolor}{ \textcolor{\coeffectcolor}{\gamma}_{{\mathrm{2}}} }\! \cdot \! \rho    \mathop{,} \;   \ottmv{x_{{\mathrm{1}}}}   \mapsto ^{\textcolor{\coeffectcolor}{ \textcolor{\coeffectcolor}{q} } }  \ottnt{W}    \vdash_{\mathit{coeff} }  \ottnt{M_{{\mathrm{1}}}}  \Downarrow  \ottnt{T} }%
\ottpremise{ \textcolor{\coeffectcolor}{ \textcolor{\coeffectcolor}{\gamma} } \equiv \textcolor{\coeffectcolor}{  \textcolor{\coeffectcolor}{  \textcolor{\coeffectcolor}{ \textcolor{\coeffectcolor}{q} \cdot \textcolor{\coeffectcolor}{\gamma}_{{\mathrm{1}}} }  \ottsym{+} \textcolor{\coeffectcolor}{\gamma}_{{\mathrm{2}}} }  } }%
\ottpremise{ \textcolor{\coeffectcolor}{ \textcolor{\coeffectcolor}{q} }\;  \textcolor{\coeffectcolor}{\mathop{\leq_{\mathit{co} } } } \; \textcolor{\coeffectcolor}{  \textcolor{\coeffectcolor}{1}  } }%
}{
  \textcolor{\coeffectcolor}{ \textcolor{\coeffectcolor}{\gamma} }\! \cdot \! \rho   \vdash_{\mathit{coeff} }   \ottkw{case}_{\textcolor{\coeffectcolor}{ \textcolor{\coeffectcolor}{q} } }\;  \ottnt{V} \; \ottkw{of}\; \ottkw{inl} \; \ottmv{x_{{\mathrm{1}}}}  \rightarrow\;  \ottnt{M_{{\mathrm{1}}}}  ;  \ottkw{inr} \; \ottmv{x_{{\mathrm{2}}}}  \rightarrow\;  \ottnt{M_{{\mathrm{2}}}}   \Downarrow  \ottnt{T} }{%
{\ottdrulename{eval\_coeff\_comp\_case\_inl}}{}%
}}

\newcommand{\ottdruleevalXXcoeffXXcompXXcaseXXinr}[1]{\ottdrule[#1]{%
\ottpremise{  \textcolor{\coeffectcolor}{ \textcolor{\coeffectcolor}{\gamma}_{{\mathrm{1}}} }\! \cdot \! \rho   \vdash_{\mathit{coeff} }  \ottnt{V}  \Downarrow  \ottkw{inr} \, \ottnt{W} }%
\ottpremise{   \textcolor{\coeffectcolor}{ \textcolor{\coeffectcolor}{\gamma}_{{\mathrm{2}}} }\! \cdot \! \rho    \mathop{,} \;   \ottmv{x_{{\mathrm{2}}}}   \mapsto ^{\textcolor{\coeffectcolor}{ \textcolor{\coeffectcolor}{q} } }  \ottnt{W}    \vdash_{\mathit{coeff} }  \ottnt{M_{{\mathrm{2}}}}  \Downarrow  \ottnt{T} }%
\ottpremise{ \textcolor{\coeffectcolor}{ \textcolor{\coeffectcolor}{\gamma} } \equiv \textcolor{\coeffectcolor}{  \textcolor{\coeffectcolor}{  \textcolor{\coeffectcolor}{ \textcolor{\coeffectcolor}{q} \cdot \textcolor{\coeffectcolor}{\gamma}_{{\mathrm{1}}} }  \ottsym{+} \textcolor{\coeffectcolor}{\gamma}_{{\mathrm{2}}} }  } }%
\ottpremise{ \textcolor{\coeffectcolor}{ \textcolor{\coeffectcolor}{q} }\;  \textcolor{\coeffectcolor}{\mathop{\leq_{\mathit{co} } } } \; \textcolor{\coeffectcolor}{  \textcolor{\coeffectcolor}{1}  } }%
}{
  \textcolor{\coeffectcolor}{ \textcolor{\coeffectcolor}{\gamma} }\! \cdot \! \rho   \vdash_{\mathit{coeff} }   \ottkw{case}_{\textcolor{\coeffectcolor}{ \textcolor{\coeffectcolor}{q} } }\;  \ottnt{V} \; \ottkw{of}\; \ottkw{inl} \; \ottmv{x_{{\mathrm{1}}}}  \rightarrow\;  \ottnt{M_{{\mathrm{1}}}}  ;  \ottkw{inr} \; \ottmv{x_{{\mathrm{2}}}}  \rightarrow\;  \ottnt{M_{{\mathrm{2}}}}   \Downarrow  \ottnt{T} }{%
{\ottdrulename{eval\_coeff\_comp\_case\_inr}}{}%
}}

\newcommand{\ottdruleevalXXcoeffXXcompXXctensor}[1]{\ottdrule[#1]{%
\ottpremise{ \textcolor{\coeffectcolor}{ \textcolor{\coeffectcolor}{\gamma} } \equiv \textcolor{\coeffectcolor}{  \textcolor{\coeffectcolor}{ \textcolor{\coeffectcolor}{\gamma}_{{\mathrm{1}}} \ottsym{+} \textcolor{\coeffectcolor}{\gamma}_{{\mathrm{2}}} }  } }%
}{
  \textcolor{\coeffectcolor}{ \textcolor{\coeffectcolor}{\gamma} }\! \cdot \! \rho   \vdash_{\mathit{coeff} }  \ottsym{(}  \ottnt{M_{{\mathrm{1}}}}  \ottsym{,}  \ottnt{M_{{\mathrm{2}}}}  \ottsym{)}  \Downarrow  \ottsym{(}   \mathbf{clo}(  \textcolor{\coeffectcolor}{ \textcolor{\coeffectcolor}{\gamma}_{{\mathrm{1}}} }\! \cdot \! \rho  , \{  \ottnt{M_{{\mathrm{1}}}}  \} )   \ottsym{,}   \mathbf{clo}(  \textcolor{\coeffectcolor}{ \textcolor{\coeffectcolor}{\gamma}_{{\mathrm{2}}} }\! \cdot \! \rho  , \{  \ottnt{M_{{\mathrm{2}}}}  \} )   \ottsym{)} }{%
{\ottdrulename{eval\_coeff\_comp\_ctensor}}{}%
}}

\newcommand{\ottdruleevalXXcoeffXXcompXXcsplit}[1]{\ottdrule[#1]{%
\ottpremise{  \textcolor{\coeffectcolor}{ \textcolor{\coeffectcolor}{\gamma}_{{\mathrm{1}}} }\! \cdot \! \rho   \vdash_{\mathit{coeff} }  \ottsym{(}  \ottnt{M_{{\mathrm{1}}}}  \ottsym{,}  \ottnt{M_{{\mathrm{2}}}}  \ottsym{)}  \Downarrow  \ottsym{(}  \ottnt{W_{{\mathrm{1}}}}  \ottsym{,}  \ottnt{W_{{\mathrm{2}}}}  \ottsym{)} }%
\ottpremise{     \textcolor{\coeffectcolor}{ \textcolor{\coeffectcolor}{\gamma}_{{\mathrm{2}}} }\! \cdot \! \rho    \mathop{,} \;   \ottmv{x_{{\mathrm{1}}}}   \mapsto ^{\textcolor{\coeffectcolor}{ \textcolor{\coeffectcolor}{q} } }  \ottnt{W_{{\mathrm{1}}}}      \mathop{,} \;   \ottmv{x_{{\mathrm{2}}}}   \mapsto ^{\textcolor{\coeffectcolor}{ \textcolor{\coeffectcolor}{q} } }  \ottnt{W_{{\mathrm{2}}}}    \vdash_{\mathit{coeff} }  \ottnt{N}  \Downarrow  \ottnt{T} }%
\ottpremise{ \textcolor{\coeffectcolor}{ \textcolor{\coeffectcolor}{\gamma} } \equiv \textcolor{\coeffectcolor}{  \textcolor{\coeffectcolor}{  \textcolor{\coeffectcolor}{ \textcolor{\coeffectcolor}{q} \cdot \textcolor{\coeffectcolor}{\gamma}_{{\mathrm{1}}} }  \ottsym{+} \textcolor{\coeffectcolor}{\gamma}_{{\mathrm{2}}} }  } }%
\ottpremise{ \textcolor{\coeffectcolor}{ \textcolor{\coeffectcolor}{q} }\;  \textcolor{\coeffectcolor}{\mathop{\leq_{\mathit{co} } } } \; \textcolor{\coeffectcolor}{  \textcolor{\coeffectcolor}{1}  } }%
}{
  \textcolor{\coeffectcolor}{ \textcolor{\coeffectcolor}{\gamma} }\! \cdot \! \rho   \vdash_{\mathit{coeff} }   \ottkw{case}_{\textcolor{\coeffectcolor}{ \textcolor{\coeffectcolor}{q} } }\;  \ottnt{M} \; \ottkw{of}\;( \ottmv{x_{{\mathrm{1}}}} , \ottmv{x_{{\mathrm{2}}}} )\; \rightarrow\;  \ottnt{N}   \Downarrow  \ottnt{T} }{%
{\ottdrulename{eval\_coeff\_comp\_csplit}}{}%
}}

\newcommand{\ottdruleevalXXcoeffXXcompXXcsub}[1]{\ottdrule[#1]{%
\ottpremise{  \textcolor{\coeffectcolor}{ \textcolor{\coeffectcolor}{\gamma}' }\! \cdot \! \rho   \vdash_{\mathit{coeff} }  \ottnt{M}  \Downarrow  \ottnt{T} }%
\ottpremise{ \textcolor{\coeffectcolor}{ \textcolor{\coeffectcolor}{\gamma} }\; \textcolor{\coeffectcolor}{\mathop{\leq_{\mathit{co} } } } \; \textcolor{\coeffectcolor}{ \textcolor{\coeffectcolor}{\gamma}' } }%
}{
  \textcolor{\coeffectcolor}{ \textcolor{\coeffectcolor}{\gamma} }\! \cdot \! \rho   \vdash_{\mathit{coeff} }  \ottnt{M}  \Downarrow  \ottnt{T} }{%
{\ottdrulename{eval\_coeff\_comp\_csub}}{}%
}}

\newcommand{\ottdefnevalXXcoeffXXcomp}[1]{\begin{ottdefnblock}[#1]{$ \mu  \vdash_{\mathit{coeff} }  \ottnt{M}  \Downarrow  \ottnt{T} $}{\ottcom{environment-based coeffect semantics for CBPV (large-step)}}
\ottusedrule{\ottdruleevalXXcoeffXXcompXXabs{}}
\ottusedrule{\ottdruleevalXXcoeffXXcompXXcpair{}}
\ottusedrule{\ottdruleevalXXcoeffXXcompXXcunit{}}
\ottusedrule{\ottdruleevalXXcoeffXXcompXXappXXabs{}}
\ottusedrule{\ottdruleevalXXcoeffXXcompXXsplit{}}
\ottusedrule{\ottdruleevalXXcoeffXXcompXXreturn{}}
\ottusedrule{\ottdruleevalXXcoeffXXcompXXletinXXret{}}
\ottusedrule{\ottdruleevalXXcoeffXXcompXXforceXXthunk{}}
\ottusedrule{\ottdruleevalXXcoeffXXcompXXfst{}}
\ottusedrule{\ottdruleevalXXcoeffXXcompXXsnd{}}
\ottusedrule{\ottdruleevalXXcoeffXXcompXXsequence{}}
\ottusedrule{\ottdruleevalXXcoeffXXcompXXcaseXXinl{}}
\ottusedrule{\ottdruleevalXXcoeffXXcompXXcaseXXinr{}}
\ottusedrule{\ottdruleevalXXcoeffXXcompXXctensor{}}
\ottusedrule{\ottdruleevalXXcoeffXXcompXXcsplit{}}
\ottusedrule{\ottdruleevalXXcoeffXXcompXXcsub{}}
\end{ottdefnblock}}

\newcommand{\ottdruleevalXXlinXXvalXXvar}[1]{\ottdrule[#1]{%
}{
     \textcolor{\coeffectcolor}{  \textcolor{\coeffectcolor}{\overline{0} }  }\! \cdot \! \rho_{{\mathrm{1}}}    \mathop{,} \;   \ottmv{x}   \mapsto ^{\textcolor{\coeffectcolor}{  \textcolor{\coeffectcolor}{1}  } }  \ottnt{W}      \mathop{,} \;   \textcolor{\coeffectcolor}{  \textcolor{\coeffectcolor}{\overline{0} }  }\! \cdot \! \rho_{{\mathrm{2}}}    \vdash_{\mathit{lin} }  \ottmv{x}  \Downarrow  \ottnt{W} }{%
{\ottdrulename{eval\_lin\_val\_var}}{}%
}}

\newcommand{\ottdruleevalXXlinXXvalXXunit}[1]{\ottdrule[#1]{%
}{
  \textcolor{\coeffectcolor}{  \textcolor{\coeffectcolor}{\overline{0} }  }\! \cdot \! \rho   \vdash_{\mathit{lin} }  \ottsym{()}  \Downarrow  \ottsym{()} }{%
{\ottdrulename{eval\_lin\_val\_unit}}{}%
}}

\newcommand{\ottdruleevalXXlinXXvalXXthunk}[1]{\ottdrule[#1]{%
}{
  \textcolor{\coeffectcolor}{ \textcolor{\coeffectcolor}{\gamma} }\! \cdot \! \rho   \vdash_{\mathit{lin} }  \ottsym{\{}  \ottnt{M}  \ottsym{\}}  \Downarrow   \mathbf{clo}(  \textcolor{\coeffectcolor}{ \textcolor{\coeffectcolor}{\gamma} }\! \cdot \! \rho  , \{  \ottnt{M}  \} )  }{%
{\ottdrulename{eval\_lin\_val\_thunk}}{}%
}}

\newcommand{\ottdruleevalXXlinXXvalXXvpair}[1]{\ottdrule[#1]{%
\ottpremise{  \textcolor{\coeffectcolor}{ \textcolor{\coeffectcolor}{\gamma}_{{\mathrm{1}}} }\! \cdot \! \rho   \vdash_{\mathit{lin} }  \ottnt{V_{{\mathrm{1}}}}  \Downarrow  \ottnt{W_{{\mathrm{1}}}} }%
\ottpremise{  \textcolor{\coeffectcolor}{ \textcolor{\coeffectcolor}{\gamma}_{{\mathrm{2}}} }\! \cdot \! \rho   \vdash_{\mathit{lin} }  \ottnt{V_{{\mathrm{2}}}}  \Downarrow  \ottnt{W_{{\mathrm{2}}}} }%
\ottpremise{ \textcolor{\coeffectcolor}{ \textcolor{\coeffectcolor}{\gamma} } \equiv \textcolor{\coeffectcolor}{  \textcolor{\coeffectcolor}{ \textcolor{\coeffectcolor}{\gamma}_{{\mathrm{1}}} \ottsym{+} \textcolor{\coeffectcolor}{\gamma}_{{\mathrm{2}}} }  } }%
}{
  \textcolor{\coeffectcolor}{ \textcolor{\coeffectcolor}{\gamma} }\! \cdot \! \rho   \vdash_{\mathit{lin} }  \ottsym{(}  \ottnt{V_{{\mathrm{1}}}}  \ottsym{,}  \ottnt{V_{{\mathrm{2}}}}  \ottsym{)}  \Downarrow  \ottsym{(}  \ottnt{W_{{\mathrm{1}}}}  \ottsym{,}  \ottnt{W_{{\mathrm{2}}}}  \ottsym{)} }{%
{\ottdrulename{eval\_lin\_val\_vpair}}{}%
}}

\newcommand{\ottdruleevalXXlinXXvalXXinl}[1]{\ottdrule[#1]{%
\ottpremise{  \textcolor{\coeffectcolor}{ \textcolor{\coeffectcolor}{\gamma} }\! \cdot \! \rho   \vdash_{\mathit{lin} }  \ottnt{V}  \Downarrow  \ottnt{W} }%
}{
  \textcolor{\coeffectcolor}{ \textcolor{\coeffectcolor}{\gamma} }\! \cdot \! \rho   \vdash_{\mathit{lin} }  \ottkw{inl} \, \ottnt{V}  \Downarrow  \ottkw{inl} \, \ottnt{W} }{%
{\ottdrulename{eval\_lin\_val\_inl}}{}%
}}

\newcommand{\ottdruleevalXXlinXXvalXXinr}[1]{\ottdrule[#1]{%
\ottpremise{  \textcolor{\coeffectcolor}{ \textcolor{\coeffectcolor}{\gamma} }\! \cdot \! \rho   \vdash_{\mathit{lin} }  \ottnt{V}  \Downarrow  \ottnt{W} }%
}{
  \textcolor{\coeffectcolor}{ \textcolor{\coeffectcolor}{\gamma} }\! \cdot \! \rho   \vdash_{\mathit{lin} }  \ottkw{inr} \, \ottnt{V}  \Downarrow  \ottkw{inr} \, \ottnt{W} }{%
{\ottdrulename{eval\_lin\_val\_inr}}{}%
}}

\newcommand{\ottdruleevalXXlinXXvalXXvwith}[1]{\ottdrule[#1]{%
}{
  \textcolor{\coeffectcolor}{ \textcolor{\coeffectcolor}{\gamma} }\! \cdot \! \rho   \vdash_{\mathit{lin} }   \langle  \ottnt{V_{{\mathrm{1}}}}  ,  \ottnt{V_{{\mathrm{2}}}}  \rangle   \Downarrow   \mathbf{clo}(  \textcolor{\coeffectcolor}{ \textcolor{\coeffectcolor}{\gamma} }\! \cdot \! \rho  , \langle  \ottnt{V_{{\mathrm{1}}}}  ,  \ottnt{V_{{\mathrm{2}}}}  \rangle )  }{%
{\ottdrulename{eval\_lin\_val\_vwith}}{}%
}}

\newcommand{\ottdruleevalXXlinXXvalXXvfst}[1]{\ottdrule[#1]{%
\ottpremise{  \textcolor{\coeffectcolor}{ \textcolor{\coeffectcolor}{\gamma} }\! \cdot \! \rho   \vdash_{\mathit{lin} }  \ottnt{V}  \Downarrow   \mathbf{clo}(  \textcolor{\coeffectcolor}{ \textcolor{\coeffectcolor}{\gamma}' }\! \cdot \! \rho'  , \langle  \ottnt{V_{{\mathrm{1}}}}  ,  \ottnt{V_{{\mathrm{2}}}}  \rangle )  }%
\ottpremise{  \textcolor{\coeffectcolor}{ \textcolor{\coeffectcolor}{\gamma}' }\! \cdot \! \rho'   \vdash_{\mathit{lin} }  \ottnt{V_{{\mathrm{1}}}}  \Downarrow  \ottnt{W} }%
}{
  \textcolor{\coeffectcolor}{ \textcolor{\coeffectcolor}{\gamma} }\! \cdot \! \rho   \vdash_{\mathit{lin} }  \ottnt{V}  \ottsym{.}  \ottsym{1}  \Downarrow  \ottnt{W} }{%
{\ottdrulename{eval\_lin\_val\_vfst}}{}%
}}

\newcommand{\ottdruleevalXXlinXXvalXXvsnd}[1]{\ottdrule[#1]{%
\ottpremise{  \textcolor{\coeffectcolor}{ \textcolor{\coeffectcolor}{\gamma} }\! \cdot \! \rho   \vdash_{\mathit{lin} }  \ottnt{V}  \Downarrow   \mathbf{clo}(  \textcolor{\coeffectcolor}{ \textcolor{\coeffectcolor}{\gamma}' }\! \cdot \! \rho'  , \langle  \ottnt{V_{{\mathrm{1}}}}  ,  \ottnt{V_{{\mathrm{2}}}}  \rangle )  }%
\ottpremise{  \textcolor{\coeffectcolor}{ \textcolor{\coeffectcolor}{\gamma}' }\! \cdot \! \rho'   \vdash_{\mathit{lin} }  \ottnt{V_{{\mathrm{2}}}}  \Downarrow  \ottnt{W} }%
}{
  \textcolor{\coeffectcolor}{ \textcolor{\coeffectcolor}{\gamma} }\! \cdot \! \rho   \vdash_{\mathit{lin} }  \ottnt{V}  \ottsym{.}  \ottsym{2}  \Downarrow  \ottnt{W} }{%
{\ottdrulename{eval\_lin\_val\_vsnd}}{}%
}}

\newcommand{\ottdruleevalXXlinXXvalXXsub}[1]{\ottdrule[#1]{%
\ottpremise{  \textcolor{\coeffectcolor}{ \textcolor{\coeffectcolor}{\gamma}' }\! \cdot \! \rho   \vdash_{\mathit{lin} }  \ottnt{V}  \Downarrow  \ottnt{W} }%
\ottpremise{ \textcolor{\coeffectcolor}{ \textcolor{\coeffectcolor}{\gamma} }\; \textcolor{\coeffectcolor}{\mathop{\leq_{\mathit{co} } } } \; \textcolor{\coeffectcolor}{ \textcolor{\coeffectcolor}{\gamma}' } }%
}{
  \textcolor{\coeffectcolor}{ \textcolor{\coeffectcolor}{\gamma} }\! \cdot \! \rho   \vdash_{\mathit{lin} }  \ottnt{V}  \Downarrow  \ottnt{W} }{%
{\ottdrulename{eval\_lin\_val\_sub}}{}%
}}

\newcommand{\ottdefnevalXXlinXXval}[1]{\begin{ottdefnblock}[#1]{$ \mu  \vdash_{\mathit{lin} }  \ottnt{V}  \Downarrow  \ottnt{W} $}{\ottcom{environment-based resource counting semantics for CBPV (large-step)}}
\ottusedrule{\ottdruleevalXXlinXXvalXXvar{}}
\ottusedrule{\ottdruleevalXXlinXXvalXXunit{}}
\ottusedrule{\ottdruleevalXXlinXXvalXXthunk{}}
\ottusedrule{\ottdruleevalXXlinXXvalXXvpair{}}
\ottusedrule{\ottdruleevalXXlinXXvalXXinl{}}
\ottusedrule{\ottdruleevalXXlinXXvalXXinr{}}
\ottusedrule{\ottdruleevalXXlinXXvalXXvwith{}}
\ottusedrule{\ottdruleevalXXlinXXvalXXvfst{}}
\ottusedrule{\ottdruleevalXXlinXXvalXXvsnd{}}
\ottusedrule{\ottdruleevalXXlinXXvalXXsub{}}
\end{ottdefnblock}}

\newcommand{\ottdruleevalXXlinXXcompXXabs}[1]{\ottdrule[#1]{%
\ottpremise{ \textcolor{\coeffectcolor}{ \textcolor{\coeffectcolor}{q}' }\;  \textcolor{\coeffectcolor}{\mathop{\leq_{\mathit{co} } } } \; \textcolor{\coeffectcolor}{ \textcolor{\coeffectcolor}{q} } }%
}{
  \textcolor{\coeffectcolor}{ \textcolor{\coeffectcolor}{\gamma} }\! \cdot \! \rho   \vdash_{\mathit{lin} }   \lambda  \ottmv{x} ^{\textcolor{\coeffectcolor}{ \textcolor{\coeffectcolor}{q} } }. \ottnt{M}   \Downarrow   \mathbf{clo}(   \textcolor{\coeffectcolor}{ \textcolor{\coeffectcolor}{\gamma} }\! \cdot \! \rho  ,   \lambda  \ottmv{x} ^{\textcolor{\coeffectcolor}{ \textcolor{\coeffectcolor}{q}' } }. \ottnt{M}   )  }{%
{\ottdrulename{eval\_lin\_comp\_abs}}{}%
}}

\newcommand{\ottdruleevalXXlinXXcompXXcpair}[1]{\ottdrule[#1]{%
}{
  \textcolor{\coeffectcolor}{ \textcolor{\coeffectcolor}{\gamma} }\! \cdot \! \rho   \vdash_{\mathit{lin} }   \langle  \ottnt{M_{{\mathrm{1}}}} , \ottnt{M_{{\mathrm{2}}}}  \rangle   \Downarrow   \mathbf{clo}(   \textcolor{\coeffectcolor}{ \textcolor{\coeffectcolor}{\gamma} }\! \cdot \! \rho  ,   \langle  \ottnt{M_{{\mathrm{1}}}} , \ottnt{M_{{\mathrm{2}}}}  \rangle   )  }{%
{\ottdrulename{eval\_lin\_comp\_cpair}}{}%
}}

\newcommand{\ottdruleevalXXlinXXcompXXcunit}[1]{\ottdrule[#1]{%
}{
  \textcolor{\coeffectcolor}{  \textcolor{\coeffectcolor}{\overline{0} }  }\! \cdot \! \rho   \vdash_{\mathit{lin} }   \langle\rangle   \Downarrow  \ottsym{<>} }{%
{\ottdrulename{eval\_lin\_comp\_cunit}}{}%
}}

\newcommand{\ottdruleevalXXlinXXcompXXappXXabs}[1]{\ottdrule[#1]{%
\ottpremise{  \textcolor{\coeffectcolor}{ \textcolor{\coeffectcolor}{\gamma}_{{\mathrm{1}}} }\! \cdot \! \rho   \vdash_{\mathit{lin} }  \ottnt{M}  \Downarrow   \mathbf{clo}(   \textcolor{\coeffectcolor}{ \textcolor{\coeffectcolor}{\gamma}' }\! \cdot \! \rho'  ,   \lambda  \ottmv{x} ^{\textcolor{\coeffectcolor}{ \textcolor{\coeffectcolor}{q} } }. \ottnt{M'}   )  }%
\ottpremise{  \textcolor{\coeffectcolor}{ \textcolor{\coeffectcolor}{\gamma}_{{\mathrm{2}}} }\! \cdot \! \rho   \vdash_{\mathit{lin} }  \ottnt{V}  \Downarrow  \ottnt{W} }%
\ottpremise{  \ottsym{(}   \textcolor{\coeffectcolor}{ \textcolor{\coeffectcolor}{\gamma}' }\! \cdot \! \rho'   \ottsym{)}   \mathop{,} \;  \ottsym{(}   \ottmv{x}   \mapsto ^{\textcolor{\coeffectcolor}{ \textcolor{\coeffectcolor}{q} } }  \ottnt{W}   \ottsym{)}   \vdash_{\mathit{lin} }  \ottnt{M'}  \Downarrow  \ottnt{T} }%
\ottpremise{ \textcolor{\coeffectcolor}{ \textcolor{\coeffectcolor}{\gamma} } \equiv \textcolor{\coeffectcolor}{  \textcolor{\coeffectcolor}{ \textcolor{\coeffectcolor}{\gamma}_{{\mathrm{1}}} \ottsym{+}  \textcolor{\coeffectcolor}{ \textcolor{\coeffectcolor}{q} \cdot \textcolor{\coeffectcolor}{\gamma}_{{\mathrm{2}}} }  }  } }%
\ottpremise{ \textcolor{\coeffectcolor}{ \textcolor{\coeffectcolor}{q} }  \neq  \textcolor{\coeffectcolor}{  \textcolor{\coeffectcolor}{0}  } }%
}{
  \textcolor{\coeffectcolor}{ \textcolor{\coeffectcolor}{\gamma} }\! \cdot \! \rho   \vdash_{\mathit{lin} }  \ottnt{M} \, \ottnt{V}  \Downarrow  \ottnt{T} }{%
{\ottdrulename{eval\_lin\_comp\_app\_abs}}{}%
}}

\newcommand{\ottdruleevalXXlinXXcompXXforceXXthunk}[1]{\ottdrule[#1]{%
\ottpremise{  \textcolor{\coeffectcolor}{ \textcolor{\coeffectcolor}{\gamma} }\! \cdot \! \rho   \vdash_{\mathit{lin} }  \ottnt{V}  \Downarrow   \mathbf{clo}(  \textcolor{\coeffectcolor}{ \textcolor{\coeffectcolor}{\gamma}' }\! \cdot \! \rho'  , \{  \ottnt{M}  \} )  }%
\ottpremise{  \textcolor{\coeffectcolor}{ \textcolor{\coeffectcolor}{\gamma}' }\! \cdot \! \rho'   \vdash_{\mathit{lin} }  \ottnt{M}  \Downarrow  \ottnt{T} }%
}{
  \textcolor{\coeffectcolor}{ \textcolor{\coeffectcolor}{\gamma} }\! \cdot \! \rho   \vdash_{\mathit{lin} }  \ottnt{V}  \ottsym{!}  \Downarrow  \ottnt{T} }{%
{\ottdrulename{eval\_lin\_comp\_force\_thunk}}{}%
}}

\newcommand{\ottdruleevalXXlinXXcompXXreturn}[1]{\ottdrule[#1]{%
\ottpremise{  \textcolor{\coeffectcolor}{ \textcolor{\coeffectcolor}{\gamma}' }\! \cdot \! \rho   \vdash_{\mathit{lin} }  \ottnt{V}  \Downarrow  \ottnt{W} }%
\ottpremise{ \textcolor{\coeffectcolor}{ \textcolor{\coeffectcolor}{\gamma} } \equiv \textcolor{\coeffectcolor}{  \textcolor{\coeffectcolor}{ \textcolor{\coeffectcolor}{q} \cdot \textcolor{\coeffectcolor}{\gamma}' }  } }%
\ottpremise{ \textcolor{\coeffectcolor}{ \textcolor{\coeffectcolor}{q} }  \neq  \textcolor{\coeffectcolor}{  \textcolor{\coeffectcolor}{0}  } }%
}{
  \textcolor{\coeffectcolor}{ \textcolor{\coeffectcolor}{\gamma} }\! \cdot \! \rho   \vdash_{\mathit{lin} }   \ottkw{return} _{\textcolor{\coeffectcolor}{ \textcolor{\coeffectcolor}{q} } }\;  \ottnt{V}   \Downarrow   \ottkw{return} _{\textcolor{\coeffectcolor}{ \textcolor{\coeffectcolor}{q} } }  \ottnt{W}  }{%
{\ottdrulename{eval\_lin\_comp\_return}}{}%
}}

\newcommand{\ottdruleevalXXlinXXcompXXletinXXret}[1]{\ottdrule[#1]{%
\ottpremise{ \textcolor{\coeffectcolor}{ \textcolor{\coeffectcolor}{q}' }\; = \; \textcolor{\coeffectcolor}{  \textcolor{\coeffectcolor}{q}_{{\mathrm{2}}} \ \|\ \textcolor{\coeffectcolor}{1}  } }%
\ottpremise{  \textcolor{\coeffectcolor}{ \textcolor{\coeffectcolor}{\gamma}_{{\mathrm{1}}} }\! \cdot \! \rho   \vdash_{\mathit{lin} }  \ottnt{M}  \Downarrow   \ottkw{return} _{\textcolor{\coeffectcolor}{ \textcolor{\coeffectcolor}{q}_{{\mathrm{1}}} } }  \ottnt{W}  }%
\ottpremise{   \textcolor{\coeffectcolor}{ \textcolor{\coeffectcolor}{\gamma}_{{\mathrm{2}}} }\! \cdot \! \rho    \mathop{,} \;   \ottmv{x}   \mapsto ^{\textcolor{\coeffectcolor}{  \textcolor{\coeffectcolor}{ \textcolor{\coeffectcolor}{q}_{{\mathrm{1}}}   \cdot   \textcolor{\coeffectcolor}{q}' }  } }  \ottnt{W}    \vdash_{\mathit{lin} }  \ottnt{N}  \Downarrow  \ottnt{T} }%
}{
  \textcolor{\coeffectcolor}{  \textcolor{\coeffectcolor}{  \textcolor{\coeffectcolor}{ \textcolor{\coeffectcolor}{q}' \cdot \textcolor{\coeffectcolor}{\gamma}_{{\mathrm{1}}} }  \ottsym{+} \textcolor{\coeffectcolor}{\gamma}_{{\mathrm{2}}} }  }\! \cdot \! \rho   \vdash_{\mathit{lin} }   \ottmv{x}  \leftarrow^{\textcolor{\coeffectcolor}{ \textcolor{\coeffectcolor}{q}_{{\mathrm{2}}} } }  \ottnt{M} \ \ottkw{in}\  \ottnt{N}   \Downarrow  \ottnt{T} }{%
{\ottdrulename{eval\_lin\_comp\_letin\_ret}}{}%
}}

\newcommand{\ottdruleevalXXlinXXcompXXsplit}[1]{\ottdrule[#1]{%
\ottpremise{  \textcolor{\coeffectcolor}{ \textcolor{\coeffectcolor}{\gamma}_{{\mathrm{1}}} }\! \cdot \! \rho   \vdash_{\mathit{lin} }  \ottnt{V}  \Downarrow  \ottsym{(}  \ottnt{W_{{\mathrm{1}}}}  \ottsym{,}  \ottnt{W_{{\mathrm{2}}}}  \ottsym{)} }%
\ottpremise{     \textcolor{\coeffectcolor}{ \textcolor{\coeffectcolor}{\gamma}_{{\mathrm{2}}} }\! \cdot \! \rho    \mathop{,} \;   \ottmv{x_{{\mathrm{1}}}}   \mapsto ^{\textcolor{\coeffectcolor}{ \textcolor{\coeffectcolor}{q} } }  \ottnt{W_{{\mathrm{1}}}}      \mathop{,} \;   \ottmv{x_{{\mathrm{2}}}}   \mapsto ^{\textcolor{\coeffectcolor}{ \textcolor{\coeffectcolor}{q} } }  \ottnt{W_{{\mathrm{2}}}}    \vdash_{\mathit{lin} }  \ottnt{N}  \Downarrow  \ottnt{T} }%
\ottpremise{ \textcolor{\coeffectcolor}{ \textcolor{\coeffectcolor}{\gamma} } \equiv \textcolor{\coeffectcolor}{  \textcolor{\coeffectcolor}{  \textcolor{\coeffectcolor}{ \textcolor{\coeffectcolor}{q} \cdot \textcolor{\coeffectcolor}{\gamma}_{{\mathrm{1}}} }  \ottsym{+} \textcolor{\coeffectcolor}{\gamma}_{{\mathrm{2}}} }  } }%
\ottpremise{ \textcolor{\coeffectcolor}{ \textcolor{\coeffectcolor}{q} }  \neq  \textcolor{\coeffectcolor}{  \textcolor{\coeffectcolor}{0}  } }%
}{
  \textcolor{\coeffectcolor}{ \textcolor{\coeffectcolor}{\gamma} }\! \cdot \! \rho   \vdash_{\mathit{lin} }   \ottkw{case}_{\textcolor{\coeffectcolor}{ \textcolor{\coeffectcolor}{q} } } \;  \ottnt{V} \; \ottkw{of}\;( \ottmv{x_{{\mathrm{1}}}} , \ottmv{x_{{\mathrm{2}}}} )\; \rightarrow\;  \ottnt{N}   \Downarrow  \ottnt{T} }{%
{\ottdrulename{eval\_lin\_comp\_split}}{}%
}}

\newcommand{\ottdruleevalXXlinXXcompXXfst}[1]{\ottdrule[#1]{%
\ottpremise{  \textcolor{\coeffectcolor}{ \textcolor{\coeffectcolor}{\gamma} }\! \cdot \! \rho   \vdash_{\mathit{lin} }  \ottnt{M}  \Downarrow   \mathbf{clo}(   \textcolor{\coeffectcolor}{ \textcolor{\coeffectcolor}{\gamma}' }\! \cdot \! \rho'  ,   \langle  \ottnt{M_{{\mathrm{1}}}} , \ottnt{M_{{\mathrm{2}}}}  \rangle   )  }%
\ottpremise{  \textcolor{\coeffectcolor}{ \textcolor{\coeffectcolor}{\gamma}' }\! \cdot \! \rho'   \vdash_{\mathit{lin} }  \ottnt{M_{{\mathrm{1}}}}  \Downarrow  \ottnt{T} }%
}{
  \textcolor{\coeffectcolor}{ \textcolor{\coeffectcolor}{\gamma} }\! \cdot \! \rho   \vdash_{\mathit{lin} }   \ottnt{M}  . 1   \Downarrow  \ottnt{T} }{%
{\ottdrulename{eval\_lin\_comp\_fst}}{}%
}}

\newcommand{\ottdruleevalXXlinXXcompXXsnd}[1]{\ottdrule[#1]{%
\ottpremise{  \textcolor{\coeffectcolor}{ \textcolor{\coeffectcolor}{\gamma} }\! \cdot \! \rho   \vdash_{\mathit{lin} }  \ottnt{M}  \Downarrow   \mathbf{clo}(   \textcolor{\coeffectcolor}{ \textcolor{\coeffectcolor}{\gamma}' }\! \cdot \! \rho'  ,   \langle  \ottnt{M_{{\mathrm{1}}}} , \ottnt{M_{{\mathrm{2}}}}  \rangle   )  }%
\ottpremise{  \textcolor{\coeffectcolor}{ \textcolor{\coeffectcolor}{\gamma}' }\! \cdot \! \rho'   \vdash_{\mathit{lin} }  \ottnt{M_{{\mathrm{2}}}}  \Downarrow  \ottnt{T} }%
}{
  \textcolor{\coeffectcolor}{ \textcolor{\coeffectcolor}{\gamma} }\! \cdot \! \rho   \vdash_{\mathit{lin} }   \ottnt{M}  . 2   \Downarrow  \ottnt{T} }{%
{\ottdrulename{eval\_lin\_comp\_snd}}{}%
}}

\newcommand{\ottdruleevalXXlinXXcompXXsequence}[1]{\ottdrule[#1]{%
\ottpremise{  \textcolor{\coeffectcolor}{ \textcolor{\coeffectcolor}{\gamma}_{{\mathrm{1}}} }\! \cdot \! \rho   \vdash_{\mathit{lin} }  \ottnt{V}  \Downarrow  \ottsym{()} }%
\ottpremise{  \textcolor{\coeffectcolor}{ \textcolor{\coeffectcolor}{\gamma}_{{\mathrm{2}}} }\! \cdot \! \rho   \vdash_{\mathit{lin} }  \ottnt{N}  \Downarrow  \ottnt{T} }%
\ottpremise{ \textcolor{\coeffectcolor}{ \textcolor{\coeffectcolor}{\gamma} } \equiv \textcolor{\coeffectcolor}{  \textcolor{\coeffectcolor}{ \textcolor{\coeffectcolor}{\gamma}_{{\mathrm{1}}} \ottsym{+} \textcolor{\coeffectcolor}{\gamma}_{{\mathrm{2}}} }  } }%
}{
  \textcolor{\coeffectcolor}{ \textcolor{\coeffectcolor}{\gamma} }\! \cdot \! \rho   \vdash_{\mathit{lin} }   \ottnt{V}  ;  \ottnt{N}   \Downarrow  \ottnt{T} }{%
{\ottdrulename{eval\_lin\_comp\_sequence}}{}%
}}

\newcommand{\ottdruleevalXXlinXXcompXXcaseXXinl}[1]{\ottdrule[#1]{%
\ottpremise{  \textcolor{\coeffectcolor}{ \textcolor{\coeffectcolor}{\gamma}_{{\mathrm{1}}} }\! \cdot \! \rho   \vdash_{\mathit{lin} }  \ottnt{V}  \Downarrow  \ottkw{inl} \, \ottnt{W} }%
\ottpremise{   \textcolor{\coeffectcolor}{ \textcolor{\coeffectcolor}{\gamma}_{{\mathrm{2}}} }\! \cdot \! \rho    \mathop{,} \;   \ottmv{x_{{\mathrm{1}}}}   \mapsto ^{\textcolor{\coeffectcolor}{ \textcolor{\coeffectcolor}{q} } }  \ottnt{W}    \vdash_{\mathit{lin} }  \ottnt{M_{{\mathrm{1}}}}  \Downarrow  \ottnt{T} }%
\ottpremise{ \textcolor{\coeffectcolor}{ \textcolor{\coeffectcolor}{\gamma} } \equiv \textcolor{\coeffectcolor}{  \textcolor{\coeffectcolor}{  \textcolor{\coeffectcolor}{ \textcolor{\coeffectcolor}{q} \cdot \textcolor{\coeffectcolor}{\gamma}_{{\mathrm{1}}} }  \ottsym{+} \textcolor{\coeffectcolor}{\gamma}_{{\mathrm{2}}} }  } }%
\ottpremise{ \textcolor{\coeffectcolor}{ \textcolor{\coeffectcolor}{q} }\;  \textcolor{\coeffectcolor}{\mathop{\leq_{\mathit{co} } } } \; \textcolor{\coeffectcolor}{  \textcolor{\coeffectcolor}{1}  } }%
}{
  \textcolor{\coeffectcolor}{ \textcolor{\coeffectcolor}{\gamma} }\! \cdot \! \rho   \vdash_{\mathit{lin} }   \ottkw{case}_{\textcolor{\coeffectcolor}{ \textcolor{\coeffectcolor}{q} } }\;  \ottnt{V} \; \ottkw{of}\; \ottkw{inl} \; \ottmv{x_{{\mathrm{1}}}}  \rightarrow\;  \ottnt{M_{{\mathrm{1}}}}  ;  \ottkw{inr} \; \ottmv{x_{{\mathrm{2}}}}  \rightarrow\;  \ottnt{M_{{\mathrm{2}}}}   \Downarrow  \ottnt{T} }{%
{\ottdrulename{eval\_lin\_comp\_case\_inl}}{}%
}}

\newcommand{\ottdruleevalXXlinXXcompXXcaseXXinr}[1]{\ottdrule[#1]{%
\ottpremise{  \textcolor{\coeffectcolor}{ \textcolor{\coeffectcolor}{\gamma}_{{\mathrm{1}}} }\! \cdot \! \rho   \vdash_{\mathit{lin} }  \ottnt{V}  \Downarrow  \ottkw{inr} \, \ottnt{W} }%
\ottpremise{   \textcolor{\coeffectcolor}{ \textcolor{\coeffectcolor}{\gamma}_{{\mathrm{2}}} }\! \cdot \! \rho    \mathop{,} \;   \ottmv{x_{{\mathrm{2}}}}   \mapsto ^{\textcolor{\coeffectcolor}{ \textcolor{\coeffectcolor}{q} } }  \ottnt{W}    \vdash_{\mathit{lin} }  \ottnt{M_{{\mathrm{2}}}}  \Downarrow  \ottnt{T} }%
\ottpremise{ \textcolor{\coeffectcolor}{ \textcolor{\coeffectcolor}{\gamma} } \equiv \textcolor{\coeffectcolor}{  \textcolor{\coeffectcolor}{  \textcolor{\coeffectcolor}{ \textcolor{\coeffectcolor}{q} \cdot \textcolor{\coeffectcolor}{\gamma}_{{\mathrm{1}}} }  \ottsym{+} \textcolor{\coeffectcolor}{\gamma}_{{\mathrm{2}}} }  } }%
\ottpremise{ \textcolor{\coeffectcolor}{ \textcolor{\coeffectcolor}{q} }\;  \textcolor{\coeffectcolor}{\mathop{\leq_{\mathit{co} } } } \; \textcolor{\coeffectcolor}{  \textcolor{\coeffectcolor}{1}  } }%
}{
  \textcolor{\coeffectcolor}{ \textcolor{\coeffectcolor}{\gamma} }\! \cdot \! \rho   \vdash_{\mathit{lin} }   \ottkw{case}_{\textcolor{\coeffectcolor}{ \textcolor{\coeffectcolor}{q} } }\;  \ottnt{V} \; \ottkw{of}\; \ottkw{inl} \; \ottmv{x_{{\mathrm{1}}}}  \rightarrow\;  \ottnt{M_{{\mathrm{1}}}}  ;  \ottkw{inr} \; \ottmv{x_{{\mathrm{2}}}}  \rightarrow\;  \ottnt{M_{{\mathrm{2}}}}   \Downarrow  \ottnt{T} }{%
{\ottdrulename{eval\_lin\_comp\_case\_inr}}{}%
}}

\newcommand{\ottdruleevalXXlinXXcompXXctensor}[1]{\ottdrule[#1]{%
\ottpremise{ \textcolor{\coeffectcolor}{ \textcolor{\coeffectcolor}{\gamma} } \equiv \textcolor{\coeffectcolor}{  \textcolor{\coeffectcolor}{ \textcolor{\coeffectcolor}{\gamma}_{{\mathrm{1}}} \ottsym{+} \textcolor{\coeffectcolor}{\gamma}_{{\mathrm{2}}} }  } }%
}{
  \textcolor{\coeffectcolor}{ \textcolor{\coeffectcolor}{\gamma} }\! \cdot \! \rho   \vdash_{\mathit{lin} }  \ottsym{(}  \ottnt{M_{{\mathrm{1}}}}  \ottsym{,}  \ottnt{M_{{\mathrm{2}}}}  \ottsym{)}  \Downarrow  \ottsym{(}   \mathbf{clo}(  \textcolor{\coeffectcolor}{ \textcolor{\coeffectcolor}{\gamma}_{{\mathrm{1}}} }\! \cdot \! \rho  ,  \ottnt{M_{{\mathrm{1}}}} )   \ottsym{,}   \mathbf{clo}(  \textcolor{\coeffectcolor}{ \textcolor{\coeffectcolor}{\gamma}_{{\mathrm{2}}} }\! \cdot \! \rho  ,  \ottnt{M_{{\mathrm{2}}}} )   \ottsym{)} }{%
{\ottdrulename{eval\_lin\_comp\_ctensor}}{}%
}}

\newcommand{\ottdruleevalXXlinXXcompXXcsplit}[1]{\ottdrule[#1]{%
\ottpremise{  \textcolor{\coeffectcolor}{ \textcolor{\coeffectcolor}{\gamma}_{{\mathrm{1}}} }\! \cdot \! \rho   \vdash_{\mathit{lin} }  \ottsym{(}  \ottnt{M_{{\mathrm{1}}}}  \ottsym{,}  \ottnt{M_{{\mathrm{2}}}}  \ottsym{)}  \Downarrow  \ottsym{(}  \ottnt{W_{{\mathrm{1}}}}  \ottsym{,}  \ottnt{W_{{\mathrm{2}}}}  \ottsym{)} }%
\ottpremise{     \textcolor{\coeffectcolor}{ \textcolor{\coeffectcolor}{\gamma}_{{\mathrm{2}}} }\! \cdot \! \rho    \mathop{,} \;   \ottmv{x_{{\mathrm{1}}}}   \mapsto ^{\textcolor{\coeffectcolor}{ \textcolor{\coeffectcolor}{q} } }  \ottnt{W_{{\mathrm{1}}}}      \mathop{,} \;   \ottmv{x_{{\mathrm{2}}}}   \mapsto ^{\textcolor{\coeffectcolor}{ \textcolor{\coeffectcolor}{q} } }  \ottnt{W_{{\mathrm{2}}}}    \vdash_{\mathit{lin} }  \ottnt{N}  \Downarrow  \ottnt{T} }%
\ottpremise{ \textcolor{\coeffectcolor}{ \textcolor{\coeffectcolor}{\gamma} }\; \textcolor{\coeffectcolor}{\mathop{\leq_{\mathit{co} } } } \; \textcolor{\coeffectcolor}{  \textcolor{\coeffectcolor}{  \textcolor{\coeffectcolor}{ \textcolor{\coeffectcolor}{q} \cdot \textcolor{\coeffectcolor}{\gamma}_{{\mathrm{1}}} }  \ottsym{+} \textcolor{\coeffectcolor}{\gamma}_{{\mathrm{2}}} }  } }%
\ottpremise{ \textcolor{\coeffectcolor}{ \textcolor{\coeffectcolor}{q} }  \neq  \textcolor{\coeffectcolor}{  \textcolor{\coeffectcolor}{0}  } }%
}{
  \textcolor{\coeffectcolor}{ \textcolor{\coeffectcolor}{\gamma} }\! \cdot \! \rho   \vdash_{\mathit{lin} }   \ottkw{case}_{\textcolor{\coeffectcolor}{ \textcolor{\coeffectcolor}{q} } }\;  \ottnt{M} \; \ottkw{of}\;( \ottmv{x_{{\mathrm{1}}}} , \ottmv{x_{{\mathrm{2}}}} )\; \rightarrow\;  \ottnt{N}   \Downarrow  \ottnt{T} }{%
{\ottdrulename{eval\_lin\_comp\_csplit}}{}%
}}

\newcommand{\ottdruleevalXXlinXXcompXXcsplitXXzero}[1]{\ottdrule[#1]{%
\ottpremise{  \ottsym{(}    \textcolor{\coeffectcolor}{ \textcolor{\coeffectcolor}{\gamma} }\! \cdot \! \rho    \mathop{,} \;   \ottmv{x_{{\mathrm{1}}}}   \mapsto ^{\textcolor{\coeffectcolor}{  \textcolor{\coeffectcolor}{0}  } }  \mathop{\lightning}    \ottsym{)}   \mathop{,} \;   \ottmv{x_{{\mathrm{2}}}}   \mapsto ^{\textcolor{\coeffectcolor}{  \textcolor{\coeffectcolor}{0}  } }  \mathop{\lightning}    \vdash_{\mathit{lin} }  \ottnt{N}  \Downarrow  \ottnt{T} }%
}{
  \textcolor{\coeffectcolor}{ \textcolor{\coeffectcolor}{\gamma} }\! \cdot \! \rho   \vdash_{\mathit{lin} }   \ottkw{case}_{\textcolor{\coeffectcolor}{  \textcolor{\coeffectcolor}{0}  } }\;  \ottnt{M} \; \ottkw{of}\;( \ottmv{x_{{\mathrm{1}}}} , \ottmv{x_{{\mathrm{2}}}} )\; \rightarrow\;  \ottnt{N}   \Downarrow  \ottnt{T} }{%
{\ottdrulename{eval\_lin\_comp\_csplit\_zero}}{}%
}}

\newcommand{\ottdruleevalXXlinXXcompXXappXXabsXXzero}[1]{\ottdrule[#1]{%
\ottpremise{  \textcolor{\coeffectcolor}{ \textcolor{\coeffectcolor}{\gamma} }\! \cdot \! \rho   \vdash_{\mathit{lin} }  \ottnt{M}  \Downarrow   \mathbf{clo}(   \textcolor{\coeffectcolor}{ \textcolor{\coeffectcolor}{\gamma}' }\! \cdot \! \rho'  ,   \lambda  \ottmv{x} ^{\textcolor{\coeffectcolor}{  \textcolor{\coeffectcolor}{0}  } }. \ottnt{M'}   )  }%
\ottpremise{  \ottsym{(}   \textcolor{\coeffectcolor}{ \textcolor{\coeffectcolor}{\gamma}' }\! \cdot \! \rho'   \ottsym{)}   \mathop{,} \;  \ottsym{(}   \ottmv{x}   \mapsto ^{\textcolor{\coeffectcolor}{  \textcolor{\coeffectcolor}{0}  } }  \mathop{\lightning}   \ottsym{)}   \vdash_{\mathit{lin} }  \ottnt{M'}  \Downarrow  \ottnt{T} }%
}{
  \textcolor{\coeffectcolor}{ \textcolor{\coeffectcolor}{\gamma} }\! \cdot \! \rho   \vdash_{\mathit{lin} }  \ottnt{M} \, \ottnt{V}  \Downarrow  \ottnt{T} }{%
{\ottdrulename{eval\_lin\_comp\_app\_abs\_zero}}{}%
}}

\newcommand{\ottdruleevalXXlinXXcompXXsplitXXzero}[1]{\ottdrule[#1]{%
\ottpremise{     \textcolor{\coeffectcolor}{ \textcolor{\coeffectcolor}{\gamma} }\! \cdot \! \rho    \mathop{,} \;   \ottmv{x_{{\mathrm{1}}}}   \mapsto ^{\textcolor{\coeffectcolor}{  \textcolor{\coeffectcolor}{0}  } }  \mathop{\lightning}      \mathop{,} \;   \ottmv{x_{{\mathrm{2}}}}   \mapsto ^{\textcolor{\coeffectcolor}{  \textcolor{\coeffectcolor}{0}  } }  \mathop{\lightning}    \vdash_{\mathit{lin} }  \ottnt{N}  \Downarrow  \ottnt{T} }%
}{
  \textcolor{\coeffectcolor}{ \textcolor{\coeffectcolor}{\gamma} }\! \cdot \! \rho   \vdash_{\mathit{lin} }   \ottkw{case}_{\textcolor{\coeffectcolor}{  \textcolor{\coeffectcolor}{0}  } } \;  \ottnt{V} \; \ottkw{of}\;( \ottmv{x_{{\mathrm{1}}}} , \ottmv{x_{{\mathrm{2}}}} )\; \rightarrow\;  \ottnt{N}   \Downarrow  \ottnt{T} }{%
{\ottdrulename{eval\_lin\_comp\_split\_zero}}{}%
}}

\newcommand{\ottdruleevalXXlinXXcompXXretXXzero}[1]{\ottdrule[#1]{%
}{
  \textcolor{\coeffectcolor}{  \textcolor{\coeffectcolor}{\overline{0} }  }\! \cdot \! \rho   \vdash_{\mathit{lin} }   \ottkw{return} _{\textcolor{\coeffectcolor}{  \textcolor{\coeffectcolor}{0}  } }\;  \ottnt{V}   \Downarrow   \ottkw{return} _{\textcolor{\coeffectcolor}{  \textcolor{\coeffectcolor}{0}  } }  \mathop{\lightning}  }{%
{\ottdrulename{eval\_lin\_comp\_ret\_zero}}{}%
}}

\newcommand{\ottdruleevalXXlinXXcompXXsub}[1]{\ottdrule[#1]{%
\ottpremise{  \textcolor{\coeffectcolor}{ \textcolor{\coeffectcolor}{\gamma}' }\! \cdot \! \rho   \vdash_{\mathit{lin} }  \ottnt{M}  \Downarrow  \ottnt{T} }%
\ottpremise{ \textcolor{\coeffectcolor}{ \textcolor{\coeffectcolor}{\gamma} }\; \textcolor{\coeffectcolor}{\mathop{\leq_{\mathit{co} } } } \; \textcolor{\coeffectcolor}{ \textcolor{\coeffectcolor}{\gamma}' } }%
}{
  \textcolor{\coeffectcolor}{ \textcolor{\coeffectcolor}{\gamma} }\! \cdot \! \rho   \vdash_{\mathit{lin} }  \ottnt{M}  \Downarrow  \ottnt{T} }{%
{\ottdrulename{eval\_lin\_comp\_sub}}{}%
}}

\newcommand{\ottdefnevalXXlinXXcomp}[1]{\begin{ottdefnblock}[#1]{$ \mu  \vdash_{\mathit{lin} }  \ottnt{M}  \Downarrow  \ottnt{T} $}{\ottcom{environment-based resource counting semantics for CBPV (large-step)}}
\ottusedrule{\ottdruleevalXXlinXXcompXXabs{}}
\ottusedrule{\ottdruleevalXXlinXXcompXXcpair{}}
\ottusedrule{\ottdruleevalXXlinXXcompXXcunit{}}
\ottusedrule{\ottdruleevalXXlinXXcompXXappXXabs{}}
\ottusedrule{\ottdruleevalXXlinXXcompXXforceXXthunk{}}
\ottusedrule{\ottdruleevalXXlinXXcompXXreturn{}}
\ottusedrule{\ottdruleevalXXlinXXcompXXletinXXret{}}
\ottusedrule{\ottdruleevalXXlinXXcompXXsplit{}}
\ottusedrule{\ottdruleevalXXlinXXcompXXfst{}}
\ottusedrule{\ottdruleevalXXlinXXcompXXsnd{}}
\ottusedrule{\ottdruleevalXXlinXXcompXXsequence{}}
\ottusedrule{\ottdruleevalXXlinXXcompXXcaseXXinl{}}
\ottusedrule{\ottdruleevalXXlinXXcompXXcaseXXinr{}}
\ottusedrule{\ottdruleevalXXlinXXcompXXctensor{}}
\ottusedrule{\ottdruleevalXXlinXXcompXXcsplit{}}
\ottusedrule{\ottdruleevalXXlinXXcompXXcsplitXXzero{}}
\ottusedrule{\ottdruleevalXXlinXXcompXXappXXabsXXzero{}}
\ottusedrule{\ottdruleevalXXlinXXcompXXsplitXXzero{}}
\ottusedrule{\ottdruleevalXXlinXXcompXXretXXzero{}}
\ottusedrule{\ottdruleevalXXlinXXcompXXsub{}}
\end{ottdefnblock}}

\newcommand{\ottdefnsJInstrEnv}{
\ottdefnevalXXcoeffXXval{}\ottdefnevalXXcoeffXXcomp{}\ottdefnevalXXlinXXval{}\ottdefnevalXXlinXXcomp{}}


\newcommand{\ottdefnstlcXXtyping}[1]{\begin{ottdefnblock}[#1]{$\Gamma  \vdash  \ottnt{e}  \ottsym{:}  \tau$}{}
\end{ottdefnblock}}

\newcommand{\ottdefnsJSTLC}{
\ottdefnstlcXXtyping{}}

\newcommand{\ottdrulecbpvXXvar}[1]{\ottdrule[#1]{%
\ottpremise{\ottmv{x}  \ottsym{:}  \ottnt{A} \, \in \, \Gamma}%
}{
\Gamma  \vdash  \ottmv{x}  \ottsym{:}  \ottnt{A}}{%
{\ottdrulename{cbpv\_var}}{}%
}}

\newcommand{\ottdrulecbpvXXunit}[1]{\ottdrule[#1]{%
}{
\Gamma  \vdash  \ottsym{()}  \ottsym{:}  \ottkw{unit}}{%
{\ottdrulename{cbpv\_unit}}{}%
}}

\newcommand{\ottdrulecbpvXXthunk}[1]{\ottdrule[#1]{%
\ottpremise{\Gamma  \vdash  \ottnt{M}  \ottsym{:}  \ottnt{B}}%
}{
\Gamma  \vdash  \ottsym{\{}  \ottnt{M}  \ottsym{\}}  \ottsym{:}  \ottkw{U} \, \ottnt{B}}{%
{\ottdrulename{cbpv\_thunk}}{}%
}}

\newcommand{\ottdrulecbpvXXpair}[1]{\ottdrule[#1]{%
\ottpremise{\Gamma  \vdash  \ottnt{V_{{\mathrm{1}}}}  \ottsym{:}  \ottnt{A_{{\mathrm{1}}}}}%
\ottpremise{\Gamma  \vdash  \ottnt{V_{{\mathrm{2}}}}  \ottsym{:}  \ottnt{A_{{\mathrm{2}}}}}%
}{
\Gamma  \vdash  \ottsym{(}  \ottnt{V_{{\mathrm{1}}}}  \ottsym{,}  \ottnt{V_{{\mathrm{2}}}}  \ottsym{)}  \ottsym{:}   \ottnt{A_{{\mathrm{1}}}} \times \ottnt{A_{{\mathrm{2}}}} }{%
{\ottdrulename{cbpv\_pair}}{}%
}}

\newcommand{\ottdrulecbpvXXinl}[1]{\ottdrule[#1]{%
\ottpremise{\Gamma  \vdash  \ottnt{V}  \ottsym{:}  \ottnt{A_{{\mathrm{1}}}}}%
}{
\Gamma  \vdash  \ottkw{inl} \, \ottnt{V}  \ottsym{:}  \ottnt{A_{{\mathrm{1}}}}  \ottsym{+}  \ottnt{A_{{\mathrm{2}}}}}{%
{\ottdrulename{cbpv\_inl}}{}%
}}

\newcommand{\ottdrulecbpvXXinr}[1]{\ottdrule[#1]{%
\ottpremise{\Gamma  \vdash  \ottnt{V}  \ottsym{:}  \ottnt{A_{{\mathrm{2}}}}}%
}{
\Gamma  \vdash  \ottkw{inr} \, \ottnt{V}  \ottsym{:}  \ottnt{A_{{\mathrm{1}}}}  \ottsym{+}  \ottnt{A_{{\mathrm{2}}}}}{%
{\ottdrulename{cbpv\_inr}}{}%
}}

\newcommand{\ottdefnvalXXtyping}[1]{\begin{ottdefnblock}[#1]{$\Gamma  \vdash  \ottnt{V}  \ottsym{:}  \ottnt{A}$}{\ottcom{value effect typing rules}}
\ottusedrule{\ottdrulecbpvXXvar{}}
\ottusedrule{\ottdrulecbpvXXunit{}}
\ottusedrule{\ottdrulecbpvXXthunk{}}
\ottusedrule{\ottdrulecbpvXXpair{}}
\ottusedrule{\ottdrulecbpvXXinl{}}
\ottusedrule{\ottdrulecbpvXXinr{}}
\end{ottdefnblock}}

\newcommand{\ottdrulecbpvXXabs}[1]{\ottdrule[#1]{%
\ottpremise{ \Gamma   \mathop{,}   \ottmv{x}  \ottsym{:}  \ottnt{A}   \vdash  \ottnt{M}  \ottsym{:}  \ottnt{B}}%
}{
\Gamma  \vdash   \lambda  \ottmv{x} . \ottnt{M}   \ottsym{:}  \ottnt{A}  \to  \ottnt{B}}{%
{\ottdrulename{cbpv\_abs}}{}%
}}

\newcommand{\ottdrulecbpvXXapp}[1]{\ottdrule[#1]{%
\ottpremise{\Gamma  \vdash  \ottnt{M}  \ottsym{:}  \ottnt{A}  \to  \ottnt{B}}%
\ottpremise{\Gamma  \vdash  \ottnt{V}  \ottsym{:}  \ottnt{A}}%
}{
\Gamma  \vdash  \ottnt{M} \, \ottnt{V}  \ottsym{:}  \ottnt{B}}{%
{\ottdrulename{cbpv\_app}}{}%
}}

\newcommand{\ottdrulecbpvXXforce}[1]{\ottdrule[#1]{%
\ottpremise{\Gamma  \vdash  \ottnt{V}  \ottsym{:}  \ottkw{U} \, \ottnt{B}}%
}{
\Gamma  \vdash  \ottnt{V}  \ottsym{!}  \ottsym{:}  \ottnt{B}}{%
{\ottdrulename{cbpv\_force}}{}%
}}

\newcommand{\ottdrulecbpvXXret}[1]{\ottdrule[#1]{%
\ottpremise{\Gamma  \vdash  \ottnt{V}  \ottsym{:}  \ottnt{A}}%
}{
\Gamma  \vdash  \ottkw{return} \, \ottnt{V}  \ottsym{:}  \ottkw{F} \, \ottnt{A}}{%
{\ottdrulename{cbpv\_ret}}{}%
}}

\newcommand{\ottdrulecbpvXXletin}[1]{\ottdrule[#1]{%
\ottpremise{\Gamma  \vdash  \ottnt{M}  \ottsym{:}  \ottkw{F} \, \ottnt{A}}%
\ottpremise{ \Gamma   \mathop{,}   \ottmv{x}  \ottsym{:}  \ottnt{A}   \vdash  \ottnt{N}  \ottsym{:}  \ottnt{B}}%
}{
\Gamma  \vdash  \ottmv{x}  \leftarrow  \ottnt{M} \, \ottkw{in} \, \ottnt{N}  \ottsym{:}  \ottnt{B}}{%
{\ottdrulename{cbpv\_letin}}{}%
}}

\newcommand{\ottdrulecbpvXXsplit}[1]{\ottdrule[#1]{%
\ottpremise{\Gamma  \vdash  \ottnt{V}  \ottsym{:}   \ottnt{A_{{\mathrm{1}}}} \times \ottnt{A_{{\mathrm{2}}}} }%
\ottpremise{  \Gamma   \mathop{,}   \ottmv{x_{{\mathrm{1}}}}  \ottsym{:}  \ottnt{A_{{\mathrm{1}}}}    \mathop{,}   \ottmv{x_{{\mathrm{2}}}}  \ottsym{:}  \ottnt{A_{{\mathrm{2}}}}   \vdash  \ottnt{N}  \ottsym{:}  \ottnt{B}}%
}{
\Gamma  \vdash   \ottkw{let}\; ( \ottmv{x_{{\mathrm{1}}}} ,  \ottmv{x_{{\mathrm{2}}}} ) =  \ottnt{V} \; \ottkw{in}\;  \ottnt{N}   \ottsym{:}  \ottnt{B}}{%
{\ottdrulename{cbpv\_split}}{}%
}}

\newcommand{\ottdrulecbpvXXsequence}[1]{\ottdrule[#1]{%
\ottpremise{\Gamma  \vdash  \ottnt{V}  \ottsym{:}  \ottkw{unit}}%
\ottpremise{\Gamma  \vdash  \ottnt{N}  \ottsym{:}  \ottnt{B}}%
}{
\Gamma  \vdash   \ottnt{V}  ;  \ottnt{N}   \ottsym{:}  \ottnt{B}}{%
{\ottdrulename{cbpv\_sequence}}{}%
}}

\newcommand{\ottdrulecbpvXXcunit}[1]{\ottdrule[#1]{%
}{
\Gamma  \vdash   \langle\rangle   \ottsym{:}   \top }{%
{\ottdrulename{cbpv\_cunit}}{}%
}}

\newcommand{\ottdrulecbpvXXcpair}[1]{\ottdrule[#1]{%
\ottpremise{\Gamma  \vdash  \ottnt{M_{{\mathrm{1}}}}  \ottsym{:}  \ottnt{B_{{\mathrm{1}}}}}%
\ottpremise{\Gamma  \vdash  \ottnt{M_{{\mathrm{2}}}}  \ottsym{:}  \ottnt{B_{{\mathrm{2}}}}}%
}{
\Gamma  \vdash   \langle  \ottnt{M_{{\mathrm{1}}}} , \ottnt{M_{{\mathrm{2}}}}  \rangle   \ottsym{:}   \ottnt{B_{{\mathrm{1}}}}   \mathop{\&}   \ottnt{B_{{\mathrm{2}}}} }{%
{\ottdrulename{cbpv\_cpair}}{}%
}}

\newcommand{\ottdrulecbpvXXfst}[1]{\ottdrule[#1]{%
\ottpremise{\Gamma  \vdash  \ottnt{M}  \ottsym{:}   \ottnt{B_{{\mathrm{1}}}}   \mathop{\&}   \ottnt{B_{{\mathrm{2}}}} }%
}{
\Gamma  \vdash   \ottnt{M}  . 1   \ottsym{:}  \ottnt{B_{{\mathrm{1}}}}}{%
{\ottdrulename{cbpv\_fst}}{}%
}}

\newcommand{\ottdrulecbpvXXsnd}[1]{\ottdrule[#1]{%
\ottpremise{\Gamma  \vdash  \ottnt{M}  \ottsym{:}   \ottnt{B_{{\mathrm{1}}}}   \mathop{\&}   \ottnt{B_{{\mathrm{2}}}} }%
}{
\Gamma  \vdash   \ottnt{M}  . 2   \ottsym{:}  \ottnt{B_{{\mathrm{2}}}}}{%
{\ottdrulename{cbpv\_snd}}{}%
}}

\newcommand{\ottdrulecbpvXXcase}[1]{\ottdrule[#1]{%
\ottpremise{\Gamma  \vdash  \ottnt{V}  \ottsym{:}  \ottnt{A_{{\mathrm{1}}}}  \ottsym{+}  \ottnt{A_{{\mathrm{2}}}}}%
\ottpremise{ \Gamma   \mathop{,}   \ottmv{x_{{\mathrm{1}}}}  \ottsym{:}  \ottnt{A_{{\mathrm{1}}}}   \vdash  \ottnt{M_{{\mathrm{1}}}}  \ottsym{:}  \ottnt{B}}%
\ottpremise{ \Gamma   \mathop{,}   \ottmv{x_{{\mathrm{2}}}}  \ottsym{:}  \ottnt{A_{{\mathrm{2}}}}   \vdash  \ottnt{M_{{\mathrm{2}}}}  \ottsym{:}  \ottnt{B}}%
}{
\Gamma  \vdash  \ottkw{case} \, \ottnt{V} \, \ottkw{of} \, \ottkw{inl} \, \ottmv{x_{{\mathrm{1}}}}  \to  \ottnt{M_{{\mathrm{1}}}}  \mathsf{;} \, \ottkw{inr} \, \ottmv{x_{{\mathrm{2}}}}  \to  \ottnt{M_{{\mathrm{2}}}}  \ottsym{:}  \ottnt{B}}{%
{\ottdrulename{cbpv\_case}}{}%
}}

\newcommand{\ottdefncompXXtyping}[1]{\begin{ottdefnblock}[#1]{$\Gamma  \vdash  \ottnt{M}  \ottsym{:}  \ottnt{B}$}{\ottcom{computation typing rules}}
\ottusedrule{\ottdrulecbpvXXabs{}}
\ottusedrule{\ottdrulecbpvXXapp{}}
\ottusedrule{\ottdrulecbpvXXforce{}}
\ottusedrule{\ottdrulecbpvXXret{}}
\ottusedrule{\ottdrulecbpvXXletin{}}
\ottusedrule{\ottdrulecbpvXXsplit{}}
\ottusedrule{\ottdrulecbpvXXsequence{}}
\ottusedrule{\ottdrulecbpvXXcunit{}}
\ottusedrule{\ottdrulecbpvXXcpair{}}
\ottusedrule{\ottdrulecbpvXXfst{}}
\ottusedrule{\ottdrulecbpvXXsnd{}}
\ottusedrule{\ottdrulecbpvXXcase{}}
\end{ottdefnblock}}

\newcommand{\ottdefnsJCBPV}{
\ottdefnvalXXtyping{}\ottdefncompXXtyping{}}

\newcommand{\ottdruleeffXXvar}[1]{\ottdrule[#1]{%
\ottpremise{\ottmv{x}  \ottsym{:}  \ottnt{A} \, \in \, \Gamma}%
}{
\Gamma  \vdash_{\mathit{eff} }  \ottmv{x}  \ottsym{:}  \ottnt{A}}{%
{\ottdrulename{eff\_var}}{}%
}}

\newcommand{\ottdruleeffXXunit}[1]{\ottdrule[#1]{%
}{
\Gamma  \vdash_{\mathit{eff} }  \ottsym{()}  \ottsym{:}  \ottkw{unit}}{%
{\ottdrulename{eff\_unit}}{}%
}}

\newcommand{\ottdruleeffXXthunk}[1]{\ottdrule[#1]{%
\ottpremise{ \Gamma   \vdash_{\mathit{eff} }   \ottnt{M} \; :^{\textcolor{\effectcolor}{  \textcolor{\effectcolor}{\phi}  } }\;  \ottnt{B} }%
}{
\Gamma  \vdash_{\mathit{eff} }  \ottsym{\{}  \ottnt{M}  \ottsym{\}}  \ottsym{:}   \ottkw{U}_{\color{\effectcolor}{ \textcolor{\effectcolor}{\phi} } }\;  \ottnt{B} }{%
{\ottdrulename{eff\_thunk}}{}%
}}

\newcommand{\ottdruleeffXXpair}[1]{\ottdrule[#1]{%
\ottpremise{\Gamma  \vdash_{\mathit{eff} }  \ottnt{V_{{\mathrm{1}}}}  \ottsym{:}  \ottnt{A_{{\mathrm{1}}}}}%
\ottpremise{\Gamma  \vdash_{\mathit{eff} }  \ottnt{V_{{\mathrm{2}}}}  \ottsym{:}  \ottnt{A_{{\mathrm{2}}}}}%
}{
\Gamma  \vdash_{\mathit{eff} }  \ottsym{(}  \ottnt{V_{{\mathrm{1}}}}  \ottsym{,}  \ottnt{V_{{\mathrm{2}}}}  \ottsym{)}  \ottsym{:}   \ottnt{A_{{\mathrm{1}}}} \times \ottnt{A_{{\mathrm{2}}}} }{%
{\ottdrulename{eff\_pair}}{}%
}}

\newcommand{\ottdruleeffXXinl}[1]{\ottdrule[#1]{%
\ottpremise{\Gamma  \vdash_{\mathit{eff} }  \ottnt{V}  \ottsym{:}  \ottnt{A_{{\mathrm{1}}}}}%
}{
\Gamma  \vdash_{\mathit{eff} }  \ottkw{inl} \, \ottnt{V}  \ottsym{:}  \ottnt{A_{{\mathrm{1}}}}  \ottsym{+}  \ottnt{A_{{\mathrm{2}}}}}{%
{\ottdrulename{eff\_inl}}{}%
}}

\newcommand{\ottdruleeffXXinr}[1]{\ottdrule[#1]{%
\ottpremise{\Gamma  \vdash_{\mathit{eff} }  \ottnt{V}  \ottsym{:}  \ottnt{A_{{\mathrm{2}}}}}%
}{
\Gamma  \vdash_{\mathit{eff} }  \ottkw{inr} \, \ottnt{V}  \ottsym{:}  \ottnt{A_{{\mathrm{1}}}}  \ottsym{+}  \ottnt{A_{{\mathrm{2}}}}}{%
{\ottdrulename{eff\_inr}}{}%
}}

\newcommand{\ottdefnvalXXeffXXtyping}[1]{\begin{ottdefnblock}[#1]{$\Gamma  \vdash_{\mathit{eff} }  \ottnt{V}  \ottsym{:}  \ottnt{A}$}{\ottcom{value effect typing rules}}
\ottusedrule{\ottdruleeffXXvar{}}
\ottusedrule{\ottdruleeffXXunit{}}
\ottusedrule{\ottdruleeffXXthunk{}}
\ottusedrule{\ottdruleeffXXpair{}}
\ottusedrule{\ottdruleeffXXinl{}}
\ottusedrule{\ottdruleeffXXinr{}}
\end{ottdefnblock}}

\newcommand{\ottdruleeffXXabs}[1]{\ottdrule[#1]{%
\ottpremise{  \Gamma   \mathop{,}   \ottmv{x}  \ottsym{:}  \ottnt{A}    \vdash_{\mathit{eff} }   \ottnt{M} \; :^{\textcolor{\effectcolor}{  \textcolor{\effectcolor}{\phi}  } }\;  \ottnt{B} }%
}{
 \Gamma   \vdash_{\mathit{eff} }    \lambda  \ottmv{x} . \ottnt{M}  \; :^{\textcolor{\effectcolor}{  \textcolor{\effectcolor}{\phi}  } }\;  \ottnt{A}  \to  \ottnt{B} }{%
{\ottdrulename{eff\_abs}}{}%
}}

\newcommand{\ottdruleeffXXapp}[1]{\ottdrule[#1]{%
\ottpremise{ \Gamma   \vdash_{\mathit{eff} }   \ottnt{M} \; :^{\textcolor{\effectcolor}{  \textcolor{\effectcolor}{\phi}  } }\;  \ottnt{A}  \to  \ottnt{B} }%
\ottpremise{\Gamma  \vdash_{\mathit{eff} }  \ottnt{V}  \ottsym{:}  \ottnt{A}}%
}{
 \Gamma   \vdash_{\mathit{eff} }   \ottnt{M} \, \ottnt{V} \; :^{\textcolor{\effectcolor}{  \textcolor{\effectcolor}{\phi}  } }\;  \ottnt{B} }{%
{\ottdrulename{eff\_app}}{}%
}}

\newcommand{\ottdruleeffXXforce}[1]{\ottdrule[#1]{%
\ottpremise{\Gamma  \vdash_{\mathit{eff} }  \ottnt{V}  \ottsym{:}   \ottkw{U}_{\color{\effectcolor}{ \textcolor{\effectcolor}{\phi} } }\;  \ottnt{B} }%
}{
 \Gamma   \vdash_{\mathit{eff} }   \ottnt{V}  \ottsym{!} \; :^{\textcolor{\effectcolor}{  \textcolor{\effectcolor}{\phi}  } }\;  \ottnt{B} }{%
{\ottdrulename{eff\_force}}{}%
}}

\newcommand{\ottdruleeffXXret}[1]{\ottdrule[#1]{%
\ottpremise{\Gamma  \vdash_{\mathit{eff} }  \ottnt{V}  \ottsym{:}  \ottnt{A}}%
}{
 \Gamma   \vdash_{\mathit{eff} }   \ottkw{return} \, \ottnt{V} \; :^{\textcolor{\effectcolor}{   \textcolor{\effectcolor}{\varepsilon}   } }\;  \ottkw{F} \, \ottnt{A} }{%
{\ottdrulename{eff\_ret}}{}%
}}

\newcommand{\ottdruleeffXXletin}[1]{\ottdrule[#1]{%
\ottpremise{ \Gamma   \vdash_{\mathit{eff} }   \ottnt{M} \; :^{\textcolor{\effectcolor}{  \textcolor{\effectcolor}{\phi}_{{\mathrm{1}}}  } }\;  \ottkw{F} \, \ottnt{A} }%
\ottpremise{  \Gamma   \mathop{,}   \ottmv{x}  \ottsym{:}  \ottnt{A}    \vdash_{\mathit{eff} }   \ottnt{N} \; :^{\textcolor{\effectcolor}{  \textcolor{\effectcolor}{\phi}_{{\mathrm{2}}}  } }\;  \ottnt{B} }%
}{
 \Gamma   \vdash_{\mathit{eff} }   \ottmv{x}  \leftarrow  \ottnt{M} \, \ottkw{in} \, \ottnt{N} \; :^{\textcolor{\effectcolor}{   \textcolor{\effectcolor}{ \textcolor{\effectcolor}{\phi}_{{\mathrm{1}}}  \cdot  \textcolor{\effectcolor}{\phi}_{{\mathrm{2}}} }   } }\;  \ottnt{B} }{%
{\ottdrulename{eff\_letin}}{}%
}}

\newcommand{\ottdruleeffXXsplit}[1]{\ottdrule[#1]{%
\ottpremise{\Gamma  \vdash_{\mathit{eff} }  \ottnt{V}  \ottsym{:}   \ottnt{A_{{\mathrm{1}}}} \times \ottnt{A_{{\mathrm{2}}}} }%
\ottpremise{   \Gamma   \mathop{,}   \ottmv{x_{{\mathrm{1}}}}  \ottsym{:}  \ottnt{A_{{\mathrm{1}}}}    \mathop{,}   \ottmv{x_{{\mathrm{2}}}}  \ottsym{:}  \ottnt{A_{{\mathrm{2}}}}    \vdash_{\mathit{eff} }   \ottnt{N} \; :^{\textcolor{\effectcolor}{  \textcolor{\effectcolor}{\phi}  } }\;  \ottnt{B} }%
}{
 \Gamma   \vdash_{\mathit{eff} }    \ottkw{let}\; ( \ottmv{x_{{\mathrm{1}}}} ,  \ottmv{x_{{\mathrm{2}}}} ) =  \ottnt{V} \; \ottkw{in}\;  \ottnt{N}  \; :^{\textcolor{\effectcolor}{  \textcolor{\effectcolor}{\phi}  } }\;  \ottnt{B} }{%
{\ottdrulename{eff\_split}}{}%
}}

\newcommand{\ottdruleeffXXsequence}[1]{\ottdrule[#1]{%
\ottpremise{\Gamma  \vdash_{\mathit{eff} }  \ottnt{V}  \ottsym{:}  \ottkw{unit}}%
\ottpremise{ \Gamma   \vdash_{\mathit{eff} }   \ottnt{N} \; :^{\textcolor{\effectcolor}{  \textcolor{\effectcolor}{\phi}  } }\;  \ottnt{B} }%
}{
 \Gamma   \vdash_{\mathit{eff} }    \ottnt{V}  ;  \ottnt{N}  \; :^{\textcolor{\effectcolor}{  \textcolor{\effectcolor}{\phi}  } }\;  \ottnt{B} }{%
{\ottdrulename{eff\_sequence}}{}%
}}

\newcommand{\ottdruleeffXXcunit}[1]{\ottdrule[#1]{%
\ottpremise{ \textcolor{\effectcolor}{  \textcolor{\effectcolor}{\varepsilon}  \;  \textcolor{\effectcolor}{\mathop{\leq_{\mathit{eff} } } } \;  \textcolor{\effectcolor}{\phi}  } }%
}{
 \Gamma   \vdash_{\mathit{eff} }    \langle\rangle  \; :^{\textcolor{\effectcolor}{  \textcolor{\effectcolor}{\phi}  } }\;   \top  }{%
{\ottdrulename{eff\_cunit}}{}%
}}

\newcommand{\ottdruleeffXXcpair}[1]{\ottdrule[#1]{%
\ottpremise{ \Gamma   \vdash_{\mathit{eff} }   \ottnt{M_{{\mathrm{1}}}} \; :^{\textcolor{\effectcolor}{  \textcolor{\effectcolor}{\phi}  } }\;  \ottnt{B_{{\mathrm{1}}}} }%
\ottpremise{ \Gamma   \vdash_{\mathit{eff} }   \ottnt{M_{{\mathrm{2}}}} \; :^{\textcolor{\effectcolor}{  \textcolor{\effectcolor}{\phi}  } }\;  \ottnt{B_{{\mathrm{2}}}} }%
}{
 \Gamma   \vdash_{\mathit{eff} }    \langle  \ottnt{M_{{\mathrm{1}}}} , \ottnt{M_{{\mathrm{2}}}}  \rangle  \; :^{\textcolor{\effectcolor}{  \textcolor{\effectcolor}{\phi}  } }\;   \ottnt{B_{{\mathrm{1}}}}   \mathop{\&}   \ottnt{B_{{\mathrm{2}}}}  }{%
{\ottdrulename{eff\_cpair}}{}%
}}

\newcommand{\ottdruleeffXXfst}[1]{\ottdrule[#1]{%
\ottpremise{ \Gamma   \vdash_{\mathit{eff} }   \ottnt{M} \; :^{\textcolor{\effectcolor}{  \textcolor{\effectcolor}{\phi}  } }\;   \ottnt{B_{{\mathrm{1}}}}   \mathop{\&}   \ottnt{B_{{\mathrm{2}}}}  }%
}{
 \Gamma   \vdash_{\mathit{eff} }    \ottnt{M}  . 1  \; :^{\textcolor{\effectcolor}{  \textcolor{\effectcolor}{\phi}  } }\;  \ottnt{B_{{\mathrm{1}}}} }{%
{\ottdrulename{eff\_fst}}{}%
}}

\newcommand{\ottdruleeffXXsnd}[1]{\ottdrule[#1]{%
\ottpremise{ \Gamma   \vdash_{\mathit{eff} }   \ottnt{M} \; :^{\textcolor{\effectcolor}{  \textcolor{\effectcolor}{\phi}  } }\;   \ottnt{B_{{\mathrm{1}}}}   \mathop{\&}   \ottnt{B_{{\mathrm{2}}}}  }%
}{
 \Gamma   \vdash_{\mathit{eff} }    \ottnt{M}  . 2  \; :^{\textcolor{\effectcolor}{  \textcolor{\effectcolor}{\phi}  } }\;  \ottnt{B_{{\mathrm{2}}}} }{%
{\ottdrulename{eff\_snd}}{}%
}}

\newcommand{\ottdruleeffXXcase}[1]{\ottdrule[#1]{%
\ottpremise{\Gamma  \vdash_{\mathit{eff} }  \ottnt{V}  \ottsym{:}  \ottnt{A_{{\mathrm{1}}}}  \ottsym{+}  \ottnt{A_{{\mathrm{2}}}}}%
\ottpremise{  \Gamma   \mathop{,}   \ottmv{x_{{\mathrm{1}}}}  \ottsym{:}  \ottnt{A_{{\mathrm{1}}}}    \vdash_{\mathit{eff} }   \ottnt{M_{{\mathrm{1}}}} \; :^{\textcolor{\effectcolor}{  \textcolor{\effectcolor}{\phi}  } }\;  \ottnt{B} }%
\ottpremise{  \Gamma   \mathop{,}   \ottmv{x_{{\mathrm{2}}}}  \ottsym{:}  \ottnt{A_{{\mathrm{2}}}}    \vdash_{\mathit{eff} }   \ottnt{M_{{\mathrm{2}}}} \; :^{\textcolor{\effectcolor}{  \textcolor{\effectcolor}{\phi}  } }\;  \ottnt{B} }%
}{
 \Gamma   \vdash_{\mathit{eff} }   \ottkw{case} \, \ottnt{V} \, \ottkw{of} \, \ottkw{inl} \, \ottmv{x_{{\mathrm{1}}}}  \to  \ottnt{M_{{\mathrm{1}}}}  \mathsf{;} \, \ottkw{inr} \, \ottmv{x_{{\mathrm{2}}}}  \to  \ottnt{M_{{\mathrm{2}}}} \; :^{\textcolor{\effectcolor}{  \textcolor{\effectcolor}{\phi}  } }\;  \ottnt{B} }{%
{\ottdrulename{eff\_case}}{}%
}}

\newcommand{\ottdruleeffXXtick}[1]{\ottdrule[#1]{%
}{
 \Gamma   \vdash_{\mathit{eff} }    \textcolor{\effectcolor}{ \ottkw{tick} }  \; :^{\textcolor{\effectcolor}{   \textcolor{\effectcolor}{\ottkw{Tick} }   } }\;  \ottkw{F} \, \ottkw{unit} }{%
{\ottdrulename{eff\_tick}}{}%
}}

\newcommand{\ottdruleeffXXsub}[1]{\ottdrule[#1]{%
\ottpremise{ \Gamma   \vdash_{\mathit{eff} }   \ottnt{M} \; :^{\textcolor{\effectcolor}{  \textcolor{\effectcolor}{\phi}_{{\mathrm{1}}}  } }\;  \ottnt{B} }%
\ottpremise{ \textcolor{\effectcolor}{ \textcolor{\effectcolor}{\phi}_{{\mathrm{1}}} \;  \textcolor{\effectcolor}{\mathop{\leq_{\mathit{eff} } } } \;  \textcolor{\effectcolor}{\phi}_{{\mathrm{2}}}  } }%
}{
 \Gamma   \vdash_{\mathit{eff} }   \ottnt{M} \; :^{\textcolor{\effectcolor}{  \textcolor{\effectcolor}{\phi}_{{\mathrm{2}}}  } }\;  \ottnt{B} }{%
{\ottdrulename{eff\_sub}}{}%
}}

\newcommand{\ottdefncompXXeffXXtyping}[1]{\begin{ottdefnblock}[#1]{$ \Gamma   \vdash_{\mathit{eff} }   \ottnt{M} \; :^{\textcolor{\effectcolor}{  \textcolor{\effectcolor}{\phi}  } }\;  \ottnt{B} $}{\ottcom{computation effect typing rules}}
\ottusedrule{\ottdruleeffXXabs{}}
\ottusedrule{\ottdruleeffXXapp{}}
\ottusedrule{\ottdruleeffXXforce{}}
\ottusedrule{\ottdruleeffXXret{}}
\ottusedrule{\ottdruleeffXXletin{}}
\ottusedrule{\ottdruleeffXXsplit{}}
\ottusedrule{\ottdruleeffXXsequence{}}
\ottusedrule{\ottdruleeffXXcunit{}}
\ottusedrule{\ottdruleeffXXcpair{}}
\ottusedrule{\ottdruleeffXXfst{}}
\ottusedrule{\ottdruleeffXXsnd{}}
\ottusedrule{\ottdruleeffXXcase{}}
\ottusedrule{\ottdruleeffXXtick{}}
\ottusedrule{\ottdruleeffXXsub{}}
\end{ottdefnblock}}

\newcommand{\ottdrulelamXXeffXXvar}[1]{\ottdrule[#1]{%
\ottpremise{\ottmv{x}  \ottsym{:}  \tau \, \in \, \Gamma}%
}{
 \Gamma   \vdash_{\mathit{eff} }   \ottmv{x}  :^{\textcolor{\effectcolor}{  \textcolor{\effectcolor}{\varepsilon}  } }  \tau }{%
{\ottdrulename{lam\_eff\_var}}{}%
}}

\newcommand{\ottdrulelamXXeffXXabs}[1]{\ottdrule[#1]{%
\ottpremise{  \Gamma   \mathop{,}   \ottmv{x}  \ottsym{:}  \tau_{{\mathrm{1}}}    \vdash_{\mathit{eff} }   \ottnt{e}  :^{\textcolor{\effectcolor}{ \textcolor{\effectcolor}{\phi} } }  \tau_{{\mathrm{2}}} }%
}{
 \Gamma   \vdash_{\mathit{eff} }    \lambda  \ottmv{x} . \ottnt{e}   :^{\textcolor{\effectcolor}{  \textcolor{\effectcolor}{\varepsilon}  } }   \tau_{{\mathrm{1}}}  \stackrel{\textcolor{\effectcolor}{ \textcolor{\effectcolor}{\phi} } }{\rightarrow}  \tau_{{\mathrm{2}}}  }{%
{\ottdrulename{lam\_eff\_abs}}{}%
}}

\newcommand{\ottdrulelamXXeffXXapp}[1]{\ottdrule[#1]{%
\ottpremise{ \Gamma   \vdash_{\mathit{eff} }   \ottnt{e_{{\mathrm{1}}}}  :^{\textcolor{\effectcolor}{ \textcolor{\effectcolor}{\phi}_{{\mathrm{1}}} } }   \tau_{{\mathrm{1}}}  \stackrel{\textcolor{\effectcolor}{ \textcolor{\effectcolor}{\phi}_{{\mathrm{3}}} } }{\rightarrow}  \tau_{{\mathrm{2}}}  }%
\ottpremise{ \Gamma   \vdash_{\mathit{eff} }   \ottnt{e_{{\mathrm{2}}}}  :^{\textcolor{\effectcolor}{ \textcolor{\effectcolor}{\phi}_{{\mathrm{2}}} } }  \tau_{{\mathrm{1}}} }%
}{
 \Gamma   \vdash_{\mathit{eff} }   \ottnt{e_{{\mathrm{1}}}} \, \ottnt{e_{{\mathrm{2}}}}  :^{\textcolor{\effectcolor}{  \textcolor{\effectcolor}{  \textcolor{\effectcolor}{ \textcolor{\effectcolor}{\phi}_{{\mathrm{1}}}  \cdot  \textcolor{\effectcolor}{\phi}_{{\mathrm{2}}} }   \cdot  \textcolor{\effectcolor}{\phi}_{{\mathrm{3}}} }  } }  \tau_{{\mathrm{2}}} }{%
{\ottdrulename{lam\_eff\_app}}{}%
}}

\newcommand{\ottdrulelamXXeffXXunit}[1]{\ottdrule[#1]{%
}{
 \Gamma   \vdash_{\mathit{eff} }   \ottsym{()}  :^{\textcolor{\effectcolor}{  \textcolor{\effectcolor}{\varepsilon}  } }  \ottkw{unit} }{%
{\ottdrulename{lam\_eff\_unit}}{}%
}}

\newcommand{\ottdrulelamXXeffXXsequence}[1]{\ottdrule[#1]{%
\ottpremise{ \Gamma   \vdash_{\mathit{eff} }   \ottnt{e_{{\mathrm{1}}}}  :^{\textcolor{\effectcolor}{ \textcolor{\effectcolor}{\phi}_{{\mathrm{1}}} } }  \ottkw{unit} }%
\ottpremise{ \Gamma   \vdash_{\mathit{eff} }   \ottnt{e_{{\mathrm{2}}}}  :^{\textcolor{\effectcolor}{ \textcolor{\effectcolor}{\phi}_{{\mathrm{2}}} } }  \tau }%
}{
 \Gamma   \vdash_{\mathit{eff} }    \ottnt{e_{{\mathrm{1}}}}  ;  \ottnt{e_{{\mathrm{2}}}}   :^{\textcolor{\effectcolor}{  \textcolor{\effectcolor}{ \textcolor{\effectcolor}{\phi}_{{\mathrm{1}}}  \cdot  \textcolor{\effectcolor}{\phi}_{{\mathrm{2}}} }  } }  \tau }{%
{\ottdrulename{lam\_eff\_sequence}}{}%
}}

\newcommand{\ottdrulelamXXeffXXpair}[1]{\ottdrule[#1]{%
\ottpremise{ \Gamma   \vdash_{\mathit{eff} }   \ottnt{e_{{\mathrm{1}}}}  :^{\textcolor{\effectcolor}{ \textcolor{\effectcolor}{\phi}_{{\mathrm{1}}} } }  \tau_{{\mathrm{1}}} }%
\ottpremise{ \Gamma   \vdash_{\mathit{eff} }   \ottnt{e_{{\mathrm{2}}}}  :^{\textcolor{\effectcolor}{ \textcolor{\effectcolor}{\phi}_{{\mathrm{2}}} } }  \tau_{{\mathrm{2}}} }%
}{
 \Gamma   \vdash_{\mathit{eff} }   \ottsym{(}  \ottnt{e_{{\mathrm{1}}}}  \ottsym{,}  \ottnt{e_{{\mathrm{2}}}}  \ottsym{)}  :^{\textcolor{\effectcolor}{  \textcolor{\effectcolor}{ \textcolor{\effectcolor}{\phi}_{{\mathrm{1}}}  \cdot  \textcolor{\effectcolor}{\phi}_{{\mathrm{2}}} }  } }   \tau_{{\mathrm{1}}}  \otimes  \tau_{{\mathrm{2}}}  }{%
{\ottdrulename{lam\_eff\_pair}}{}%
}}

\newcommand{\ottdrulelamXXeffXXsplit}[1]{\ottdrule[#1]{%
\ottpremise{ \Gamma   \vdash_{\mathit{eff} }   \ottnt{e_{{\mathrm{1}}}}  :^{\textcolor{\effectcolor}{ \textcolor{\effectcolor}{\phi}_{{\mathrm{1}}} } }   \tau_{{\mathrm{1}}}  \otimes  \tau_{{\mathrm{2}}}  }%
\ottpremise{   \Gamma   \mathop{,}   \ottmv{x_{{\mathrm{1}}}}  \ottsym{:}  \tau_{{\mathrm{1}}}    \mathop{,}   \ottmv{x_{{\mathrm{2}}}}  \ottsym{:}  \tau_{{\mathrm{2}}}    \vdash_{\mathit{eff} }   \ottnt{e_{{\mathrm{2}}}}  :^{\textcolor{\effectcolor}{ \textcolor{\effectcolor}{\phi}_{{\mathrm{2}}} } }  \tau }%
}{
 \Gamma   \vdash_{\mathit{eff} }    \ottkw{let}\; ( \ottmv{x_{{\mathrm{1}}}} ,  \ottmv{x_{{\mathrm{2}}}} ) =  \ottnt{e_{{\mathrm{1}}}} \; \ottkw{in}\;  \ottnt{e_{{\mathrm{2}}}}   :^{\textcolor{\effectcolor}{  \textcolor{\effectcolor}{ \textcolor{\effectcolor}{\phi}_{{\mathrm{1}}}  \cdot  \textcolor{\effectcolor}{\phi}_{{\mathrm{2}}} }  } }  \tau }{%
{\ottdrulename{lam\_eff\_split}}{}%
}}

\newcommand{\ottdrulelamXXeffXXinl}[1]{\ottdrule[#1]{%
\ottpremise{ \Gamma   \vdash_{\mathit{eff} }   \ottnt{e}  :^{\textcolor{\effectcolor}{ \textcolor{\effectcolor}{\phi} } }  \tau_{{\mathrm{1}}} }%
}{
 \Gamma   \vdash_{\mathit{eff} }   \ottkw{inl} \, \ottnt{e}  :^{\textcolor{\effectcolor}{ \textcolor{\effectcolor}{\phi} } }  \tau_{{\mathrm{1}}}  \ottsym{+}  \tau_{{\mathrm{2}}} }{%
{\ottdrulename{lam\_eff\_inl}}{}%
}}

\newcommand{\ottdrulelamXXeffXXinr}[1]{\ottdrule[#1]{%
\ottpremise{ \Gamma   \vdash_{\mathit{eff} }   \ottnt{e}  :^{\textcolor{\effectcolor}{ \textcolor{\effectcolor}{\phi} } }  \tau_{{\mathrm{2}}} }%
}{
 \Gamma   \vdash_{\mathit{eff} }   \ottkw{inr} \, \ottnt{e}  :^{\textcolor{\effectcolor}{ \textcolor{\effectcolor}{\phi} } }  \tau_{{\mathrm{1}}}  \ottsym{+}  \tau_{{\mathrm{2}}} }{%
{\ottdrulename{lam\_eff\_inr}}{}%
}}

\newcommand{\ottdrulelamXXeffXXcase}[1]{\ottdrule[#1]{%
\ottpremise{ \Gamma   \vdash_{\mathit{eff} }   \ottnt{e}  :^{\textcolor{\effectcolor}{ \textcolor{\effectcolor}{\phi}_{{\mathrm{1}}} } }  \tau_{{\mathrm{1}}}  \ottsym{+}  \tau_{{\mathrm{2}}} }%
\ottpremise{  \Gamma   \mathop{,}   \ottmv{x}  \ottsym{:}  \tau_{{\mathrm{1}}}    \vdash_{\mathit{eff} }   \ottnt{e_{{\mathrm{1}}}}  :^{\textcolor{\effectcolor}{ \textcolor{\effectcolor}{\phi}_{{\mathrm{2}}} } }  \tau }%
\ottpremise{  \Gamma   \mathop{,}   \ottmv{x}  \ottsym{:}  \tau_{{\mathrm{2}}}    \vdash_{\mathit{eff} }   \ottnt{e_{{\mathrm{2}}}}  :^{\textcolor{\effectcolor}{ \textcolor{\effectcolor}{\phi}_{{\mathrm{2}}} } }  \tau }%
}{
 \Gamma   \vdash_{\mathit{eff} }    \ottkw{case}\;  \ottnt{e} \; \ottkw{of}\; \ottkw{inl}\;  \ottmv{x_{{\mathrm{1}}}}  \rightarrow  \ottnt{e_{{\mathrm{1}}}}  \ottkw{;} \ottkw{inr}\;  \ottmv{x_{{\mathrm{2}}}}  \rightarrow  \ottnt{e_{{\mathrm{2}}}}   :^{\textcolor{\effectcolor}{  \textcolor{\effectcolor}{ \textcolor{\effectcolor}{\phi}_{{\mathrm{1}}}  \cdot  \textcolor{\effectcolor}{\phi}_{{\mathrm{2}}} }  } }  \tau }{%
{\ottdrulename{lam\_eff\_case}}{}%
}}

\newcommand{\ottdrulelamXXeffXXtick}[1]{\ottdrule[#1]{%
}{
 \Gamma   \vdash_{\mathit{eff} }    \textcolor{\effectcolor}{ \ottkw{tick} }   :^{\textcolor{\effectcolor}{  \textcolor{\effectcolor}{\ottkw{Tick} }  } }  \ottkw{unit} }{%
{\ottdrulename{lam\_eff\_tick}}{}%
}}

\newcommand{\ottdrulelamXXeffXXsub}[1]{\ottdrule[#1]{%
\ottpremise{ \Gamma   \vdash_{\mathit{eff} }   \ottnt{e}  :^{\textcolor{\effectcolor}{ \textcolor{\effectcolor}{\phi}_{{\mathrm{1}}} } }  \tau }%
\ottpremise{ \textcolor{\effectcolor}{ \textcolor{\effectcolor}{\phi}_{{\mathrm{1}}} \;  \textcolor{\effectcolor}{\mathop{\leq_{\mathit{eff} } } } \;  \textcolor{\effectcolor}{\phi}_{{\mathrm{2}}}  } }%
}{
 \Gamma   \vdash_{\mathit{eff} }   \ottnt{e}  :^{\textcolor{\effectcolor}{ \textcolor{\effectcolor}{\phi}_{{\mathrm{2}}} } }  \tau }{%
{\ottdrulename{lam\_eff\_sub}}{}%
}}

\newcommand{\ottdrulelamXXeffXXprint}[1]{\ottdrule[#1]{%
}{
 \Gamma   \vdash_{\mathit{eff} }   \ottkw{print} \, \ottmv{s}  :^{\textcolor{\effectcolor}{  \textcolor{\effectcolor}{[  \ottmv{s}  ] }  } }  \ottkw{unit} }{%
{\ottdrulename{lam\_eff\_print}}{}%
}}

\newcommand{\ottdefnlamXXeffXXtyping}[1]{\begin{ottdefnblock}[#1]{$ \Gamma   \vdash_{\mathit{eff} }   \ottnt{e}  :^{\textcolor{\effectcolor}{ \textcolor{\effectcolor}{\phi} } }  \tau $}{\ottcom{Effect calculus from Wadler / Thiemann, CBV semantics}}
\ottusedrule{\ottdrulelamXXeffXXvar{}}
\ottusedrule{\ottdrulelamXXeffXXabs{}}
\ottusedrule{\ottdrulelamXXeffXXapp{}}
\ottusedrule{\ottdrulelamXXeffXXunit{}}
\ottusedrule{\ottdrulelamXXeffXXsequence{}}
\ottusedrule{\ottdrulelamXXeffXXpair{}}
\ottusedrule{\ottdrulelamXXeffXXsplit{}}
\ottusedrule{\ottdrulelamXXeffXXinl{}}
\ottusedrule{\ottdrulelamXXeffXXinr{}}
\ottusedrule{\ottdrulelamXXeffXXcase{}}
\ottusedrule{\ottdrulelamXXeffXXtick{}}
\ottusedrule{\ottdrulelamXXeffXXsub{}}
\ottusedrule{\ottdrulelamXXeffXXprint{}}
\end{ottdefnblock}}

\newcommand{\ottdrulelamXXeffnXXvar}[1]{\ottdrule[#1]{%
\ottpremise{ \ottmv{x}  :^{\textcolor{\effectcolor}{ \textcolor{\effectcolor}{\phi}_{{\mathrm{1}}} } }  \tau  \, \in \, \Gamma}%
\ottpremise{ \textcolor{\effectcolor}{ \textcolor{\effectcolor}{\phi}_{{\mathrm{1}}} \;  \textcolor{\effectcolor}{\mathop{\leq_{\mathit{eff} } } } \;  \textcolor{\effectcolor}{\phi}_{{\mathrm{2}}}  } }%
}{
 \Gamma   \vdash_{\mathit{effn} }   \ottmv{x}  :^{\textcolor{\effectcolor}{ \textcolor{\effectcolor}{\phi}_{{\mathrm{2}}} } }  \tau }{%
{\ottdrulename{lam\_effn\_var}}{}%
}}

\newcommand{\ottdrulelamXXeffnXXabs}[1]{\ottdrule[#1]{%
\ottpremise{  \Gamma   \mathop{,}    \ottmv{x}  :^{\textcolor{\effectcolor}{ \textcolor{\effectcolor}{\phi}_{{\mathrm{1}}} } }  \tau_{{\mathrm{1}}}     \vdash_{\mathit{effn} }   \ottnt{e}  :^{\textcolor{\effectcolor}{ \textcolor{\effectcolor}{\phi}_{{\mathrm{2}}} } }  \tau_{{\mathrm{2}}} }%
}{
 \Gamma   \vdash_{\mathit{effn} }    \lambda  \ottmv{x} . \ottnt{e}   :^{\textcolor{\effectcolor}{ \textcolor{\effectcolor}{\phi} } }   \tau_{{\mathrm{1}}} ^{\textcolor{\effectcolor}{ \textcolor{\effectcolor}{\phi}_{{\mathrm{1}}} } }  \to   \tau_{{\mathrm{2}}} ^{\textcolor{\effectcolor}{ \textcolor{\effectcolor}{\phi}_{{\mathrm{2}}} } }  }{%
{\ottdrulename{lam\_effn\_abs}}{}%
}}

\newcommand{\ottdrulelamXXeffnXXapp}[1]{\ottdrule[#1]{%
\ottpremise{ \Gamma   \vdash_{\mathit{effn} }   \ottnt{e_{{\mathrm{1}}}}  :^{\textcolor{\effectcolor}{ \textcolor{\effectcolor}{\phi}_{{\mathrm{1}}} } }   \tau_{{\mathrm{1}}} ^{\textcolor{\effectcolor}{ \textcolor{\effectcolor}{\phi}_{{\mathrm{3}}} } }  \to   \tau_{{\mathrm{2}}} ^{\textcolor{\effectcolor}{ \textcolor{\effectcolor}{\phi}_{{\mathrm{4}}} } }  }%
\ottpremise{ \Gamma   \vdash_{\mathit{effn} }   \ottnt{e_{{\mathrm{2}}}}  :^{\textcolor{\effectcolor}{ \textcolor{\effectcolor}{\phi}_{{\mathrm{2}}} } }  \tau_{{\mathrm{1}}} }%
\ottpremise{ \textcolor{\effectcolor}{ \textcolor{\effectcolor}{\phi}_{{\mathrm{2}}} \;  \textcolor{\effectcolor}{\mathop{\leq_{\mathit{eff} } } } \;  \textcolor{\effectcolor}{\phi}_{{\mathrm{3}}}  } }%
\ottpremise{ \textcolor{\effectcolor}{  \textcolor{\effectcolor}{ \textcolor{\effectcolor}{\phi}_{{\mathrm{1}}}  \cdot  \textcolor{\effectcolor}{\phi}_{{\mathrm{4}}} }  \;  \textcolor{\effectcolor}{\mathop{\leq_{\mathit{eff} } } } \;  \textcolor{\effectcolor}{\phi}  } }%
}{
 \Gamma   \vdash_{\mathit{effn} }   \ottnt{e_{{\mathrm{1}}}} \, \ottnt{e_{{\mathrm{2}}}}  :^{\textcolor{\effectcolor}{ \textcolor{\effectcolor}{\phi} } }  \tau_{{\mathrm{2}}} }{%
{\ottdrulename{lam\_effn\_app}}{}%
}}

\newcommand{\ottdrulelamXXeffnXXunit}[1]{\ottdrule[#1]{%
}{
 \Gamma   \vdash_{\mathit{effn} }   \ottsym{()}  :^{\textcolor{\effectcolor}{ \textcolor{\effectcolor}{\phi} } }  \ottkw{unit} }{%
{\ottdrulename{lam\_effn\_unit}}{}%
}}

\newcommand{\ottdrulelamXXeffnXXsequence}[1]{\ottdrule[#1]{%
\ottpremise{ \Gamma   \vdash_{\mathit{effn} }   \ottnt{e_{{\mathrm{1}}}}  :^{\textcolor{\effectcolor}{ \textcolor{\effectcolor}{\phi}_{{\mathrm{1}}} } }  \ottkw{unit} }%
\ottpremise{ \Gamma   \vdash_{\mathit{effn} }   \ottnt{e_{{\mathrm{2}}}}  :^{\textcolor{\effectcolor}{ \textcolor{\effectcolor}{\phi}_{{\mathrm{2}}} } }  \tau }%
\ottpremise{ \textcolor{\effectcolor}{  \textcolor{\effectcolor}{ \textcolor{\effectcolor}{\phi}_{{\mathrm{1}}}  \cdot  \textcolor{\effectcolor}{\phi}_{{\mathrm{2}}} }  \;  \textcolor{\effectcolor}{\mathop{\leq_{\mathit{eff} } } } \;  \textcolor{\effectcolor}{\phi}  } }%
}{
 \Gamma   \vdash_{\mathit{effn} }    \ottnt{e_{{\mathrm{1}}}}  ;  \ottnt{e_{{\mathrm{2}}}}   :^{\textcolor{\effectcolor}{ \textcolor{\effectcolor}{\phi} } }  \tau }{%
{\ottdrulename{lam\_effn\_sequence}}{}%
}}

\newcommand{\ottdrulelamXXeffnXXpair}[1]{\ottdrule[#1]{%
\ottpremise{ \Gamma   \vdash_{\mathit{effn} }   \ottnt{e_{{\mathrm{1}}}}  :^{\textcolor{\effectcolor}{ \textcolor{\effectcolor}{\phi} } }  \tau_{{\mathrm{1}}} }%
\ottpremise{ \Gamma   \vdash_{\mathit{effn} }   \ottnt{e_{{\mathrm{2}}}}  :^{\textcolor{\effectcolor}{ \textcolor{\effectcolor}{\phi} } }  \tau_{{\mathrm{2}}} }%
}{
 \Gamma   \vdash_{\mathit{effn} }   \ottsym{(}  \ottnt{e_{{\mathrm{1}}}}  \ottsym{,}  \ottnt{e_{{\mathrm{2}}}}  \ottsym{)}  :^{\textcolor{\effectcolor}{ \textcolor{\effectcolor}{\phi} } }   \tau_{{\mathrm{1}}}  \otimes  \tau_{{\mathrm{2}}}  }{%
{\ottdrulename{lam\_effn\_pair}}{}%
}}

\newcommand{\ottdrulelamXXeffnXXsplit}[1]{\ottdrule[#1]{%
\ottpremise{ \Gamma   \vdash_{\mathit{effn} }   \ottnt{e_{{\mathrm{1}}}}  :^{\textcolor{\effectcolor}{ \textcolor{\effectcolor}{\phi} } }   \tau_{{\mathrm{1}}}  \otimes  \tau_{{\mathrm{2}}}  }%
\ottpremise{   \Gamma   \mathop{,}    \ottmv{x_{{\mathrm{1}}}}  :^{\textcolor{\effectcolor}{  \textcolor{\effectcolor}{\varepsilon}  } }  \tau_{{\mathrm{1}}}     \mathop{,}    \ottmv{x_{{\mathrm{2}}}}  :^{\textcolor{\effectcolor}{  \textcolor{\effectcolor}{\varepsilon}  } }  \tau_{{\mathrm{2}}}     \vdash_{\mathit{effn} }   \ottnt{e_{{\mathrm{2}}}}  :^{\textcolor{\effectcolor}{ \textcolor{\effectcolor}{\phi} } }  \tau }%
}{
 \Gamma   \vdash_{\mathit{effn} }    \ottkw{let}\; ( \ottmv{x_{{\mathrm{1}}}} ,  \ottmv{x_{{\mathrm{2}}}} ) =  \ottnt{e_{{\mathrm{1}}}} \; \ottkw{in}\;  \ottnt{e_{{\mathrm{2}}}}   :^{\textcolor{\effectcolor}{ \textcolor{\effectcolor}{\phi} } }  \tau }{%
{\ottdrulename{lam\_effn\_split}}{}%
}}

\newcommand{\ottdrulelamXXeffnXXwith}[1]{\ottdrule[#1]{%
\ottpremise{ \Gamma   \vdash_{\mathit{effn} }   \ottnt{e_{{\mathrm{1}}}}  :^{\textcolor{\effectcolor}{ \textcolor{\effectcolor}{\phi}_{{\mathrm{1}}} } }  \tau_{{\mathrm{1}}} }%
\ottpremise{ \Gamma   \vdash_{\mathit{effn} }   \ottnt{e_{{\mathrm{2}}}}  :^{\textcolor{\effectcolor}{ \textcolor{\effectcolor}{\phi}_{{\mathrm{2}}} } }  \tau_{{\mathrm{2}}} }%
}{
 \Gamma   \vdash_{\mathit{effn} }   \ottsym{(}  \ottnt{e_{{\mathrm{1}}}}  \ottsym{,}  \ottnt{e_{{\mathrm{2}}}}  \ottsym{)}  :^{\textcolor{\effectcolor}{ \textcolor{\effectcolor}{\phi} } }  \ottsym{(}   \tau_{{\mathrm{1}}} ^{\textcolor{\effectcolor}{ \textcolor{\effectcolor}{\phi}_{{\mathrm{1}}} } }  \mathop{\&}   \tau_{{\mathrm{2}}} ^{\textcolor{\effectcolor}{ \textcolor{\effectcolor}{\phi}_{{\mathrm{2}}} } }   \ottsym{)} }{%
{\ottdrulename{lam\_effn\_with}}{}%
}}

\newcommand{\ottdrulelamXXeffnXXfst}[1]{\ottdrule[#1]{%
\ottpremise{ \Gamma   \vdash_{\mathit{effn} }   \ottnt{e}  :^{\textcolor{\effectcolor}{ \textcolor{\effectcolor}{\phi} } }  \ottsym{(}   \tau_{{\mathrm{1}}} ^{\textcolor{\effectcolor}{ \textcolor{\effectcolor}{\phi}_{{\mathrm{1}}} } }  \mathop{\&}   \tau_{{\mathrm{2}}} ^{\textcolor{\effectcolor}{ \textcolor{\effectcolor}{\phi}_{{\mathrm{2}}} } }   \ottsym{)} }%
\ottpremise{ \textcolor{\effectcolor}{ \textcolor{\effectcolor}{\phi}_{{\mathrm{1}}} \;  \textcolor{\effectcolor}{\mathop{\leq_{\mathit{eff} } } } \;  \textcolor{\effectcolor}{\phi}  } }%
}{
 \Gamma   \vdash_{\mathit{effn} }   \ottnt{e}  \ottsym{.}  \ottsym{1}  :^{\textcolor{\effectcolor}{ \textcolor{\effectcolor}{\phi} } }  \tau_{{\mathrm{1}}} }{%
{\ottdrulename{lam\_effn\_fst}}{}%
}}

\newcommand{\ottdrulelamXXeffnXXsnd}[1]{\ottdrule[#1]{%
\ottpremise{ \Gamma   \vdash_{\mathit{effn} }   \ottnt{e}  :^{\textcolor{\effectcolor}{ \textcolor{\effectcolor}{\phi} } }  \ottsym{(}   \tau_{{\mathrm{1}}} ^{\textcolor{\effectcolor}{ \textcolor{\effectcolor}{\phi}_{{\mathrm{1}}} } }  \mathop{\&}   \tau_{{\mathrm{2}}} ^{\textcolor{\effectcolor}{ \textcolor{\effectcolor}{\phi}_{{\mathrm{2}}} } }   \ottsym{)} }%
\ottpremise{ \textcolor{\effectcolor}{ \textcolor{\effectcolor}{\phi}_{{\mathrm{2}}} \;  \textcolor{\effectcolor}{\mathop{\leq_{\mathit{eff} } } } \;  \textcolor{\effectcolor}{\phi}  } }%
}{
 \Gamma   \vdash_{\mathit{effn} }   \ottnt{e}  \ottsym{.}  \ottsym{2}  :^{\textcolor{\effectcolor}{ \textcolor{\effectcolor}{\phi} } }  \tau_{{\mathrm{1}}} }{%
{\ottdrulename{lam\_effn\_snd}}{}%
}}

\newcommand{\ottdrulelamXXeffnXXinl}[1]{\ottdrule[#1]{%
\ottpremise{ \Gamma   \vdash_{\mathit{effn} }   \ottnt{e}  :^{\textcolor{\effectcolor}{ \textcolor{\effectcolor}{\phi} } }  \tau_{{\mathrm{1}}} }%
}{
 \Gamma   \vdash_{\mathit{effn} }   \ottkw{inl} \, \ottnt{e}  :^{\textcolor{\effectcolor}{ \textcolor{\effectcolor}{\phi} } }  \tau_{{\mathrm{1}}}  \ottsym{+}  \tau_{{\mathrm{2}}} }{%
{\ottdrulename{lam\_effn\_inl}}{}%
}}

\newcommand{\ottdrulelamXXeffnXXinr}[1]{\ottdrule[#1]{%
\ottpremise{ \Gamma   \vdash_{\mathit{effn} }   \ottnt{e}  :^{\textcolor{\effectcolor}{ \textcolor{\effectcolor}{\phi} } }  \tau_{{\mathrm{2}}} }%
}{
 \Gamma   \vdash_{\mathit{effn} }   \ottkw{inr} \, \ottnt{e}  :^{\textcolor{\effectcolor}{ \textcolor{\effectcolor}{\phi} } }  \tau_{{\mathrm{1}}}  \ottsym{+}  \tau_{{\mathrm{2}}} }{%
{\ottdrulename{lam\_effn\_inr}}{}%
}}

\newcommand{\ottdrulelamXXeffnXXcase}[1]{\ottdrule[#1]{%
\ottpremise{ \Gamma   \vdash_{\mathit{effn} }   \ottnt{e}  :^{\textcolor{\effectcolor}{ \textcolor{\effectcolor}{\phi} } }  \tau_{{\mathrm{1}}}  \ottsym{+}  \tau_{{\mathrm{2}}} }%
\ottpremise{  \Gamma   \mathop{,}    \ottmv{x}  :^{\textcolor{\effectcolor}{  \textcolor{\effectcolor}{\varepsilon}  } }  \tau_{{\mathrm{1}}}     \vdash_{\mathit{effn} }   \ottnt{e_{{\mathrm{1}}}}  :^{\textcolor{\effectcolor}{ \textcolor{\effectcolor}{\phi} } }  \tau }%
\ottpremise{  \Gamma   \mathop{,}    \ottmv{x}  :^{\textcolor{\effectcolor}{  \textcolor{\effectcolor}{\varepsilon}  } }  \tau_{{\mathrm{2}}}     \vdash_{\mathit{effn} }   \ottnt{e_{{\mathrm{2}}}}  :^{\textcolor{\effectcolor}{ \textcolor{\effectcolor}{\phi} } }  \tau }%
}{
 \Gamma   \vdash_{\mathit{effn} }    \ottkw{case}\;  \ottnt{e} \; \ottkw{of}\; \ottkw{inl}\;  \ottmv{x_{{\mathrm{1}}}}  \rightarrow  \ottnt{e_{{\mathrm{1}}}}  \ottkw{;} \ottkw{inr}\;  \ottmv{x_{{\mathrm{2}}}}  \rightarrow  \ottnt{e_{{\mathrm{2}}}}   :^{\textcolor{\effectcolor}{ \textcolor{\effectcolor}{\phi} } }  \tau }{%
{\ottdrulename{lam\_effn\_case}}{}%
}}

\newcommand{\ottdrulelamXXeffnXXtick}[1]{\ottdrule[#1]{%
}{
 \Gamma   \vdash_{\mathit{effn} }    \textcolor{\effectcolor}{ \ottkw{tick} }   :^{\textcolor{\effectcolor}{  \textcolor{\effectcolor}{\ottkw{Tick} }  } }  \ottkw{unit} }{%
{\ottdrulename{lam\_effn\_tick}}{}%
}}

\newcommand{\ottdrulelamXXeffnXXsub}[1]{\ottdrule[#1]{%
\ottpremise{ \Gamma   \vdash_{\mathit{effn} }   \ottnt{e}  :^{\textcolor{\effectcolor}{ \textcolor{\effectcolor}{\phi}_{{\mathrm{1}}} } }  \tau }%
\ottpremise{ \textcolor{\effectcolor}{ \textcolor{\effectcolor}{\phi}_{{\mathrm{1}}} \;  \textcolor{\effectcolor}{\mathop{\leq_{\mathit{eff} } } } \;  \textcolor{\effectcolor}{\phi}_{{\mathrm{2}}}  } }%
}{
 \Gamma   \vdash_{\mathit{effn} }   \ottnt{e}  :^{\textcolor{\effectcolor}{ \textcolor{\effectcolor}{\phi}_{{\mathrm{2}}} } }  \tau }{%
{\ottdrulename{lam\_effn\_sub}}{}%
}}

\newcommand{\ottdefnlamXXeffnXXtyping}[1]{\begin{ottdefnblock}[#1]{$ \Gamma   \vdash_{\mathit{effn} }   \ottnt{e}  :^{\textcolor{\effectcolor}{ \textcolor{\effectcolor}{\phi} } }  \tau $}{\ottcom{CBN Effect calculus, not realistic but why not?}}
\ottusedrule{\ottdrulelamXXeffnXXvar{}}
\ottusedrule{\ottdrulelamXXeffnXXabs{}}
\ottusedrule{\ottdrulelamXXeffnXXapp{}}
\ottusedrule{\ottdrulelamXXeffnXXunit{}}
\ottusedrule{\ottdrulelamXXeffnXXsequence{}}
\ottusedrule{\ottdrulelamXXeffnXXpair{}}
\ottusedrule{\ottdrulelamXXeffnXXsplit{}}
\ottusedrule{\ottdrulelamXXeffnXXwith{}}
\ottusedrule{\ottdrulelamXXeffnXXfst{}}
\ottusedrule{\ottdrulelamXXeffnXXsnd{}}
\ottusedrule{\ottdrulelamXXeffnXXinl{}}
\ottusedrule{\ottdrulelamXXeffnXXinr{}}
\ottusedrule{\ottdrulelamXXeffnXXcase{}}
\ottusedrule{\ottdrulelamXXeffnXXtick{}}
\ottusedrule{\ottdrulelamXXeffnXXsub{}}
\end{ottdefnblock}}

\newcommand{\ottdefnsJEff}{
\ottdefnvalXXeffXXtyping{}\ottdefncompXXeffXXtyping{}\ottdefnlamXXeffXXtyping{}\ottdefnlamXXeffnXXtyping{}}

\newcommand{\ottdrulelamXXmonXXtick}[1]{\ottdrule[#1]{%
}{
\Gamma  \vdash_{\mathit{mon} }   \textcolor{\effectcolor}{ \ottkw{tick} }   \ottsym{:}   \ottkw{T}_{\textcolor{\effectcolor}{  \textcolor{\effectcolor}{\ottkw{Tick} }  } }\;  \ottkw{unit} }{%
{\ottdrulename{lam\_mon\_tick}}{}%
}}

\newcommand{\ottdrulelamXXmonXXreturn}[1]{\ottdrule[#1]{%
\ottpremise{\Gamma  \vdash_{\mathit{mon} }  \ottnt{e}  \ottsym{:}  \tau}%
}{
\Gamma  \vdash_{\mathit{mon} }  \ottkw{return} \, \ottnt{e}  \ottsym{:}   \ottkw{T}_{\textcolor{\effectcolor}{  \textcolor{\effectcolor}{\varepsilon}  } }\;  \tau }{%
{\ottdrulename{lam\_mon\_return}}{}%
}}

\newcommand{\ottdrulelamXXmonXXbind}[1]{\ottdrule[#1]{%
\ottpremise{\Gamma  \vdash_{\mathit{mon} }  \ottnt{e_{{\mathrm{1}}}}  \ottsym{:}   \ottkw{T}_{\textcolor{\effectcolor}{ \textcolor{\effectcolor}{\phi}_{{\mathrm{1}}} } }\;  \tau_{{\mathrm{1}}} }%
\ottpremise{ \Gamma   \mathop{,}   \ottmv{x}  \ottsym{:}  \tau_{{\mathrm{1}}}   \vdash_{\mathit{mon} }  \ottnt{e_{{\mathrm{2}}}}  \ottsym{:}   \ottkw{T}_{\textcolor{\effectcolor}{ \textcolor{\effectcolor}{\phi}_{{\mathrm{2}}} } }\;  \tau_{{\mathrm{2}}} }%
}{
\Gamma  \vdash_{\mathit{mon} }  \ottkw{bind} \, \ottmv{x}  \ottsym{=}  \ottnt{e_{{\mathrm{1}}}} \, \ottkw{in} \, \ottnt{e_{{\mathrm{2}}}}  \ottsym{:}   \ottkw{T}_{\textcolor{\effectcolor}{  \textcolor{\effectcolor}{ \textcolor{\effectcolor}{\phi}_{{\mathrm{1}}}  \cdot  \textcolor{\effectcolor}{\phi}_{{\mathrm{2}}} }  } }\;  \tau_{{\mathrm{2}}} }{%
{\ottdrulename{lam\_mon\_bind}}{}%
}}

\newcommand{\ottdrulelamXXmonXXcoerce}[1]{\ottdrule[#1]{%
\ottpremise{\Gamma  \vdash_{\mathit{mon} }  \ottnt{e}  \ottsym{:}   \ottkw{T}_{\textcolor{\effectcolor}{ \textcolor{\effectcolor}{\phi}_{{\mathrm{1}}} } }\;  \tau }%
\ottpremise{ \textcolor{\effectcolor}{ \textcolor{\effectcolor}{\phi}_{{\mathrm{1}}} \;  \textcolor{\effectcolor}{\mathop{\leq_{\mathit{eff} } } } \;  \textcolor{\effectcolor}{\phi}_{{\mathrm{2}}}  } }%
}{
\Gamma  \vdash_{\mathit{mon} }  \ottkw{coerce} \, \ottnt{e}  \ottsym{:}   \ottkw{T}_{\textcolor{\effectcolor}{ \textcolor{\effectcolor}{\phi}_{{\mathrm{2}}} } }\;  \tau }{%
{\ottdrulename{lam\_mon\_coerce}}{}%
}}

\newcommand{\ottdrulelamXXmonXXvar}[1]{\ottdrule[#1]{%
\ottpremise{\ottmv{x}  \ottsym{:}  \tau \, \in \, \Gamma}%
}{
\Gamma  \vdash_{\mathit{mon} }  \ottmv{x}  \ottsym{:}  \tau}{%
{\ottdrulename{lam\_mon\_var}}{}%
}}

\newcommand{\ottdrulelamXXmonXXabs}[1]{\ottdrule[#1]{%
\ottpremise{ \Gamma   \mathop{,}   \ottmv{x}  \ottsym{:}  \tau_{{\mathrm{1}}}   \vdash_{\mathit{mon} }  \ottnt{e}  \ottsym{:}  \tau_{{\mathrm{2}}}}%
}{
\Gamma  \vdash_{\mathit{mon} }   \lambda  \ottmv{x} . \ottnt{e}   \ottsym{:}  \tau_{{\mathrm{1}}}  \to  \tau_{{\mathrm{2}}}}{%
{\ottdrulename{lam\_mon\_abs}}{}%
}}

\newcommand{\ottdrulelamXXmonXXapp}[1]{\ottdrule[#1]{%
\ottpremise{\Gamma  \vdash_{\mathit{mon} }  \ottnt{e_{{\mathrm{1}}}}  \ottsym{:}  \tau_{{\mathrm{1}}}  \to  \tau_{{\mathrm{2}}}}%
\ottpremise{\Gamma  \vdash_{\mathit{mon} }  \ottnt{e_{{\mathrm{2}}}}  \ottsym{:}  \tau_{{\mathrm{1}}}}%
}{
\Gamma  \vdash_{\mathit{mon} }  \ottnt{e_{{\mathrm{1}}}} \, \ottnt{e_{{\mathrm{2}}}}  \ottsym{:}  \tau_{{\mathrm{2}}}}{%
{\ottdrulename{lam\_mon\_app}}{}%
}}

\newcommand{\ottdrulelamXXmonXXunit}[1]{\ottdrule[#1]{%
}{
\Gamma  \vdash_{\mathit{mon} }  \ottsym{()}  \ottsym{:}  \ottkw{unit}}{%
{\ottdrulename{lam\_mon\_unit}}{}%
}}

\newcommand{\ottdrulelamXXmonXXsequence}[1]{\ottdrule[#1]{%
\ottpremise{\Gamma  \vdash_{\mathit{mon} }  \ottnt{e_{{\mathrm{1}}}}  \ottsym{:}  \ottkw{unit}}%
\ottpremise{\Gamma  \vdash_{\mathit{mon} }  \ottnt{e_{{\mathrm{2}}}}  \ottsym{:}  \tau}%
}{
\Gamma  \vdash_{\mathit{mon} }   \ottnt{e_{{\mathrm{1}}}}  ;  \ottnt{e_{{\mathrm{2}}}}   \ottsym{:}  \tau}{%
{\ottdrulename{lam\_mon\_sequence}}{}%
}}

\newcommand{\ottdrulelamXXmonXXpair}[1]{\ottdrule[#1]{%
\ottpremise{\Gamma  \vdash_{\mathit{mon} }  \ottnt{e_{{\mathrm{1}}}}  \ottsym{:}  \tau_{{\mathrm{1}}}}%
\ottpremise{\Gamma  \vdash_{\mathit{mon} }  \ottnt{e_{{\mathrm{2}}}}  \ottsym{:}  \tau_{{\mathrm{2}}}}%
}{
\Gamma  \vdash_{\mathit{mon} }  \ottsym{(}  \ottnt{e_{{\mathrm{1}}}}  \ottsym{,}  \ottnt{e_{{\mathrm{2}}}}  \ottsym{)}  \ottsym{:}   \tau_{{\mathrm{1}}}  \otimes  \tau_{{\mathrm{2}}} }{%
{\ottdrulename{lam\_mon\_pair}}{}%
}}

\newcommand{\ottdrulelamXXmonXXsplit}[1]{\ottdrule[#1]{%
\ottpremise{\Gamma  \vdash_{\mathit{mon} }  \ottnt{e_{{\mathrm{1}}}}  \ottsym{:}   \tau_{{\mathrm{1}}}  \otimes  \tau_{{\mathrm{2}}} }%
\ottpremise{  \Gamma   \mathop{,}   \ottmv{x_{{\mathrm{1}}}}  \ottsym{:}  \tau_{{\mathrm{1}}}    \mathop{,}   \ottmv{x_{{\mathrm{2}}}}  \ottsym{:}  \tau_{{\mathrm{2}}}   \vdash_{\mathit{mon} }  \ottnt{e_{{\mathrm{2}}}}  \ottsym{:}  \tau}%
}{
\Gamma  \vdash_{\mathit{mon} }   \ottkw{let}\; ( \ottmv{x_{{\mathrm{1}}}} ,  \ottmv{x_{{\mathrm{2}}}} ) =  \ottnt{e_{{\mathrm{1}}}} \; \ottkw{in}\;  \ottnt{e_{{\mathrm{2}}}}   \ottsym{:}  \tau}{%
{\ottdrulename{lam\_mon\_split}}{}%
}}

\newcommand{\ottdrulelamXXmonXXwith}[1]{\ottdrule[#1]{%
\ottpremise{\Gamma  \vdash_{\mathit{mon} }  \ottnt{e_{{\mathrm{1}}}}  \ottsym{:}  \tau_{{\mathrm{1}}}}%
\ottpremise{\Gamma  \vdash_{\mathit{mon} }  \ottnt{e_{{\mathrm{2}}}}  \ottsym{:}  \tau_{{\mathrm{2}}}}%
}{
\Gamma  \vdash_{\mathit{mon} }   \langle  \ottnt{e_{{\mathrm{1}}}} , \ottnt{e_{{\mathrm{2}}}}  \rangle   \ottsym{:}   \tau_{{\mathrm{1}}}  \mathop{\&}   \tau_{{\mathrm{2}}} }{%
{\ottdrulename{lam\_mon\_with}}{}%
}}

\newcommand{\ottdrulelamXXmonXXfst}[1]{\ottdrule[#1]{%
\ottpremise{\Gamma  \vdash_{\mathit{mon} }  \ottnt{e}  \ottsym{:}   \tau_{{\mathrm{1}}}  \mathop{\&}   \tau_{{\mathrm{2}}} }%
}{
\Gamma  \vdash_{\mathit{mon} }  \ottnt{e}  \ottsym{.}  \ottsym{1}  \ottsym{:}  \tau_{{\mathrm{1}}}}{%
{\ottdrulename{lam\_mon\_fst}}{}%
}}

\newcommand{\ottdrulelamXXmonXXsnd}[1]{\ottdrule[#1]{%
\ottpremise{\Gamma  \vdash_{\mathit{mon} }  \ottnt{e}  \ottsym{:}   \tau_{{\mathrm{1}}}  \mathop{\&}   \tau_{{\mathrm{2}}} }%
}{
\Gamma  \vdash_{\mathit{mon} }  \ottnt{e}  \ottsym{.}  \ottsym{1}  \ottsym{:}  \tau_{{\mathrm{2}}}}{%
{\ottdrulename{lam\_mon\_snd}}{}%
}}

\newcommand{\ottdrulelamXXmonXXinl}[1]{\ottdrule[#1]{%
\ottpremise{\Gamma  \vdash_{\mathit{mon} }  \ottnt{e}  \ottsym{:}  \tau_{{\mathrm{1}}}}%
}{
\Gamma  \vdash_{\mathit{mon} }  \ottkw{inl} \, \ottnt{e}  \ottsym{:}  \tau_{{\mathrm{1}}}  \ottsym{+}  \tau_{{\mathrm{2}}}}{%
{\ottdrulename{lam\_mon\_inl}}{}%
}}

\newcommand{\ottdrulelamXXmonXXinr}[1]{\ottdrule[#1]{%
\ottpremise{\Gamma  \vdash_{\mathit{mon} }  \ottnt{e}  \ottsym{:}  \tau_{{\mathrm{2}}}}%
}{
\Gamma  \vdash_{\mathit{mon} }  \ottkw{inr} \, \ottnt{e}  \ottsym{:}  \tau_{{\mathrm{1}}}  \ottsym{+}  \tau_{{\mathrm{2}}}}{%
{\ottdrulename{lam\_mon\_inr}}{}%
}}

\newcommand{\ottdrulelamXXmonXXcase}[1]{\ottdrule[#1]{%
\ottpremise{\Gamma  \vdash_{\mathit{mon} }  \ottnt{e}  \ottsym{:}  \tau_{{\mathrm{1}}}  \ottsym{+}  \tau_{{\mathrm{2}}}}%
\ottpremise{ \Gamma   \mathop{,}   \ottmv{x}  \ottsym{:}  \tau_{{\mathrm{1}}}   \vdash_{\mathit{mon} }  \ottnt{e_{{\mathrm{1}}}}  \ottsym{:}  \tau}%
\ottpremise{ \Gamma   \mathop{,}   \ottmv{x}  \ottsym{:}  \tau_{{\mathrm{2}}}   \vdash_{\mathit{mon} }  \ottnt{e_{{\mathrm{2}}}}  \ottsym{:}  \tau}%
}{
\Gamma  \vdash_{\mathit{mon} }   \ottkw{case}\;  \ottnt{e} \; \ottkw{of}\; \ottkw{inl}\;  \ottmv{x_{{\mathrm{1}}}}  \rightarrow  \ottnt{e_{{\mathrm{1}}}}  \ottkw{;} \ottkw{inr}\;  \ottmv{x_{{\mathrm{2}}}}  \rightarrow  \ottnt{e_{{\mathrm{2}}}}   \ottsym{:}  \tau}{%
{\ottdrulename{lam\_mon\_case}}{}%
}}

\newcommand{\ottdefnmonXXtyping}[1]{\begin{ottdefnblock}[#1]{$\Gamma  \vdash_{\mathit{mon} }  \ottnt{e}  \ottsym{:}  \tau$}{\ottcom{Monadic calculus from Wadler / Thiemann, CBN semantics}}
\ottusedrule{\ottdrulelamXXmonXXtick{}}
\ottusedrule{\ottdrulelamXXmonXXreturn{}}
\ottusedrule{\ottdrulelamXXmonXXbind{}}
\ottusedrule{\ottdrulelamXXmonXXcoerce{}}
\ottusedrule{\ottdrulelamXXmonXXvar{}}
\ottusedrule{\ottdrulelamXXmonXXabs{}}
\ottusedrule{\ottdrulelamXXmonXXapp{}}
\ottusedrule{\ottdrulelamXXmonXXunit{}}
\ottusedrule{\ottdrulelamXXmonXXsequence{}}
\ottusedrule{\ottdrulelamXXmonXXpair{}}
\ottusedrule{\ottdrulelamXXmonXXsplit{}}
\ottusedrule{\ottdrulelamXXmonXXwith{}}
\ottusedrule{\ottdrulelamXXmonXXfst{}}
\ottusedrule{\ottdrulelamXXmonXXsnd{}}
\ottusedrule{\ottdrulelamXXmonXXinl{}}
\ottusedrule{\ottdrulelamXXmonXXinr{}}
\ottusedrule{\ottdrulelamXXmonXXcase{}}
\end{ottdefnblock}}

\newcommand{\ottdefnsJMonEff}{
\ottdefnmonXXtyping{}}


\newcommand{\ottdefnsJContextOps}{
}

\newcommand{\ottdrulecoeffXXvar}[1]{\ottdrule[#1]{%
}{
   \textcolor{\coeffectcolor}{  \textcolor{\coeffectcolor}{\overline{0} }  }\! \cdot \! \Gamma_{{\mathrm{1}}}    \mathop{,}    \ottmv{x}  :^{\textcolor{\coeffectcolor}{  \textcolor{\coeffectcolor}{1}  } }  \ottnt{A}     \mathop{,}    \textcolor{\coeffectcolor}{  \textcolor{\coeffectcolor}{\overline{0} }  }\! \cdot \! \Gamma_{{\mathrm{2}}}    \vdash_{\mathit{coeff} }  \ottmv{x}  \ottsym{:}  \ottnt{A}}{%
{\ottdrulename{coeff\_var}}{}%
}}

\newcommand{\ottdrulecoeffXXthunk}[1]{\ottdrule[#1]{%
\ottpremise{ \textcolor{\coeffectcolor}{ \textcolor{\coeffectcolor}{\gamma} }\! \cdot \! \Gamma   \vdash_{\mathit{coeff} }  \ottnt{M}  \ottsym{:}  \ottnt{B}}%
}{
 \textcolor{\coeffectcolor}{ \textcolor{\coeffectcolor}{\gamma} }\! \cdot \! \Gamma   \vdash_{\mathit{coeff} }  \ottsym{\{}  \ottnt{M}  \ottsym{\}}  \ottsym{:}  \ottkw{U} \, \ottnt{B}}{%
{\ottdrulename{coeff\_thunk}}{}%
}}

\newcommand{\ottdrulecoeffXXunit}[1]{\ottdrule[#1]{%
}{
 \textcolor{\coeffectcolor}{  \textcolor{\coeffectcolor}{\overline{0} }  }\! \cdot \! \Gamma   \vdash_{\mathit{coeff} }  \ottsym{()}  \ottsym{:}  \ottkw{unit}}{%
{\ottdrulename{coeff\_unit}}{}%
}}

\newcommand{\ottdrulecoeffXXpair}[1]{\ottdrule[#1]{%
\ottpremise{ \textcolor{\coeffectcolor}{ \textcolor{\coeffectcolor}{\gamma}_{{\mathrm{1}}} }\! \cdot \! \Gamma   \vdash_{\mathit{coeff} }  \ottnt{V_{{\mathrm{1}}}}  \ottsym{:}  \ottnt{A_{{\mathrm{1}}}}}%
\ottpremise{ \textcolor{\coeffectcolor}{ \textcolor{\coeffectcolor}{\gamma}_{{\mathrm{2}}} }\! \cdot \! \Gamma   \vdash_{\mathit{coeff} }  \ottnt{V_{{\mathrm{2}}}}  \ottsym{:}  \ottnt{A_{{\mathrm{2}}}}}%
}{
 \textcolor{\coeffectcolor}{  \textcolor{\coeffectcolor}{ \textcolor{\coeffectcolor}{\gamma}_{{\mathrm{1}}} \ottsym{+} \textcolor{\coeffectcolor}{\gamma}_{{\mathrm{2}}} }  }\! \cdot \! \Gamma   \vdash_{\mathit{coeff} }  \ottsym{(}  \ottnt{V_{{\mathrm{1}}}}  \ottsym{,}  \ottnt{V_{{\mathrm{2}}}}  \ottsym{)}  \ottsym{:}   \ottnt{A_{{\mathrm{1}}}} \times \ottnt{A_{{\mathrm{2}}}} }{%
{\ottdrulename{coeff\_pair}}{}%
}}

\newcommand{\ottdrulecoeffXXinl}[1]{\ottdrule[#1]{%
\ottpremise{ \textcolor{\coeffectcolor}{ \textcolor{\coeffectcolor}{\gamma} }\! \cdot \! \Gamma   \vdash_{\mathit{coeff} }  \ottnt{V}  \ottsym{:}  \ottnt{A_{{\mathrm{1}}}}}%
}{
 \textcolor{\coeffectcolor}{ \textcolor{\coeffectcolor}{\gamma} }\! \cdot \! \Gamma   \vdash_{\mathit{coeff} }  \ottkw{inl} \, \ottnt{V}  \ottsym{:}  \ottnt{A_{{\mathrm{1}}}}  \ottsym{+}  \ottnt{A_{{\mathrm{2}}}}}{%
{\ottdrulename{coeff\_inl}}{}%
}}

\newcommand{\ottdrulecoeffXXinr}[1]{\ottdrule[#1]{%
\ottpremise{ \textcolor{\coeffectcolor}{ \textcolor{\coeffectcolor}{\gamma} }\! \cdot \! \Gamma   \vdash_{\mathit{coeff} }  \ottnt{V}  \ottsym{:}  \ottnt{A_{{\mathrm{2}}}}}%
}{
 \textcolor{\coeffectcolor}{ \textcolor{\coeffectcolor}{\gamma} }\! \cdot \! \Gamma   \vdash_{\mathit{coeff} }  \ottkw{inr} \, \ottnt{V}  \ottsym{:}  \ottnt{A_{{\mathrm{1}}}}  \ottsym{+}  \ottnt{A_{{\mathrm{2}}}}}{%
{\ottdrulename{coeff\_inr}}{}%
}}

\newcommand{\ottdrulecoeffXXvwith}[1]{\ottdrule[#1]{%
\ottpremise{ \textcolor{\coeffectcolor}{ \textcolor{\coeffectcolor}{\gamma} }\! \cdot \! \Gamma   \vdash_{\mathit{coeff} }  \ottnt{V_{{\mathrm{1}}}}  \ottsym{:}  \ottnt{A_{{\mathrm{1}}}}}%
\ottpremise{ \textcolor{\coeffectcolor}{ \textcolor{\coeffectcolor}{\gamma} }\! \cdot \! \Gamma   \vdash_{\mathit{coeff} }  \ottnt{V_{{\mathrm{2}}}}  \ottsym{:}  \ottnt{A_{{\mathrm{2}}}}}%
}{
 \textcolor{\coeffectcolor}{ \textcolor{\coeffectcolor}{\gamma} }\! \cdot \! \Gamma   \vdash_{\mathit{coeff} }   \langle  \ottnt{V_{{\mathrm{1}}}}  ,  \ottnt{V_{{\mathrm{2}}}}  \rangle   \ottsym{:}   \ottnt{A_{{\mathrm{1}}}} \mathop{\&} \ottnt{A_{{\mathrm{2}}}} }{%
{\ottdrulename{coeff\_vwith}}{}%
}}

\newcommand{\ottdrulecoeffXXvfst}[1]{\ottdrule[#1]{%
\ottpremise{ \textcolor{\coeffectcolor}{ \textcolor{\coeffectcolor}{\gamma} }\! \cdot \! \Gamma   \vdash_{\mathit{coeff} }  \ottnt{V}  \ottsym{:}   \ottnt{A_{{\mathrm{1}}}} \mathop{\&} \ottnt{A_{{\mathrm{2}}}} }%
}{
 \textcolor{\coeffectcolor}{ \textcolor{\coeffectcolor}{\gamma} }\! \cdot \! \Gamma   \vdash_{\mathit{coeff} }  \ottnt{V}  \ottsym{.}  \ottsym{1}  \ottsym{:}  \ottnt{A_{{\mathrm{1}}}}}{%
{\ottdrulename{coeff\_vfst}}{}%
}}

\newcommand{\ottdrulecoeffXXvsnd}[1]{\ottdrule[#1]{%
\ottpremise{ \textcolor{\coeffectcolor}{ \textcolor{\coeffectcolor}{\gamma} }\! \cdot \! \Gamma   \vdash_{\mathit{coeff} }  \ottnt{V}  \ottsym{:}   \ottnt{A_{{\mathrm{1}}}} \mathop{\&} \ottnt{A_{{\mathrm{2}}}} }%
}{
 \textcolor{\coeffectcolor}{ \textcolor{\coeffectcolor}{\gamma} }\! \cdot \! \Gamma   \vdash_{\mathit{coeff} }  \ottnt{V}  \ottsym{.}  \ottsym{2}  \ottsym{:}  \ottnt{A_{{\mathrm{2}}}}}{%
{\ottdrulename{coeff\_vsnd}}{}%
}}

\newcommand{\ottdrulecoeffXXvsub}[1]{\ottdrule[#1]{%
\ottpremise{ \textcolor{\coeffectcolor}{ \textcolor{\coeffectcolor}{\gamma}' }\! \cdot \! \Gamma   \vdash_{\mathit{coeff} }  \ottnt{V}  \ottsym{:}  \ottnt{A}}%
\ottpremise{ \textcolor{\coeffectcolor}{ \textcolor{\coeffectcolor}{\gamma} }\; \textcolor{\coeffectcolor}{\mathop{\leq_{\mathit{co} } } } \; \textcolor{\coeffectcolor}{ \textcolor{\coeffectcolor}{\gamma}' } }%
}{
 \textcolor{\coeffectcolor}{ \textcolor{\coeffectcolor}{\gamma} }\! \cdot \! \Gamma   \vdash_{\mathit{coeff} }  \ottnt{V}  \ottsym{:}  \ottnt{A}}{%
{\ottdrulename{coeff\_vsub}}{}%
}}

\newcommand{\ottdefnvalXXcoeffXXtyping}[1]{\begin{ottdefnblock}[#1]{$\Psi  \vdash_{\mathit{coeff} }  \ottnt{V}  \ottsym{:}  \ottnt{A}$}{\ottcom{value typing rules (coeffect)}}
\ottusedrule{\ottdrulecoeffXXvar{}}
\ottusedrule{\ottdrulecoeffXXthunk{}}
\ottusedrule{\ottdrulecoeffXXunit{}}
\ottusedrule{\ottdrulecoeffXXpair{}}
\ottusedrule{\ottdrulecoeffXXinl{}}
\ottusedrule{\ottdrulecoeffXXinr{}}
\ottusedrule{\ottdrulecoeffXXvwith{}}
\ottusedrule{\ottdrulecoeffXXvfst{}}
\ottusedrule{\ottdrulecoeffXXvsnd{}}
\ottusedrule{\ottdrulecoeffXXvsub{}}
\end{ottdefnblock}}

\newcommand{\ottdrulecoeffXXabs}[1]{\ottdrule[#1]{%
\ottpremise{  \textcolor{\coeffectcolor}{ \textcolor{\coeffectcolor}{\gamma} }\! \cdot \! \Gamma    \mathop{,}    \ottmv{x}  :^{\textcolor{\coeffectcolor}{ \textcolor{\coeffectcolor}{q} } }  \ottnt{A}    \vdash_{\mathit{coeff} }  \ottnt{M}  \ottsym{:}  \ottnt{B}}%
\ottpremise{ \textcolor{\coeffectcolor}{ \textcolor{\coeffectcolor}{q}' }\;  \textcolor{\coeffectcolor}{\mathop{\leq_{\mathit{co} } } } \; \textcolor{\coeffectcolor}{ \textcolor{\coeffectcolor}{q} } }%
}{
 \textcolor{\coeffectcolor}{ \textcolor{\coeffectcolor}{\gamma} }\! \cdot \! \Gamma   \vdash_{\mathit{coeff} }   \lambda  \ottmv{x} ^{\textcolor{\coeffectcolor}{ \textcolor{\coeffectcolor}{q} } }. \ottnt{M}   \ottsym{:}   \ottnt{A} ^{\textcolor{\coeffectcolor}{ \textcolor{\coeffectcolor}{q}' } } \rightarrow  \ottnt{B} }{%
{\ottdrulename{coeff\_abs}}{}%
}}

\newcommand{\ottdrulecoeffXXapp}[1]{\ottdrule[#1]{%
\ottpremise{ \textcolor{\coeffectcolor}{ \textcolor{\coeffectcolor}{\gamma}_{{\mathrm{1}}} }\! \cdot \! \Gamma   \vdash_{\mathit{coeff} }  \ottnt{M}  \ottsym{:}   \ottnt{A} ^{\textcolor{\coeffectcolor}{ \textcolor{\coeffectcolor}{q} } } \rightarrow  \ottnt{B} }%
\ottpremise{ \textcolor{\coeffectcolor}{ \textcolor{\coeffectcolor}{\gamma}_{{\mathrm{2}}} }\! \cdot \! \Gamma   \vdash_{\mathit{coeff} }  \ottnt{V}  \ottsym{:}  \ottnt{A}}%
}{
 \textcolor{\coeffectcolor}{  \textcolor{\coeffectcolor}{ \textcolor{\coeffectcolor}{\gamma}_{{\mathrm{1}}} \ottsym{+} \ottsym{(}   \textcolor{\coeffectcolor}{ \textcolor{\coeffectcolor}{q} \cdot \textcolor{\coeffectcolor}{\gamma}_{{\mathrm{2}}} }   \ottsym{)} }  }\! \cdot \! \Gamma   \vdash_{\mathit{coeff} }  \ottnt{M} \, \ottnt{V}  \ottsym{:}  \ottnt{B}}{%
{\ottdrulename{coeff\_app}}{}%
}}

\newcommand{\ottdrulecoeffXXforce}[1]{\ottdrule[#1]{%
\ottpremise{ \textcolor{\coeffectcolor}{ \textcolor{\coeffectcolor}{\gamma} }\! \cdot \! \Gamma   \vdash_{\mathit{coeff} }  \ottnt{V}  \ottsym{:}  \ottkw{U} \, \ottnt{B}}%
}{
 \textcolor{\coeffectcolor}{ \textcolor{\coeffectcolor}{\gamma} }\! \cdot \! \Gamma   \vdash_{\mathit{coeff} }  \ottnt{V}  \ottsym{!}  \ottsym{:}  \ottnt{B}}{%
{\ottdrulename{coeff\_force}}{}%
}}

\newcommand{\ottdrulecoeffXXsplit}[1]{\ottdrule[#1]{%
\ottpremise{ \textcolor{\coeffectcolor}{ \textcolor{\coeffectcolor}{\gamma}_{{\mathrm{1}}} }\! \cdot \! \Gamma   \vdash_{\mathit{coeff} }  \ottnt{V}  \ottsym{:}   \ottnt{A_{{\mathrm{1}}}} \times \ottnt{A_{{\mathrm{2}}}} }%
\ottpremise{    \textcolor{\coeffectcolor}{ \textcolor{\coeffectcolor}{\gamma}_{{\mathrm{2}}} }\! \cdot \! \Gamma    \mathop{,}    \ottmv{x_{{\mathrm{1}}}}  :^{\textcolor{\coeffectcolor}{ \textcolor{\coeffectcolor}{q} } }  \ottnt{A_{{\mathrm{1}}}}      \mathop{,}    \ottmv{x_{{\mathrm{2}}}}  :^{\textcolor{\coeffectcolor}{ \textcolor{\coeffectcolor}{q} } }  \ottnt{A_{{\mathrm{2}}}}    \vdash_{\mathit{coeff} }  \ottnt{N}  \ottsym{:}  \ottnt{B}}%
}{
 \textcolor{\coeffectcolor}{  \textcolor{\coeffectcolor}{ \ottsym{(}   \textcolor{\coeffectcolor}{ \textcolor{\coeffectcolor}{q} \cdot \textcolor{\coeffectcolor}{\gamma}_{{\mathrm{1}}} }   \ottsym{)} \ottsym{+} \textcolor{\coeffectcolor}{\gamma}_{{\mathrm{2}}} }  }\! \cdot \! \Gamma   \vdash_{\mathit{coeff} }   \ottkw{case}_{\textcolor{\coeffectcolor}{ \textcolor{\coeffectcolor}{q} } } \;  \ottnt{V} \; \ottkw{of}\;( \ottmv{x_{{\mathrm{1}}}} , \ottmv{x_{{\mathrm{2}}}} )\; \rightarrow\;  \ottnt{N}   \ottsym{:}  \ottnt{B}}{%
{\ottdrulename{coeff\_split}}{}%
}}

\newcommand{\ottdrulecoeffXXret}[1]{\ottdrule[#1]{%
\ottpremise{ \textcolor{\coeffectcolor}{ \textcolor{\coeffectcolor}{\gamma} }\! \cdot \! \Gamma   \vdash_{\mathit{coeff} }  \ottnt{V}  \ottsym{:}  \ottnt{A}}%
}{
 \textcolor{\coeffectcolor}{  \textcolor{\coeffectcolor}{ \textcolor{\coeffectcolor}{q} \cdot \textcolor{\coeffectcolor}{\gamma} }  }\! \cdot \! \Gamma   \vdash_{\mathit{coeff} }   \ottkw{return} _{\textcolor{\coeffectcolor}{ \textcolor{\coeffectcolor}{q} } }\;  \ottnt{V}   \ottsym{:}   \ottkw{F}_{\color{\coeffectcolor}{ \textcolor{\coeffectcolor}{q} } }\;  \ottnt{A} }{%
{\ottdrulename{coeff\_ret}}{}%
}}

\newcommand{\ottdrulecoeffXXletin}[1]{\ottdrule[#1]{%
\ottpremise{ \textcolor{\coeffectcolor}{ \textcolor{\coeffectcolor}{\gamma}_{{\mathrm{1}}} }\! \cdot \! \Gamma   \vdash_{\mathit{coeff} }  \ottnt{M}  \ottsym{:}   \ottkw{F}_{\color{\coeffectcolor}{ \textcolor{\coeffectcolor}{q}_{{\mathrm{1}}} } }\;  \ottnt{A} }%
\ottpremise{  \textcolor{\coeffectcolor}{ \textcolor{\coeffectcolor}{\gamma}_{{\mathrm{2}}} }\! \cdot \! \Gamma    \mathop{,}    \ottmv{x}  :^{\textcolor{\coeffectcolor}{  \textcolor{\coeffectcolor}{ \textcolor{\coeffectcolor}{q}_{{\mathrm{1}}}   \cdot   \textcolor{\coeffectcolor}{q}_{{\mathrm{2}}} }  } }  \ottnt{A}    \vdash_{\mathit{coeff} }  \ottnt{N}  \ottsym{:}  \ottnt{B}}%
}{
 \textcolor{\coeffectcolor}{  \textcolor{\coeffectcolor}{ \ottsym{(}   \textcolor{\coeffectcolor}{ \textcolor{\coeffectcolor}{q}_{{\mathrm{2}}} \cdot \textcolor{\coeffectcolor}{\gamma}_{{\mathrm{1}}} }   \ottsym{)} \ottsym{+} \textcolor{\coeffectcolor}{\gamma}_{{\mathrm{2}}} }  }\! \cdot \! \Gamma   \vdash_{\mathit{coeff} }   \ottmv{x}  \leftarrow^{\textcolor{\coeffectcolor}{ \textcolor{\coeffectcolor}{q}_{{\mathrm{2}}} } }  \ottnt{M} \ \ottkw{in}\  \ottnt{N}   \ottsym{:}  \ottnt{B}}{%
{\ottdrulename{coeff\_letin}}{}%
}}

\newcommand{\ottdrulecoeffXXcpair}[1]{\ottdrule[#1]{%
\ottpremise{ \textcolor{\coeffectcolor}{ \textcolor{\coeffectcolor}{\gamma} }\! \cdot \! \Gamma   \vdash_{\mathit{coeff} }  \ottnt{M_{{\mathrm{1}}}}  \ottsym{:}  \ottnt{B_{{\mathrm{1}}}}}%
\ottpremise{ \textcolor{\coeffectcolor}{ \textcolor{\coeffectcolor}{\gamma} }\! \cdot \! \Gamma   \vdash_{\mathit{coeff} }  \ottnt{M_{{\mathrm{2}}}}  \ottsym{:}  \ottnt{B_{{\mathrm{2}}}}}%
}{
 \textcolor{\coeffectcolor}{ \textcolor{\coeffectcolor}{\gamma} }\! \cdot \! \Gamma   \vdash_{\mathit{coeff} }   \langle  \ottnt{M_{{\mathrm{1}}}} , \ottnt{M_{{\mathrm{2}}}}  \rangle   \ottsym{:}   \ottnt{B_{{\mathrm{1}}}}   \mathop{\&}   \ottnt{B_{{\mathrm{2}}}} }{%
{\ottdrulename{coeff\_cpair}}{}%
}}

\newcommand{\ottdrulecoeffXXfst}[1]{\ottdrule[#1]{%
\ottpremise{ \textcolor{\coeffectcolor}{ \textcolor{\coeffectcolor}{\gamma} }\! \cdot \! \Gamma   \vdash_{\mathit{coeff} }  \ottnt{M}  \ottsym{:}   \ottnt{B_{{\mathrm{1}}}}   \mathop{\&}   \ottnt{B_{{\mathrm{2}}}} }%
}{
 \textcolor{\coeffectcolor}{ \textcolor{\coeffectcolor}{\gamma} }\! \cdot \! \Gamma   \vdash_{\mathit{coeff} }   \ottnt{M}  . 1   \ottsym{:}  \ottnt{B_{{\mathrm{1}}}}}{%
{\ottdrulename{coeff\_fst}}{}%
}}

\newcommand{\ottdrulecoeffXXsnd}[1]{\ottdrule[#1]{%
\ottpremise{ \textcolor{\coeffectcolor}{ \textcolor{\coeffectcolor}{\gamma} }\! \cdot \! \Gamma   \vdash_{\mathit{coeff} }  \ottnt{M}  \ottsym{:}   \ottnt{B_{{\mathrm{1}}}}   \mathop{\&}   \ottnt{B_{{\mathrm{2}}}} }%
}{
 \textcolor{\coeffectcolor}{ \textcolor{\coeffectcolor}{\gamma} }\! \cdot \! \Gamma   \vdash_{\mathit{coeff} }   \ottnt{M}  . 2   \ottsym{:}  \ottnt{B_{{\mathrm{2}}}}}{%
{\ottdrulename{coeff\_snd}}{}%
}}

\newcommand{\ottdrulecoeffXXsequence}[1]{\ottdrule[#1]{%
\ottpremise{ \textcolor{\coeffectcolor}{ \textcolor{\coeffectcolor}{\gamma}_{{\mathrm{1}}} }\! \cdot \! \Gamma   \vdash_{\mathit{coeff} }  \ottnt{V}  \ottsym{:}  \ottkw{unit}}%
\ottpremise{ \textcolor{\coeffectcolor}{ \textcolor{\coeffectcolor}{\gamma}_{{\mathrm{2}}} }\! \cdot \! \Gamma   \vdash_{\mathit{coeff} }  \ottnt{N}  \ottsym{:}  \ottnt{B}}%
}{
 \textcolor{\coeffectcolor}{  \textcolor{\coeffectcolor}{ \textcolor{\coeffectcolor}{\gamma}_{{\mathrm{1}}} \ottsym{+} \textcolor{\coeffectcolor}{\gamma}_{{\mathrm{2}}} }  }\! \cdot \! \Gamma   \vdash_{\mathit{coeff} }   \ottnt{V}  ;  \ottnt{N}   \ottsym{:}  \ottnt{B}}{%
{\ottdrulename{coeff\_sequence}}{}%
}}

\newcommand{\ottdrulecoeffXXcunit}[1]{\ottdrule[#1]{%
}{
 \textcolor{\coeffectcolor}{  \textcolor{\coeffectcolor}{\overline{0} }  }\! \cdot \! \Gamma   \vdash_{\mathit{coeff} }   \langle\rangle   \ottsym{:}   \top }{%
{\ottdrulename{coeff\_cunit}}{}%
}}

\newcommand{\ottdrulecoeffXXcase}[1]{\ottdrule[#1]{%
\ottpremise{ \textcolor{\coeffectcolor}{ \textcolor{\coeffectcolor}{\gamma}_{{\mathrm{1}}} }\! \cdot \! \Gamma   \vdash_{\mathit{coeff} }  \ottnt{V}  \ottsym{:}  \ottnt{A_{{\mathrm{1}}}}  \ottsym{+}  \ottnt{A_{{\mathrm{2}}}}}%
\ottpremise{  \textcolor{\coeffectcolor}{ \textcolor{\coeffectcolor}{\gamma}_{{\mathrm{2}}} }\! \cdot \! \Gamma    \mathop{,}    \ottmv{x_{{\mathrm{1}}}}  :^{\textcolor{\coeffectcolor}{ \textcolor{\coeffectcolor}{q} } }  \ottnt{A_{{\mathrm{1}}}}    \vdash_{\mathit{coeff} }  \ottnt{M_{{\mathrm{1}}}}  \ottsym{:}  \ottnt{B}}%
\ottpremise{  \textcolor{\coeffectcolor}{ \textcolor{\coeffectcolor}{\gamma}_{{\mathrm{2}}} }\! \cdot \! \Gamma    \mathop{,}    \ottmv{x_{{\mathrm{2}}}}  :^{\textcolor{\coeffectcolor}{ \textcolor{\coeffectcolor}{q} } }  \ottnt{A_{{\mathrm{2}}}}    \vdash_{\mathit{coeff} }  \ottnt{M_{{\mathrm{2}}}}  \ottsym{:}  \ottnt{B}}%
\ottpremise{ \textcolor{\coeffectcolor}{ \textcolor{\coeffectcolor}{q} }\;  \textcolor{\coeffectcolor}{\mathop{\leq_{\mathit{co} } } } \; \textcolor{\coeffectcolor}{  \textcolor{\coeffectcolor}{1}  } }%
}{
 \textcolor{\coeffectcolor}{  \textcolor{\coeffectcolor}{  \textcolor{\coeffectcolor}{ \textcolor{\coeffectcolor}{q} \cdot \textcolor{\coeffectcolor}{\gamma}_{{\mathrm{1}}} }  \ottsym{+} \textcolor{\coeffectcolor}{\gamma}_{{\mathrm{2}}} }  }\! \cdot \! \Gamma   \vdash_{\mathit{coeff} }   \ottkw{case}_{\textcolor{\coeffectcolor}{ \textcolor{\coeffectcolor}{q} } }\;  \ottnt{V} \; \ottkw{of}\; \ottkw{inl} \; \ottmv{x_{{\mathrm{1}}}}  \rightarrow\;  \ottnt{M_{{\mathrm{1}}}}  ;  \ottkw{inr} \; \ottmv{x_{{\mathrm{2}}}}  \rightarrow\;  \ottnt{M_{{\mathrm{2}}}}   \ottsym{:}  \ottnt{B}}{%
{\ottdrulename{coeff\_case}}{}%
}}

\newcommand{\ottdrulecoeffXXctensor}[1]{\ottdrule[#1]{%
\ottpremise{ \textcolor{\coeffectcolor}{ \textcolor{\coeffectcolor}{\gamma}_{{\mathrm{1}}} }\! \cdot \! \Gamma   \vdash_{\mathit{coeff} }  \ottnt{M_{{\mathrm{1}}}}  \ottsym{:}  \ottnt{B_{{\mathrm{1}}}}}%
\ottpremise{ \textcolor{\coeffectcolor}{ \textcolor{\coeffectcolor}{\gamma}_{{\mathrm{2}}} }\! \cdot \! \Gamma   \vdash_{\mathit{coeff} }  \ottnt{M_{{\mathrm{2}}}}  \ottsym{:}  \ottnt{B_{{\mathrm{2}}}}}%
}{
 \textcolor{\coeffectcolor}{  \textcolor{\coeffectcolor}{ \textcolor{\coeffectcolor}{\gamma}_{{\mathrm{1}}} \ottsym{+} \textcolor{\coeffectcolor}{\gamma}_{{\mathrm{2}}} }  }\! \cdot \! \Gamma   \vdash_{\mathit{coeff} }  \ottsym{(}  \ottnt{M_{{\mathrm{1}}}}  \ottsym{,}  \ottnt{M_{{\mathrm{2}}}}  \ottsym{)}  \ottsym{:}   \ottnt{B_{{\mathrm{1}}}}  \times  \ottnt{B_{{\mathrm{2}}}} }{%
{\ottdrulename{coeff\_ctensor}}{}%
}}

\newcommand{\ottdrulecoeffXXcsplit}[1]{\ottdrule[#1]{%
\ottpremise{ \textcolor{\coeffectcolor}{ \textcolor{\coeffectcolor}{\gamma}_{{\mathrm{1}}} }\! \cdot \! \Gamma   \vdash_{\mathit{coeff} }  \ottnt{M}  \ottsym{:}   \ottnt{B_{{\mathrm{1}}}}  \times  \ottnt{B_{{\mathrm{2}}}} }%
\ottpremise{    \textcolor{\coeffectcolor}{ \textcolor{\coeffectcolor}{\gamma}_{{\mathrm{2}}} }\! \cdot \! \Gamma    \mathop{,}    \ottmv{x_{{\mathrm{1}}}}  :^{\textcolor{\coeffectcolor}{ \textcolor{\coeffectcolor}{q} } }  \ottkw{U} \, \ottnt{B_{{\mathrm{1}}}}      \mathop{,}    \ottmv{x_{{\mathrm{2}}}}  :^{\textcolor{\coeffectcolor}{ \textcolor{\coeffectcolor}{q} } }  \ottkw{U} \, \ottnt{B_{{\mathrm{2}}}}    \vdash_{\mathit{coeff} }  \ottnt{N}  \ottsym{:}  \ottnt{B}}%
}{
 \textcolor{\coeffectcolor}{  \textcolor{\coeffectcolor}{  \textcolor{\coeffectcolor}{ \textcolor{\coeffectcolor}{q} \cdot \textcolor{\coeffectcolor}{\gamma}_{{\mathrm{1}}} }  \ottsym{+} \textcolor{\coeffectcolor}{\gamma}_{{\mathrm{2}}} }  }\! \cdot \! \Gamma   \vdash_{\mathit{coeff} }   \ottkw{case}_{\textcolor{\coeffectcolor}{ \textcolor{\coeffectcolor}{q} } }\;  \ottnt{M} \; \ottkw{of}\;( \ottmv{x_{{\mathrm{1}}}} , \ottmv{x_{{\mathrm{2}}}} )\; \rightarrow\;  \ottnt{N}   \ottsym{:}  \ottnt{B}}{%
{\ottdrulename{coeff\_csplit}}{}%
}}

\newcommand{\ottdrulecoeffXXcsub}[1]{\ottdrule[#1]{%
\ottpremise{ \textcolor{\coeffectcolor}{ \textcolor{\coeffectcolor}{\gamma}' }\! \cdot \! \Gamma   \vdash_{\mathit{coeff} }  \ottnt{M}  \ottsym{:}  \ottnt{B}}%
\ottpremise{ \textcolor{\coeffectcolor}{ \textcolor{\coeffectcolor}{\gamma} }\; \textcolor{\coeffectcolor}{\mathop{\leq_{\mathit{co} } } } \; \textcolor{\coeffectcolor}{ \textcolor{\coeffectcolor}{\gamma}' } }%
}{
 \textcolor{\coeffectcolor}{ \textcolor{\coeffectcolor}{\gamma} }\! \cdot \! \Gamma   \vdash_{\mathit{coeff} }  \ottnt{M}  \ottsym{:}  \ottnt{B}}{%
{\ottdrulename{coeff\_csub}}{}%
}}

\newcommand{\ottdefncompXXcoeffXXtyping}[1]{\begin{ottdefnblock}[#1]{$\Psi  \vdash_{\mathit{coeff} }  \ottnt{M}  \ottsym{:}  \ottnt{B}$}{\ottcom{computation typing rules (coeffect)}}
\ottusedrule{\ottdrulecoeffXXabs{}}
\ottusedrule{\ottdrulecoeffXXapp{}}
\ottusedrule{\ottdrulecoeffXXforce{}}
\ottusedrule{\ottdrulecoeffXXsplit{}}
\ottusedrule{\ottdrulecoeffXXret{}}
\ottusedrule{\ottdrulecoeffXXletin{}}
\ottusedrule{\ottdrulecoeffXXcpair{}}
\ottusedrule{\ottdrulecoeffXXfst{}}
\ottusedrule{\ottdrulecoeffXXsnd{}}
\ottusedrule{\ottdrulecoeffXXsequence{}}
\ottusedrule{\ottdrulecoeffXXcunit{}}
\ottusedrule{\ottdrulecoeffXXcase{}}
\ottusedrule{\ottdrulecoeffXXctensor{}}
\ottusedrule{\ottdrulecoeffXXcsplit{}}
\ottusedrule{\ottdrulecoeffXXcsub{}}
\end{ottdefnblock}}

\newcommand{\ottdefnsJCoeff}{
\ottdefnvalXXcoeffXXtyping{}\ottdefncompXXcoeffXXtyping{}}

\newcommand{\ottdrulelinXXvar}[1]{\ottdrule[#1]{%
}{
   \textcolor{\coeffectcolor}{  \textcolor{\coeffectcolor}{\overline{0} }  }\! \cdot \! \Gamma_{{\mathrm{1}}}    \mathop{,}    \ottmv{x}  :^{\textcolor{\coeffectcolor}{  \textcolor{\coeffectcolor}{1}  } }  \ottnt{A}     \mathop{,}    \textcolor{\coeffectcolor}{  \textcolor{\coeffectcolor}{\overline{0} }  }\! \cdot \! \Gamma_{{\mathrm{2}}}    \vdash_{\mathit{lin} }  \ottmv{x}  \ottsym{:}  \ottnt{A}}{%
{\ottdrulename{lin\_var}}{}%
}}

\newcommand{\ottdrulelinXXthunk}[1]{\ottdrule[#1]{%
\ottpremise{ \textcolor{\coeffectcolor}{ \textcolor{\coeffectcolor}{\gamma} }\! \cdot \! \Gamma   \vdash_{\mathit{lin} }  \ottnt{M}  \ottsym{:}  \ottnt{B}}%
}{
 \textcolor{\coeffectcolor}{ \textcolor{\coeffectcolor}{\gamma} }\! \cdot \! \Gamma   \vdash_{\mathit{lin} }  \ottsym{\{}  \ottnt{M}  \ottsym{\}}  \ottsym{:}  \ottkw{U} \, \ottnt{B}}{%
{\ottdrulename{lin\_thunk}}{}%
}}

\newcommand{\ottdrulelinXXunit}[1]{\ottdrule[#1]{%
}{
 \textcolor{\coeffectcolor}{  \textcolor{\coeffectcolor}{\overline{0} }  }\! \cdot \! \Gamma   \vdash_{\mathit{lin} }  \ottsym{()}  \ottsym{:}  \ottkw{unit}}{%
{\ottdrulename{lin\_unit}}{}%
}}

\newcommand{\ottdrulelinXXpair}[1]{\ottdrule[#1]{%
\ottpremise{ \textcolor{\coeffectcolor}{ \textcolor{\coeffectcolor}{\gamma}_{{\mathrm{1}}} }\! \cdot \! \Gamma   \vdash_{\mathit{lin} }  \ottnt{V_{{\mathrm{1}}}}  \ottsym{:}  \ottnt{A_{{\mathrm{1}}}}}%
\ottpremise{ \textcolor{\coeffectcolor}{ \textcolor{\coeffectcolor}{\gamma}_{{\mathrm{2}}} }\! \cdot \! \Gamma   \vdash_{\mathit{lin} }  \ottnt{V_{{\mathrm{2}}}}  \ottsym{:}  \ottnt{A_{{\mathrm{2}}}}}%
}{
 \textcolor{\coeffectcolor}{  \textcolor{\coeffectcolor}{ \textcolor{\coeffectcolor}{\gamma}_{{\mathrm{1}}} \ottsym{+} \textcolor{\coeffectcolor}{\gamma}_{{\mathrm{2}}} }  }\! \cdot \! \Gamma   \vdash_{\mathit{lin} }  \ottsym{(}  \ottnt{V_{{\mathrm{1}}}}  \ottsym{,}  \ottnt{V_{{\mathrm{2}}}}  \ottsym{)}  \ottsym{:}   \ottnt{A_{{\mathrm{1}}}} \times \ottnt{A_{{\mathrm{2}}}} }{%
{\ottdrulename{lin\_pair}}{}%
}}

\newcommand{\ottdrulelinXXinl}[1]{\ottdrule[#1]{%
\ottpremise{ \textcolor{\coeffectcolor}{ \textcolor{\coeffectcolor}{\gamma} }\! \cdot \! \Gamma   \vdash_{\mathit{lin} }  \ottnt{V}  \ottsym{:}  \ottnt{A_{{\mathrm{1}}}}}%
}{
 \textcolor{\coeffectcolor}{ \textcolor{\coeffectcolor}{\gamma} }\! \cdot \! \Gamma   \vdash_{\mathit{lin} }  \ottkw{inl} \, \ottnt{V}  \ottsym{:}  \ottnt{A_{{\mathrm{1}}}}  \ottsym{+}  \ottnt{A_{{\mathrm{2}}}}}{%
{\ottdrulename{lin\_inl}}{}%
}}

\newcommand{\ottdrulelinXXinr}[1]{\ottdrule[#1]{%
\ottpremise{ \textcolor{\coeffectcolor}{ \textcolor{\coeffectcolor}{\gamma} }\! \cdot \! \Gamma   \vdash_{\mathit{lin} }  \ottnt{V}  \ottsym{:}  \ottnt{A_{{\mathrm{2}}}}}%
}{
 \textcolor{\coeffectcolor}{ \textcolor{\coeffectcolor}{\gamma} }\! \cdot \! \Gamma   \vdash_{\mathit{lin} }  \ottkw{inr} \, \ottnt{V}  \ottsym{:}  \ottnt{A_{{\mathrm{1}}}}  \ottsym{+}  \ottnt{A_{{\mathrm{2}}}}}{%
{\ottdrulename{lin\_inr}}{}%
}}

\newcommand{\ottdrulelinXXvsub}[1]{\ottdrule[#1]{%
\ottpremise{ \textcolor{\coeffectcolor}{ \textcolor{\coeffectcolor}{\gamma}' }\! \cdot \! \Gamma   \vdash_{\mathit{lin} }  \ottnt{V}  \ottsym{:}  \ottnt{A}}%
\ottpremise{ \textcolor{\coeffectcolor}{ \textcolor{\coeffectcolor}{\gamma} }\; \textcolor{\coeffectcolor}{\mathop{\leq_{\mathit{co} } } } \; \textcolor{\coeffectcolor}{ \textcolor{\coeffectcolor}{\gamma}' } }%
}{
 \textcolor{\coeffectcolor}{ \textcolor{\coeffectcolor}{\gamma} }\! \cdot \! \Gamma   \vdash_{\mathit{lin} }  \ottnt{V}  \ottsym{:}  \ottnt{A}}{%
{\ottdrulename{lin\_vsub}}{}%
}}

\newcommand{\ottdefnvalXXlinXXtyping}[1]{\begin{ottdefnblock}[#1]{$\Psi  \vdash_{\mathit{lin} }  \ottnt{V}  \ottsym{:}  \ottnt{A}$}{\ottcom{value typing rules (resource usage)}}
\ottusedrule{\ottdrulelinXXvar{}}
\ottusedrule{\ottdrulelinXXthunk{}}
\ottusedrule{\ottdrulelinXXunit{}}
\ottusedrule{\ottdrulelinXXpair{}}
\ottusedrule{\ottdrulelinXXinl{}}
\ottusedrule{\ottdrulelinXXinr{}}
\ottusedrule{\ottdrulelinXXvsub{}}
\end{ottdefnblock}}

\newcommand{\ottdrulelinXXabs}[1]{\ottdrule[#1]{%
\ottpremise{  \textcolor{\coeffectcolor}{ \textcolor{\coeffectcolor}{\gamma} }\! \cdot \! \Gamma    \mathop{,}    \ottmv{x}  :^{\textcolor{\coeffectcolor}{ \textcolor{\coeffectcolor}{q} } }  \ottnt{A}    \vdash_{\mathit{lin} }  \ottnt{M}  \ottsym{:}  \ottnt{B}}%
\ottpremise{ \textcolor{\coeffectcolor}{ \textcolor{\coeffectcolor}{q} }\;  \textcolor{\coeffectcolor}{\mathop{\leq_{\mathit{co} } } } \; \textcolor{\coeffectcolor}{ \textcolor{\coeffectcolor}{q}' } }%
}{
 \textcolor{\coeffectcolor}{ \textcolor{\coeffectcolor}{\gamma} }\! \cdot \! \Gamma   \vdash_{\mathit{lin} }   \lambda  \ottmv{x} ^{\textcolor{\coeffectcolor}{ \textcolor{\coeffectcolor}{q} } }. \ottnt{M}   \ottsym{:}   \ottnt{A} ^{\textcolor{\coeffectcolor}{ \textcolor{\coeffectcolor}{q}' } } \rightarrow  \ottnt{B} }{%
{\ottdrulename{lin\_abs}}{}%
}}

\newcommand{\ottdrulelinXXapp}[1]{\ottdrule[#1]{%
\ottpremise{ \textcolor{\coeffectcolor}{ \textcolor{\coeffectcolor}{\gamma}_{{\mathrm{1}}} }\! \cdot \! \Gamma   \vdash_{\mathit{lin} }  \ottnt{M}  \ottsym{:}   \ottnt{A} ^{\textcolor{\coeffectcolor}{ \textcolor{\coeffectcolor}{q} } } \rightarrow  \ottnt{B} }%
\ottpremise{ \textcolor{\coeffectcolor}{ \textcolor{\coeffectcolor}{\gamma}_{{\mathrm{2}}} }\! \cdot \! \Gamma   \vdash_{\mathit{lin} }  \ottnt{V}  \ottsym{:}  \ottnt{A}}%
}{
 \textcolor{\coeffectcolor}{  \textcolor{\coeffectcolor}{ \textcolor{\coeffectcolor}{\gamma}_{{\mathrm{1}}} \ottsym{+} \ottsym{(}   \textcolor{\coeffectcolor}{ \textcolor{\coeffectcolor}{q} \cdot \textcolor{\coeffectcolor}{\gamma}_{{\mathrm{2}}} }   \ottsym{)} }  }\! \cdot \! \Gamma   \vdash_{\mathit{lin} }  \ottnt{M} \, \ottnt{V}  \ottsym{:}  \ottnt{B}}{%
{\ottdrulename{lin\_app}}{}%
}}

\newcommand{\ottdrulelinXXforce}[1]{\ottdrule[#1]{%
\ottpremise{ \textcolor{\coeffectcolor}{ \textcolor{\coeffectcolor}{\gamma} }\! \cdot \! \Gamma   \vdash_{\mathit{lin} }  \ottnt{V}  \ottsym{:}  \ottkw{U} \, \ottnt{B}}%
}{
 \textcolor{\coeffectcolor}{ \textcolor{\coeffectcolor}{\gamma} }\! \cdot \! \Gamma   \vdash_{\mathit{lin} }  \ottnt{V}  \ottsym{!}  \ottsym{:}  \ottnt{B}}{%
{\ottdrulename{lin\_force}}{}%
}}

\newcommand{\ottdrulelinXXsplit}[1]{\ottdrule[#1]{%
\ottpremise{ \textcolor{\coeffectcolor}{ \textcolor{\coeffectcolor}{\gamma}_{{\mathrm{1}}} }\! \cdot \! \Gamma   \vdash_{\mathit{lin} }  \ottnt{V}  \ottsym{:}   \ottnt{A_{{\mathrm{1}}}} \times \ottnt{A_{{\mathrm{2}}}} }%
\ottpremise{    \textcolor{\coeffectcolor}{ \textcolor{\coeffectcolor}{\gamma}_{{\mathrm{2}}} }\! \cdot \! \Gamma    \mathop{,}    \ottmv{x_{{\mathrm{1}}}}  :^{\textcolor{\coeffectcolor}{ \textcolor{\coeffectcolor}{q} } }  \ottnt{A_{{\mathrm{1}}}}      \mathop{,}    \ottmv{x_{{\mathrm{2}}}}  :^{\textcolor{\coeffectcolor}{ \textcolor{\coeffectcolor}{q} } }  \ottnt{A_{{\mathrm{2}}}}    \vdash_{\mathit{lin} }  \ottnt{N}  \ottsym{:}  \ottnt{B}}%
}{
 \textcolor{\coeffectcolor}{  \textcolor{\coeffectcolor}{ \ottsym{(}   \textcolor{\coeffectcolor}{ \textcolor{\coeffectcolor}{q} \cdot \textcolor{\coeffectcolor}{\gamma}_{{\mathrm{1}}} }   \ottsym{)} \ottsym{+} \textcolor{\coeffectcolor}{\gamma}_{{\mathrm{2}}} }  }\! \cdot \! \Gamma   \vdash_{\mathit{lin} }   \ottkw{case}_{\textcolor{\coeffectcolor}{ \textcolor{\coeffectcolor}{q} } } \;  \ottnt{V} \; \ottkw{of}\;( \ottmv{x_{{\mathrm{1}}}} , \ottmv{x_{{\mathrm{2}}}} )\; \rightarrow\;  \ottnt{N}   \ottsym{:}  \ottnt{B}}{%
{\ottdrulename{lin\_split}}{}%
}}

\newcommand{\ottdrulelinXXret}[1]{\ottdrule[#1]{%
\ottpremise{ \textcolor{\coeffectcolor}{ \textcolor{\coeffectcolor}{\gamma} }\! \cdot \! \Gamma   \vdash_{\mathit{lin} }  \ottnt{V}  \ottsym{:}  \ottnt{A}}%
}{
 \textcolor{\coeffectcolor}{  \textcolor{\coeffectcolor}{ \textcolor{\coeffectcolor}{q} \cdot \textcolor{\coeffectcolor}{\gamma} }  }\! \cdot \! \Gamma   \vdash_{\mathit{lin} }   \ottkw{return} _{\textcolor{\coeffectcolor}{ \textcolor{\coeffectcolor}{q} } }\;  \ottnt{V}   \ottsym{:}   \ottkw{F}_{\color{\coeffectcolor}{ \textcolor{\coeffectcolor}{q} } }\;  \ottnt{A} }{%
{\ottdrulename{lin\_ret}}{}%
}}

\newcommand{\ottdrulelinXXletin}[1]{\ottdrule[#1]{%
\ottpremise{ \textcolor{\coeffectcolor}{ \textcolor{\coeffectcolor}{q}' }\; = \; \textcolor{\coeffectcolor}{  \textcolor{\coeffectcolor}{q}_{{\mathrm{2}}} \ \|\ \textcolor{\coeffectcolor}{1}  } }%
\ottpremise{ \textcolor{\coeffectcolor}{ \textcolor{\coeffectcolor}{\gamma}_{{\mathrm{1}}} }\! \cdot \! \Gamma   \vdash_{\mathit{lin} }  \ottnt{M}  \ottsym{:}   \ottkw{F}_{\color{\coeffectcolor}{ \textcolor{\coeffectcolor}{q}_{{\mathrm{1}}} } }\;  \ottnt{A} }%
\ottpremise{  \textcolor{\coeffectcolor}{ \textcolor{\coeffectcolor}{\gamma}_{{\mathrm{2}}} }\! \cdot \! \Gamma    \mathop{,}    \ottmv{x}  :^{\textcolor{\coeffectcolor}{  \textcolor{\coeffectcolor}{ \textcolor{\coeffectcolor}{q}_{{\mathrm{1}}}   \cdot   \textcolor{\coeffectcolor}{q}' }  } }  \ottnt{A}    \vdash_{\mathit{lin} }  \ottnt{N}  \ottsym{:}  \ottnt{B}}%
}{
 \textcolor{\coeffectcolor}{  \textcolor{\coeffectcolor}{ \ottsym{(}   \textcolor{\coeffectcolor}{ \textcolor{\coeffectcolor}{q}' \cdot \textcolor{\coeffectcolor}{\gamma}_{{\mathrm{1}}} }   \ottsym{)} \ottsym{+} \textcolor{\coeffectcolor}{\gamma}_{{\mathrm{2}}} }  }\! \cdot \! \Gamma   \vdash_{\mathit{lin} }   \ottmv{x}  \leftarrow^{\textcolor{\coeffectcolor}{ \textcolor{\coeffectcolor}{q}_{{\mathrm{2}}} } }  \ottnt{M} \ \ottkw{in}\  \ottnt{N}   \ottsym{:}  \ottnt{B}}{%
{\ottdrulename{lin\_letin}}{}%
}}

\newcommand{\ottdrulelinXXletinZero}[1]{\ottdrule[#1]{%
\ottpremise{ \textcolor{\coeffectcolor}{ \textcolor{\coeffectcolor}{q}' }\; = \; \textcolor{\coeffectcolor}{  \textcolor{\coeffectcolor}{q}_{{\mathrm{2}}} \ \|\ \textcolor{\coeffectcolor}{1}  } }%
\ottpremise{ \textcolor{\coeffectcolor}{ \textcolor{\coeffectcolor}{\gamma}_{{\mathrm{1}}} }\! \cdot \! \Gamma   \vdash_{\mathit{lin} }  \ottnt{M}  \ottsym{:}   \ottkw{F}_{\color{\coeffectcolor}{ \textcolor{\coeffectcolor}{q}_{{\mathrm{1}}} } }\;  \ottnt{A} }%
\ottpremise{  \textcolor{\coeffectcolor}{ \textcolor{\coeffectcolor}{\gamma}_{{\mathrm{2}}} }\! \cdot \! \Gamma    \mathop{,}    \ottmv{x}  :^{\textcolor{\coeffectcolor}{  \textcolor{\coeffectcolor}{ \textcolor{\coeffectcolor}{q}_{{\mathrm{1}}}   \cdot   \textcolor{\coeffectcolor}{q}' }  } }  \ottnt{A}    \vdash_{\mathit{lin} }  \ottnt{N}  \ottsym{:}  \ottnt{B}}%
}{
 \textcolor{\coeffectcolor}{  \textcolor{\coeffectcolor}{ \ottsym{(}   \textcolor{\coeffectcolor}{ \textcolor{\coeffectcolor}{q}' \cdot \textcolor{\coeffectcolor}{\gamma}_{{\mathrm{1}}} }   \ottsym{)} \ottsym{+} \textcolor{\coeffectcolor}{\gamma}_{{\mathrm{2}}} }  }\! \cdot \! \Gamma   \vdash_{\mathit{lin} }   \ottmv{x}  \leftarrow^{\textcolor{\coeffectcolor}{ \textcolor{\coeffectcolor}{q}_{{\mathrm{2}}} } }  \ottnt{M} \ \ottkw{in}\  \ottnt{N}   \ottsym{:}  \ottnt{B}}{%
{\ottdrulename{lin\_letin0}}{}%
}}

\newcommand{\ottdrulelinXXcpair}[1]{\ottdrule[#1]{%
\ottpremise{ \textcolor{\coeffectcolor}{ \textcolor{\coeffectcolor}{\gamma} }\! \cdot \! \Gamma   \vdash_{\mathit{lin} }  \ottnt{M_{{\mathrm{1}}}}  \ottsym{:}  \ottnt{B_{{\mathrm{1}}}}}%
\ottpremise{ \textcolor{\coeffectcolor}{ \textcolor{\coeffectcolor}{\gamma} }\! \cdot \! \Gamma   \vdash_{\mathit{lin} }  \ottnt{M_{{\mathrm{2}}}}  \ottsym{:}  \ottnt{B_{{\mathrm{2}}}}}%
}{
 \textcolor{\coeffectcolor}{ \textcolor{\coeffectcolor}{\gamma} }\! \cdot \! \Gamma   \vdash_{\mathit{lin} }   \langle  \ottnt{M_{{\mathrm{1}}}} , \ottnt{M_{{\mathrm{2}}}}  \rangle   \ottsym{:}   \ottnt{B_{{\mathrm{1}}}}   \mathop{\&}   \ottnt{B_{{\mathrm{2}}}} }{%
{\ottdrulename{lin\_cpair}}{}%
}}

\newcommand{\ottdrulelinXXfst}[1]{\ottdrule[#1]{%
\ottpremise{ \textcolor{\coeffectcolor}{ \textcolor{\coeffectcolor}{\gamma} }\! \cdot \! \Gamma   \vdash_{\mathit{lin} }  \ottnt{M}  \ottsym{:}   \ottnt{B_{{\mathrm{1}}}}   \mathop{\&}   \ottnt{B_{{\mathrm{2}}}} }%
}{
 \textcolor{\coeffectcolor}{ \textcolor{\coeffectcolor}{\gamma} }\! \cdot \! \Gamma   \vdash_{\mathit{lin} }   \ottnt{M}  . 1   \ottsym{:}  \ottnt{B_{{\mathrm{1}}}}}{%
{\ottdrulename{lin\_fst}}{}%
}}

\newcommand{\ottdrulelinXXsnd}[1]{\ottdrule[#1]{%
\ottpremise{ \textcolor{\coeffectcolor}{ \textcolor{\coeffectcolor}{\gamma} }\! \cdot \! \Gamma   \vdash_{\mathit{lin} }  \ottnt{M}  \ottsym{:}   \ottnt{B_{{\mathrm{1}}}}   \mathop{\&}   \ottnt{B_{{\mathrm{2}}}} }%
}{
 \textcolor{\coeffectcolor}{ \textcolor{\coeffectcolor}{\gamma} }\! \cdot \! \Gamma   \vdash_{\mathit{lin} }   \ottnt{M}  . 2   \ottsym{:}  \ottnt{B_{{\mathrm{2}}}}}{%
{\ottdrulename{lin\_snd}}{}%
}}

\newcommand{\ottdrulelinXXsequence}[1]{\ottdrule[#1]{%
\ottpremise{ \textcolor{\coeffectcolor}{ \textcolor{\coeffectcolor}{\gamma}_{{\mathrm{1}}} }\! \cdot \! \Gamma   \vdash_{\mathit{lin} }  \ottnt{V}  \ottsym{:}  \ottkw{unit}}%
\ottpremise{ \textcolor{\coeffectcolor}{ \textcolor{\coeffectcolor}{\gamma}_{{\mathrm{2}}} }\! \cdot \! \Gamma   \vdash_{\mathit{lin} }  \ottnt{N}  \ottsym{:}  \ottnt{B}}%
}{
 \textcolor{\coeffectcolor}{  \textcolor{\coeffectcolor}{ \textcolor{\coeffectcolor}{\gamma}_{{\mathrm{1}}} \ottsym{+} \textcolor{\coeffectcolor}{\gamma}_{{\mathrm{2}}} }  }\! \cdot \! \Gamma   \vdash_{\mathit{lin} }   \ottnt{V}  ;  \ottnt{N}   \ottsym{:}  \ottnt{B}}{%
{\ottdrulename{lin\_sequence}}{}%
}}

\newcommand{\ottdrulelinXXcase}[1]{\ottdrule[#1]{%
\ottpremise{ \textcolor{\coeffectcolor}{ \textcolor{\coeffectcolor}{\gamma}_{{\mathrm{1}}} }\! \cdot \! \Gamma   \vdash_{\mathit{lin} }  \ottnt{V}  \ottsym{:}  \ottnt{A_{{\mathrm{1}}}}  \ottsym{+}  \ottnt{A_{{\mathrm{2}}}}}%
\ottpremise{  \textcolor{\coeffectcolor}{ \textcolor{\coeffectcolor}{\gamma}_{{\mathrm{2}}} }\! \cdot \! \Gamma    \mathop{,}    \ottmv{x_{{\mathrm{1}}}}  :^{\textcolor{\coeffectcolor}{ \textcolor{\coeffectcolor}{q} } }  \ottnt{A_{{\mathrm{1}}}}    \vdash_{\mathit{lin} }  \ottnt{M_{{\mathrm{1}}}}  \ottsym{:}  \ottnt{B}}%
\ottpremise{  \textcolor{\coeffectcolor}{ \textcolor{\coeffectcolor}{\gamma}_{{\mathrm{2}}} }\! \cdot \! \Gamma    \mathop{,}    \ottmv{x_{{\mathrm{2}}}}  :^{\textcolor{\coeffectcolor}{ \textcolor{\coeffectcolor}{q} } }  \ottnt{A_{{\mathrm{2}}}}    \vdash_{\mathit{lin} }  \ottnt{M_{{\mathrm{2}}}}  \ottsym{:}  \ottnt{B}}%
\ottpremise{ \textcolor{\coeffectcolor}{ \textcolor{\coeffectcolor}{q} }\;  \textcolor{\coeffectcolor}{\mathop{\leq_{\mathit{co} } } } \; \textcolor{\coeffectcolor}{  \textcolor{\coeffectcolor}{1}  } }%
}{
 \textcolor{\coeffectcolor}{  \textcolor{\coeffectcolor}{  \textcolor{\coeffectcolor}{ \textcolor{\coeffectcolor}{q} \cdot \textcolor{\coeffectcolor}{\gamma}_{{\mathrm{1}}} }  \ottsym{+} \textcolor{\coeffectcolor}{\gamma}_{{\mathrm{2}}} }  }\! \cdot \! \Gamma   \vdash_{\mathit{lin} }   \ottkw{case}_{\textcolor{\coeffectcolor}{ \textcolor{\coeffectcolor}{q} } }\;  \ottnt{V} \; \ottkw{of}\; \ottkw{inl} \; \ottmv{x_{{\mathrm{1}}}}  \rightarrow\;  \ottnt{M_{{\mathrm{1}}}}  ;  \ottkw{inr} \; \ottmv{x_{{\mathrm{2}}}}  \rightarrow\;  \ottnt{M_{{\mathrm{2}}}}   \ottsym{:}  \ottnt{B}}{%
{\ottdrulename{lin\_case}}{}%
}}

\newcommand{\ottdrulelinXXcsub}[1]{\ottdrule[#1]{%
\ottpremise{ \textcolor{\coeffectcolor}{ \textcolor{\coeffectcolor}{\gamma}' }\! \cdot \! \Gamma   \vdash_{\mathit{lin} }  \ottnt{M}  \ottsym{:}  \ottnt{B}}%
\ottpremise{ \textcolor{\coeffectcolor}{ \textcolor{\coeffectcolor}{\gamma} }\; \textcolor{\coeffectcolor}{\mathop{\leq_{\mathit{co} } } } \; \textcolor{\coeffectcolor}{ \textcolor{\coeffectcolor}{\gamma}' } }%
}{
 \textcolor{\coeffectcolor}{ \textcolor{\coeffectcolor}{\gamma} }\! \cdot \! \Gamma   \vdash_{\mathit{lin} }  \ottnt{M}  \ottsym{:}  \ottnt{B}}{%
{\ottdrulename{lin\_csub}}{}%
}}

\newcommand{\ottdefncompXXlinXXtyping}[1]{\begin{ottdefnblock}[#1]{$\Psi  \vdash_{\mathit{lin} }  \ottnt{M}  \ottsym{:}  \ottnt{B}$}{\ottcom{computation typing rules (resource usage)}}
\ottusedrule{\ottdrulelinXXabs{}}
\ottusedrule{\ottdrulelinXXapp{}}
\ottusedrule{\ottdrulelinXXforce{}}
\ottusedrule{\ottdrulelinXXsplit{}}
\ottusedrule{\ottdrulelinXXret{}}
\ottusedrule{\ottdrulelinXXletin{}}
\ottusedrule{\ottdrulelinXXletinZero{}}
\ottusedrule{\ottdrulelinXXcpair{}}
\ottusedrule{\ottdrulelinXXfst{}}
\ottusedrule{\ottdrulelinXXsnd{}}
\ottusedrule{\ottdrulelinXXsequence{}}
\ottusedrule{\ottdrulelinXXcase{}}
\ottusedrule{\ottdrulelinXXcsub{}}
\end{ottdefnblock}}

\newcommand{\ottdefnsJLin}{
\ottdefnvalXXlinXXtyping{}\ottdefncompXXlinXXtyping{}}

\newcommand{\ottdrulecbvXXfullXXvar}[1]{\ottdrule[#1]{%
}{
      \textcolor{\coeffectcolor}{  \textcolor{\coeffectcolor}{\overline{0} }  }\! \cdot \! \Gamma_{{\mathrm{1}}}     \mathop{,}    \ottmv{x}  :^{\textcolor{\coeffectcolor}{  \textcolor{\coeffectcolor}{1}  } }  \tau      \mathop{,}    \textcolor{\coeffectcolor}{  \textcolor{\coeffectcolor}{\overline{0} }  }\! \cdot \! \Gamma_{{\mathrm{2}}}     \ottsym{\mbox{$\mid$}-cbvfull}   \ottmv{x}  :  \tau }{%
{\ottdrulename{cbv\_full\_var}}{}%
}}

\newcommand{\ottdrulecbvXXfullXXabs}[1]{\ottdrule[#1]{%
\ottpremise{   \textcolor{\coeffectcolor}{ \textcolor{\coeffectcolor}{\gamma} }\! \cdot \! \Gamma    \mathop{,}   \ottsym{(}   \ottmv{x}  :^{\textcolor{\coeffectcolor}{ \textcolor{\coeffectcolor}{q} } }  \tau_{{\mathrm{1}}}   \ottsym{)}    \ottsym{\mbox{$\mid$}-cbvfull}   \ottnt{e}  :  \tau_{{\mathrm{2}}} }%
\ottpremise{ \textcolor{\coeffectcolor}{ \textcolor{\coeffectcolor}{q}' }\;  \textcolor{\coeffectcolor}{\mathop{\leq_{\mathit{co} } } } \; \textcolor{\coeffectcolor}{ \textcolor{\coeffectcolor}{q} } }%
}{
  \textcolor{\coeffectcolor}{ \textcolor{\coeffectcolor}{\gamma} }\! \cdot \! \Gamma    \ottsym{\mbox{$\mid$}-cbvfull}    \lambda^{\textcolor{\coeffectcolor}{ \textcolor{\coeffectcolor}{q} } }  \ottmv{x} . \ottnt{e}   :   \tau_{{\mathrm{1}}} ^{\textcolor{\coeffectcolor}{ \textcolor{\coeffectcolor}{q}' } } \rightarrow  \tau_{{\mathrm{2}}}  }{%
{\ottdrulename{cbv\_full\_abs}}{}%
}}

\newcommand{\ottdrulecbvXXfullXXapp}[1]{\ottdrule[#1]{%
\ottpremise{ \textcolor{\coeffectcolor}{ \textcolor{\coeffectcolor}{q}' }\; = \; \textcolor{\coeffectcolor}{  \textcolor{\coeffectcolor}{q} \ \|\ \textcolor{\coeffectcolor}{1}  } }%
\ottpremise{  \textcolor{\coeffectcolor}{ \textcolor{\coeffectcolor}{\gamma}_{{\mathrm{1}}} }\! \cdot \! \Gamma    \ottsym{\mbox{$\mid$}-cbvfull}   \ottnt{e_{{\mathrm{1}}}}  :   \tau_{{\mathrm{1}}} ^{\textcolor{\coeffectcolor}{ \textcolor{\coeffectcolor}{q}' } } \rightarrow  \tau_{{\mathrm{2}}}  }%
\ottpremise{  \textcolor{\coeffectcolor}{ \textcolor{\coeffectcolor}{\gamma}_{{\mathrm{2}}} }\! \cdot \! \Gamma    \ottsym{\mbox{$\mid$}-cbvfull}   \ottnt{e_{{\mathrm{2}}}}  :  \tau_{{\mathrm{1}}} }%
\ottpremise{ \textcolor{\coeffectcolor}{ \textcolor{\coeffectcolor}{\gamma} } \equiv \textcolor{\coeffectcolor}{  \textcolor{\coeffectcolor}{ \textcolor{\coeffectcolor}{\gamma}_{{\mathrm{1}}} \ottsym{+}  \textcolor{\coeffectcolor}{ \textcolor{\coeffectcolor}{q}' \cdot \textcolor{\coeffectcolor}{\gamma}_{{\mathrm{2}}} }  }  } }%
}{
  \textcolor{\coeffectcolor}{ \textcolor{\coeffectcolor}{\gamma} }\! \cdot \! \Gamma    \ottsym{\mbox{$\mid$}-cbvfull}    \ottnt{e_{{\mathrm{1}}}} ^{ \textcolor{\coeffectcolor}{q} }  \ottnt{e_{{\mathrm{2}}}}   :  \tau_{{\mathrm{2}}} }{%
{\ottdrulename{cbv\_full\_app}}{}%
}}

\newcommand{\ottdrulecbvXXfullXXunit}[1]{\ottdrule[#1]{%
}{
  \textcolor{\coeffectcolor}{  \textcolor{\coeffectcolor}{\overline{0} }  }\! \cdot \! \Gamma    \ottsym{\mbox{$\mid$}-cbvfull}   \ottsym{()}  :  \ottkw{unit} }{%
{\ottdrulename{cbv\_full\_unit}}{}%
}}

\newcommand{\ottdrulecbvXXfullXXsequence}[1]{\ottdrule[#1]{%
\ottpremise{  \textcolor{\coeffectcolor}{ \textcolor{\coeffectcolor}{\gamma}_{{\mathrm{1}}} }\! \cdot \! \Gamma    \ottsym{\mbox{$\mid$}-cbvfull}   \ottnt{e_{{\mathrm{1}}}}  :  \ottkw{unit} }%
\ottpremise{  \textcolor{\coeffectcolor}{ \textcolor{\coeffectcolor}{\gamma}_{{\mathrm{2}}} }\! \cdot \! \Gamma    \ottsym{\mbox{$\mid$}-cbvfull}   \ottnt{e_{{\mathrm{2}}}}  :  \tau }%
\ottpremise{ \textcolor{\coeffectcolor}{ \textcolor{\coeffectcolor}{\gamma} } \equiv \textcolor{\coeffectcolor}{  \textcolor{\coeffectcolor}{ \textcolor{\coeffectcolor}{\gamma}_{{\mathrm{1}}} \ottsym{+} \textcolor{\coeffectcolor}{\gamma}_{{\mathrm{2}}} }  } }%
}{
  \textcolor{\coeffectcolor}{ \textcolor{\coeffectcolor}{\gamma} }\! \cdot \! \Gamma    \ottsym{\mbox{$\mid$}-cbvfull}    \ottnt{e_{{\mathrm{1}}}}  ;  \ottnt{e_{{\mathrm{2}}}}   :  \tau }{%
{\ottdrulename{cbv\_full\_sequence}}{}%
}}

\newcommand{\ottdrulecbvXXfullXXseq}[1]{\ottdrule[#1]{%
\ottpremise{  \textcolor{\coeffectcolor}{ \textcolor{\coeffectcolor}{\gamma}_{{\mathrm{1}}} }\! \cdot \! \Gamma    \ottsym{\mbox{$\mid$}-cbvfull}   \ottnt{e_{{\mathrm{1}}}}  :  \ottkw{unit} }%
\ottpremise{  \textcolor{\coeffectcolor}{ \textcolor{\coeffectcolor}{\gamma}_{{\mathrm{2}}} }\! \cdot \! \Gamma    \ottsym{\mbox{$\mid$}-cbvfull}   \ottnt{e_{{\mathrm{2}}}}  :  \tau }%
\ottpremise{ \textcolor{\coeffectcolor}{ \textcolor{\coeffectcolor}{\gamma} } \equiv \textcolor{\coeffectcolor}{  \textcolor{\coeffectcolor}{ \textcolor{\coeffectcolor}{\gamma}_{{\mathrm{1}}} \ottsym{+} \textcolor{\coeffectcolor}{\gamma}_{{\mathrm{2}}} }  } }%
}{
  \textcolor{\coeffectcolor}{ \textcolor{\coeffectcolor}{\gamma} }\! \cdot \! \Gamma    \ottsym{\mbox{$\mid$}-cbvfull}    \ottnt{e_{{\mathrm{1}}}} ; \ottnt{e_{{\mathrm{2}}}}   :  \tau }{%
{\ottdrulename{cbv\_full\_seq}}{}%
}}

\newcommand{\ottdrulecbvXXfullXXpair}[1]{\ottdrule[#1]{%
\ottpremise{  \textcolor{\coeffectcolor}{ \textcolor{\coeffectcolor}{\gamma}_{{\mathrm{1}}} }\! \cdot \! \Gamma    \ottsym{\mbox{$\mid$}-cbvfull}   \ottnt{e_{{\mathrm{1}}}}  :  \tau_{{\mathrm{1}}} }%
\ottpremise{  \textcolor{\coeffectcolor}{ \textcolor{\coeffectcolor}{\gamma}_{{\mathrm{2}}} }\! \cdot \! \Gamma    \ottsym{\mbox{$\mid$}-cbvfull}   \ottnt{e_{{\mathrm{2}}}}  :  \tau_{{\mathrm{2}}} }%
\ottpremise{ \textcolor{\coeffectcolor}{ \textcolor{\coeffectcolor}{\gamma} } \equiv \textcolor{\coeffectcolor}{  \textcolor{\coeffectcolor}{ \textcolor{\coeffectcolor}{\gamma}_{{\mathrm{1}}} \ottsym{+} \textcolor{\coeffectcolor}{\gamma}_{{\mathrm{2}}} }  } }%
}{
  \textcolor{\coeffectcolor}{ \textcolor{\coeffectcolor}{\gamma} }\! \cdot \! \Gamma    \ottsym{\mbox{$\mid$}-cbvfull}   \ottsym{(}  \ottnt{e_{{\mathrm{1}}}}  \ottsym{,}  \ottnt{e_{{\mathrm{2}}}}  \ottsym{)}  :   \tau_{{\mathrm{1}}}  \otimes  \tau_{{\mathrm{2}}}  }{%
{\ottdrulename{cbv\_full\_pair}}{}%
}}

\newcommand{\ottdrulecbvXXfullXXsplit}[1]{\ottdrule[#1]{%
\ottpremise{ \textcolor{\coeffectcolor}{ \textcolor{\coeffectcolor}{q}' }\; = \; \textcolor{\coeffectcolor}{  \textcolor{\coeffectcolor}{q} \ \|\ \textcolor{\coeffectcolor}{1}  } }%
\ottpremise{  \textcolor{\coeffectcolor}{ \textcolor{\coeffectcolor}{\gamma}_{{\mathrm{1}}} }\! \cdot \! \Gamma    \ottsym{\mbox{$\mid$}-cbvfull}   \ottnt{e_{{\mathrm{1}}}}  :   \tau_{{\mathrm{1}}}  \otimes  \tau_{{\mathrm{2}}}  }%
\ottpremise{    \textcolor{\coeffectcolor}{ \textcolor{\coeffectcolor}{\gamma}_{{\mathrm{2}}} }\! \cdot \! \Gamma    \mathop{,}    \ottmv{x_{{\mathrm{1}}}}  :^{\textcolor{\coeffectcolor}{ \textcolor{\coeffectcolor}{q}' } }  \tau_{{\mathrm{1}}}     \mathop{,}    \ottmv{x_{{\mathrm{2}}}}  :^{\textcolor{\coeffectcolor}{ \textcolor{\coeffectcolor}{q}' } }  \tau_{{\mathrm{2}}}     \ottsym{\mbox{$\mid$}-cbvfull}   \ottnt{e_{{\mathrm{2}}}}  :  \tau }%
\ottpremise{ \textcolor{\coeffectcolor}{ \textcolor{\coeffectcolor}{\gamma} } \equiv \textcolor{\coeffectcolor}{  \textcolor{\coeffectcolor}{  \textcolor{\coeffectcolor}{ \textcolor{\coeffectcolor}{q} \cdot \textcolor{\coeffectcolor}{\gamma}_{{\mathrm{1}}} }  \ottsym{+} \textcolor{\coeffectcolor}{\gamma}_{{\mathrm{2}}} }  } }%
}{
  \textcolor{\coeffectcolor}{ \textcolor{\coeffectcolor}{\gamma} }\! \cdot \! \Gamma    \ottsym{\mbox{$\mid$}-cbvfull}    \ottkw{let}_{ \textcolor{\coeffectcolor}{q} }\; ( \ottmv{x_{{\mathrm{1}}}} ,  \ottmv{x_{{\mathrm{2}}}} ) =  \ottnt{e_{{\mathrm{1}}}} \; \ottkw{in}\;  \ottnt{e_{{\mathrm{2}}}}   :  \tau }{%
{\ottdrulename{cbv\_full\_split}}{}%
}}

\newcommand{\ottdrulecbvXXfullXXinl}[1]{\ottdrule[#1]{%
\ottpremise{  \textcolor{\coeffectcolor}{ \textcolor{\coeffectcolor}{\gamma} }\! \cdot \! \Gamma    \ottsym{\mbox{$\mid$}-cbvfull}   \ottnt{e}  :  \tau_{{\mathrm{1}}} }%
}{
  \textcolor{\coeffectcolor}{ \textcolor{\coeffectcolor}{\gamma} }\! \cdot \! \Gamma    \ottsym{\mbox{$\mid$}-cbvfull}   \ottkw{inl} \, \ottnt{e}  :  \tau_{{\mathrm{1}}}  \ottsym{+}  \tau_{{\mathrm{2}}} }{%
{\ottdrulename{cbv\_full\_inl}}{}%
}}

\newcommand{\ottdrulecbvXXfullXXinr}[1]{\ottdrule[#1]{%
\ottpremise{  \textcolor{\coeffectcolor}{ \textcolor{\coeffectcolor}{\gamma} }\! \cdot \! \Gamma    \ottsym{\mbox{$\mid$}-cbvfull}   \ottnt{e}  :  \tau_{{\mathrm{2}}} }%
}{
  \textcolor{\coeffectcolor}{ \textcolor{\coeffectcolor}{\gamma} }\! \cdot \! \Gamma    \ottsym{\mbox{$\mid$}-cbvfull}   \ottkw{inr} \, \ottnt{e}  :  \tau_{{\mathrm{1}}}  \ottsym{+}  \tau_{{\mathrm{2}}} }{%
{\ottdrulename{cbv\_full\_inr}}{}%
}}

\newcommand{\ottdrulecbvXXfullXXcase}[1]{\ottdrule[#1]{%
\ottpremise{  \textcolor{\coeffectcolor}{ \textcolor{\coeffectcolor}{\gamma}_{{\mathrm{1}}} }\! \cdot \! \Gamma    \ottsym{\mbox{$\mid$}-cbvfull}   \ottnt{e}  :  \tau_{{\mathrm{1}}}  \ottsym{+}  \tau_{{\mathrm{2}}} }%
\ottpremise{   \textcolor{\coeffectcolor}{ \textcolor{\coeffectcolor}{\gamma}_{{\mathrm{2}}} }\! \cdot \! \Gamma    \mathop{,}    \ottmv{x_{{\mathrm{1}}}}  :^{\textcolor{\coeffectcolor}{ \textcolor{\coeffectcolor}{q} } }  \tau_{{\mathrm{1}}}     \ottsym{\mbox{$\mid$}-cbvfull}   \ottnt{e_{{\mathrm{1}}}}  :  \tau }%
\ottpremise{   \textcolor{\coeffectcolor}{ \textcolor{\coeffectcolor}{\gamma}_{{\mathrm{2}}} }\! \cdot \! \Gamma    \mathop{,}    \ottmv{x_{{\mathrm{2}}}}  :^{\textcolor{\coeffectcolor}{ \textcolor{\coeffectcolor}{q} } }  \tau_{{\mathrm{2}}}     \ottsym{\mbox{$\mid$}-cbvfull}   \ottnt{e_{{\mathrm{2}}}}  :  \tau }%
\ottpremise{ \textcolor{\coeffectcolor}{ \textcolor{\coeffectcolor}{\gamma} } \equiv \textcolor{\coeffectcolor}{  \textcolor{\coeffectcolor}{  \textcolor{\coeffectcolor}{ \textcolor{\coeffectcolor}{q} \cdot \textcolor{\coeffectcolor}{\gamma}_{{\mathrm{1}}} }  \ottsym{+} \textcolor{\coeffectcolor}{\gamma}_{{\mathrm{2}}} }  } }%
\ottpremise{ \textcolor{\coeffectcolor}{ \textcolor{\coeffectcolor}{q} }\;  \textcolor{\coeffectcolor}{\mathop{\leq_{\mathit{co} } } } \; \textcolor{\coeffectcolor}{  \textcolor{\coeffectcolor}{1}  } }%
}{
  \textcolor{\coeffectcolor}{ \textcolor{\coeffectcolor}{\gamma} }\! \cdot \! \Gamma    \ottsym{\mbox{$\mid$}-cbvfull}    \ottkw{case}_{\textcolor{\coeffectcolor}{ \textcolor{\coeffectcolor}{q} } }\;  \ottnt{e} \; \ottkw{of}\;\ottkw{inl}\; \ottmv{x_{{\mathrm{1}}}}  \rightarrow\;  \ottnt{e_{{\mathrm{1}}}}  ; \ottkw{inr}\; \ottmv{x_{{\mathrm{2}}}}  \rightarrow\;  \ottnt{e_{{\mathrm{2}}}}   :  \tau }{%
{\ottdrulename{cbv\_full\_case}}{}%
}}

\newcommand{\ottdrulecbvXXfullXXbox}[1]{\ottdrule[#1]{%
\ottpremise{  \textcolor{\coeffectcolor}{ \textcolor{\coeffectcolor}{\gamma}_{{\mathrm{1}}} }\! \cdot \! \Gamma    \ottsym{\mbox{$\mid$}-cbvfull}   \ottnt{e}  :  \tau }%
\ottpremise{ \textcolor{\coeffectcolor}{ \textcolor{\coeffectcolor}{\gamma} } \equiv \textcolor{\coeffectcolor}{  \textcolor{\coeffectcolor}{ \textcolor{\coeffectcolor}{q} \cdot \textcolor{\coeffectcolor}{\gamma}_{{\mathrm{1}}} }  } }%
\ottpremise{ \textcolor{\coeffectcolor}{ \textcolor{\coeffectcolor}{q}' }\;  \textcolor{\coeffectcolor}{\mathop{\leq_{\mathit{co} } } } \; \textcolor{\coeffectcolor}{ \textcolor{\coeffectcolor}{q} } }%
}{
  \textcolor{\coeffectcolor}{ \textcolor{\coeffectcolor}{\gamma} }\! \cdot \! \Gamma    \ottsym{\mbox{$\mid$}-cbvfull}    \ottkw{box} _{\textcolor{\coeffectcolor}{ \textcolor{\coeffectcolor}{q} } }\  \ottnt{e}   :   \square_{\textcolor{\coeffectcolor}{ \textcolor{\coeffectcolor}{q}' } }\;  \tau  }{%
{\ottdrulename{cbv\_full\_box}}{}%
}}

\newcommand{\ottdrulecbvXXfullXXunbox}[1]{\ottdrule[#1]{%
\ottpremise{ \textcolor{\coeffectcolor}{ \textcolor{\coeffectcolor}{q}'_{{\mathrm{2}}} }\; = \; \textcolor{\coeffectcolor}{  \textcolor{\coeffectcolor}{q}_{{\mathrm{2}}} \ \|\ \textcolor{\coeffectcolor}{1}  } }%
\ottpremise{  \textcolor{\coeffectcolor}{ \textcolor{\coeffectcolor}{\gamma}_{{\mathrm{1}}} }\! \cdot \! \Gamma    \ottsym{\mbox{$\mid$}-cbvfull}   \ottnt{e_{{\mathrm{1}}}}  :   \square_{\textcolor{\coeffectcolor}{ \textcolor{\coeffectcolor}{q}_{{\mathrm{1}}} } }\;  \tau  }%
\ottpremise{   \textcolor{\coeffectcolor}{ \textcolor{\coeffectcolor}{\gamma}_{{\mathrm{2}}} }\! \cdot \! \Gamma    \mathop{,}    \ottmv{x}  :^{\textcolor{\coeffectcolor}{  \textcolor{\coeffectcolor}{ \textcolor{\coeffectcolor}{q}_{{\mathrm{1}}}   \cdot   \textcolor{\coeffectcolor}{q}'_{{\mathrm{2}}} }  } }  \tau     \ottsym{\mbox{$\mid$}-cbvfull}   \ottnt{e_{{\mathrm{2}}}}  :  \tau' }%
\ottpremise{ \textcolor{\coeffectcolor}{ \textcolor{\coeffectcolor}{\gamma} } \equiv \textcolor{\coeffectcolor}{  \textcolor{\coeffectcolor}{  \textcolor{\coeffectcolor}{ \textcolor{\coeffectcolor}{q}'_{{\mathrm{2}}} \cdot \textcolor{\coeffectcolor}{\gamma}_{{\mathrm{1}}} }  \ottsym{+} \textcolor{\coeffectcolor}{\gamma}_{{\mathrm{2}}} }  } }%
}{
  \textcolor{\coeffectcolor}{ \textcolor{\coeffectcolor}{\gamma} }\! \cdot \! \Gamma    \ottsym{\mbox{$\mid$}-cbvfull}    \ottkw{unbox} _{\textcolor{\coeffectcolor}{ \textcolor{\coeffectcolor}{q}_{{\mathrm{2}}} } }\  \ottmv{x}  =  \ottnt{e_{{\mathrm{1}}}} \  \ottkw{in} \  \ottnt{e_{{\mathrm{2}}}}   :  \tau' }{%
{\ottdrulename{cbv\_full\_unbox}}{}%
}}

\newcommand{\ottdrulecbvXXfullXXsub}[1]{\ottdrule[#1]{%
\ottpremise{  \textcolor{\coeffectcolor}{ \textcolor{\coeffectcolor}{\gamma}' }\! \cdot \! \Gamma    \ottsym{\mbox{$\mid$}-cbvfull}   \ottnt{e}  :  \tau }%
\ottpremise{ \textcolor{\coeffectcolor}{ \textcolor{\coeffectcolor}{\gamma} }\; \textcolor{\coeffectcolor}{\mathop{\leq_{\mathit{co} } } } \; \textcolor{\coeffectcolor}{ \textcolor{\coeffectcolor}{\gamma}' } }%
}{
  \textcolor{\coeffectcolor}{ \textcolor{\coeffectcolor}{\gamma} }\! \cdot \! \Gamma    \ottsym{\mbox{$\mid$}-cbvfull}   \ottnt{e}  :  \tau }{%
{\ottdrulename{cbv\_full\_sub}}{}%
}}

\newcommand{\ottdrulecbvXXfullXXprev}[1]{\ottdrule[#1]{%
\ottpremise{  \textcolor{\coeffectcolor}{ \textcolor{\coeffectcolor}{\gamma} }\! \cdot \! \Gamma    \ottsym{\mbox{$\mid$}-cbvfull}   \ottnt{e}  :  \tau }%
}{
  \textcolor{\coeffectcolor}{  \textcolor{\coeffectcolor}{  \overline{1}  \ottsym{+} \textcolor{\coeffectcolor}{\gamma} }  }\! \cdot \! \Gamma    \ottsym{\mbox{$\mid$}-cbvfull}   \ottkw{prev} \, \ottnt{e}  :  \tau }{%
{\ottdrulename{cbv\_full\_prev}}{}%
}}

\newcommand{\ottdrulecbvXXfullXXappv}[1]{\ottdrule[#1]{%
\ottpremise{  \textcolor{\coeffectcolor}{ \textcolor{\coeffectcolor}{\gamma}_{{\mathrm{1}}} }\! \cdot \! \Gamma    \ottsym{\mbox{$\mid$}-cbvfull}   \ottnt{e_{{\mathrm{1}}}}  :   \tau_{{\mathrm{1}}} ^{\textcolor{\coeffectcolor}{ \textcolor{\coeffectcolor}{q} } } \rightarrow  \tau_{{\mathrm{2}}}  }%
\ottpremise{  \textcolor{\coeffectcolor}{ \textcolor{\coeffectcolor}{\gamma}_{{\mathrm{2}}} }\! \cdot \! \Gamma    \ottsym{\mbox{$\mid$}-cbvfull}   \ottnt{e_{{\mathrm{2}}}}  :  \tau_{{\mathrm{1}}} }%
\ottpremise{ \textcolor{\coeffectcolor}{ \textcolor{\coeffectcolor}{\gamma} } \equiv \textcolor{\coeffectcolor}{  \textcolor{\coeffectcolor}{ \textcolor{\coeffectcolor}{\gamma}_{{\mathrm{1}}} \ottsym{+} \ottsym{(}   \textcolor{\coeffectcolor}{  \textcolor{\coeffectcolor}{q}  \wedge \textcolor{\coeffectcolor}{1}  \cdot \textcolor{\coeffectcolor}{\gamma}_{{\mathrm{2}}} }   \ottsym{)} }  } }%
}{
  \textcolor{\coeffectcolor}{ \textcolor{\coeffectcolor}{\gamma} }\! \cdot \! \Gamma    \ottsym{\mbox{$\mid$}-cbvfull}   \ottnt{e_{{\mathrm{1}}}} \, \ottnt{e_{{\mathrm{2}}}}  :  \tau_{{\mathrm{2}}} }{%
{\ottdrulename{cbv\_full\_appv}}{}%
}}

\newcommand{\ottdrulecbvXXfullXXret}[1]{\ottdrule[#1]{%
\ottpremise{  \textcolor{\coeffectcolor}{ \textcolor{\coeffectcolor}{\gamma} }\! \cdot \! \Gamma    \ottsym{\mbox{$\mid$}-cbvfull}   \ottnt{e}  :  \tau }%
\ottpremise{ \textcolor{\effectcolor}{  \textcolor{\effectcolor}{\varepsilon}  \;  \textcolor{\effectcolor}{\mathop{\leq_{\mathit{eff} } } } \;  \textcolor{\effectcolor}{\phi}  } }%
}{
  \textcolor{\coeffectcolor}{ \textcolor{\coeffectcolor}{\gamma} }\! \cdot \! \Gamma    \ottsym{\mbox{$\mid$}-cbvfull}   ret \, \ottnt{e}  :   \ottkw{T}_{\textcolor{\effectcolor}{ \textcolor{\effectcolor}{\phi} } }\;  \tau  }{%
{\ottdrulename{cbv\_full\_ret}}{}%
}}

\newcommand{\ottdrulecbvXXfullXXbind}[1]{\ottdrule[#1]{%
\ottpremise{ \textcolor{\coeffectcolor}{ \textcolor{\coeffectcolor}{q}' }\; = \; \textcolor{\coeffectcolor}{  \textcolor{\coeffectcolor}{q} \ \|\ \textcolor{\coeffectcolor}{1}  } }%
\ottpremise{  \textcolor{\coeffectcolor}{ \textcolor{\coeffectcolor}{\gamma}_{{\mathrm{1}}} }\! \cdot \! \Gamma    \ottsym{\mbox{$\mid$}-cbvfull}   \ottnt{e_{{\mathrm{1}}}}  :   \ottkw{T}_{\textcolor{\effectcolor}{ \textcolor{\effectcolor}{\phi}_{{\mathrm{1}}} } }\;  \tau_{{\mathrm{1}}}  }%
\ottpremise{   \textcolor{\coeffectcolor}{ \textcolor{\coeffectcolor}{\gamma} }\! \cdot \! \Gamma    \mathop{,}    \ottmv{x}  :^{\textcolor{\coeffectcolor}{ \textcolor{\coeffectcolor}{q}' } }  \tau_{{\mathrm{1}}}     \ottsym{\mbox{$\mid$}-cbvfull}   \ottnt{e_{{\mathrm{2}}}}  :   \ottkw{T}_{\textcolor{\effectcolor}{ \textcolor{\effectcolor}{\phi}_{{\mathrm{2}}} } }\;  \tau_{{\mathrm{2}}}  }%
\ottpremise{ \textcolor{\coeffectcolor}{ \textcolor{\coeffectcolor}{\gamma} } \equiv \textcolor{\coeffectcolor}{  \textcolor{\coeffectcolor}{  \textcolor{\coeffectcolor}{ \textcolor{\coeffectcolor}{q}' \cdot \textcolor{\coeffectcolor}{\gamma}_{{\mathrm{1}}} }  \ottsym{+} \textcolor{\coeffectcolor}{\gamma}_{{\mathrm{2}}} }  } }%
\ottpremise{ \textcolor{\effectcolor}{  \textcolor{\effectcolor}{ \textcolor{\effectcolor}{\phi}_{{\mathrm{1}}}  \cdot  \textcolor{\effectcolor}{\phi}_{{\mathrm{2}}} }  \;  \textcolor{\effectcolor}{\mathop{\leq_{\mathit{eff} } } } \;  \textcolor{\effectcolor}{\phi}  } }%
}{
  \textcolor{\coeffectcolor}{ \textcolor{\coeffectcolor}{\gamma} }\! \cdot \! \Gamma    \ottsym{\mbox{$\mid$}-cbvfull}   \ottkw{bind} \, \ottmv{x}  \ottsym{=}  \textcolor{\coeffectcolor}{q} \, \ottnt{e_{{\mathrm{1}}}} \, \ottkw{in} \, \ottnt{e_{{\mathrm{2}}}}  :   \ottkw{T}_{\textcolor{\effectcolor}{ \textcolor{\effectcolor}{\phi} } }\;  \tau_{{\mathrm{2}}}  }{%
{\ottdrulename{cbv\_full\_bind}}{}%
}}

\newcommand{\ottdrulecbvXXfullXXtick}[1]{\ottdrule[#1]{%
\ottpremise{ \textcolor{\effectcolor}{  \textcolor{\effectcolor}{\ottkw{Tick} }  \;  \textcolor{\effectcolor}{\mathop{\leq_{\mathit{eff} } } } \;  \textcolor{\effectcolor}{\phi}  } }%
}{
  \textcolor{\coeffectcolor}{  \textcolor{\coeffectcolor}{\overline{0} }  }\! \cdot \! \Gamma    \ottsym{\mbox{$\mid$}-cbvfull}    \textcolor{\effectcolor}{ \ottkw{tick} }   :   \ottkw{T}_{\textcolor{\effectcolor}{ \textcolor{\effectcolor}{\phi} } }\;  \ottkw{unit}  }{%
{\ottdrulename{cbv\_full\_tick}}{}%
}}

\newcommand{\ottdefncbvXXfullXXtyping}[1]{\begin{ottdefnblock}[#1]{$ \Psi   \ottsym{\mbox{$\mid$}-cbvfull}   \ottnt{e}  :  \tau $}{\ottcom{Graded effect + graded coeffect CBV system}}
\ottusedrule{\ottdrulecbvXXfullXXvar{}}
\ottusedrule{\ottdrulecbvXXfullXXabs{}}
\ottusedrule{\ottdrulecbvXXfullXXapp{}}
\ottusedrule{\ottdrulecbvXXfullXXunit{}}
\ottusedrule{\ottdrulecbvXXfullXXsequence{}}
\ottusedrule{\ottdrulecbvXXfullXXseq{}}
\ottusedrule{\ottdrulecbvXXfullXXpair{}}
\ottusedrule{\ottdrulecbvXXfullXXsplit{}}
\ottusedrule{\ottdrulecbvXXfullXXinl{}}
\ottusedrule{\ottdrulecbvXXfullXXinr{}}
\ottusedrule{\ottdrulecbvXXfullXXcase{}}
\ottusedrule{\ottdrulecbvXXfullXXbox{}}
\ottusedrule{\ottdrulecbvXXfullXXunbox{}}
\ottusedrule{\ottdrulecbvXXfullXXsub{}}
\ottusedrule{\ottdrulecbvXXfullXXprev{}}
\ottusedrule{\ottdrulecbvXXfullXXappv{}}
\ottusedrule{\ottdrulecbvXXfullXXret{}}
\ottusedrule{\ottdrulecbvXXfullXXbind{}}
\ottusedrule{\ottdrulecbvXXfullXXtick{}}
\end{ottdefnblock}}

\newcommand{\ottdefnsJCBVFull}{
\ottdefncbvXXfullXXtyping{}}

\newcommand{\ottdrulecbnXXfullXXvar}[1]{\ottdrule[#1]{%
}{
      \textcolor{\coeffectcolor}{  \textcolor{\coeffectcolor}{\overline{0} }  }\! \cdot \! \Gamma_{{\mathrm{1}}}     \mathop{,}    \ottmv{x}  :^{\textcolor{\coeffectcolor}{  \textcolor{\coeffectcolor}{1}  } }  \tau      \mathop{,}    \textcolor{\coeffectcolor}{  \textcolor{\coeffectcolor}{\overline{0} }  }\! \cdot \! \Gamma_{{\mathrm{2}}}     \ottsym{\mbox{$\mid$}-cbnfull}   \ottmv{x}  :  \tau }{%
{\ottdrulename{cbn\_full\_var}}{}%
}}

\newcommand{\ottdrulecbnXXfullXXabs}[1]{\ottdrule[#1]{%
\ottpremise{   \textcolor{\coeffectcolor}{ \textcolor{\coeffectcolor}{\gamma} }\! \cdot \! \Gamma    \mathop{,}   \ottsym{(}   \ottmv{x}  :^{\textcolor{\coeffectcolor}{ \textcolor{\coeffectcolor}{q} } }  \tau_{{\mathrm{1}}}   \ottsym{)}    \ottsym{\mbox{$\mid$}-cbnfull}   \ottnt{e}  :  \tau_{{\mathrm{2}}} }%
}{
  \textcolor{\coeffectcolor}{ \textcolor{\coeffectcolor}{\gamma} }\! \cdot \! \Gamma    \ottsym{\mbox{$\mid$}-cbnfull}    \lambda^{\textcolor{\coeffectcolor}{ \textcolor{\coeffectcolor}{q} } }  \ottmv{x} . \ottnt{e}   :   \tau_{{\mathrm{1}}} ^{\textcolor{\coeffectcolor}{ \textcolor{\coeffectcolor}{q} } } \rightarrow  \tau_{{\mathrm{2}}}  }{%
{\ottdrulename{cbn\_full\_abs}}{}%
}}

\newcommand{\ottdrulecbnXXfullXXapp}[1]{\ottdrule[#1]{%
\ottpremise{  \textcolor{\coeffectcolor}{ \textcolor{\coeffectcolor}{\gamma}_{{\mathrm{1}}} }\! \cdot \! \Gamma    \ottsym{\mbox{$\mid$}-cbnfull}   \ottnt{e_{{\mathrm{1}}}}  :   \tau_{{\mathrm{1}}} ^{\textcolor{\coeffectcolor}{ \textcolor{\coeffectcolor}{q} } } \rightarrow  \tau_{{\mathrm{2}}}  }%
\ottpremise{  \textcolor{\coeffectcolor}{ \textcolor{\coeffectcolor}{\gamma}_{{\mathrm{2}}} }\! \cdot \! \Gamma    \ottsym{\mbox{$\mid$}-cbnfull}   \ottnt{e_{{\mathrm{2}}}}  :  \tau_{{\mathrm{1}}} }%
\ottpremise{ \textcolor{\coeffectcolor}{ \textcolor{\coeffectcolor}{\gamma} } \equiv \textcolor{\coeffectcolor}{  \textcolor{\coeffectcolor}{ \textcolor{\coeffectcolor}{\gamma}_{{\mathrm{1}}} \ottsym{+}  \textcolor{\coeffectcolor}{ \textcolor{\coeffectcolor}{q} \cdot \textcolor{\coeffectcolor}{\gamma}_{{\mathrm{2}}} }  }  } }%
}{
  \textcolor{\coeffectcolor}{ \textcolor{\coeffectcolor}{\gamma} }\! \cdot \! \Gamma    \ottsym{\mbox{$\mid$}-cbnfull}    \ottnt{e_{{\mathrm{1}}}} ^{ \textcolor{\coeffectcolor}{q} }  \ottnt{e_{{\mathrm{2}}}}   :  \tau_{{\mathrm{2}}} }{%
{\ottdrulename{cbn\_full\_app}}{}%
}}

\newcommand{\ottdrulecbnXXfullXXunit}[1]{\ottdrule[#1]{%
}{
  \textcolor{\coeffectcolor}{  \textcolor{\coeffectcolor}{\overline{0} }  }\! \cdot \! \Gamma    \ottsym{\mbox{$\mid$}-cbnfull}   \ottsym{()}  :  \ottkw{unit} }{%
{\ottdrulename{cbn\_full\_unit}}{}%
}}

\newcommand{\ottdrulecbnXXfullXXsequence}[1]{\ottdrule[#1]{%
\ottpremise{  \textcolor{\coeffectcolor}{ \textcolor{\coeffectcolor}{\gamma}_{{\mathrm{1}}} }\! \cdot \! \Gamma    \ottsym{\mbox{$\mid$}-cbnfull}   \ottnt{e_{{\mathrm{1}}}}  :  \ottkw{unit} }%
\ottpremise{  \textcolor{\coeffectcolor}{ \textcolor{\coeffectcolor}{\gamma}_{{\mathrm{2}}} }\! \cdot \! \Gamma    \ottsym{\mbox{$\mid$}-cbnfull}   \ottnt{e_{{\mathrm{2}}}}  :  \tau }%
\ottpremise{ \textcolor{\coeffectcolor}{ \textcolor{\coeffectcolor}{\gamma} } \equiv \textcolor{\coeffectcolor}{  \textcolor{\coeffectcolor}{ \textcolor{\coeffectcolor}{\gamma}_{{\mathrm{1}}} \ottsym{+} \textcolor{\coeffectcolor}{\gamma}_{{\mathrm{2}}} }  } }%
}{
  \textcolor{\coeffectcolor}{ \textcolor{\coeffectcolor}{\gamma} }\! \cdot \! \Gamma    \ottsym{\mbox{$\mid$}-cbnfull}    \ottnt{e_{{\mathrm{1}}}}  ;  \ottnt{e_{{\mathrm{2}}}}   :  \tau }{%
{\ottdrulename{cbn\_full\_sequence}}{}%
}}

\newcommand{\ottdrulecbnXXfullXXseq}[1]{\ottdrule[#1]{%
\ottpremise{  \textcolor{\coeffectcolor}{ \textcolor{\coeffectcolor}{\gamma}_{{\mathrm{1}}} }\! \cdot \! \Gamma    \ottsym{\mbox{$\mid$}-cbnfull}   \ottnt{e_{{\mathrm{1}}}}  :  \ottkw{unit} }%
\ottpremise{  \textcolor{\coeffectcolor}{ \textcolor{\coeffectcolor}{\gamma}_{{\mathrm{2}}} }\! \cdot \! \Gamma    \ottsym{\mbox{$\mid$}-cbnfull}   \ottnt{e_{{\mathrm{2}}}}  :  \tau }%
\ottpremise{ \textcolor{\coeffectcolor}{ \textcolor{\coeffectcolor}{\gamma} } \equiv \textcolor{\coeffectcolor}{  \textcolor{\coeffectcolor}{ \textcolor{\coeffectcolor}{\gamma}_{{\mathrm{1}}} \ottsym{+} \textcolor{\coeffectcolor}{\gamma}_{{\mathrm{2}}} }  } }%
}{
  \textcolor{\coeffectcolor}{ \textcolor{\coeffectcolor}{\gamma} }\! \cdot \! \Gamma    \ottsym{\mbox{$\mid$}-cbnfull}    \ottnt{e_{{\mathrm{1}}}} ; \ottnt{e_{{\mathrm{2}}}}   :  \tau }{%
{\ottdrulename{cbn\_full\_seq}}{}%
}}

\newcommand{\ottdrulecbnXXfullXXinl}[1]{\ottdrule[#1]{%
\ottpremise{  \textcolor{\coeffectcolor}{ \textcolor{\coeffectcolor}{\gamma} }\! \cdot \! \Gamma    \ottsym{\mbox{$\mid$}-cbnfull}   \ottnt{e}  :  \tau_{{\mathrm{1}}} }%
}{
  \textcolor{\coeffectcolor}{ \textcolor{\coeffectcolor}{\gamma} }\! \cdot \! \Gamma    \ottsym{\mbox{$\mid$}-cbnfull}   \ottkw{inl} \, \ottnt{e}  :  \tau_{{\mathrm{1}}}  \ottsym{+}  \tau_{{\mathrm{2}}} }{%
{\ottdrulename{cbn\_full\_inl}}{}%
}}

\newcommand{\ottdrulecbnXXfullXXinr}[1]{\ottdrule[#1]{%
\ottpremise{  \textcolor{\coeffectcolor}{ \textcolor{\coeffectcolor}{\gamma} }\! \cdot \! \Gamma    \ottsym{\mbox{$\mid$}-cbnfull}   \ottnt{e}  :  \tau_{{\mathrm{2}}} }%
}{
  \textcolor{\coeffectcolor}{ \textcolor{\coeffectcolor}{\gamma} }\! \cdot \! \Gamma    \ottsym{\mbox{$\mid$}-cbnfull}   \ottkw{inr} \, \ottnt{e}  :  \tau_{{\mathrm{1}}}  \ottsym{+}  \tau_{{\mathrm{2}}} }{%
{\ottdrulename{cbn\_full\_inr}}{}%
}}

\newcommand{\ottdrulecbnXXfullXXwith}[1]{\ottdrule[#1]{%
\ottpremise{  \textcolor{\coeffectcolor}{ \textcolor{\coeffectcolor}{\gamma} }\! \cdot \! \Gamma    \ottsym{\mbox{$\mid$}-cbnfull}   \ottnt{e_{{\mathrm{1}}}}  :  \tau_{{\mathrm{1}}} }%
\ottpremise{  \textcolor{\coeffectcolor}{ \textcolor{\coeffectcolor}{\gamma} }\! \cdot \! \Gamma    \ottsym{\mbox{$\mid$}-cbnfull}   \ottnt{e_{{\mathrm{2}}}}  :  \tau_{{\mathrm{2}}} }%
}{
  \textcolor{\coeffectcolor}{ \textcolor{\coeffectcolor}{\gamma} }\! \cdot \! \Gamma    \ottsym{\mbox{$\mid$}-cbnfull}    \langle  \ottnt{e_{{\mathrm{1}}}} , \ottnt{e_{{\mathrm{2}}}}  \rangle   :   \tau_{{\mathrm{1}}}  \mathop{\&}   \tau_{{\mathrm{2}}}  }{%
{\ottdrulename{cbn\_full\_with}}{}%
}}

\newcommand{\ottdrulecbnXXfullXXfst}[1]{\ottdrule[#1]{%
\ottpremise{  \textcolor{\coeffectcolor}{ \textcolor{\coeffectcolor}{\gamma} }\! \cdot \! \Gamma    \ottsym{\mbox{$\mid$}-cbnfull}   \ottnt{e}  :   \tau_{{\mathrm{1}}}  \mathop{\&}   \tau_{{\mathrm{2}}}  }%
}{
  \textcolor{\coeffectcolor}{ \textcolor{\coeffectcolor}{\gamma} }\! \cdot \! \Gamma    \ottsym{\mbox{$\mid$}-cbnfull}   \ottnt{e}  \ottsym{.}  \ottsym{1}  :  \tau_{{\mathrm{1}}} }{%
{\ottdrulename{cbn\_full\_fst}}{}%
}}

\newcommand{\ottdrulecbnXXfullXXsnd}[1]{\ottdrule[#1]{%
\ottpremise{  \textcolor{\coeffectcolor}{ \textcolor{\coeffectcolor}{\gamma} }\! \cdot \! \Gamma    \ottsym{\mbox{$\mid$}-cbnfull}   \ottnt{e}  :   \tau_{{\mathrm{1}}}  \mathop{\&}   \tau_{{\mathrm{2}}}  }%
}{
  \textcolor{\coeffectcolor}{ \textcolor{\coeffectcolor}{\gamma} }\! \cdot \! \Gamma    \ottsym{\mbox{$\mid$}-cbnfull}   \ottnt{e}  \ottsym{.}  \ottsym{2}  :  \tau_{{\mathrm{2}}} }{%
{\ottdrulename{cbn\_full\_snd}}{}%
}}

\newcommand{\ottdrulecbnXXfullXXcase}[1]{\ottdrule[#1]{%
\ottpremise{  \textcolor{\coeffectcolor}{ \textcolor{\coeffectcolor}{\gamma}_{{\mathrm{1}}} }\! \cdot \! \Gamma    \ottsym{\mbox{$\mid$}-cbnfull}   \ottnt{e}  :  \tau_{{\mathrm{1}}}  \ottsym{+}  \tau_{{\mathrm{2}}} }%
\ottpremise{   \textcolor{\coeffectcolor}{ \textcolor{\coeffectcolor}{\gamma}_{{\mathrm{2}}} }\! \cdot \! \Gamma    \mathop{,}    \ottmv{x_{{\mathrm{1}}}}  :^{\textcolor{\coeffectcolor}{ \textcolor{\coeffectcolor}{q} } }  \tau_{{\mathrm{1}}}     \ottsym{\mbox{$\mid$}-cbnfull}   \ottnt{e_{{\mathrm{1}}}}  :  \tau }%
\ottpremise{   \textcolor{\coeffectcolor}{ \textcolor{\coeffectcolor}{\gamma}_{{\mathrm{2}}} }\! \cdot \! \Gamma    \mathop{,}    \ottmv{x_{{\mathrm{2}}}}  :^{\textcolor{\coeffectcolor}{ \textcolor{\coeffectcolor}{q} } }  \tau_{{\mathrm{2}}}     \ottsym{\mbox{$\mid$}-cbnfull}   \ottnt{e_{{\mathrm{2}}}}  :  \tau }%
\ottpremise{ \textcolor{\coeffectcolor}{ \textcolor{\coeffectcolor}{\gamma} } \equiv \textcolor{\coeffectcolor}{  \textcolor{\coeffectcolor}{  \textcolor{\coeffectcolor}{ \textcolor{\coeffectcolor}{q} \cdot \textcolor{\coeffectcolor}{\gamma}_{{\mathrm{1}}} }  \ottsym{+} \textcolor{\coeffectcolor}{\gamma}_{{\mathrm{2}}} }  } }%
\ottpremise{ \textcolor{\coeffectcolor}{ \textcolor{\coeffectcolor}{q} }\;  \textcolor{\coeffectcolor}{\mathop{\leq_{\mathit{co} } } } \; \textcolor{\coeffectcolor}{  \textcolor{\coeffectcolor}{1}  } }%
}{
  \textcolor{\coeffectcolor}{ \textcolor{\coeffectcolor}{\gamma} }\! \cdot \! \Gamma    \ottsym{\mbox{$\mid$}-cbnfull}    \ottkw{case}_{\textcolor{\coeffectcolor}{ \textcolor{\coeffectcolor}{q} } }\;  \ottnt{e} \; \ottkw{of}\;\ottkw{inl}\; \ottmv{x_{{\mathrm{1}}}}  \rightarrow\;  \ottnt{e_{{\mathrm{1}}}}  ; \ottkw{inr}\; \ottmv{x_{{\mathrm{2}}}}  \rightarrow\;  \ottnt{e_{{\mathrm{2}}}}   :  \tau }{%
{\ottdrulename{cbn\_full\_case}}{}%
}}

\newcommand{\ottdrulecbnXXfullXXbox}[1]{\ottdrule[#1]{%
\ottpremise{  \textcolor{\coeffectcolor}{ \textcolor{\coeffectcolor}{\gamma}_{{\mathrm{1}}} }\! \cdot \! \Gamma    \ottsym{\mbox{$\mid$}-cbnfull}   \ottnt{e}  :  \tau }%
\ottpremise{ \textcolor{\coeffectcolor}{ \textcolor{\coeffectcolor}{\gamma} } \equiv \textcolor{\coeffectcolor}{  \textcolor{\coeffectcolor}{ \textcolor{\coeffectcolor}{q} \cdot \textcolor{\coeffectcolor}{\gamma}_{{\mathrm{1}}} }  } }%
}{
  \textcolor{\coeffectcolor}{ \textcolor{\coeffectcolor}{\gamma} }\! \cdot \! \Gamma    \ottsym{\mbox{$\mid$}-cbnfull}    \ottkw{box} _{\textcolor{\coeffectcolor}{ \textcolor{\coeffectcolor}{q} } }\  \ottnt{e}   :   \square_{\textcolor{\coeffectcolor}{ \textcolor{\coeffectcolor}{q}' } }\;  \tau  }{%
{\ottdrulename{cbn\_full\_box}}{}%
}}

\newcommand{\ottdrulecbnXXfullXXunbox}[1]{\ottdrule[#1]{%
\ottpremise{ \textcolor{\coeffectcolor}{ \textcolor{\coeffectcolor}{q}'_{{\mathrm{2}}} }\; = \; \textcolor{\coeffectcolor}{  \textcolor{\coeffectcolor}{q}_{{\mathrm{2}}} \ \|\ \textcolor{\coeffectcolor}{1}  } }%
\ottpremise{  \textcolor{\coeffectcolor}{ \textcolor{\coeffectcolor}{\gamma}_{{\mathrm{1}}} }\! \cdot \! \Gamma    \ottsym{\mbox{$\mid$}-cbnfull}   \ottnt{e_{{\mathrm{1}}}}  :   \square_{\textcolor{\coeffectcolor}{ \textcolor{\coeffectcolor}{q}_{{\mathrm{1}}} } }\;  \tau  }%
\ottpremise{   \textcolor{\coeffectcolor}{ \textcolor{\coeffectcolor}{\gamma}_{{\mathrm{2}}} }\! \cdot \! \Gamma    \mathop{,}    \ottmv{x}  :^{\textcolor{\coeffectcolor}{  \textcolor{\coeffectcolor}{ \textcolor{\coeffectcolor}{q}_{{\mathrm{1}}}   \cdot   \textcolor{\coeffectcolor}{q}'_{{\mathrm{2}}} }  } }  \tau     \ottsym{\mbox{$\mid$}-cbnfull}   \ottnt{e_{{\mathrm{2}}}}  :  \tau' }%
\ottpremise{ \textcolor{\coeffectcolor}{ \textcolor{\coeffectcolor}{\gamma} } \equiv \textcolor{\coeffectcolor}{  \textcolor{\coeffectcolor}{  \textcolor{\coeffectcolor}{ \textcolor{\coeffectcolor}{q}'_{{\mathrm{2}}} \cdot \textcolor{\coeffectcolor}{\gamma}_{{\mathrm{1}}} }  \ottsym{+} \textcolor{\coeffectcolor}{\gamma}_{{\mathrm{2}}} }  } }%
}{
  \textcolor{\coeffectcolor}{ \textcolor{\coeffectcolor}{\gamma} }\! \cdot \! \Gamma    \ottsym{\mbox{$\mid$}-cbnfull}    \ottkw{unbox} _{\textcolor{\coeffectcolor}{ \textcolor{\coeffectcolor}{q}_{{\mathrm{2}}} } }\  \ottmv{x}  =  \ottnt{e_{{\mathrm{1}}}} \  \ottkw{in} \  \ottnt{e_{{\mathrm{2}}}}   :  \tau' }{%
{\ottdrulename{cbn\_full\_unbox}}{}%
}}

\newcommand{\ottdrulecbnXXfullXXsub}[1]{\ottdrule[#1]{%
\ottpremise{  \textcolor{\coeffectcolor}{ \textcolor{\coeffectcolor}{\gamma}' }\! \cdot \! \Gamma    \ottsym{\mbox{$\mid$}-cbnfull}   \ottnt{e}  :  \tau }%
\ottpremise{ \textcolor{\coeffectcolor}{ \textcolor{\coeffectcolor}{\gamma} }\; \textcolor{\coeffectcolor}{\mathop{\leq_{\mathit{co} } } } \; \textcolor{\coeffectcolor}{ \textcolor{\coeffectcolor}{\gamma}' } }%
}{
  \textcolor{\coeffectcolor}{ \textcolor{\coeffectcolor}{\gamma} }\! \cdot \! \Gamma    \ottsym{\mbox{$\mid$}-cbnfull}   \ottnt{e}  :  \tau }{%
{\ottdrulename{cbn\_full\_sub}}{}%
}}

\newcommand{\ottdrulecbnXXfullXXret}[1]{\ottdrule[#1]{%
\ottpremise{  \textcolor{\coeffectcolor}{ \textcolor{\coeffectcolor}{\gamma} }\! \cdot \! \Gamma    \ottsym{\mbox{$\mid$}-cbnfull}   \ottnt{e}  :  \tau }%
}{
  \textcolor{\coeffectcolor}{ \textcolor{\coeffectcolor}{\gamma} }\! \cdot \! \Gamma    \ottsym{\mbox{$\mid$}-cbnfull}   ret \, \ottnt{e}  :   \ottkw{T}_{\textcolor{\effectcolor}{  \textcolor{\effectcolor}{\varepsilon}  } }\;  \tau  }{%
{\ottdrulename{cbn\_full\_ret}}{}%
}}

\newcommand{\ottdrulecbnXXfullXXbind}[1]{\ottdrule[#1]{%
\ottpremise{ \textcolor{\coeffectcolor}{ \textcolor{\coeffectcolor}{q}' }\; = \; \textcolor{\coeffectcolor}{  \textcolor{\coeffectcolor}{q} \ \|\ \textcolor{\coeffectcolor}{1}  } }%
\ottpremise{  \textcolor{\coeffectcolor}{ \textcolor{\coeffectcolor}{\gamma}_{{\mathrm{1}}} }\! \cdot \! \Gamma    \ottsym{\mbox{$\mid$}-cbnfull}   \ottnt{e_{{\mathrm{1}}}}  :   \ottkw{T}_{\textcolor{\effectcolor}{ \textcolor{\effectcolor}{\phi}_{{\mathrm{1}}} } }\;  \tau_{{\mathrm{1}}}  }%
\ottpremise{   \textcolor{\coeffectcolor}{ \textcolor{\coeffectcolor}{\gamma}_{{\mathrm{2}}} }\! \cdot \! \Gamma    \mathop{,}    \ottmv{x}  :^{\textcolor{\coeffectcolor}{ \textcolor{\coeffectcolor}{q}' } }  \tau_{{\mathrm{1}}}     \ottsym{\mbox{$\mid$}-cbnfull}   \ottnt{e_{{\mathrm{2}}}}  :   \ottkw{T}_{\textcolor{\effectcolor}{ \textcolor{\effectcolor}{\phi}_{{\mathrm{2}}} } }\;  \tau_{{\mathrm{2}}}  }%
\ottpremise{ \textcolor{\coeffectcolor}{ \textcolor{\coeffectcolor}{\gamma} } \equiv \textcolor{\coeffectcolor}{  \textcolor{\coeffectcolor}{  \textcolor{\coeffectcolor}{ \textcolor{\coeffectcolor}{q}' \cdot \textcolor{\coeffectcolor}{\gamma}_{{\mathrm{1}}} }  \ottsym{+} \textcolor{\coeffectcolor}{\gamma}_{{\mathrm{2}}} }  } }%
\ottpremise{ \textcolor{\effectcolor}{  \textcolor{\effectcolor}{ \textcolor{\effectcolor}{\phi}_{{\mathrm{1}}}  \cdot  \textcolor{\effectcolor}{\phi}_{{\mathrm{2}}} }  \; \equiv \;  \textcolor{\effectcolor}{\phi} } }%
}{
  \textcolor{\coeffectcolor}{ \textcolor{\coeffectcolor}{\gamma} }\! \cdot \! \Gamma    \ottsym{\mbox{$\mid$}-cbnfull}   \ottkw{bind} \, \ottmv{x}  \ottsym{=}  \textcolor{\coeffectcolor}{q} \, \ottnt{e_{{\mathrm{1}}}} \, \ottkw{in} \, \ottnt{e_{{\mathrm{2}}}}  :   \ottkw{T}_{\textcolor{\effectcolor}{ \textcolor{\effectcolor}{\phi} } }\;  \tau_{{\mathrm{2}}}  }{%
{\ottdrulename{cbn\_full\_bind}}{}%
}}

\newcommand{\ottdrulecbnXXfullXXtick}[1]{\ottdrule[#1]{%
}{
  \textcolor{\coeffectcolor}{  \textcolor{\coeffectcolor}{\overline{0} }  }\! \cdot \! \Gamma    \ottsym{\mbox{$\mid$}-cbnfull}    \textcolor{\effectcolor}{ \ottkw{tick} }   :   \ottkw{T}_{\textcolor{\coeffectcolor}{  \textcolor{\effectcolor}{ \ottkw{tick} }  } }\;  \ottkw{unit}  }{%
{\ottdrulename{cbn\_full\_tick}}{}%
}}

\newcommand{\ottdefncbnXXfullXXtyping}[1]{\begin{ottdefnblock}[#1]{$ \Psi   \ottsym{\mbox{$\mid$}-cbnfull}   \ottnt{e}  :  \tau $}{\ottcom{Graded effect + graded coeffect CBN system}}
\ottusedrule{\ottdrulecbnXXfullXXvar{}}
\ottusedrule{\ottdrulecbnXXfullXXabs{}}
\ottusedrule{\ottdrulecbnXXfullXXapp{}}
\ottusedrule{\ottdrulecbnXXfullXXunit{}}
\ottusedrule{\ottdrulecbnXXfullXXsequence{}}
\ottusedrule{\ottdrulecbnXXfullXXseq{}}
\ottusedrule{\ottdrulecbnXXfullXXinl{}}
\ottusedrule{\ottdrulecbnXXfullXXinr{}}
\ottusedrule{\ottdrulecbnXXfullXXwith{}}
\ottusedrule{\ottdrulecbnXXfullXXfst{}}
\ottusedrule{\ottdrulecbnXXfullXXsnd{}}
\ottusedrule{\ottdrulecbnXXfullXXcase{}}
\ottusedrule{\ottdrulecbnXXfullXXbox{}}
\ottusedrule{\ottdrulecbnXXfullXXunbox{}}
\ottusedrule{\ottdrulecbnXXfullXXsub{}}
\ottusedrule{\ottdrulecbnXXfullXXret{}}
\ottusedrule{\ottdrulecbnXXfullXXbind{}}
\ottusedrule{\ottdrulecbnXXfullXXtick{}}
\end{ottdefnblock}}

\newcommand{\ottdefnsJCBNFull}{
\ottdefncbnXXfullXXtyping{}}

\newcommand{\ottdrulecbncoeffXXvar}[1]{\ottdrule[#1]{%
}{
      \textcolor{\coeffectcolor}{  \textcolor{\coeffectcolor}{\overline{0} }  }\! \cdot \! \Gamma_{{\mathrm{1}}}     \mathop{,}    \ottmv{x}  :^{\textcolor{\coeffectcolor}{  \textcolor{\coeffectcolor}{1}  } }  \tau      \mathop{,}    \textcolor{\coeffectcolor}{  \textcolor{\coeffectcolor}{\overline{0} }  }\! \cdot \! \Gamma_{{\mathrm{2}}}     \vdash_{\mathit{cbncoeff} }   \ottmv{x}  : \tau }{%
{\ottdrulename{cbncoeff\_var}}{}%
}}

\newcommand{\ottdrulecbncoeffXXabs}[1]{\ottdrule[#1]{%
\ottpremise{   \textcolor{\coeffectcolor}{ \textcolor{\coeffectcolor}{\gamma} }\! \cdot \! \Gamma    \mathop{,}   \ottsym{(}   \ottmv{x}  :^{\textcolor{\coeffectcolor}{ \textcolor{\coeffectcolor}{q} } }  \tau_{{\mathrm{1}}}   \ottsym{)}    \vdash_{\mathit{cbncoeff} }   \ottnt{e}  : \tau_{{\mathrm{2}}} }%
\ottpremise{ \textcolor{\coeffectcolor}{ \textcolor{\coeffectcolor}{q}' }\;  \textcolor{\coeffectcolor}{\mathop{\leq_{\mathit{co} } } } \; \textcolor{\coeffectcolor}{ \textcolor{\coeffectcolor}{q} } }%
}{
  \textcolor{\coeffectcolor}{ \textcolor{\coeffectcolor}{\gamma} }\! \cdot \! \Gamma    \vdash_{\mathit{cbncoeff} }    \lambda^{\textcolor{\coeffectcolor}{ \textcolor{\coeffectcolor}{q} } }  \ottmv{x} . \ottnt{e}   :  \tau_{{\mathrm{1}}} ^{\textcolor{\coeffectcolor}{ \textcolor{\coeffectcolor}{q}' } } \rightarrow  \tau_{{\mathrm{2}}}  }{%
{\ottdrulename{cbncoeff\_abs}}{}%
}}

\newcommand{\ottdrulecbncoeffXXapp}[1]{\ottdrule[#1]{%
\ottpremise{  \textcolor{\coeffectcolor}{ \textcolor{\coeffectcolor}{\gamma}_{{\mathrm{1}}} }\! \cdot \! \Gamma    \vdash_{\mathit{cbncoeff} }   \ottnt{e_{{\mathrm{1}}}}  :  \tau_{{\mathrm{1}}} ^{\textcolor{\coeffectcolor}{ \textcolor{\coeffectcolor}{q} } } \rightarrow  \tau_{{\mathrm{2}}}  }%
\ottpremise{  \textcolor{\coeffectcolor}{ \textcolor{\coeffectcolor}{\gamma}_{{\mathrm{2}}} }\! \cdot \! \Gamma    \vdash_{\mathit{cbncoeff} }   \ottnt{e_{{\mathrm{2}}}}  : \tau_{{\mathrm{1}}} }%
}{
  \textcolor{\coeffectcolor}{  \textcolor{\coeffectcolor}{ \textcolor{\coeffectcolor}{\gamma}_{{\mathrm{1}}} \ottsym{+}  \textcolor{\coeffectcolor}{ \textcolor{\coeffectcolor}{q} \cdot \textcolor{\coeffectcolor}{\gamma}_{{\mathrm{2}}} }  }  }\! \cdot \! \Gamma    \vdash_{\mathit{cbncoeff} }   \ottnt{e_{{\mathrm{1}}}} \, \ottnt{e_{{\mathrm{2}}}}  : \tau_{{\mathrm{2}}} }{%
{\ottdrulename{cbncoeff\_app}}{}%
}}

\newcommand{\ottdrulecbncoeffXXunit}[1]{\ottdrule[#1]{%
}{
  \textcolor{\coeffectcolor}{  \textcolor{\coeffectcolor}{\overline{0} }  }\! \cdot \! \Gamma    \vdash_{\mathit{cbncoeff} }   \ottsym{()}  : \ottkw{unit} }{%
{\ottdrulename{cbncoeff\_unit}}{}%
}}

\newcommand{\ottdrulecbncoeffXXseq}[1]{\ottdrule[#1]{%
\ottpremise{  \textcolor{\coeffectcolor}{ \textcolor{\coeffectcolor}{\gamma}_{{\mathrm{1}}} }\! \cdot \! \Gamma    \vdash_{\mathit{cbncoeff} }   \ottnt{e_{{\mathrm{1}}}}  : \ottkw{unit} }%
\ottpremise{  \textcolor{\coeffectcolor}{ \textcolor{\coeffectcolor}{\gamma}_{{\mathrm{2}}} }\! \cdot \! \Gamma    \vdash_{\mathit{cbncoeff} }   \ottnt{e_{{\mathrm{2}}}}  : \tau }%
}{
  \textcolor{\coeffectcolor}{  \textcolor{\coeffectcolor}{ \textcolor{\coeffectcolor}{\gamma}_{{\mathrm{1}}} \ottsym{+} \textcolor{\coeffectcolor}{\gamma}_{{\mathrm{2}}} }  }\! \cdot \! \Gamma    \vdash_{\mathit{cbncoeff} }    \ottnt{e_{{\mathrm{1}}}}  ;  \ottnt{e_{{\mathrm{2}}}}   : \tau }{%
{\ottdrulename{cbncoeff\_seq}}{}%
}}

\newcommand{\ottdrulecbncoeffXXwith}[1]{\ottdrule[#1]{%
\ottpremise{  \textcolor{\coeffectcolor}{ \textcolor{\coeffectcolor}{\gamma} }\! \cdot \! \Gamma    \vdash_{\mathit{cbncoeff} }   \ottnt{e_{{\mathrm{1}}}}  : \tau_{{\mathrm{1}}} }%
\ottpremise{  \textcolor{\coeffectcolor}{ \textcolor{\coeffectcolor}{\gamma} }\! \cdot \! \Gamma    \vdash_{\mathit{cbncoeff} }   \ottnt{e_{{\mathrm{2}}}}  : \tau_{{\mathrm{2}}} }%
}{
  \textcolor{\coeffectcolor}{ \textcolor{\coeffectcolor}{\gamma} }\! \cdot \! \Gamma    \vdash_{\mathit{cbncoeff} }    \langle  \ottnt{e_{{\mathrm{1}}}} , \ottnt{e_{{\mathrm{2}}}}  \rangle   :  \tau_{{\mathrm{1}}}  \mathop{\&}   \tau_{{\mathrm{2}}}  }{%
{\ottdrulename{cbncoeff\_with}}{}%
}}

\newcommand{\ottdrulecbncoeffXXfst}[1]{\ottdrule[#1]{%
\ottpremise{  \textcolor{\coeffectcolor}{ \textcolor{\coeffectcolor}{\gamma} }\! \cdot \! \Gamma    \vdash_{\mathit{cbncoeff} }   \ottnt{e}  :  \tau_{{\mathrm{1}}}  \mathop{\&}   \tau_{{\mathrm{2}}}  }%
}{
  \textcolor{\coeffectcolor}{ \textcolor{\coeffectcolor}{\gamma} }\! \cdot \! \Gamma    \vdash_{\mathit{cbncoeff} }   \ottnt{e}  \ottsym{.}  \ottsym{1}  : \tau_{{\mathrm{1}}} }{%
{\ottdrulename{cbncoeff\_fst}}{}%
}}

\newcommand{\ottdrulecbncoeffXXsnd}[1]{\ottdrule[#1]{%
\ottpremise{  \textcolor{\coeffectcolor}{ \textcolor{\coeffectcolor}{\gamma} }\! \cdot \! \Gamma    \vdash_{\mathit{cbncoeff} }   \ottnt{e}  :  \tau_{{\mathrm{1}}}  \mathop{\&}   \tau_{{\mathrm{2}}}  }%
}{
  \textcolor{\coeffectcolor}{ \textcolor{\coeffectcolor}{\gamma} }\! \cdot \! \Gamma    \vdash_{\mathit{cbncoeff} }   \ottnt{e}  \ottsym{.}  \ottsym{2}  : \tau_{{\mathrm{2}}} }{%
{\ottdrulename{cbncoeff\_snd}}{}%
}}

\newcommand{\ottdrulecbncoeffXXinl}[1]{\ottdrule[#1]{%
\ottpremise{  \textcolor{\coeffectcolor}{ \textcolor{\coeffectcolor}{\gamma} }\! \cdot \! \Gamma    \vdash_{\mathit{cbncoeff} }   \ottnt{e}  : \tau_{{\mathrm{1}}} }%
}{
  \textcolor{\coeffectcolor}{ \textcolor{\coeffectcolor}{\gamma} }\! \cdot \! \Gamma    \vdash_{\mathit{cbncoeff} }   \ottkw{inl} \, \ottnt{e}  : \tau_{{\mathrm{1}}}  \ottsym{+}  \tau_{{\mathrm{2}}} }{%
{\ottdrulename{cbncoeff\_inl}}{}%
}}

\newcommand{\ottdrulecbncoeffXXinr}[1]{\ottdrule[#1]{%
\ottpremise{  \textcolor{\coeffectcolor}{ \textcolor{\coeffectcolor}{\gamma} }\! \cdot \! \Gamma    \vdash_{\mathit{cbncoeff} }   \ottnt{e}  : \tau_{{\mathrm{2}}} }%
}{
  \textcolor{\coeffectcolor}{ \textcolor{\coeffectcolor}{\gamma} }\! \cdot \! \Gamma    \vdash_{\mathit{cbncoeff} }   \ottkw{inr} \, \ottnt{e}  : \tau_{{\mathrm{1}}}  \ottsym{+}  \tau_{{\mathrm{2}}} }{%
{\ottdrulename{cbncoeff\_inr}}{}%
}}

\newcommand{\ottdrulecbncoeffXXcase}[1]{\ottdrule[#1]{%
\ottpremise{  \textcolor{\coeffectcolor}{ \textcolor{\coeffectcolor}{\gamma}_{{\mathrm{1}}} }\! \cdot \! \Gamma    \vdash_{\mathit{cbncoeff} }   \ottnt{e}  : \tau_{{\mathrm{1}}}  \ottsym{+}  \tau_{{\mathrm{2}}} }%
\ottpremise{   \textcolor{\coeffectcolor}{ \textcolor{\coeffectcolor}{\gamma}_{{\mathrm{2}}} }\! \cdot \! \Gamma    \mathop{,}    \ottmv{x_{{\mathrm{1}}}}  :^{\textcolor{\coeffectcolor}{ \textcolor{\coeffectcolor}{q} } }  \tau_{{\mathrm{1}}}     \vdash_{\mathit{cbncoeff} }   \ottnt{e_{{\mathrm{1}}}}  : \tau }%
\ottpremise{   \textcolor{\coeffectcolor}{ \textcolor{\coeffectcolor}{\gamma}_{{\mathrm{2}}} }\! \cdot \! \Gamma    \mathop{,}    \ottmv{x_{{\mathrm{2}}}}  :^{\textcolor{\coeffectcolor}{ \textcolor{\coeffectcolor}{q} } }  \tau_{{\mathrm{2}}}     \vdash_{\mathit{cbncoeff} }   \ottnt{e_{{\mathrm{2}}}}  : \tau }%
\ottpremise{ \textcolor{\coeffectcolor}{ \textcolor{\coeffectcolor}{q} }\;  \textcolor{\coeffectcolor}{\mathop{\leq_{\mathit{co} } } } \; \textcolor{\coeffectcolor}{  \textcolor{\coeffectcolor}{1}  } }%
}{
  \textcolor{\coeffectcolor}{  \textcolor{\coeffectcolor}{  \textcolor{\coeffectcolor}{ \textcolor{\coeffectcolor}{q} \cdot \textcolor{\coeffectcolor}{\gamma}_{{\mathrm{1}}} }  \ottsym{+} \textcolor{\coeffectcolor}{\gamma}_{{\mathrm{2}}} }  }\! \cdot \! \Gamma    \vdash_{\mathit{cbncoeff} }    \ottkw{case}_{\textcolor{\coeffectcolor}{ \textcolor{\coeffectcolor}{q} } }\;  \ottnt{e} \; \ottkw{of}\;\ottkw{inl}\; \ottmv{x_{{\mathrm{1}}}}  \rightarrow\;  \ottnt{e_{{\mathrm{1}}}}  ; \ottkw{inr}\; \ottmv{x_{{\mathrm{2}}}}  \rightarrow\;  \ottnt{e_{{\mathrm{2}}}}   : \tau }{%
{\ottdrulename{cbncoeff\_case}}{}%
}}

\newcommand{\ottdrulecbncoeffXXbox}[1]{\ottdrule[#1]{%
\ottpremise{  \textcolor{\coeffectcolor}{ \textcolor{\coeffectcolor}{\gamma}_{{\mathrm{1}}} }\! \cdot \! \Gamma    \vdash_{\mathit{cbncoeff} }   \ottnt{e}  : \tau }%
}{
  \textcolor{\coeffectcolor}{  \textcolor{\coeffectcolor}{ \textcolor{\coeffectcolor}{q} \cdot \textcolor{\coeffectcolor}{\gamma}_{{\mathrm{1}}} }  }\! \cdot \! \Gamma    \vdash_{\mathit{cbncoeff} }    \ottkw{box} _{\textcolor{\coeffectcolor}{ \textcolor{\coeffectcolor}{q} } }\  \ottnt{e}   :  \square_{\textcolor{\coeffectcolor}{ \textcolor{\coeffectcolor}{q} } }\;  \tau  }{%
{\ottdrulename{cbncoeff\_box}}{}%
}}

\newcommand{\ottdrulecbncoeffXXunbox}[1]{\ottdrule[#1]{%
\ottpremise{ \textcolor{\coeffectcolor}{ \textcolor{\coeffectcolor}{q}' }\; = \; \textcolor{\coeffectcolor}{  \textcolor{\coeffectcolor}{q}_{{\mathrm{2}}} \ \|\ \textcolor{\coeffectcolor}{1}  } }%
\ottpremise{  \textcolor{\coeffectcolor}{ \textcolor{\coeffectcolor}{\gamma}_{{\mathrm{1}}} }\! \cdot \! \Gamma    \vdash_{\mathit{cbncoeff} }   \ottnt{e_{{\mathrm{1}}}}  :  \square_{\textcolor{\coeffectcolor}{ \textcolor{\coeffectcolor}{q}_{{\mathrm{1}}} } }\;  \tau  }%
\ottpremise{   \textcolor{\coeffectcolor}{ \textcolor{\coeffectcolor}{\gamma}_{{\mathrm{2}}} }\! \cdot \! \Gamma    \mathop{,}    \ottmv{x}  :^{\textcolor{\coeffectcolor}{  \textcolor{\coeffectcolor}{ \textcolor{\coeffectcolor}{q}_{{\mathrm{1}}}   \cdot   \textcolor{\coeffectcolor}{q}' }  } }  \tau     \vdash_{\mathit{cbncoeff} }   \ottnt{e_{{\mathrm{2}}}}  : \tau' }%
}{
  \textcolor{\coeffectcolor}{  \textcolor{\coeffectcolor}{  \textcolor{\coeffectcolor}{ \textcolor{\coeffectcolor}{q}' \cdot \textcolor{\coeffectcolor}{\gamma}_{{\mathrm{1}}} }  \ottsym{+} \textcolor{\coeffectcolor}{\gamma}_{{\mathrm{2}}} }  }\! \cdot \! \Gamma    \vdash_{\mathit{cbncoeff} }    \ottkw{unbox} _{\textcolor{\coeffectcolor}{ \textcolor{\coeffectcolor}{q}_{{\mathrm{2}}} } }\  \ottmv{x}  =  \ottnt{e_{{\mathrm{1}}}} \  \ottkw{in} \  \ottnt{e_{{\mathrm{2}}}}   : \tau' }{%
{\ottdrulename{cbncoeff\_unbox}}{}%
}}

\newcommand{\ottdrulecbncoeffXXsub}[1]{\ottdrule[#1]{%
\ottpremise{  \textcolor{\coeffectcolor}{ \textcolor{\coeffectcolor}{\gamma}_{{\mathrm{1}}} }\! \cdot \! \Gamma    \vdash_{\mathit{cbncoeff} }   \ottnt{e}  : \tau }%
\ottpremise{ \textcolor{\coeffectcolor}{ \textcolor{\coeffectcolor}{\gamma}_{{\mathrm{2}}} }\; \textcolor{\coeffectcolor}{\mathop{\leq_{\mathit{co} } } } \; \textcolor{\coeffectcolor}{ \textcolor{\coeffectcolor}{\gamma}_{{\mathrm{1}}} } }%
}{
  \textcolor{\coeffectcolor}{ \textcolor{\coeffectcolor}{\gamma}_{{\mathrm{2}}} }\! \cdot \! \Gamma    \vdash_{\mathit{cbncoeff} }   \ottnt{e}  : \tau }{%
{\ottdrulename{cbncoeff\_sub}}{}%
}}

\newcommand{\ottdefncbncoeffXXtyping}[1]{\begin{ottdefnblock}[#1]{$ \Psi   \vdash_{\mathit{cbncoeff} }   \ottnt{e}  : \tau $}{\ottcom{Graded CBN coeffect system}}
\ottusedrule{\ottdrulecbncoeffXXvar{}}
\ottusedrule{\ottdrulecbncoeffXXabs{}}
\ottusedrule{\ottdrulecbncoeffXXapp{}}
\ottusedrule{\ottdrulecbncoeffXXunit{}}
\ottusedrule{\ottdrulecbncoeffXXseq{}}
\ottusedrule{\ottdrulecbncoeffXXwith{}}
\ottusedrule{\ottdrulecbncoeffXXfst{}}
\ottusedrule{\ottdrulecbncoeffXXsnd{}}
\ottusedrule{\ottdrulecbncoeffXXinl{}}
\ottusedrule{\ottdrulecbncoeffXXinr{}}
\ottusedrule{\ottdrulecbncoeffXXcase{}}
\ottusedrule{\ottdrulecbncoeffXXbox{}}
\ottusedrule{\ottdrulecbncoeffXXunbox{}}
\ottusedrule{\ottdrulecbncoeffXXsub{}}
\end{ottdefnblock}}

\newcommand{\ottdefnsJLambdaCBNCoeff}{
\ottdefncbncoeffXXtyping{}}

\newcommand{\ottdrulecbvcoeffXXvar}[1]{\ottdrule[#1]{%
}{
      \textcolor{\coeffectcolor}{  \textcolor{\coeffectcolor}{\overline{0} }  }\! \cdot \! \Gamma_{{\mathrm{1}}}     \mathop{,}    \ottmv{x}  :^{\textcolor{\coeffectcolor}{  \textcolor{\coeffectcolor}{1}  } }  \tau      \mathop{,}    \textcolor{\coeffectcolor}{  \textcolor{\coeffectcolor}{\overline{0} }  }\! \cdot \! \Gamma_{{\mathrm{2}}}     \vdash_{\mathit{cbvcoeff} }   \ottmv{x}  : \tau }{%
{\ottdrulename{cbvcoeff\_var}}{}%
}}

\newcommand{\ottdrulecbvcoeffXXabs}[1]{\ottdrule[#1]{%
\ottpremise{   \textcolor{\coeffectcolor}{ \textcolor{\coeffectcolor}{\gamma} }\! \cdot \! \Gamma    \mathop{,}   \ottsym{(}   \ottmv{x}  :^{\textcolor{\coeffectcolor}{ \textcolor{\coeffectcolor}{q} } }  \tau_{{\mathrm{1}}}   \ottsym{)}    \vdash_{\mathit{cbvcoeff} }   \ottnt{e}  : \tau_{{\mathrm{2}}} }%
\ottpremise{ \textcolor{\coeffectcolor}{ \textcolor{\coeffectcolor}{q}' }\;  \textcolor{\coeffectcolor}{\mathop{\leq_{\mathit{co} } } } \; \textcolor{\coeffectcolor}{ \textcolor{\coeffectcolor}{q} } }%
}{
  \textcolor{\coeffectcolor}{ \textcolor{\coeffectcolor}{\gamma} }\! \cdot \! \Gamma    \vdash_{\mathit{cbvcoeff} }    \lambda^{\textcolor{\coeffectcolor}{ \textcolor{\coeffectcolor}{q} } }  \ottmv{x} . \ottnt{e}   :  \tau_{{\mathrm{1}}} ^{\textcolor{\coeffectcolor}{ \textcolor{\coeffectcolor}{q}' } } \rightarrow  \tau_{{\mathrm{2}}}  }{%
{\ottdrulename{cbvcoeff\_abs}}{}%
}}

\newcommand{\ottdrulecbvcoeffXXapp}[1]{\ottdrule[#1]{%
\ottpremise{ \textcolor{\coeffectcolor}{ \textcolor{\coeffectcolor}{q}' }\; = \; \textcolor{\coeffectcolor}{  \textcolor{\coeffectcolor}{q} \ \|\ \textcolor{\coeffectcolor}{1}  } }%
\ottpremise{  \textcolor{\coeffectcolor}{ \textcolor{\coeffectcolor}{\gamma}_{{\mathrm{1}}} }\! \cdot \! \Gamma    \vdash_{\mathit{cbvcoeff} }   \ottnt{e_{{\mathrm{1}}}}  :  \tau_{{\mathrm{1}}} ^{\textcolor{\coeffectcolor}{ \textcolor{\coeffectcolor}{q}' } } \rightarrow  \tau_{{\mathrm{2}}}  }%
\ottpremise{  \textcolor{\coeffectcolor}{ \textcolor{\coeffectcolor}{\gamma}_{{\mathrm{2}}} }\! \cdot \! \Gamma    \vdash_{\mathit{cbvcoeff} }   \ottnt{e_{{\mathrm{2}}}}  : \tau_{{\mathrm{1}}} }%
}{
  \textcolor{\coeffectcolor}{  \textcolor{\coeffectcolor}{ \textcolor{\coeffectcolor}{\gamma}_{{\mathrm{1}}} \ottsym{+}  \textcolor{\coeffectcolor}{ \textcolor{\coeffectcolor}{q}' \cdot \textcolor{\coeffectcolor}{\gamma}_{{\mathrm{2}}} }  }  }\! \cdot \! \Gamma    \vdash_{\mathit{cbvcoeff} }    \ottnt{e_{{\mathrm{1}}}} ^{ \textcolor{\coeffectcolor}{q} }  \ottnt{e_{{\mathrm{2}}}}   : \tau_{{\mathrm{2}}} }{%
{\ottdrulename{cbvcoeff\_app}}{}%
}}

\newcommand{\ottdrulecbvcoeffXXunit}[1]{\ottdrule[#1]{%
}{
  \textcolor{\coeffectcolor}{  \textcolor{\coeffectcolor}{\overline{0} }  }\! \cdot \! \Gamma    \vdash_{\mathit{cbvcoeff} }   \ottsym{()}  : \ottkw{unit} }{%
{\ottdrulename{cbvcoeff\_unit}}{}%
}}

\newcommand{\ottdrulecbvcoeffXXseq}[1]{\ottdrule[#1]{%
\ottpremise{  \textcolor{\coeffectcolor}{ \textcolor{\coeffectcolor}{\gamma}_{{\mathrm{1}}} }\! \cdot \! \Gamma    \vdash_{\mathit{cbvcoeff} }   \ottnt{e_{{\mathrm{1}}}}  : \ottkw{unit} }%
\ottpremise{  \textcolor{\coeffectcolor}{ \textcolor{\coeffectcolor}{\gamma}_{{\mathrm{2}}} }\! \cdot \! \Gamma    \vdash_{\mathit{cbvcoeff} }   \ottnt{e_{{\mathrm{2}}}}  : \tau }%
}{
  \textcolor{\coeffectcolor}{  \textcolor{\coeffectcolor}{ \textcolor{\coeffectcolor}{\gamma}_{{\mathrm{1}}} \ottsym{+} \textcolor{\coeffectcolor}{\gamma}_{{\mathrm{2}}} }  }\! \cdot \! \Gamma    \vdash_{\mathit{cbvcoeff} }    \ottnt{e_{{\mathrm{1}}}}  ;  \ottnt{e_{{\mathrm{2}}}}   : \tau }{%
{\ottdrulename{cbvcoeff\_seq}}{}%
}}

\newcommand{\ottdrulecbvcoeffXXpair}[1]{\ottdrule[#1]{%
\ottpremise{  \textcolor{\coeffectcolor}{ \textcolor{\coeffectcolor}{\gamma}_{{\mathrm{1}}} }\! \cdot \! \Gamma    \vdash_{\mathit{cbvcoeff} }   \ottnt{e_{{\mathrm{1}}}}  : \tau_{{\mathrm{1}}} }%
\ottpremise{  \textcolor{\coeffectcolor}{ \textcolor{\coeffectcolor}{\gamma}_{{\mathrm{2}}} }\! \cdot \! \Gamma    \vdash_{\mathit{cbvcoeff} }   \ottnt{e_{{\mathrm{2}}}}  : \tau_{{\mathrm{2}}} }%
}{
  \textcolor{\coeffectcolor}{  \textcolor{\coeffectcolor}{ \textcolor{\coeffectcolor}{\gamma}_{{\mathrm{1}}} \ottsym{+} \textcolor{\coeffectcolor}{\gamma}_{{\mathrm{2}}} }  }\! \cdot \! \Gamma    \vdash_{\mathit{cbvcoeff} }   \ottsym{(}  \ottnt{e_{{\mathrm{1}}}}  \ottsym{,}  \ottnt{e_{{\mathrm{2}}}}  \ottsym{)}  :  \tau_{{\mathrm{1}}}  \otimes  \tau_{{\mathrm{2}}}  }{%
{\ottdrulename{cbvcoeff\_pair}}{}%
}}

\newcommand{\ottdrulecbvcoeffXXsplit}[1]{\ottdrule[#1]{%
\ottpremise{ \textcolor{\coeffectcolor}{ \textcolor{\coeffectcolor}{q}' }\; = \; \textcolor{\coeffectcolor}{  \textcolor{\coeffectcolor}{q} \ \|\ \textcolor{\coeffectcolor}{1}  } }%
\ottpremise{  \textcolor{\coeffectcolor}{ \textcolor{\coeffectcolor}{\gamma} }\! \cdot \! \Gamma    \vdash_{\mathit{cbvcoeff} }   \ottnt{e_{{\mathrm{1}}}}  :  \tau_{{\mathrm{1}}}  \otimes  \tau_{{\mathrm{2}}}  }%
\ottpremise{    \textcolor{\coeffectcolor}{ \textcolor{\coeffectcolor}{\gamma}_{{\mathrm{2}}} }\! \cdot \! \Gamma    \mathop{,}    \ottmv{x_{{\mathrm{1}}}}  :^{\textcolor{\coeffectcolor}{ \textcolor{\coeffectcolor}{q}' } }  \tau_{{\mathrm{1}}}     \mathop{,}    \ottmv{x_{{\mathrm{2}}}}  :^{\textcolor{\coeffectcolor}{ \textcolor{\coeffectcolor}{q}' } }  \tau_{{\mathrm{2}}}     \vdash_{\mathit{cbvcoeff} }   \ottnt{e_{{\mathrm{2}}}}  : \tau }%
}{
  \textcolor{\coeffectcolor}{  \textcolor{\coeffectcolor}{  \textcolor{\coeffectcolor}{ \textcolor{\coeffectcolor}{q}' \cdot \textcolor{\coeffectcolor}{\gamma}_{{\mathrm{1}}} }  \ottsym{+} \textcolor{\coeffectcolor}{\gamma}_{{\mathrm{2}}} }  }\! \cdot \! \Gamma    \vdash_{\mathit{cbvcoeff} }    \ottkw{let}_{ \textcolor{\coeffectcolor}{q} }\; ( \ottmv{x_{{\mathrm{1}}}} ,  \ottmv{x_{{\mathrm{2}}}} ) =  \ottnt{e_{{\mathrm{1}}}} \; \ottkw{in}\;  \ottnt{e_{{\mathrm{2}}}}   : \tau }{%
{\ottdrulename{cbvcoeff\_split}}{}%
}}

\newcommand{\ottdrulecbvcoeffXXinl}[1]{\ottdrule[#1]{%
\ottpremise{  \textcolor{\coeffectcolor}{ \textcolor{\coeffectcolor}{\gamma} }\! \cdot \! \Gamma    \vdash_{\mathit{cbvcoeff} }   \ottnt{e}  : \tau_{{\mathrm{1}}} }%
}{
  \textcolor{\coeffectcolor}{ \textcolor{\coeffectcolor}{\gamma} }\! \cdot \! \Gamma    \vdash_{\mathit{cbvcoeff} }   \ottkw{inl} \, \ottnt{e}  : \tau_{{\mathrm{1}}}  \ottsym{+}  \tau_{{\mathrm{2}}} }{%
{\ottdrulename{cbvcoeff\_inl}}{}%
}}

\newcommand{\ottdrulecbvcoeffXXinr}[1]{\ottdrule[#1]{%
\ottpremise{  \textcolor{\coeffectcolor}{ \textcolor{\coeffectcolor}{\gamma} }\! \cdot \! \Gamma    \vdash_{\mathit{cbvcoeff} }   \ottnt{e}  : \tau_{{\mathrm{2}}} }%
}{
  \textcolor{\coeffectcolor}{ \textcolor{\coeffectcolor}{\gamma} }\! \cdot \! \Gamma    \vdash_{\mathit{cbvcoeff} }   \ottkw{inr} \, \ottnt{e}  : \tau_{{\mathrm{1}}}  \ottsym{+}  \tau_{{\mathrm{2}}} }{%
{\ottdrulename{cbvcoeff\_inr}}{}%
}}

\newcommand{\ottdrulecbvcoeffXXcase}[1]{\ottdrule[#1]{%
\ottpremise{  \textcolor{\coeffectcolor}{ \textcolor{\coeffectcolor}{\gamma}_{{\mathrm{1}}} }\! \cdot \! \Gamma    \vdash_{\mathit{cbvcoeff} }   \ottnt{e}  : \tau_{{\mathrm{1}}}  \ottsym{+}  \tau_{{\mathrm{2}}} }%
\ottpremise{   \textcolor{\coeffectcolor}{ \textcolor{\coeffectcolor}{\gamma}_{{\mathrm{2}}} }\! \cdot \! \Gamma    \mathop{,}    \ottmv{x_{{\mathrm{1}}}}  :^{\textcolor{\coeffectcolor}{ \textcolor{\coeffectcolor}{q} } }  \tau_{{\mathrm{1}}}     \vdash_{\mathit{cbvcoeff} }   \ottnt{e_{{\mathrm{1}}}}  : \tau }%
\ottpremise{   \textcolor{\coeffectcolor}{ \textcolor{\coeffectcolor}{\gamma}_{{\mathrm{2}}} }\! \cdot \! \Gamma    \mathop{,}    \ottmv{x_{{\mathrm{2}}}}  :^{\textcolor{\coeffectcolor}{ \textcolor{\coeffectcolor}{q} } }  \tau_{{\mathrm{2}}}     \vdash_{\mathit{cbvcoeff} }   \ottnt{e_{{\mathrm{2}}}}  : \tau }%
\ottpremise{ \textcolor{\coeffectcolor}{ \textcolor{\coeffectcolor}{q} }\;  \textcolor{\coeffectcolor}{\mathop{\leq_{\mathit{co} } } } \; \textcolor{\coeffectcolor}{  \textcolor{\coeffectcolor}{1}  } }%
}{
  \textcolor{\coeffectcolor}{  \textcolor{\coeffectcolor}{  \textcolor{\coeffectcolor}{ \textcolor{\coeffectcolor}{q} \cdot \textcolor{\coeffectcolor}{\gamma}_{{\mathrm{1}}} }  \ottsym{+} \textcolor{\coeffectcolor}{\gamma}_{{\mathrm{2}}} }  }\! \cdot \! \Gamma    \vdash_{\mathit{cbvcoeff} }    \ottkw{case}_{\textcolor{\coeffectcolor}{ \textcolor{\coeffectcolor}{q} } }\;  \ottnt{e} \; \ottkw{of}\;\ottkw{inl}\; \ottmv{x_{{\mathrm{1}}}}  \rightarrow\;  \ottnt{e_{{\mathrm{1}}}}  ; \ottkw{inr}\; \ottmv{x_{{\mathrm{2}}}}  \rightarrow\;  \ottnt{e_{{\mathrm{2}}}}   : \tau }{%
{\ottdrulename{cbvcoeff\_case}}{}%
}}

\newcommand{\ottdrulecbvcoeffXXbox}[1]{\ottdrule[#1]{%
\ottpremise{ \textcolor{\coeffectcolor}{ \textcolor{\coeffectcolor}{q}' }\; = \; \textcolor{\coeffectcolor}{  \textcolor{\coeffectcolor}{q} \ \|\ \textcolor{\coeffectcolor}{1}  } }%
\ottpremise{  \textcolor{\coeffectcolor}{ \textcolor{\coeffectcolor}{\gamma}_{{\mathrm{1}}} }\! \cdot \! \Gamma    \vdash_{\mathit{cbvcoeff} }   \ottnt{e}  : \tau }%
}{
  \textcolor{\coeffectcolor}{  \textcolor{\coeffectcolor}{ \textcolor{\coeffectcolor}{q}' \cdot \textcolor{\coeffectcolor}{\gamma}_{{\mathrm{1}}} }  }\! \cdot \! \Gamma    \vdash_{\mathit{cbvcoeff} }    \ottkw{box} _{\textcolor{\coeffectcolor}{ \textcolor{\coeffectcolor}{q} } }\  \ottnt{e}   :  \square_{\textcolor{\coeffectcolor}{ \textcolor{\coeffectcolor}{q}' } }\;  \tau  }{%
{\ottdrulename{cbvcoeff\_box}}{}%
}}

\newcommand{\ottdrulecbvcoeffXXunbox}[1]{\ottdrule[#1]{%
\ottpremise{ \textcolor{\coeffectcolor}{ \textcolor{\coeffectcolor}{q}' }\; = \; \textcolor{\coeffectcolor}{  \textcolor{\coeffectcolor}{q}_{{\mathrm{2}}} \ \|\ \textcolor{\coeffectcolor}{1}  } }%
\ottpremise{  \textcolor{\coeffectcolor}{ \textcolor{\coeffectcolor}{\gamma}_{{\mathrm{1}}} }\! \cdot \! \Gamma    \vdash_{\mathit{cbvcoeff} }   \ottnt{e_{{\mathrm{1}}}}  :  \square_{\textcolor{\coeffectcolor}{ \textcolor{\coeffectcolor}{q}_{{\mathrm{1}}} } }\;  \tau  }%
\ottpremise{   \textcolor{\coeffectcolor}{ \textcolor{\coeffectcolor}{\gamma}_{{\mathrm{2}}} }\! \cdot \! \Gamma    \mathop{,}    \ottmv{x}  :^{\textcolor{\coeffectcolor}{  \textcolor{\coeffectcolor}{ \textcolor{\coeffectcolor}{q}_{{\mathrm{1}}}   \cdot   \textcolor{\coeffectcolor}{q}' }  } }  \tau     \vdash_{\mathit{cbvcoeff} }   \ottnt{e_{{\mathrm{2}}}}  : \tau' }%
}{
  \textcolor{\coeffectcolor}{  \textcolor{\coeffectcolor}{  \textcolor{\coeffectcolor}{ \textcolor{\coeffectcolor}{q}' \cdot \textcolor{\coeffectcolor}{\gamma}_{{\mathrm{1}}} }  \ottsym{+} \textcolor{\coeffectcolor}{\gamma}_{{\mathrm{2}}} }  }\! \cdot \! \Gamma    \vdash_{\mathit{cbvcoeff} }    \ottkw{unbox} _{\textcolor{\coeffectcolor}{ \textcolor{\coeffectcolor}{q}_{{\mathrm{2}}} } }\  \ottmv{x}  =  \ottnt{e_{{\mathrm{1}}}} \  \ottkw{in} \  \ottnt{e_{{\mathrm{2}}}}   : \tau' }{%
{\ottdrulename{cbvcoeff\_unbox}}{}%
}}

\newcommand{\ottdrulecbvcoeffXXsub}[1]{\ottdrule[#1]{%
\ottpremise{  \textcolor{\coeffectcolor}{ \textcolor{\coeffectcolor}{\gamma}_{{\mathrm{1}}} }\! \cdot \! \Gamma    \vdash_{\mathit{cbvcoeff} }   \ottnt{e}  : \tau }%
\ottpremise{ \textcolor{\coeffectcolor}{ \textcolor{\coeffectcolor}{\gamma}_{{\mathrm{2}}} }\; \textcolor{\coeffectcolor}{\mathop{\leq_{\mathit{co} } } } \; \textcolor{\coeffectcolor}{ \textcolor{\coeffectcolor}{\gamma}_{{\mathrm{1}}} } }%
}{
  \textcolor{\coeffectcolor}{ \textcolor{\coeffectcolor}{\gamma}_{{\mathrm{2}}} }\! \cdot \! \Gamma    \vdash_{\mathit{cbvcoeff} }   \ottnt{e}  : \tau }{%
{\ottdrulename{cbvcoeff\_sub}}{}%
}}

\newcommand{\ottdefncbvcoeffXXtyping}[1]{\begin{ottdefnblock}[#1]{$ \Psi   \vdash_{\mathit{cbvcoeff} }   \ottnt{e}  : \tau $}{\ottcom{Graded CBV coeffect system}}
\ottusedrule{\ottdrulecbvcoeffXXvar{}}
\ottusedrule{\ottdrulecbvcoeffXXabs{}}
\ottusedrule{\ottdrulecbvcoeffXXapp{}}
\ottusedrule{\ottdrulecbvcoeffXXunit{}}
\ottusedrule{\ottdrulecbvcoeffXXseq{}}
\ottusedrule{\ottdrulecbvcoeffXXpair{}}
\ottusedrule{\ottdrulecbvcoeffXXsplit{}}
\ottusedrule{\ottdrulecbvcoeffXXinl{}}
\ottusedrule{\ottdrulecbvcoeffXXinr{}}
\ottusedrule{\ottdrulecbvcoeffXXcase{}}
\ottusedrule{\ottdrulecbvcoeffXXbox{}}
\ottusedrule{\ottdrulecbvcoeffXXunbox{}}
\ottusedrule{\ottdrulecbvcoeffXXsub{}}
\end{ottdefnblock}}

\newcommand{\ottdefnsJLambdaCBVCoeff}{
\ottdefncbvcoeffXXtyping{}}

\newcommand{\ottdrulelamXXcomXXvar}[1]{\ottdrule[#1]{%
}{
\ottmv{x}  \ottsym{:}  \tau  \vdash_{\mathit{com} }  \ottmv{x}  \ottsym{:}  \tau}{%
{\ottdrulename{lam\_com\_var}}{}%
}}

\newcommand{\ottdrulelamXXcomXXabs}[1]{\ottdrule[#1]{%
\ottpremise{ \Gamma   \mathop{,}   \ottmv{x}  \ottsym{:}  \tau_{{\mathrm{1}}}   \vdash_{\mathit{com} }  \ottnt{e}  \ottsym{:}  \tau_{{\mathrm{2}}}}%
}{
\Gamma  \vdash_{\mathit{com} }   \lambda  \ottmv{x} . \ottnt{e}   \ottsym{:}  \tau_{{\mathrm{1}}}  \multimap  \tau_{{\mathrm{2}}}}{%
{\ottdrulename{lam\_com\_abs}}{}%
}}

\newcommand{\ottdrulelamXXcomXXapp}[1]{\ottdrule[#1]{%
\ottpremise{\Gamma_{{\mathrm{1}}}  \vdash_{\mathit{com} }  \ottnt{e_{{\mathrm{1}}}}  \ottsym{:}  \tau_{{\mathrm{1}}}  \multimap  \tau_{{\mathrm{2}}}}%
\ottpremise{\Gamma_{{\mathrm{2}}}  \vdash_{\mathit{com} }  \ottnt{e_{{\mathrm{2}}}}  \ottsym{:}  \tau_{{\mathrm{1}}}}%
}{
 \Gamma_{{\mathrm{1}}}   \mathop{,}   \Gamma_{{\mathrm{2}}}   \vdash_{\mathit{com} }  \ottnt{e_{{\mathrm{1}}}} \, \ottnt{e_{{\mathrm{2}}}}  \ottsym{:}  \tau_{{\mathrm{2}}}}{%
{\ottdrulename{lam\_com\_app}}{}%
}}

\newcommand{\ottdrulelamXXcomXXunit}[1]{\ottdrule[#1]{%
}{
\varnothing  \vdash_{\mathit{com} }  \ottsym{()}  \ottsym{:}  \ottkw{unit}}{%
{\ottdrulename{lam\_com\_unit}}{}%
}}

\newcommand{\ottdrulelamXXcomXXsequence}[1]{\ottdrule[#1]{%
\ottpremise{\Gamma_{{\mathrm{1}}}  \vdash_{\mathit{com} }  \ottnt{e_{{\mathrm{1}}}}  \ottsym{:}  \ottkw{unit}}%
\ottpremise{\Gamma_{{\mathrm{2}}}  \vdash_{\mathit{com} }  \ottnt{e_{{\mathrm{2}}}}  \ottsym{:}  \tau}%
\ottpremise{\Gamma  \leq   \Gamma_{{\mathrm{1}}}  +  \Gamma_{{\mathrm{2}}} }%
}{
\Gamma  \vdash_{\mathit{com} }   \ottnt{e_{{\mathrm{1}}}}  ;  \ottnt{e_{{\mathrm{2}}}}   \ottsym{:}  \tau}{%
{\ottdrulename{lam\_com\_sequence}}{}%
}}

\newcommand{\ottdrulelamXXcomXXpair}[1]{\ottdrule[#1]{%
\ottpremise{\Gamma_{{\mathrm{1}}}  \vdash_{\mathit{com} }  \ottnt{e_{{\mathrm{1}}}}  \ottsym{:}  \tau_{{\mathrm{1}}}}%
\ottpremise{\Gamma_{{\mathrm{2}}}  \vdash_{\mathit{com} }  \ottnt{e_{{\mathrm{2}}}}  \ottsym{:}  \tau_{{\mathrm{2}}}}%
\ottpremise{\Gamma  \leq   \Gamma_{{\mathrm{1}}}  +  \Gamma_{{\mathrm{2}}} }%
}{
\Gamma  \vdash_{\mathit{com} }  \ottsym{(}  \ottnt{e_{{\mathrm{1}}}}  \ottsym{,}  \ottnt{e_{{\mathrm{2}}}}  \ottsym{)}  \ottsym{:}   \tau_{{\mathrm{1}}}  \otimes  \tau_{{\mathrm{2}}} }{%
{\ottdrulename{lam\_com\_pair}}{}%
}}

\newcommand{\ottdrulelamXXcomXXsplit}[1]{\ottdrule[#1]{%
\ottpremise{\Gamma_{{\mathrm{1}}}  \vdash_{\mathit{com} }  \ottnt{e_{{\mathrm{1}}}}  \ottsym{:}   \tau_{{\mathrm{1}}}  \otimes  \tau_{{\mathrm{2}}} }%
\ottpremise{  \Gamma_{{\mathrm{2}}}   \mathop{,}   \ottmv{x_{{\mathrm{1}}}}  \ottsym{:}  \tau_{{\mathrm{1}}}    \mathop{,}   \ottmv{x_{{\mathrm{2}}}}  \ottsym{:}  \tau_{{\mathrm{2}}}   \vdash_{\mathit{com} }  \ottnt{e_{{\mathrm{2}}}}  \ottsym{:}  \tau}%
\ottpremise{\Gamma  \leq   \Gamma_{{\mathrm{1}}}  +  \Gamma_{{\mathrm{2}}} }%
}{
\Gamma  \vdash_{\mathit{com} }   \ottkw{let}\; ( \ottmv{x_{{\mathrm{1}}}} ,  \ottmv{x_{{\mathrm{2}}}} ) =  \ottnt{e_{{\mathrm{1}}}} \; \ottkw{in}\;  \ottnt{e_{{\mathrm{2}}}}   \ottsym{:}  \tau}{%
{\ottdrulename{lam\_com\_split}}{}%
}}

\newcommand{\ottdrulelamXXcomXXwith}[1]{\ottdrule[#1]{%
\ottpremise{\Gamma  \vdash_{\mathit{com} }  \ottnt{e_{{\mathrm{1}}}}  \ottsym{:}  \tau_{{\mathrm{1}}}}%
\ottpremise{\Gamma  \vdash_{\mathit{com} }  \ottnt{e_{{\mathrm{2}}}}  \ottsym{:}  \tau_{{\mathrm{2}}}}%
}{
\Gamma  \vdash_{\mathit{com} }   \langle  \ottnt{e_{{\mathrm{1}}}} , \ottnt{e_{{\mathrm{2}}}}  \rangle   \ottsym{:}   \tau_{{\mathrm{1}}}  \mathop{\&}   \tau_{{\mathrm{2}}} }{%
{\ottdrulename{lam\_com\_with}}{}%
}}

\newcommand{\ottdrulelamXXcomXXfst}[1]{\ottdrule[#1]{%
\ottpremise{\Gamma  \vdash_{\mathit{com} }  \ottnt{e}  \ottsym{:}   \tau_{{\mathrm{1}}}  \mathop{\&}   \tau_{{\mathrm{2}}} }%
}{
\Gamma  \vdash_{\mathit{com} }  \ottnt{e}  \ottsym{.}  \ottsym{1}  \ottsym{:}  \tau_{{\mathrm{1}}}}{%
{\ottdrulename{lam\_com\_fst}}{}%
}}

\newcommand{\ottdrulelamXXcomXXsnd}[1]{\ottdrule[#1]{%
\ottpremise{\Gamma  \vdash_{\mathit{com} }  \ottnt{e}  \ottsym{:}   \tau_{{\mathrm{1}}}  \mathop{\&}   \tau_{{\mathrm{2}}} }%
}{
\Gamma  \vdash_{\mathit{com} }  \ottnt{e}  \ottsym{.}  \ottsym{2}  \ottsym{:}  \tau_{{\mathrm{2}}}}{%
{\ottdrulename{lam\_com\_snd}}{}%
}}

\newcommand{\ottdrulelamXXcomXXinl}[1]{\ottdrule[#1]{%
\ottpremise{\Gamma  \vdash_{\mathit{com} }  \ottnt{e}  \ottsym{:}  \tau_{{\mathrm{1}}}}%
}{
\Gamma  \vdash_{\mathit{com} }  \ottkw{inl} \, \ottnt{e}  \ottsym{:}  \tau_{{\mathrm{1}}}  \ottsym{+}  \tau_{{\mathrm{2}}}}{%
{\ottdrulename{lam\_com\_inl}}{}%
}}

\newcommand{\ottdrulelamXXcomXXinr}[1]{\ottdrule[#1]{%
\ottpremise{\Gamma  \vdash_{\mathit{com} }  \ottnt{e}  \ottsym{:}  \tau_{{\mathrm{2}}}}%
}{
\Gamma  \vdash_{\mathit{com} }  \ottkw{inr} \, \ottnt{e}  \ottsym{:}  \tau_{{\mathrm{1}}}  \ottsym{+}  \tau_{{\mathrm{2}}}}{%
{\ottdrulename{lam\_com\_inr}}{}%
}}

\newcommand{\ottdrulelamXXcomXXcase}[1]{\ottdrule[#1]{%
\ottpremise{\Gamma_{{\mathrm{1}}}  \vdash_{\mathit{com} }  \ottnt{e}  \ottsym{:}  \tau_{{\mathrm{1}}}  \ottsym{+}  \tau_{{\mathrm{2}}}}%
\ottpremise{ \Gamma_{{\mathrm{2}}}   \mathop{,}   \ottmv{x}  \ottsym{:}  \tau_{{\mathrm{1}}}   \vdash_{\mathit{com} }  \ottnt{e_{{\mathrm{1}}}}  \ottsym{:}  \tau}%
\ottpremise{ \Gamma_{{\mathrm{2}}}   \mathop{,}   \ottmv{x}  \ottsym{:}  \tau_{{\mathrm{2}}}   \vdash_{\mathit{com} }  \ottnt{e_{{\mathrm{2}}}}  \ottsym{:}  \tau}%
\ottpremise{\Gamma  \leq   \Gamma_{{\mathrm{1}}}  +  \Gamma_{{\mathrm{2}}} }%
}{
\Gamma  \vdash_{\mathit{com} }   \ottkw{case}\;  \ottnt{e} \; \ottkw{of}\; \ottkw{inl}\;  \ottmv{x_{{\mathrm{1}}}}  \rightarrow  \ottnt{e_{{\mathrm{1}}}}  \ottkw{;} \ottkw{inr}\;  \ottmv{x_{{\mathrm{2}}}}  \rightarrow  \ottnt{e_{{\mathrm{2}}}}   \ottsym{:}  \tau}{%
{\ottdrulename{lam\_com\_case}}{}%
}}

\newcommand{\ottdrulelamXXcomXXextract}[1]{\ottdrule[#1]{%
\ottpremise{\Gamma  \vdash_{\mathit{com} }  \ottnt{e}  \ottsym{:}   \square_{\textcolor{\coeffectcolor}{ \textcolor{\coeffectcolor}{q} } }\;  \tau }%
\ottpremise{ \textcolor{\coeffectcolor}{ \textcolor{\coeffectcolor}{q} }\;  \textcolor{\coeffectcolor}{\mathop{\leq_{\mathit{co} } } } \; \textcolor{\coeffectcolor}{  \textcolor{\coeffectcolor}{1}  } }%
}{
\Gamma  \vdash_{\mathit{com} }  \ottkw{extract} \, \ottnt{e}  \ottsym{:}  \tau}{%
{\ottdrulename{lam\_com\_extract}}{}%
}}

\newcommand{\ottdrulelamXXcomXXdiscard}[1]{\ottdrule[#1]{%
\ottpremise{\Gamma_{{\mathrm{1}}}  \vdash_{\mathit{com} }  \ottnt{e_{{\mathrm{1}}}}  \ottsym{:}   \square_{\textcolor{\coeffectcolor}{ \textcolor{\coeffectcolor}{q} } }\;  \tau_{{\mathrm{1}}} }%
\ottpremise{\Gamma_{{\mathrm{2}}}  \vdash_{\mathit{com} }  \ottnt{e_{{\mathrm{2}}}}  \ottsym{:}  \tau_{{\mathrm{2}}}}%
\ottpremise{ \textcolor{\coeffectcolor}{ \textcolor{\coeffectcolor}{q} }\;  \textcolor{\coeffectcolor}{\mathop{\leq_{\mathit{co} } } } \; \textcolor{\coeffectcolor}{  \textcolor{\coeffectcolor}{0}  } }%
}{
 \Gamma_{{\mathrm{1}}}   \mathop{,}   \Gamma_{{\mathrm{2}}}   \vdash_{\mathit{com} }  \ottkw{discard} \, \_ \, \ottsym{=}  \ottnt{e_{{\mathrm{1}}}} \, \ottkw{in} \, \ottnt{e_{{\mathrm{2}}}}  \ottsym{:}  \tau_{{\mathrm{2}}}}{%
{\ottdrulename{lam\_com\_discard}}{}%
}}

\newcommand{\ottdrulelamXXcomXXdivide}[1]{\ottdrule[#1]{%
\ottpremise{\Gamma_{{\mathrm{1}}}  \vdash_{\mathit{com} }  \ottnt{e_{{\mathrm{1}}}}  \ottsym{:}   \square_{\textcolor{\coeffectcolor}{ \textcolor{\coeffectcolor}{q} } }\;  \tau_{{\mathrm{1}}} }%
\ottpremise{ \textcolor{\coeffectcolor}{ \textcolor{\coeffectcolor}{q} }\;  \textcolor{\coeffectcolor}{\mathop{\leq_{\mathit{co} } } } \; \textcolor{\coeffectcolor}{  \textcolor{\coeffectcolor}{ \textcolor{\coeffectcolor}{q}_{{\mathrm{1}}}  +  \textcolor{\coeffectcolor}{q}_{{\mathrm{2}}} }  } }%
\ottpremise{  \Gamma_{{\mathrm{2}}}   \mathop{,}   \ottmv{x}  \ottsym{:}   \square_{\textcolor{\coeffectcolor}{ \textcolor{\coeffectcolor}{q}_{{\mathrm{1}}} } }\;  \tau_{{\mathrm{1}}}     \mathop{,}   \ottmv{x}  \ottsym{:}   \square_{\textcolor{\coeffectcolor}{ \textcolor{\coeffectcolor}{q}_{{\mathrm{2}}} } }\;  \tau_{{\mathrm{2}}}    \vdash_{\mathit{com} }  \ottnt{e_{{\mathrm{2}}}}  \ottsym{:}  \tau_{{\mathrm{2}}}}%
}{
 \Gamma_{{\mathrm{1}}}   \mathop{,}   \Gamma_{{\mathrm{2}}}   \vdash_{\mathit{com} }   \ottkw{divide} \  \ottmv{x_{{\mathrm{1}}}} ^{ \textcolor{\coeffectcolor}{ \textcolor{\coeffectcolor}{q}_{{\mathrm{1}}} } },  \ottmv{x_{{\mathrm{2}}}} ^{ \textcolor{\coeffectcolor}{ \textcolor{\coeffectcolor}{q}_{{\mathrm{2}}} } } =  \ottnt{e_{{\mathrm{1}}}} \  \ottkw{in} \  \ottnt{e_{{\mathrm{2}}}}   \ottsym{:}  \tau_{{\mathrm{2}}}}{%
{\ottdrulename{lam\_com\_divide}}{}%
}}

\newcommand{\ottdrulelamXXcomXXextend}[1]{\ottdrule[#1]{%
\ottpremise{\Gamma_{{\mathrm{1}}}  \vdash_{\mathit{com} }  \ottnt{e_{{\mathrm{1}}}}  \ottsym{:}   \square_{\textcolor{\coeffectcolor}{ \textcolor{\coeffectcolor}{q}'_{{\mathrm{1}}} } }\;  \tau_{{\mathrm{1}}}  \, ... \, \Gamma_{\ottmv{k}}  \vdash_{\mathit{com} }  \ottnt{e_{\ottmv{k}}}  \ottsym{:}   \square_{\textcolor{\coeffectcolor}{ \textcolor{\coeffectcolor}{q}'_{\ottmv{k}} } }\;  \tau_{\ottmv{k}} }%
\ottpremise{ \textcolor{\coeffectcolor}{ \textcolor{\coeffectcolor}{q}'_{{\mathrm{1}}} }\;  \textcolor{\coeffectcolor}{\mathop{\leq_{\mathit{co} } } } \; \textcolor{\coeffectcolor}{  \textcolor{\coeffectcolor}{ \textcolor{\coeffectcolor}{q}   \cdot   \textcolor{\coeffectcolor}{q}_{{\mathrm{1}}} }  }  \, ... \,  \textcolor{\coeffectcolor}{ \textcolor{\coeffectcolor}{q}'_{\ottmv{k}} }\;  \textcolor{\coeffectcolor}{\mathop{\leq_{\mathit{co} } } } \; \textcolor{\coeffectcolor}{  \textcolor{\coeffectcolor}{ \textcolor{\coeffectcolor}{q}   \cdot   \textcolor{\coeffectcolor}{q}_{\ottmv{k}} }  } }%
\ottpremise{\ottmv{x_{{\mathrm{1}}}}  \ottsym{:}   \square_{\textcolor{\coeffectcolor}{ \textcolor{\coeffectcolor}{q}_{{\mathrm{1}}} } }\;  \tau_{{\mathrm{1}}}   \ottsym{,} \, ... \, \ottsym{,}  \ottmv{x_{\ottmv{k}}}  \ottsym{:}   \square_{\textcolor{\coeffectcolor}{ \textcolor{\coeffectcolor}{q}_{\ottmv{k}} } }\;  \tau_{\ottmv{k}}   \vdash_{\mathit{com} }  \ottnt{e'}  \ottsym{:}  \tau'}%
}{
\Gamma_{{\mathrm{1}}}  \ottsym{,} \, ... \, \ottsym{,}  \Gamma_{\ottmv{k}}  \vdash_{\mathit{com} }   \ottkw{extend} _{ \textcolor{\coeffectcolor}{ \textcolor{\coeffectcolor}{q} } }\   \ottmv{x_{{\mathrm{1}}}} ^{ \textcolor{\coeffectcolor}{ \textcolor{\coeffectcolor}{q}_{{\mathrm{1}}} } } =  \ottnt{e_{{\mathrm{1}}}}   \ottsym{,} \, ... \, \ottsym{,}   \ottmv{x_{\ottmv{k}}} ^{ \textcolor{\coeffectcolor}{ \textcolor{\coeffectcolor}{q}_{\ottmv{k}} } } =  \ottnt{e_{\ottmv{k}}}  \  \ottkw{in} \  \ottnt{e'}   \ottsym{:}   \square_{\textcolor{\coeffectcolor}{ \textcolor{\coeffectcolor}{q} } }\;  \tau' }{%
{\ottdrulename{lam\_com\_extend}}{}%
}}

\newcommand{\ottdefncomXXtyping}[1]{\begin{ottdefnblock}[#1]{$\Gamma  \vdash_{\mathit{com} }  \ottnt{e}  \ottsym{:}  \tau$}{\ottcom{Comonadic calculus for linearity, not syntax directed}}
\ottusedrule{\ottdrulelamXXcomXXvar{}}
\ottusedrule{\ottdrulelamXXcomXXabs{}}
\ottusedrule{\ottdrulelamXXcomXXapp{}}
\ottusedrule{\ottdrulelamXXcomXXunit{}}
\ottusedrule{\ottdrulelamXXcomXXsequence{}}
\ottusedrule{\ottdrulelamXXcomXXpair{}}
\ottusedrule{\ottdrulelamXXcomXXsplit{}}
\ottusedrule{\ottdrulelamXXcomXXwith{}}
\ottusedrule{\ottdrulelamXXcomXXfst{}}
\ottusedrule{\ottdrulelamXXcomXXsnd{}}
\ottusedrule{\ottdrulelamXXcomXXinl{}}
\ottusedrule{\ottdrulelamXXcomXXinr{}}
\ottusedrule{\ottdrulelamXXcomXXcase{}}
\ottusedrule{\ottdrulelamXXcomXXextract{}}
\ottusedrule{\ottdrulelamXXcomXXdiscard{}}
\ottusedrule{\ottdrulelamXXcomXXdivide{}}
\ottusedrule{\ottdrulelamXXcomXXextend{}}
\end{ottdefnblock}}

\newcommand{\ottdefnsJComLin}{
\ottdefncomXXtyping{}}

\newcommand{\ottdrulefullXXvar}[1]{\ottdrule[#1]{%
}{
   \textcolor{\coeffectcolor}{  \textcolor{\coeffectcolor}{\overline{0} }  }\! \cdot \! \Gamma_{{\mathrm{1}}}    \mathop{,}    \ottmv{x}  :^{\textcolor{\coeffectcolor}{  \textcolor{\coeffectcolor}{1}  } }  \ottnt{A}     \mathop{,}    \textcolor{\coeffectcolor}{  \textcolor{\coeffectcolor}{\overline{0} }  }\! \cdot \! \Gamma_{{\mathrm{2}}}    \vdash_{\mathit{full} }  \ottmv{x}  \ottsym{:}  \ottnt{A}}{%
{\ottdrulename{full\_var}}{}%
}}

\newcommand{\ottdrulefullXXthunk}[1]{\ottdrule[#1]{%
\ottpremise{  \textcolor{\coeffectcolor}{ \textcolor{\coeffectcolor}{\gamma} }\! \cdot \! \Gamma    \vdash_{\mathit{full} }   \ottnt{M}  :^{\textcolor{\effectcolor}{ \textcolor{\effectcolor}{\phi} } }  \ottnt{B} }%
}{
 \textcolor{\coeffectcolor}{ \textcolor{\coeffectcolor}{\gamma} }\! \cdot \! \Gamma   \vdash_{\mathit{full} }  \ottsym{\{}  \ottnt{M}  \ottsym{\}}  \ottsym{:}   \ottkw{U}_{\color{\effectcolor}{ \textcolor{\effectcolor}{\phi} } }\;  \ottnt{B} }{%
{\ottdrulename{full\_thunk}}{}%
}}

\newcommand{\ottdrulefullXXunit}[1]{\ottdrule[#1]{%
}{
 \textcolor{\coeffectcolor}{  \textcolor{\coeffectcolor}{\overline{0} }  }\! \cdot \! \Gamma   \vdash_{\mathit{full} }  \ottsym{()}  \ottsym{:}  \ottkw{unit}}{%
{\ottdrulename{full\_unit}}{}%
}}

\newcommand{\ottdrulefullXXpair}[1]{\ottdrule[#1]{%
\ottpremise{ \textcolor{\coeffectcolor}{ \textcolor{\coeffectcolor}{\gamma}_{{\mathrm{1}}} }\! \cdot \! \Gamma   \vdash_{\mathit{full} }  \ottnt{V_{{\mathrm{1}}}}  \ottsym{:}  \ottnt{A_{{\mathrm{1}}}}}%
\ottpremise{ \textcolor{\coeffectcolor}{ \textcolor{\coeffectcolor}{\gamma}_{{\mathrm{2}}} }\! \cdot \! \Gamma   \vdash_{\mathit{full} }  \ottnt{V_{{\mathrm{2}}}}  \ottsym{:}  \ottnt{A_{{\mathrm{2}}}}}%
}{
 \textcolor{\coeffectcolor}{  \textcolor{\coeffectcolor}{ \textcolor{\coeffectcolor}{\gamma}_{{\mathrm{1}}} \ottsym{+} \textcolor{\coeffectcolor}{\gamma}_{{\mathrm{2}}} }  }\! \cdot \! \Gamma   \vdash_{\mathit{full} }  \ottsym{(}  \ottnt{V_{{\mathrm{1}}}}  \ottsym{,}  \ottnt{V_{{\mathrm{2}}}}  \ottsym{)}  \ottsym{:}   \ottnt{A_{{\mathrm{1}}}} \times \ottnt{A_{{\mathrm{2}}}} }{%
{\ottdrulename{full\_pair}}{}%
}}

\newcommand{\ottdrulefullXXinl}[1]{\ottdrule[#1]{%
\ottpremise{ \textcolor{\coeffectcolor}{ \textcolor{\coeffectcolor}{\gamma} }\! \cdot \! \Gamma   \vdash_{\mathit{full} }  \ottnt{V}  \ottsym{:}  \ottnt{A_{{\mathrm{1}}}}}%
}{
 \textcolor{\coeffectcolor}{ \textcolor{\coeffectcolor}{\gamma} }\! \cdot \! \Gamma   \vdash_{\mathit{full} }  \ottkw{inl} \, \ottnt{V}  \ottsym{:}  \ottnt{A_{{\mathrm{1}}}}  \ottsym{+}  \ottnt{A_{{\mathrm{2}}}}}{%
{\ottdrulename{full\_inl}}{}%
}}

\newcommand{\ottdrulefullXXinr}[1]{\ottdrule[#1]{%
\ottpremise{ \textcolor{\coeffectcolor}{ \textcolor{\coeffectcolor}{\gamma} }\! \cdot \! \Gamma   \vdash_{\mathit{full} }  \ottnt{V}  \ottsym{:}  \ottnt{A_{{\mathrm{2}}}}}%
}{
 \textcolor{\coeffectcolor}{ \textcolor{\coeffectcolor}{\gamma} }\! \cdot \! \Gamma   \vdash_{\mathit{full} }  \ottkw{inr} \, \ottnt{V}  \ottsym{:}  \ottnt{A_{{\mathrm{1}}}}  \ottsym{+}  \ottnt{A_{{\mathrm{2}}}}}{%
{\ottdrulename{full\_inr}}{}%
}}

\newcommand{\ottdrulefullXXvsub}[1]{\ottdrule[#1]{%
\ottpremise{ \textcolor{\coeffectcolor}{ \textcolor{\coeffectcolor}{\gamma}' }\! \cdot \! \Gamma   \vdash_{\mathit{full} }  \ottnt{V}  \ottsym{:}  \ottnt{A}}%
\ottpremise{ \textcolor{\coeffectcolor}{ \textcolor{\coeffectcolor}{\gamma} }\; \textcolor{\coeffectcolor}{\mathop{\leq_{\mathit{co} } } } \; \textcolor{\coeffectcolor}{ \textcolor{\coeffectcolor}{\gamma}' } }%
}{
 \textcolor{\coeffectcolor}{ \textcolor{\coeffectcolor}{\gamma} }\! \cdot \! \Gamma   \vdash_{\mathit{full} }  \ottnt{V}  \ottsym{:}  \ottnt{A}}{%
{\ottdrulename{full\_vsub}}{}%
}}

\newcommand{\ottdefnfullXXvalXXtyping}[1]{\begin{ottdefnblock}[#1]{$\Psi  \vdash_{\mathit{full} }  \ottnt{V}  \ottsym{:}  \ottnt{A}$}{\ottcom{value effect and co-effect typing rules}}
\ottusedrule{\ottdrulefullXXvar{}}
\ottusedrule{\ottdrulefullXXthunk{}}
\ottusedrule{\ottdrulefullXXunit{}}
\ottusedrule{\ottdrulefullXXpair{}}
\ottusedrule{\ottdrulefullXXinl{}}
\ottusedrule{\ottdrulefullXXinr{}}
\ottusedrule{\ottdrulefullXXvsub{}}
\end{ottdefnblock}}

\newcommand{\ottdrulefullXXabs}[1]{\ottdrule[#1]{%
\ottpremise{   \textcolor{\coeffectcolor}{ \textcolor{\coeffectcolor}{\gamma} }\! \cdot \! \Gamma    \mathop{,}    \ottmv{x}  :^{\textcolor{\coeffectcolor}{ \textcolor{\coeffectcolor}{q} } }  \ottnt{A}     \vdash_{\mathit{full} }   \ottnt{M}  :^{\textcolor{\effectcolor}{ \textcolor{\effectcolor}{\phi} } }  \ottnt{B} }%
\ottpremise{ \textcolor{\coeffectcolor}{ \textcolor{\coeffectcolor}{q} }\;  \textcolor{\coeffectcolor}{\mathop{\leq_{\mathit{co} } } } \; \textcolor{\coeffectcolor}{ \textcolor{\coeffectcolor}{q}' } }%
}{
  \textcolor{\coeffectcolor}{ \textcolor{\coeffectcolor}{\gamma} }\! \cdot \! \Gamma    \vdash_{\mathit{full} }    \lambda  \ottmv{x} ^{\textcolor{\coeffectcolor}{ \textcolor{\coeffectcolor}{q}' } }. \ottnt{M}   :^{\textcolor{\effectcolor}{ \textcolor{\effectcolor}{\phi} } }   \ottnt{A} ^{\textcolor{\coeffectcolor}{ \textcolor{\coeffectcolor}{q}' } } \rightarrow  \ottnt{B}  }{%
{\ottdrulename{full\_abs}}{}%
}}

\newcommand{\ottdrulefullXXapp}[1]{\ottdrule[#1]{%
\ottpremise{  \textcolor{\coeffectcolor}{ \textcolor{\coeffectcolor}{\gamma}_{{\mathrm{1}}} }\! \cdot \! \Gamma    \vdash_{\mathit{full} }   \ottnt{M}  :^{\textcolor{\effectcolor}{ \textcolor{\effectcolor}{\phi} } }   \ottnt{A} ^{\textcolor{\coeffectcolor}{ \textcolor{\coeffectcolor}{q} } } \rightarrow  \ottnt{B}  }%
\ottpremise{ \textcolor{\coeffectcolor}{ \textcolor{\coeffectcolor}{\gamma}_{{\mathrm{2}}} }\! \cdot \! \Gamma   \vdash_{\mathit{full} }  \ottnt{V}  \ottsym{:}  \ottnt{A}}%
}{
  \textcolor{\coeffectcolor}{  \textcolor{\coeffectcolor}{ \textcolor{\coeffectcolor}{\gamma}_{{\mathrm{1}}} \ottsym{+}  \textcolor{\coeffectcolor}{ \textcolor{\coeffectcolor}{q} \cdot \textcolor{\coeffectcolor}{\gamma}_{{\mathrm{2}}} }  }  }\! \cdot \! \Gamma    \vdash_{\mathit{full} }   \ottnt{M} \, \ottnt{V}  :^{\textcolor{\effectcolor}{ \textcolor{\effectcolor}{\phi} } }  \ottnt{B} }{%
{\ottdrulename{full\_app}}{}%
}}

\newcommand{\ottdrulefullXXforce}[1]{\ottdrule[#1]{%
\ottpremise{ \textcolor{\coeffectcolor}{ \textcolor{\coeffectcolor}{\gamma} }\! \cdot \! \Gamma   \vdash_{\mathit{full} }  \ottnt{V}  \ottsym{:}   \ottkw{U}_{\color{\effectcolor}{ \textcolor{\effectcolor}{\phi} } }\;  \ottnt{B} }%
}{
  \textcolor{\coeffectcolor}{ \textcolor{\coeffectcolor}{\gamma} }\! \cdot \! \Gamma    \vdash_{\mathit{full} }   \ottnt{V}  \ottsym{!}  :^{\textcolor{\effectcolor}{ \textcolor{\effectcolor}{\phi} } }  \ottnt{B} }{%
{\ottdrulename{full\_force}}{}%
}}

\newcommand{\ottdrulefullXXret}[1]{\ottdrule[#1]{%
\ottpremise{ \textcolor{\coeffectcolor}{ \textcolor{\coeffectcolor}{\gamma}_{{\mathrm{1}}} }\! \cdot \! \Gamma   \vdash_{\mathit{full} }  \ottnt{V}  \ottsym{:}  \ottnt{A}}%
\ottpremise{ \textcolor{\coeffectcolor}{ \textcolor{\coeffectcolor}{q}' }\;  \textcolor{\coeffectcolor}{\mathop{\leq_{\mathit{co} } } } \; \textcolor{\coeffectcolor}{ \textcolor{\coeffectcolor}{q} } }%
}{
  \textcolor{\coeffectcolor}{  \textcolor{\coeffectcolor}{ \textcolor{\coeffectcolor}{q} \cdot \textcolor{\coeffectcolor}{\gamma}_{{\mathrm{1}}} }  }\! \cdot \! \Gamma    \vdash_{\mathit{full} }    \ottkw{return} _{\textcolor{\coeffectcolor}{ \textcolor{\coeffectcolor}{q} } }\;  \ottnt{V}   :^{\textcolor{\effectcolor}{  \textcolor{\effectcolor}{\varepsilon}  } }   \ottkw{F}_{\color{\coeffectcolor}{ \textcolor{\coeffectcolor}{q}' } }\;  \ottnt{A}  }{%
{\ottdrulename{full\_ret}}{}%
}}

\newcommand{\ottdrulefullXXletin}[1]{\ottdrule[#1]{%
\ottpremise{ \textcolor{\coeffectcolor}{ \textcolor{\coeffectcolor}{q}'_{{\mathrm{2}}} }\; = \; \textcolor{\coeffectcolor}{  \textcolor{\coeffectcolor}{q}_{{\mathrm{2}}} \ \|\ \textcolor{\coeffectcolor}{1}  } }%
\ottpremise{  \textcolor{\coeffectcolor}{ \textcolor{\coeffectcolor}{\gamma}_{{\mathrm{1}}} }\! \cdot \! \Gamma    \vdash_{\mathit{full} }   \ottnt{M_{{\mathrm{1}}}}  :^{\textcolor{\effectcolor}{ \textcolor{\effectcolor}{\phi}_{{\mathrm{1}}} } }   \ottkw{F}_{\color{\coeffectcolor}{ \textcolor{\coeffectcolor}{q}_{{\mathrm{1}}} } }\;  \ottnt{A}  }%
\ottpremise{   \textcolor{\coeffectcolor}{ \textcolor{\coeffectcolor}{\gamma}_{{\mathrm{2}}} }\! \cdot \! \Gamma    \mathop{,}    \ottmv{x}  :^{\textcolor{\coeffectcolor}{  \textcolor{\coeffectcolor}{ \textcolor{\coeffectcolor}{q}_{{\mathrm{1}}}   \cdot   \textcolor{\coeffectcolor}{q}'_{{\mathrm{2}}} }  } }  \ottnt{A}     \vdash_{\mathit{full} }   \ottnt{M_{{\mathrm{2}}}}  :^{\textcolor{\effectcolor}{ \textcolor{\effectcolor}{\phi}_{{\mathrm{2}}} } }  \ottnt{B} }%
}{
  \textcolor{\coeffectcolor}{  \textcolor{\coeffectcolor}{ \ottsym{(}   \textcolor{\coeffectcolor}{ \textcolor{\coeffectcolor}{q}'_{{\mathrm{2}}} \cdot \textcolor{\coeffectcolor}{\gamma}_{{\mathrm{1}}} }   \ottsym{)} \ottsym{+} \textcolor{\coeffectcolor}{\gamma}_{{\mathrm{2}}} }  }\! \cdot \! \Gamma    \vdash_{\mathit{full} }    \ottmv{x}  \leftarrow^{\textcolor{\coeffectcolor}{ \textcolor{\coeffectcolor}{q}_{{\mathrm{2}}} } }  \ottnt{M_{{\mathrm{1}}}} \ \ottkw{in}\  \ottnt{M_{{\mathrm{2}}}}   :^{\textcolor{\effectcolor}{  \textcolor{\effectcolor}{ \textcolor{\effectcolor}{\phi}_{{\mathrm{1}}}  \cdot  \textcolor{\effectcolor}{\phi}_{{\mathrm{2}}} }  } }  \ottnt{B} }{%
{\ottdrulename{full\_letin}}{}%
}}

\newcommand{\ottdrulefullXXletinXXzero}[1]{\ottdrule[#1]{%
\ottpremise{  \textcolor{\coeffectcolor}{ \textcolor{\coeffectcolor}{\gamma}_{{\mathrm{1}}} }\! \cdot \! \Gamma    \vdash_{\mathit{full} }   \ottnt{M}  :^{\textcolor{\effectcolor}{  \textcolor{\effectcolor}{\varepsilon}  } }   \ottkw{F}_{\color{\coeffectcolor}{ \textcolor{\coeffectcolor}{q} } }\;  \ottnt{A}  }%
\ottpremise{   \textcolor{\coeffectcolor}{ \textcolor{\coeffectcolor}{\gamma}_{{\mathrm{2}}} }\! \cdot \! \Gamma    \mathop{,}    \ottmv{x}  :^{\textcolor{\coeffectcolor}{  \textcolor{\coeffectcolor}{0}  } }  \ottnt{A}     \vdash_{\mathit{full} }   \ottnt{N}  :^{\textcolor{\effectcolor}{ \textcolor{\effectcolor}{\phi} } }  \ottnt{B} }%
}{
  \textcolor{\coeffectcolor}{ \textcolor{\coeffectcolor}{\gamma}_{{\mathrm{2}}} }\! \cdot \! \Gamma    \vdash_{\mathit{full} }    \ottmv{x}  \leftarrow^{\textcolor{\coeffectcolor}{0} }_{\textcolor{\effectcolor}{\varepsilon} }  \ottnt{M} \ \ottkw{in}\  \ottnt{N}   :^{\textcolor{\effectcolor}{ \textcolor{\effectcolor}{\phi} } }  \ottnt{B} }{%
{\ottdrulename{full\_letin\_zero}}{}%
}}

\newcommand{\ottdrulefullXXsplit}[1]{\ottdrule[#1]{%
\ottpremise{ \textcolor{\coeffectcolor}{ \textcolor{\coeffectcolor}{\gamma}_{{\mathrm{1}}} }\! \cdot \! \Gamma   \vdash_{\mathit{full} }  \ottnt{V}  \ottsym{:}   \ottnt{A_{{\mathrm{1}}}} \times \ottnt{A_{{\mathrm{2}}}} }%
\ottpremise{    \textcolor{\coeffectcolor}{ \textcolor{\coeffectcolor}{\gamma}_{{\mathrm{2}}} }\! \cdot \! \Gamma    \mathop{,}    \ottmv{x_{{\mathrm{1}}}}  :^{\textcolor{\coeffectcolor}{ \textcolor{\coeffectcolor}{q} } }  \ottnt{A_{{\mathrm{1}}}}     \mathop{,}    \ottmv{x_{{\mathrm{2}}}}  :^{\textcolor{\coeffectcolor}{ \textcolor{\coeffectcolor}{q} } }  \ottnt{A_{{\mathrm{2}}}}     \vdash_{\mathit{full} }   \ottnt{N}  :^{\textcolor{\effectcolor}{ \textcolor{\effectcolor}{\phi} } }  \ottnt{B} }%
\ottpremise{ \textcolor{\coeffectcolor}{ \textcolor{\coeffectcolor}{\gamma} } \equiv \textcolor{\coeffectcolor}{  \textcolor{\coeffectcolor}{  \textcolor{\coeffectcolor}{ \textcolor{\coeffectcolor}{q} \cdot \textcolor{\coeffectcolor}{\gamma}_{{\mathrm{1}}} }  \ottsym{+} \textcolor{\coeffectcolor}{\gamma}_{{\mathrm{2}}} }  } }%
}{
  \textcolor{\coeffectcolor}{ \textcolor{\coeffectcolor}{\gamma} }\! \cdot \! \Gamma    \vdash_{\mathit{full} }    \ottkw{case}_{\textcolor{\coeffectcolor}{ \textcolor{\coeffectcolor}{q} } } \;  \ottnt{V} \; \ottkw{of}\;( \ottmv{x_{{\mathrm{1}}}} , \ottmv{x_{{\mathrm{2}}}} )\; \rightarrow\;  \ottnt{N}   :^{\textcolor{\effectcolor}{ \textcolor{\effectcolor}{\phi} } }  \ottnt{B} }{%
{\ottdrulename{full\_split}}{}%
}}

\newcommand{\ottdrulefullXXcunit}[1]{\ottdrule[#1]{%
}{
  \textcolor{\coeffectcolor}{  \textcolor{\coeffectcolor}{\overline{0} }  }\! \cdot \! \Gamma    \vdash_{\mathit{full} }    \langle\rangle   :^{\textcolor{\effectcolor}{  \textcolor{\effectcolor}{\varepsilon}  } }   \top  }{%
{\ottdrulename{full\_cunit}}{}%
}}

\newcommand{\ottdrulefullXXcpair}[1]{\ottdrule[#1]{%
\ottpremise{  \textcolor{\coeffectcolor}{ \textcolor{\coeffectcolor}{\gamma} }\! \cdot \! \Gamma    \vdash_{\mathit{full} }   \ottnt{M_{{\mathrm{1}}}}  :^{\textcolor{\effectcolor}{ \textcolor{\effectcolor}{\phi} } }  \ottnt{B_{{\mathrm{1}}}} }%
\ottpremise{  \textcolor{\coeffectcolor}{ \textcolor{\coeffectcolor}{\gamma} }\! \cdot \! \Gamma    \vdash_{\mathit{full} }   \ottnt{M_{{\mathrm{2}}}}  :^{\textcolor{\effectcolor}{ \textcolor{\effectcolor}{\phi} } }  \ottnt{B_{{\mathrm{2}}}} }%
}{
  \textcolor{\coeffectcolor}{ \textcolor{\coeffectcolor}{\gamma} }\! \cdot \! \Gamma    \vdash_{\mathit{full} }    \langle  \ottnt{M_{{\mathrm{1}}}} , \ottnt{M_{{\mathrm{2}}}}  \rangle   :^{\textcolor{\effectcolor}{ \textcolor{\effectcolor}{\phi} } }   \ottnt{B_{{\mathrm{1}}}}   \mathop{\&}   \ottnt{B_{{\mathrm{2}}}}  }{%
{\ottdrulename{full\_cpair}}{}%
}}

\newcommand{\ottdrulefullXXfst}[1]{\ottdrule[#1]{%
\ottpremise{  \textcolor{\coeffectcolor}{ \textcolor{\coeffectcolor}{\gamma} }\! \cdot \! \Gamma    \vdash_{\mathit{full} }   \ottnt{M}  :^{\textcolor{\effectcolor}{ \textcolor{\effectcolor}{\phi} } }   \ottnt{B_{{\mathrm{1}}}}   \mathop{\&}   \ottnt{B_{{\mathrm{2}}}}  }%
}{
  \textcolor{\coeffectcolor}{ \textcolor{\coeffectcolor}{\gamma} }\! \cdot \! \Gamma    \vdash_{\mathit{full} }    \ottnt{M}  . 1   :^{\textcolor{\effectcolor}{ \textcolor{\effectcolor}{\phi} } }  \ottnt{B_{{\mathrm{1}}}} }{%
{\ottdrulename{full\_fst}}{}%
}}

\newcommand{\ottdrulefullXXsnd}[1]{\ottdrule[#1]{%
\ottpremise{  \textcolor{\coeffectcolor}{ \textcolor{\coeffectcolor}{\gamma} }\! \cdot \! \Gamma    \vdash_{\mathit{full} }   \ottnt{M}  :^{\textcolor{\effectcolor}{ \textcolor{\effectcolor}{\phi} } }   \ottnt{B_{{\mathrm{1}}}}   \mathop{\&}   \ottnt{B_{{\mathrm{2}}}}  }%
}{
  \textcolor{\coeffectcolor}{ \textcolor{\coeffectcolor}{\gamma} }\! \cdot \! \Gamma    \vdash_{\mathit{full} }    \ottnt{M}  . 2   :^{\textcolor{\effectcolor}{ \textcolor{\effectcolor}{\phi} } }  \ottnt{B_{{\mathrm{2}}}} }{%
{\ottdrulename{full\_snd}}{}%
}}

\newcommand{\ottdrulefullXXsequence}[1]{\ottdrule[#1]{%
\ottpremise{ \textcolor{\coeffectcolor}{ \textcolor{\coeffectcolor}{\gamma}_{{\mathrm{1}}} }\! \cdot \! \Gamma   \vdash_{\mathit{full} }  \ottnt{V}  \ottsym{:}  \ottkw{unit}}%
\ottpremise{  \textcolor{\coeffectcolor}{ \textcolor{\coeffectcolor}{\gamma}_{{\mathrm{2}}} }\! \cdot \! \Gamma    \vdash_{\mathit{full} }   \ottnt{N}  :^{\textcolor{\effectcolor}{ \textcolor{\effectcolor}{\phi} } }  \ottnt{B} }%
\ottpremise{ \textcolor{\coeffectcolor}{ \textcolor{\coeffectcolor}{\gamma} } \equiv \textcolor{\coeffectcolor}{  \textcolor{\coeffectcolor}{ \textcolor{\coeffectcolor}{\gamma}_{{\mathrm{1}}} \ottsym{+} \textcolor{\coeffectcolor}{\gamma}_{{\mathrm{2}}} }  } }%
}{
  \textcolor{\coeffectcolor}{ \textcolor{\coeffectcolor}{\gamma} }\! \cdot \! \Gamma    \vdash_{\mathit{full} }    \ottnt{V}  ;  \ottnt{N}   :^{\textcolor{\effectcolor}{ \textcolor{\effectcolor}{\phi} } }  \ottnt{B} }{%
{\ottdrulename{full\_sequence}}{}%
}}

\newcommand{\ottdrulefullXXcase}[1]{\ottdrule[#1]{%
\ottpremise{ \textcolor{\coeffectcolor}{ \textcolor{\coeffectcolor}{\gamma}_{{\mathrm{1}}} }\! \cdot \! \Gamma   \vdash_{\mathit{full} }  \ottnt{V}  \ottsym{:}  \ottnt{A_{{\mathrm{1}}}}  \ottsym{+}  \ottnt{A_{{\mathrm{2}}}}}%
\ottpremise{  \textcolor{\coeffectcolor}{ \textcolor{\coeffectcolor}{\gamma}_{{\mathrm{2}}} }\! \cdot \!  \Gamma   \mathop{,}   \ottmv{x_{{\mathrm{1}}}}  \ottsym{:}  \ottnt{A_{{\mathrm{1}}}}     \vdash_{\mathit{full} }   \ottnt{M_{{\mathrm{1}}}}  :^{\textcolor{\effectcolor}{ \textcolor{\effectcolor}{\phi} } }  \ottnt{B} }%
\ottpremise{  \textcolor{\coeffectcolor}{ \textcolor{\coeffectcolor}{\gamma}_{{\mathrm{2}}} }\! \cdot \!  \Gamma   \mathop{,}   \ottmv{x_{{\mathrm{2}}}}  \ottsym{:}  \ottnt{A_{{\mathrm{2}}}}     \vdash_{\mathit{full} }   \ottnt{M_{{\mathrm{2}}}}  :^{\textcolor{\effectcolor}{ \textcolor{\effectcolor}{\phi} } }  \ottnt{B} }%
\ottpremise{ \textcolor{\coeffectcolor}{ \textcolor{\coeffectcolor}{\gamma} } \equiv \textcolor{\coeffectcolor}{  \textcolor{\coeffectcolor}{  \textcolor{\coeffectcolor}{ \textcolor{\coeffectcolor}{q} \cdot \textcolor{\coeffectcolor}{\gamma}_{{\mathrm{1}}} }  \ottsym{+} \textcolor{\coeffectcolor}{\gamma}_{{\mathrm{2}}} }  } }%
\ottpremise{ \textcolor{\coeffectcolor}{ \textcolor{\coeffectcolor}{q} }\;  \textcolor{\coeffectcolor}{\mathop{\leq_{\mathit{co} } } } \; \textcolor{\coeffectcolor}{  \textcolor{\coeffectcolor}{1}  } }%
}{
  \textcolor{\coeffectcolor}{ \textcolor{\coeffectcolor}{\gamma} }\! \cdot \! \Gamma    \vdash_{\mathit{full} }   \ottkw{case} \, \ottnt{V} \, \ottkw{of} \, \ottkw{inl} \, \ottmv{x_{{\mathrm{1}}}}  \to  \ottnt{M_{{\mathrm{1}}}}  \mathsf{;} \, \ottkw{inr} \, \ottmv{x_{{\mathrm{2}}}}  \to  \ottnt{M_{{\mathrm{2}}}}  :^{\textcolor{\effectcolor}{ \textcolor{\effectcolor}{\phi} } }  \ottnt{B} }{%
{\ottdrulename{full\_case}}{}%
}}

\newcommand{\ottdrulefullXXtick}[1]{\ottdrule[#1]{%
}{
  \textcolor{\coeffectcolor}{  \textcolor{\coeffectcolor}{\overline{0} }  }\! \cdot \! \Gamma    \vdash_{\mathit{full} }    \textcolor{\effectcolor}{ \ottkw{tick} }   :^{\textcolor{\effectcolor}{  \textcolor{\effectcolor}{\ottkw{Tick} }  } }   \ottkw{F}_{\color{\coeffectcolor}{  \textcolor{\coeffectcolor}{1}  } }\;  \ottkw{unit}  }{%
{\ottdrulename{full\_tick}}{}%
}}

\newcommand{\ottdrulefullXXcsub}[1]{\ottdrule[#1]{%
\ottpremise{  \textcolor{\coeffectcolor}{ \textcolor{\coeffectcolor}{\gamma}' }\! \cdot \! \Gamma    \vdash_{\mathit{full} }   \ottnt{M}  :^{\textcolor{\effectcolor}{ \textcolor{\effectcolor}{\phi}' } }  \ottnt{B} }%
\ottpremise{ \textcolor{\coeffectcolor}{ \textcolor{\coeffectcolor}{\gamma} }\; \textcolor{\coeffectcolor}{\mathop{\leq_{\mathit{co} } } } \; \textcolor{\coeffectcolor}{ \textcolor{\coeffectcolor}{\gamma}' } }%
\ottpremise{ \textcolor{\effectcolor}{ \textcolor{\effectcolor}{\phi}' \;  \textcolor{\effectcolor}{\mathop{\leq_{\mathit{eff} } } } \;  \textcolor{\effectcolor}{\phi}  } }%
}{
  \textcolor{\coeffectcolor}{ \textcolor{\coeffectcolor}{\gamma} }\! \cdot \! \Gamma    \vdash_{\mathit{full} }   \ottnt{M}  :^{\textcolor{\effectcolor}{ \textcolor{\effectcolor}{\phi} } }  \ottnt{B} }{%
{\ottdrulename{full\_csub}}{}%
}}

\newcommand{\ottdrulefullXXdist}[1]{\ottdrule[#1]{%
\ottpremise{  \textcolor{\coeffectcolor}{ \textcolor{\coeffectcolor}{\gamma} }\! \cdot \! \Gamma    \vdash_{\mathit{full} }   \ottnt{M}  :^{\textcolor{\effectcolor}{ \textcolor{\effectcolor}{\phi}' } }   \ottkw{F}_{\color{\coeffectcolor}{ \textcolor{\coeffectcolor}{q} } }\;   \ottkw{U}_{\color{\effectcolor}{ \textcolor{\effectcolor}{\phi} } }\;   \ottkw{F}_{\color{\coeffectcolor}{  \textcolor{\coeffectcolor}{1}  } }\;  \ottnt{A}    }%
}{
  \textcolor{\coeffectcolor}{ \textcolor{\coeffectcolor}{\gamma} }\! \cdot \! \Gamma    \vdash_{\mathit{full} }   \ottkw{dist} \, \ottnt{M}  :^{\textcolor{\effectcolor}{ \textcolor{\effectcolor}{\phi}' } }   \ottkw{F}_{\color{\coeffectcolor}{  \textcolor{\coeffectcolor}{1}  } }\;   \ottkw{U}_{\color{\effectcolor}{  \textcolor{\coeffectcolor}{q}  \oslash  \textcolor{\effectcolor}{\phi}  } }\;   \ottkw{F}_{\color{\coeffectcolor}{  \textcolor{\coeffectcolor}{q}  \odot  \textcolor{\effectcolor}{\phi}  } }\;  \ottnt{A}    }{%
{\ottdrulename{full\_dist}}{}%
}}

\newcommand{\ottdefnfullXXcompXXtyping}[1]{\begin{ottdefnblock}[#1]{$ \Psi   \vdash_{\mathit{full} }   \ottnt{M}  :^{\textcolor{\effectcolor}{ \textcolor{\effectcolor}{\phi} } }  \ottnt{B} $}{\ottcom{computation (co)effect typing rules}}
\ottusedrule{\ottdrulefullXXabs{}}
\ottusedrule{\ottdrulefullXXapp{}}
\ottusedrule{\ottdrulefullXXforce{}}
\ottusedrule{\ottdrulefullXXret{}}
\ottusedrule{\ottdrulefullXXletin{}}
\ottusedrule{\ottdrulefullXXletinXXzero{}}
\ottusedrule{\ottdrulefullXXsplit{}}
\ottusedrule{\ottdrulefullXXcunit{}}
\ottusedrule{\ottdrulefullXXcpair{}}
\ottusedrule{\ottdrulefullXXfst{}}
\ottusedrule{\ottdrulefullXXsnd{}}
\ottusedrule{\ottdrulefullXXsequence{}}
\ottusedrule{\ottdrulefullXXcase{}}
\ottusedrule{\ottdrulefullXXtick{}}
\ottusedrule{\ottdrulefullXXcsub{}}
\ottusedrule{\ottdrulefullXXdist{}}
\end{ottdefnblock}}

\newcommand{\ottdefnsJEffCoeffCBPV}{
\ottdefnfullXXvalXXtyping{}\ottdefnfullXXcompXXtyping{}}

\newcommand{\ottdruleevalXXfullXXvalXXvar}[1]{\ottdrule[#1]{%
}{
     \textcolor{\coeffectcolor}{  \textcolor{\coeffectcolor}{\overline{0}_1}  }\! \cdot \! \rho_{{\mathrm{1}}}    \mathop{,} \;   \ottmv{x}   \mapsto ^{\textcolor{\coeffectcolor}{  \textcolor{\coeffectcolor}{1}  } }  \ottnt{W}      \mathop{,} \;   \textcolor{\coeffectcolor}{  \textcolor{\coeffectcolor}{\overline{0}_2}  }\! \cdot \! \rho_{{\mathrm{2}}}    \vdash_{\mathit{full} }  \ottmv{x}  \Downarrow  \ottnt{W} }{%
{\ottdrulename{eval\_full\_val\_var}}{}%
}}

\newcommand{\ottdruleevalXXfullXXvalXXunit}[1]{\ottdrule[#1]{%
}{
  \textcolor{\coeffectcolor}{  \textcolor{\coeffectcolor}{\overline{0} }  }\! \cdot \! \rho   \vdash_{\mathit{full} }  \ottsym{()}  \Downarrow  \ottsym{()} }{%
{\ottdrulename{eval\_full\_val\_unit}}{}%
}}

\newcommand{\ottdruleevalXXfullXXvalXXthunk}[1]{\ottdrule[#1]{%
}{
  \textcolor{\coeffectcolor}{ \textcolor{\coeffectcolor}{\gamma} }\! \cdot \! \rho   \vdash_{\mathit{full} }  \ottsym{\{}  \ottnt{M}  \ottsym{\}}  \Downarrow   \mathbf{clo}( \textcolor{\coeffectcolor}{\gamma} ,  \rho ,  \ottnt{M} )  }{%
{\ottdrulename{eval\_full\_val\_thunk}}{}%
}}

\newcommand{\ottdruleevalXXfullXXvalXXvpair}[1]{\ottdrule[#1]{%
\ottpremise{  \textcolor{\coeffectcolor}{ \textcolor{\coeffectcolor}{\gamma}_{{\mathrm{1}}} }\! \cdot \! \rho   \vdash_{\mathit{full} }  \ottnt{V_{{\mathrm{1}}}}  \Downarrow  \ottnt{W_{{\mathrm{1}}}} }%
\ottpremise{  \textcolor{\coeffectcolor}{ \textcolor{\coeffectcolor}{\gamma}_{{\mathrm{2}}} }\! \cdot \! \rho   \vdash_{\mathit{full} }  \ottnt{V_{{\mathrm{2}}}}  \Downarrow  \ottnt{W_{{\mathrm{2}}}} }%
}{
  \textcolor{\coeffectcolor}{  \textcolor{\coeffectcolor}{ \textcolor{\coeffectcolor}{\gamma}_{{\mathrm{1}}} \ottsym{+} \textcolor{\coeffectcolor}{\gamma}_{{\mathrm{2}}} }  }\! \cdot \! \rho   \vdash_{\mathit{full} }  \ottsym{(}  \ottnt{V_{{\mathrm{1}}}}  \ottsym{,}  \ottnt{V_{{\mathrm{2}}}}  \ottsym{)}  \Downarrow  \ottsym{(}  \ottnt{W_{{\mathrm{1}}}}  \ottsym{,}  \ottnt{W_{{\mathrm{2}}}}  \ottsym{)} }{%
{\ottdrulename{eval\_full\_val\_vpair}}{}%
}}

\newcommand{\ottdruleevalXXfullXXvalXXinl}[1]{\ottdrule[#1]{%
\ottpremise{  \textcolor{\coeffectcolor}{ \textcolor{\coeffectcolor}{\gamma} }\! \cdot \! \rho   \vdash_{\mathit{full} }  \ottnt{V}  \Downarrow  \ottnt{W} }%
}{
  \textcolor{\coeffectcolor}{ \textcolor{\coeffectcolor}{\gamma} }\! \cdot \! \rho   \vdash_{\mathit{full} }  \ottkw{inl} \, \ottnt{V}  \Downarrow  \ottkw{inl} \, \ottnt{W} }{%
{\ottdrulename{eval\_full\_val\_inl}}{}%
}}

\newcommand{\ottdruleevalXXfullXXvalXXinr}[1]{\ottdrule[#1]{%
\ottpremise{  \textcolor{\coeffectcolor}{ \textcolor{\coeffectcolor}{\gamma} }\! \cdot \! \rho   \vdash_{\mathit{full} }  \ottnt{V}  \Downarrow  \ottnt{W} }%
}{
  \textcolor{\coeffectcolor}{ \textcolor{\coeffectcolor}{\gamma} }\! \cdot \! \rho   \vdash_{\mathit{full} }  \ottkw{inr} \, \ottnt{V}  \Downarrow  \ottkw{inr} \, \ottnt{W} }{%
{\ottdrulename{eval\_full\_val\_inr}}{}%
}}

\newcommand{\ottdruleevalXXfullXXvalXXvsub}[1]{\ottdrule[#1]{%
\ottpremise{  \textcolor{\coeffectcolor}{ \textcolor{\coeffectcolor}{\gamma}' }\! \cdot \! \rho   \vdash_{\mathit{full} }  \ottnt{V}  \Downarrow  \ottnt{W} }%
\ottpremise{ \textcolor{\coeffectcolor}{ \textcolor{\coeffectcolor}{\gamma} }\; \textcolor{\coeffectcolor}{\mathop{\leq_{\mathit{co} } } } \; \textcolor{\coeffectcolor}{ \textcolor{\coeffectcolor}{\gamma}' } }%
}{
  \textcolor{\coeffectcolor}{ \textcolor{\coeffectcolor}{\gamma} }\! \cdot \! \rho   \vdash_{\mathit{full} }  \ottnt{V}  \Downarrow  \ottnt{W} }{%
{\ottdrulename{eval\_full\_val\_vsub}}{}%
}}

\newcommand{\ottdefnevalXXfullXXval}[1]{\begin{ottdefnblock}[#1]{$ \mu  \vdash_{\mathit{full} }  \ottnt{V}  \Downarrow  \ottnt{W} $}{\ottcom{environment-based effect-coeffect semantics for CBPV (large-step)}}
\ottusedrule{\ottdruleevalXXfullXXvalXXvar{}}
\ottusedrule{\ottdruleevalXXfullXXvalXXunit{}}
\ottusedrule{\ottdruleevalXXfullXXvalXXthunk{}}
\ottusedrule{\ottdruleevalXXfullXXvalXXvpair{}}
\ottusedrule{\ottdruleevalXXfullXXvalXXinl{}}
\ottusedrule{\ottdruleevalXXfullXXvalXXinr{}}
\ottusedrule{\ottdruleevalXXfullXXvalXXvsub{}}
\end{ottdefnblock}}

\newcommand{\ottdruleevalXXfullXXcompXXabs}[1]{\ottdrule[#1]{%
\ottpremise{ \textcolor{\coeffectcolor}{ \textcolor{\coeffectcolor}{q}' }\;  \textcolor{\coeffectcolor}{\mathop{\leq_{\mathit{co} } } } \; \textcolor{\coeffectcolor}{ \textcolor{\coeffectcolor}{q} } }%
}{
  \textcolor{\coeffectcolor}{ \textcolor{\coeffectcolor}{\gamma} }\! \cdot \! \rho   \vdash_{\mathit{full} }   \lambda  \ottmv{x} ^{\textcolor{\coeffectcolor}{ \textcolor{\coeffectcolor}{q} } }. \ottnt{M}   \Downarrow   \mathbf{clo}(   \textcolor{\coeffectcolor}{ \textcolor{\coeffectcolor}{\gamma} }\! \cdot \! \rho  ,   \lambda  \ottmv{x} ^{\textcolor{\coeffectcolor}{ \textcolor{\coeffectcolor}{q}' } }. \ottnt{M}   )   \mathop{\#} \textcolor{\effectcolor}{  \textcolor{\effectcolor}{\varepsilon}  } }{%
{\ottdrulename{eval\_full\_comp\_abs}}{}%
}}

\newcommand{\ottdruleevalXXfullXXcompXXcpair}[1]{\ottdrule[#1]{%
}{
  \textcolor{\coeffectcolor}{ \textcolor{\coeffectcolor}{\gamma} }\! \cdot \! \rho   \vdash_{\mathit{full} }   \langle  \ottnt{M_{{\mathrm{1}}}} , \ottnt{M_{{\mathrm{2}}}}  \rangle   \Downarrow   \mathbf{clo}(   \textcolor{\coeffectcolor}{ \textcolor{\coeffectcolor}{\gamma} }\! \cdot \! \rho  ,   \langle  \ottnt{M_{{\mathrm{1}}}} , \ottnt{M_{{\mathrm{2}}}}  \rangle   )   \mathop{\#} \textcolor{\effectcolor}{  \textcolor{\effectcolor}{\varepsilon}  } }{%
{\ottdrulename{eval\_full\_comp\_cpair}}{}%
}}

\newcommand{\ottdruleevalXXfullXXcompXXapp}[1]{\ottdrule[#1]{%
\ottpremise{  \textcolor{\coeffectcolor}{ \textcolor{\coeffectcolor}{\gamma}_{{\mathrm{1}}} }\! \cdot \! \rho   \vdash_{\mathit{full} }  \ottnt{M}  \Downarrow   \mathbf{clo}(   \textcolor{\coeffectcolor}{ \textcolor{\coeffectcolor}{\gamma}' }\! \cdot \! \rho'  ,   \lambda  \ottmv{x} ^{\textcolor{\coeffectcolor}{ \textcolor{\coeffectcolor}{q} } }. \ottnt{M'}   )   \mathop{\#} \textcolor{\effectcolor}{ \textcolor{\effectcolor}{\phi}_{{\mathrm{1}}} } }%
\ottpremise{  \textcolor{\coeffectcolor}{ \textcolor{\coeffectcolor}{\gamma}_{{\mathrm{2}}} }\! \cdot \! \rho   \vdash_{\mathit{full} }  \ottnt{V}  \Downarrow  \ottnt{W} }%
\ottpremise{   \textcolor{\coeffectcolor}{ \textcolor{\coeffectcolor}{\gamma}' }\! \cdot \! \rho'    \mathop{,} \;   \ottmv{x}   \mapsto ^{\textcolor{\coeffectcolor}{ \textcolor{\coeffectcolor}{q} } }  \ottnt{W}    \vdash_{\mathit{full} }  \ottnt{M'}  \Downarrow  \ottnt{T}  \mathop{\#} \textcolor{\effectcolor}{ \textcolor{\effectcolor}{\phi}_{{\mathrm{2}}} } }%
\ottpremise{ \textcolor{\coeffectcolor}{ \textcolor{\coeffectcolor}{\gamma} } \equiv \textcolor{\coeffectcolor}{  \textcolor{\coeffectcolor}{ \textcolor{\coeffectcolor}{\gamma}_{{\mathrm{1}}} \ottsym{+}  \textcolor{\coeffectcolor}{ \textcolor{\coeffectcolor}{q} \cdot \textcolor{\coeffectcolor}{\gamma}_{{\mathrm{2}}} }  }  } }%
\ottpremise{ \textcolor{\coeffectcolor}{ \textcolor{\coeffectcolor}{q} }  \neq  \textcolor{\coeffectcolor}{  \textcolor{\coeffectcolor}{0}  } }%
}{
  \textcolor{\coeffectcolor}{ \textcolor{\coeffectcolor}{\gamma} }\! \cdot \! \rho   \vdash_{\mathit{full} }  \ottnt{M} \, \ottnt{V}  \Downarrow  \ottnt{T}  \mathop{\#} \textcolor{\effectcolor}{  \textcolor{\effectcolor}{ \textcolor{\effectcolor}{\phi}_{{\mathrm{1}}}  \cdot  \textcolor{\effectcolor}{\phi}_{{\mathrm{2}}} }  } }{%
{\ottdrulename{eval\_full\_comp\_app}}{}%
}}

\newcommand{\ottdruleevalXXfullXXcompXXsplit}[1]{\ottdrule[#1]{%
\ottpremise{  \textcolor{\coeffectcolor}{ \textcolor{\coeffectcolor}{\gamma}_{{\mathrm{1}}} }\! \cdot \! \rho   \vdash_{\mathit{full} }  \ottnt{V}  \Downarrow  \ottsym{(}  \ottnt{W_{{\mathrm{1}}}}  \ottsym{,}  \ottnt{W_{{\mathrm{2}}}}  \ottsym{)} }%
\ottpremise{     \textcolor{\coeffectcolor}{ \textcolor{\coeffectcolor}{\gamma}_{{\mathrm{2}}} }\! \cdot \! \rho    \mathop{,} \;   \ottmv{x_{{\mathrm{1}}}}   \mapsto ^{\textcolor{\coeffectcolor}{ \textcolor{\coeffectcolor}{q} } }  \ottnt{W_{{\mathrm{1}}}}      \mathop{,} \;   \ottmv{x_{{\mathrm{2}}}}   \mapsto ^{\textcolor{\coeffectcolor}{ \textcolor{\coeffectcolor}{q} } }  \ottnt{W_{{\mathrm{2}}}}    \vdash_{\mathit{full} }  \ottnt{N}  \Downarrow  \ottnt{T}  \mathop{\#} \textcolor{\effectcolor}{ \textcolor{\effectcolor}{\phi} } }%
\ottpremise{ \textcolor{\coeffectcolor}{ \textcolor{\coeffectcolor}{\gamma} } \equiv \textcolor{\coeffectcolor}{  \textcolor{\coeffectcolor}{  \textcolor{\coeffectcolor}{ \textcolor{\coeffectcolor}{q} \cdot \textcolor{\coeffectcolor}{\gamma}_{{\mathrm{1}}} }  \ottsym{+} \textcolor{\coeffectcolor}{\gamma}_{{\mathrm{2}}} }  } }%
\ottpremise{ \textcolor{\coeffectcolor}{ \textcolor{\coeffectcolor}{q} }  \neq  \textcolor{\coeffectcolor}{  \textcolor{\coeffectcolor}{0}  } }%
}{
  \textcolor{\coeffectcolor}{ \textcolor{\coeffectcolor}{\gamma} }\! \cdot \! \rho   \vdash_{\mathit{full} }   \ottkw{case}_{\textcolor{\coeffectcolor}{ \textcolor{\coeffectcolor}{q} } } \;  \ottnt{V} \; \ottkw{of}\;( \ottmv{x_{{\mathrm{1}}}} , \ottmv{x_{{\mathrm{2}}}} )\; \rightarrow\;  \ottnt{N}   \Downarrow  \ottnt{T}  \mathop{\#} \textcolor{\effectcolor}{ \textcolor{\effectcolor}{\phi} } }{%
{\ottdrulename{eval\_full\_comp\_split}}{}%
}}

\newcommand{\ottdruleevalXXfullXXcompXXret}[1]{\ottdrule[#1]{%
\ottpremise{  \textcolor{\coeffectcolor}{ \textcolor{\coeffectcolor}{\gamma}' }\! \cdot \! \rho   \vdash_{\mathit{full} }  \ottnt{V}  \Downarrow  \ottnt{W} }%
\ottpremise{ \textcolor{\coeffectcolor}{ \textcolor{\coeffectcolor}{\gamma} } \equiv \textcolor{\coeffectcolor}{  \textcolor{\coeffectcolor}{ \textcolor{\coeffectcolor}{q} \cdot \textcolor{\coeffectcolor}{\gamma}' }  } }%
\ottpremise{ \textcolor{\coeffectcolor}{ \textcolor{\coeffectcolor}{q} }  \neq  \textcolor{\coeffectcolor}{  \textcolor{\coeffectcolor}{0}  } }%
}{
  \textcolor{\coeffectcolor}{ \textcolor{\coeffectcolor}{\gamma} }\! \cdot \! \rho   \vdash_{\mathit{full} }   \ottkw{return} _{\textcolor{\coeffectcolor}{ \textcolor{\coeffectcolor}{q} } }\;  \ottnt{V}   \Downarrow   \ottkw{return} _{\textcolor{\coeffectcolor}{ \textcolor{\coeffectcolor}{q} } }  \ottnt{W}   \mathop{\#} \textcolor{\effectcolor}{  \textcolor{\effectcolor}{\varepsilon}  } }{%
{\ottdrulename{eval\_full\_comp\_ret}}{}%
}}

\newcommand{\ottdruleevalXXfullXXcompXXletin}[1]{\ottdrule[#1]{%
\ottpremise{ \textcolor{\coeffectcolor}{ \textcolor{\coeffectcolor}{q}'_{{\mathrm{2}}} }\; = \; \textcolor{\coeffectcolor}{  \textcolor{\coeffectcolor}{q}_{{\mathrm{2}}} \ \|\ \textcolor{\coeffectcolor}{1}  } }%
\ottpremise{  \textcolor{\coeffectcolor}{ \textcolor{\coeffectcolor}{\gamma}_{{\mathrm{1}}} }\! \cdot \! \rho   \vdash_{\mathit{full} }  \ottnt{M}  \Downarrow   \ottkw{return} _{\textcolor{\coeffectcolor}{ \textcolor{\coeffectcolor}{q}_{{\mathrm{1}}} } }  \ottnt{W}   \mathop{\#} \textcolor{\effectcolor}{ \textcolor{\effectcolor}{\phi}_{{\mathrm{1}}} } }%
\ottpremise{   \textcolor{\coeffectcolor}{ \textcolor{\coeffectcolor}{\gamma}_{{\mathrm{2}}} }\! \cdot \! \rho    \mathop{,} \;   \ottmv{x}   \mapsto ^{\textcolor{\coeffectcolor}{  \textcolor{\coeffectcolor}{ \textcolor{\coeffectcolor}{q}_{{\mathrm{1}}}   \cdot   \textcolor{\coeffectcolor}{q}'_{{\mathrm{2}}} }  } }  \ottnt{W}    \vdash_{\mathit{full} }  \ottnt{N}  \Downarrow  \ottnt{T}  \mathop{\#} \textcolor{\effectcolor}{ \textcolor{\effectcolor}{\phi}_{{\mathrm{2}}} } }%
}{
  \textcolor{\coeffectcolor}{  \textcolor{\coeffectcolor}{  \textcolor{\coeffectcolor}{ \textcolor{\coeffectcolor}{q}'_{{\mathrm{2}}} \cdot \textcolor{\coeffectcolor}{\gamma}_{{\mathrm{1}}} }  \ottsym{+} \textcolor{\coeffectcolor}{\gamma}_{{\mathrm{2}}} }  }\! \cdot \! \rho   \vdash_{\mathit{full} }   \ottmv{x}  \leftarrow^{\textcolor{\coeffectcolor}{ \textcolor{\coeffectcolor}{q}_{{\mathrm{2}}} } }  \ottnt{M} \ \ottkw{in}\  \ottnt{N}   \Downarrow  \ottnt{T}  \mathop{\#} \textcolor{\effectcolor}{  \textcolor{\effectcolor}{ \textcolor{\effectcolor}{\phi}_{{\mathrm{1}}}  \cdot  \textcolor{\effectcolor}{\phi}_{{\mathrm{2}}} }  } }{%
{\ottdrulename{eval\_full\_comp\_letin}}{}%
}}

\newcommand{\ottdruleevalXXfullXXcompXXletinXXzero}[1]{\ottdrule[#1]{%
\ottpremise{   \textcolor{\coeffectcolor}{ \textcolor{\coeffectcolor}{\gamma} }\! \cdot \! \rho    \mathop{,} \;   \ottmv{x}   \mapsto ^{\textcolor{\coeffectcolor}{  \textcolor{\coeffectcolor}{0}  } }  \mathop{\lightning}    \vdash_{\mathit{full} }  \ottnt{N}  \Downarrow  \ottnt{T}  \mathop{\#} \textcolor{\effectcolor}{ \textcolor{\effectcolor}{\phi} } }%
}{
  \textcolor{\coeffectcolor}{ \textcolor{\coeffectcolor}{\gamma} }\! \cdot \! \rho   \vdash_{\mathit{full} }   \ottmv{x}  \leftarrow^{\textcolor{\coeffectcolor}{0} }_{\textcolor{\effectcolor}{\varepsilon} }  \ottnt{M} \ \ottkw{in}\  \ottnt{N}   \Downarrow  \ottnt{T}  \mathop{\#} \textcolor{\effectcolor}{ \textcolor{\effectcolor}{\phi} } }{%
{\ottdrulename{eval\_full\_comp\_letin\_zero}}{}%
}}

\newcommand{\ottdruleevalXXfullXXcompXXforce}[1]{\ottdrule[#1]{%
\ottpremise{  \textcolor{\coeffectcolor}{ \textcolor{\coeffectcolor}{\gamma} }\! \cdot \! \rho   \vdash_{\mathit{full} }  \ottnt{V}  \Downarrow   \mathbf{clo}( \textcolor{\coeffectcolor}{\gamma}' ,  \rho' ,  \ottnt{M} )  }%
\ottpremise{  \textcolor{\coeffectcolor}{ \textcolor{\coeffectcolor}{\gamma}' }\! \cdot \! \rho'   \vdash_{\mathit{full} }  \ottnt{M}  \Downarrow  \ottnt{T}  \mathop{\#} \textcolor{\effectcolor}{ \textcolor{\effectcolor}{\phi} } }%
}{
  \textcolor{\coeffectcolor}{ \textcolor{\coeffectcolor}{\gamma} }\! \cdot \! \rho   \vdash_{\mathit{full} }  \ottnt{V}  \ottsym{!}  \Downarrow  \ottnt{T}  \mathop{\#} \textcolor{\effectcolor}{ \textcolor{\effectcolor}{\phi} } }{%
{\ottdrulename{eval\_full\_comp\_force}}{}%
}}

\newcommand{\ottdruleevalXXfullXXcompXXfst}[1]{\ottdrule[#1]{%
\ottpremise{  \textcolor{\coeffectcolor}{ \textcolor{\coeffectcolor}{\gamma} }\! \cdot \! \rho   \vdash_{\mathit{full} }  \ottnt{M}  \Downarrow   \mathbf{clo}(   \textcolor{\coeffectcolor}{ \textcolor{\coeffectcolor}{\gamma}' }\! \cdot \! \rho'  ,   \langle  \ottnt{M_{{\mathrm{1}}}} , \ottnt{M_{{\mathrm{2}}}}  \rangle   )   \mathop{\#} \textcolor{\effectcolor}{ \textcolor{\effectcolor}{\phi}_{{\mathrm{1}}} } }%
\ottpremise{  \textcolor{\coeffectcolor}{ \textcolor{\coeffectcolor}{\gamma}' }\! \cdot \! \rho'   \vdash_{\mathit{full} }  \ottnt{M_{{\mathrm{1}}}}  \Downarrow  \ottnt{T}  \mathop{\#} \textcolor{\effectcolor}{ \textcolor{\effectcolor}{\phi}_{{\mathrm{2}}} } }%
}{
  \textcolor{\coeffectcolor}{ \textcolor{\coeffectcolor}{\gamma} }\! \cdot \! \rho   \vdash_{\mathit{full} }   \ottnt{M}  . 1   \Downarrow  \ottnt{T}  \mathop{\#} \textcolor{\effectcolor}{  \textcolor{\effectcolor}{ \textcolor{\effectcolor}{\phi}_{{\mathrm{1}}}  \cdot  \textcolor{\effectcolor}{\phi}_{{\mathrm{2}}} }  } }{%
{\ottdrulename{eval\_full\_comp\_fst}}{}%
}}

\newcommand{\ottdruleevalXXfullXXcompXXsnd}[1]{\ottdrule[#1]{%
\ottpremise{  \textcolor{\coeffectcolor}{ \textcolor{\coeffectcolor}{\gamma} }\! \cdot \! \rho   \vdash_{\mathit{full} }  \ottnt{M}  \Downarrow   \mathbf{clo}(   \textcolor{\coeffectcolor}{ \textcolor{\coeffectcolor}{\gamma}' }\! \cdot \! \rho'  ,   \langle  \ottnt{M_{{\mathrm{1}}}} , \ottnt{M_{{\mathrm{2}}}}  \rangle   )   \mathop{\#} \textcolor{\effectcolor}{ \textcolor{\effectcolor}{\phi}_{{\mathrm{1}}} } }%
\ottpremise{  \textcolor{\coeffectcolor}{ \textcolor{\coeffectcolor}{\gamma}' }\! \cdot \! \rho'   \vdash_{\mathit{full} }  \ottnt{M_{{\mathrm{2}}}}  \Downarrow  \ottnt{T}  \mathop{\#} \textcolor{\effectcolor}{ \textcolor{\effectcolor}{\phi}_{{\mathrm{2}}} } }%
}{
  \textcolor{\coeffectcolor}{ \textcolor{\coeffectcolor}{\gamma} }\! \cdot \! \rho   \vdash_{\mathit{full} }   \ottnt{M}  . 2   \Downarrow  \ottnt{T}  \mathop{\#} \textcolor{\effectcolor}{  \textcolor{\effectcolor}{ \textcolor{\effectcolor}{\phi}_{{\mathrm{1}}}  \cdot  \textcolor{\effectcolor}{\phi}_{{\mathrm{2}}} }  } }{%
{\ottdrulename{eval\_full\_comp\_snd}}{}%
}}

\newcommand{\ottdruleevalXXfullXXcompXXsequence}[1]{\ottdrule[#1]{%
\ottpremise{  \textcolor{\coeffectcolor}{ \textcolor{\coeffectcolor}{\gamma}_{{\mathrm{1}}} }\! \cdot \! \rho   \vdash_{\mathit{full} }  \ottnt{V}  \Downarrow  \ottsym{()} }%
\ottpremise{  \textcolor{\coeffectcolor}{ \textcolor{\coeffectcolor}{\gamma}_{{\mathrm{2}}} }\! \cdot \! \rho   \vdash_{\mathit{full} }  \ottnt{N}  \Downarrow  \ottnt{T}  \mathop{\#} \textcolor{\effectcolor}{ \textcolor{\effectcolor}{\phi} } }%
\ottpremise{ \textcolor{\coeffectcolor}{ \textcolor{\coeffectcolor}{\gamma} } \equiv \textcolor{\coeffectcolor}{  \textcolor{\coeffectcolor}{ \textcolor{\coeffectcolor}{\gamma}_{{\mathrm{1}}} \ottsym{+} \textcolor{\coeffectcolor}{\gamma}_{{\mathrm{2}}} }  } }%
}{
  \textcolor{\coeffectcolor}{ \textcolor{\coeffectcolor}{\gamma} }\! \cdot \! \rho   \vdash_{\mathit{full} }   \ottnt{V}  ;  \ottnt{N}   \Downarrow  \ottnt{T}  \mathop{\#} \textcolor{\effectcolor}{ \textcolor{\effectcolor}{\phi} } }{%
{\ottdrulename{eval\_full\_comp\_sequence}}{}%
}}

\newcommand{\ottdruleevalXXfullXXcompXXcasel}[1]{\ottdrule[#1]{%
\ottpremise{  \textcolor{\coeffectcolor}{ \textcolor{\coeffectcolor}{\gamma}_{{\mathrm{1}}} }\! \cdot \! \rho   \vdash_{\mathit{full} }  \ottnt{V}  \Downarrow  \ottkw{inl} \, \ottnt{W} }%
\ottpremise{   \textcolor{\coeffectcolor}{ \textcolor{\coeffectcolor}{\gamma}_{{\mathrm{2}}} }\! \cdot \! \rho    \mathop{,} \;   \ottmv{x_{{\mathrm{1}}}}   \mapsto ^{\textcolor{\coeffectcolor}{ \textcolor{\coeffectcolor}{q} } }  \ottnt{W}    \vdash_{\mathit{full} }  \ottnt{M_{{\mathrm{1}}}}  \Downarrow  \ottnt{T}  \mathop{\#} \textcolor{\effectcolor}{ \textcolor{\effectcolor}{\phi} } }%
\ottpremise{ \textcolor{\coeffectcolor}{ \textcolor{\coeffectcolor}{\gamma} } \equiv \textcolor{\coeffectcolor}{  \textcolor{\coeffectcolor}{  \textcolor{\coeffectcolor}{ \textcolor{\coeffectcolor}{q} \cdot \textcolor{\coeffectcolor}{\gamma}_{{\mathrm{1}}} }  \ottsym{+} \textcolor{\coeffectcolor}{\gamma}_{{\mathrm{2}}} }  } }%
\ottpremise{ \textcolor{\coeffectcolor}{ \textcolor{\coeffectcolor}{q} }\;  \textcolor{\coeffectcolor}{\mathop{\leq_{\mathit{co} } } } \; \textcolor{\coeffectcolor}{  \textcolor{\coeffectcolor}{1}  } }%
}{
  \textcolor{\coeffectcolor}{ \textcolor{\coeffectcolor}{\gamma} }\! \cdot \! \rho   \vdash_{\mathit{full} }   \ottkw{case}_{\textcolor{\coeffectcolor}{ \textcolor{\coeffectcolor}{q} } }\;  \ottnt{V} \; \ottkw{of}\; \ottkw{inl} \; \ottmv{x_{{\mathrm{1}}}}  \rightarrow\;  \ottnt{M_{{\mathrm{1}}}}  ;  \ottkw{inr} \; \ottmv{x_{{\mathrm{2}}}}  \rightarrow\;  \ottnt{M_{{\mathrm{2}}}}   \Downarrow  \ottnt{T}  \mathop{\#} \textcolor{\effectcolor}{ \textcolor{\effectcolor}{\phi} } }{%
{\ottdrulename{eval\_full\_comp\_casel}}{}%
}}

\newcommand{\ottdruleevalXXfullXXcompXXcaser}[1]{\ottdrule[#1]{%
\ottpremise{  \textcolor{\coeffectcolor}{ \textcolor{\coeffectcolor}{\gamma}_{{\mathrm{1}}} }\! \cdot \! \rho   \vdash_{\mathit{full} }  \ottnt{V}  \Downarrow  \ottkw{inr} \, \ottnt{W} }%
\ottpremise{   \textcolor{\coeffectcolor}{ \textcolor{\coeffectcolor}{\gamma}_{{\mathrm{2}}} }\! \cdot \! \rho    \mathop{,} \;   \ottmv{x_{{\mathrm{2}}}}   \mapsto ^{\textcolor{\coeffectcolor}{ \textcolor{\coeffectcolor}{q} } }  \ottnt{W}    \vdash_{\mathit{full} }  \ottnt{M_{{\mathrm{2}}}}  \Downarrow  \ottnt{T}  \mathop{\#} \textcolor{\effectcolor}{ \textcolor{\effectcolor}{\phi} } }%
\ottpremise{ \textcolor{\coeffectcolor}{ \textcolor{\coeffectcolor}{\gamma} } \equiv \textcolor{\coeffectcolor}{  \textcolor{\coeffectcolor}{  \textcolor{\coeffectcolor}{ \textcolor{\coeffectcolor}{q} \cdot \textcolor{\coeffectcolor}{\gamma}_{{\mathrm{1}}} }  \ottsym{+} \textcolor{\coeffectcolor}{\gamma}_{{\mathrm{2}}} }  } }%
\ottpremise{ \textcolor{\coeffectcolor}{ \textcolor{\coeffectcolor}{q} }\;  \textcolor{\coeffectcolor}{\mathop{\leq_{\mathit{co} } } } \; \textcolor{\coeffectcolor}{  \textcolor{\coeffectcolor}{1}  } }%
}{
  \textcolor{\coeffectcolor}{ \textcolor{\coeffectcolor}{\gamma} }\! \cdot \! \rho   \vdash_{\mathit{full} }   \ottkw{case}_{\textcolor{\coeffectcolor}{ \textcolor{\coeffectcolor}{q} } }\;  \ottnt{V} \; \ottkw{of}\; \ottkw{inl} \; \ottmv{x_{{\mathrm{1}}}}  \rightarrow\;  \ottnt{M_{{\mathrm{1}}}}  ;  \ottkw{inr} \; \ottmv{x_{{\mathrm{2}}}}  \rightarrow\;  \ottnt{M_{{\mathrm{2}}}}   \Downarrow  \ottnt{T}  \mathop{\#} \textcolor{\effectcolor}{ \textcolor{\effectcolor}{\phi} } }{%
{\ottdrulename{eval\_full\_comp\_caser}}{}%
}}

\newcommand{\ottdruleevalXXfullXXcompXXappXXabsXXzero}[1]{\ottdrule[#1]{%
\ottpremise{  \textcolor{\coeffectcolor}{ \textcolor{\coeffectcolor}{\gamma} }\! \cdot \! \rho   \vdash_{\mathit{full} }  \ottnt{M}  \Downarrow   \mathbf{clo}(   \textcolor{\coeffectcolor}{ \textcolor{\coeffectcolor}{\gamma}' }\! \cdot \! \rho'  ,  \ottnt{M'}  )   \mathop{\#} \textcolor{\effectcolor}{ \textcolor{\effectcolor}{\phi}_{{\mathrm{1}}} } }%
\ottpremise{   \textcolor{\coeffectcolor}{ \textcolor{\coeffectcolor}{\gamma}' }\! \cdot \! \rho'    \mathop{,} \;   \ottmv{x}   \mapsto ^{\textcolor{\coeffectcolor}{  \textcolor{\coeffectcolor}{0}  } }  \mathop{\lightning}    \vdash_{\mathit{full} }  \ottnt{M'}  \Downarrow  \ottnt{T}  \mathop{\#} \textcolor{\effectcolor}{ \textcolor{\effectcolor}{\phi}_{{\mathrm{2}}} } }%
}{
  \textcolor{\coeffectcolor}{ \textcolor{\coeffectcolor}{\gamma} }\! \cdot \! \rho   \vdash_{\mathit{full} }  \ottnt{M} \, \ottnt{V}  \Downarrow  \ottnt{T}  \mathop{\#} \textcolor{\effectcolor}{  \textcolor{\effectcolor}{ \textcolor{\effectcolor}{\phi}_{{\mathrm{1}}}  \cdot  \textcolor{\effectcolor}{\phi}_{{\mathrm{2}}} }  } }{%
{\ottdrulename{eval\_full\_comp\_app\_abs\_zero}}{}%
}}

\newcommand{\ottdruleevalXXfullXXcompXXsplitXXzero}[1]{\ottdrule[#1]{%
\ottpremise{     \textcolor{\coeffectcolor}{ \textcolor{\coeffectcolor}{\gamma} }\! \cdot \! \rho    \mathop{,} \;   \ottmv{x_{{\mathrm{1}}}}   \mapsto ^{\textcolor{\coeffectcolor}{  \textcolor{\coeffectcolor}{0}  } }  \mathop{\lightning}      \mathop{,} \;   \ottmv{x_{{\mathrm{2}}}}   \mapsto ^{\textcolor{\coeffectcolor}{  \textcolor{\coeffectcolor}{0}  } }  \mathop{\lightning}    \vdash_{\mathit{full} }  \ottnt{N}  \Downarrow  \ottnt{T}  \mathop{\#} \textcolor{\effectcolor}{ \textcolor{\effectcolor}{\phi} } }%
}{
  \textcolor{\coeffectcolor}{ \textcolor{\coeffectcolor}{\gamma} }\! \cdot \! \rho   \vdash_{\mathit{full} }   \ottkw{case}_{\textcolor{\coeffectcolor}{  \textcolor{\coeffectcolor}{0}  } } \;  \ottnt{V} \; \ottkw{of}\;( \ottmv{x_{{\mathrm{1}}}} , \ottmv{x_{{\mathrm{2}}}} )\; \rightarrow\;  \ottnt{N}   \Downarrow  \ottnt{T}  \mathop{\#} \textcolor{\effectcolor}{ \textcolor{\effectcolor}{\phi} } }{%
{\ottdrulename{eval\_full\_comp\_split\_zero}}{}%
}}

\newcommand{\ottdruleevalXXfullXXcompXXretXXzero}[1]{\ottdrule[#1]{%
}{
  \textcolor{\coeffectcolor}{  \textcolor{\coeffectcolor}{\overline{0} }  }\! \cdot \! \rho   \vdash_{\mathit{full} }   \ottkw{return} _{\textcolor{\coeffectcolor}{  \textcolor{\coeffectcolor}{0}  } }\;  \ottnt{V}   \Downarrow   \ottkw{return} _{\textcolor{\coeffectcolor}{  \textcolor{\coeffectcolor}{0}  } }  \mathop{\lightning}   \mathop{\#} \textcolor{\effectcolor}{  \textcolor{\effectcolor}{\varepsilon}  } }{%
{\ottdrulename{eval\_full\_comp\_ret\_zero}}{}%
}}

\newcommand{\ottdruleevalXXfullXXcompXXtick}[1]{\ottdrule[#1]{%
}{
  \textcolor{\coeffectcolor}{  \textcolor{\coeffectcolor}{\overline{0} }  }\! \cdot \! \rho   \vdash_{\mathit{full} }   \textcolor{\effectcolor}{ \ottkw{tick} }   \Downarrow   \ottkw{return} _{\textcolor{\coeffectcolor}{  \textcolor{\coeffectcolor}{1}  } }  \ottsym{()}   \mathop{\#} \textcolor{\effectcolor}{  \textcolor{\effectcolor}{\ottkw{Tick} }  } }{%
{\ottdrulename{eval\_full\_comp\_tick}}{}%
}}

\newcommand{\ottdruleevalXXfullXXcompXXcsub}[1]{\ottdrule[#1]{%
\ottpremise{  \textcolor{\coeffectcolor}{ \textcolor{\coeffectcolor}{\gamma}' }\! \cdot \! \rho   \vdash_{\mathit{full} }  \ottnt{M}  \Downarrow  \ottnt{T}  \mathop{\#} \textcolor{\effectcolor}{ \textcolor{\effectcolor}{\phi} } }%
\ottpremise{ \textcolor{\coeffectcolor}{ \textcolor{\coeffectcolor}{\gamma} }\; \textcolor{\coeffectcolor}{\mathop{\leq_{\mathit{co} } } } \; \textcolor{\coeffectcolor}{ \textcolor{\coeffectcolor}{\gamma}' } }%
}{
  \textcolor{\coeffectcolor}{ \textcolor{\coeffectcolor}{\gamma} }\! \cdot \! \rho   \vdash_{\mathit{full} }  \ottnt{M}  \Downarrow  \ottnt{T}  \mathop{\#} \textcolor{\effectcolor}{ \textcolor{\effectcolor}{\phi} } }{%
{\ottdrulename{eval\_full\_comp\_csub}}{}%
}}

\newcommand{\ottdruleevalXXfullXXcompXXdist}[1]{\ottdrule[#1]{%
\ottpremise{  \textcolor{\coeffectcolor}{ \textcolor{\coeffectcolor}{\gamma} }\! \cdot \! \rho   \vdash_{\mathit{full} }  \ottnt{M}  \Downarrow   \ottkw{return} _{\textcolor{\coeffectcolor}{ \textcolor{\coeffectcolor}{q} } }   \mathbf{clo}( \textcolor{\coeffectcolor}{\gamma}' ,  \rho' ,  \ottnt{M'} )    \mathop{\#} \textcolor{\effectcolor}{ \textcolor{\effectcolor}{\phi}' } }%
}{
  \textcolor{\coeffectcolor}{ \textcolor{\coeffectcolor}{\gamma} }\! \cdot \! \rho   \vdash_{\mathit{full} }  \ottkw{dist} \, \ottnt{M}  \Downarrow   \ottkw{return} _{\textcolor{\coeffectcolor}{  \textcolor{\coeffectcolor}{1}  } }   \mathbf{dclo}( \textcolor{\coeffectcolor}{\gamma}' ,  \rho' ,  \textcolor{\coeffectcolor}{q} ,  \ottnt{M'} )    \mathop{\#} \textcolor{\effectcolor}{ \textcolor{\effectcolor}{\phi}' } }{%
{\ottdrulename{eval\_full\_comp\_dist}}{}%
}}

\newcommand{\ottdruleevalXXfullXXcompXXdistXXforce}[1]{\ottdrule[#1]{%
\ottpremise{  \textcolor{\coeffectcolor}{ \textcolor{\coeffectcolor}{\gamma} }\! \cdot \! \rho   \vdash_{\mathit{full} }  \ottnt{V}  \Downarrow   \mathbf{dclo}( \textcolor{\coeffectcolor}{\gamma}' ,  \rho' ,  \textcolor{\coeffectcolor}{q} ,  \ottnt{M'} )  }%
\ottpremise{  \textcolor{\coeffectcolor}{ \textcolor{\coeffectcolor}{\gamma}' }\! \cdot \! \rho'   \vdash_{\mathit{full} }  \ottnt{M'}  \Downarrow   \ottkw{return} _{\textcolor{\coeffectcolor}{  \textcolor{\coeffectcolor}{1}  } }  \ottnt{W}   \mathop{\#} \textcolor{\effectcolor}{ \textcolor{\effectcolor}{\phi} } }%
}{
  \textcolor{\coeffectcolor}{ \textcolor{\coeffectcolor}{\gamma} }\! \cdot \! \rho   \vdash_{\mathit{full} }  \ottnt{V}  \ottsym{!}  \Downarrow   \ottkw{return} _{\textcolor{\coeffectcolor}{ \ottsym{(}   \textcolor{\coeffectcolor}{q}  \odot  \textcolor{\effectcolor}{\phi}   \ottsym{)} } }  \ottnt{W}   \mathop{\#} \textcolor{\effectcolor}{  \textcolor{\coeffectcolor}{q}  \oslash  \textcolor{\effectcolor}{\phi}  } }{%
{\ottdrulename{eval\_full\_comp\_dist\_force}}{}%
}}

\newcommand{\ottdruleevalXXfullXXcompXXdistXXdist}[1]{\ottdrule[#1]{%
\ottpremise{  \textcolor{\coeffectcolor}{ \textcolor{\coeffectcolor}{\gamma} }\! \cdot \! \rho   \vdash_{\mathit{full} }  \ottnt{M}  \Downarrow   \ottkw{return} _{\textcolor{\coeffectcolor}{  \textcolor{\coeffectcolor}{1}  } }   \mathbf{dclo}( \textcolor{\coeffectcolor}{\gamma}' ,  \rho' ,   \textcolor{\coeffectcolor}{1}  ,  \ottnt{M'} )    \mathop{\#} \textcolor{\effectcolor}{ \textcolor{\effectcolor}{\phi}' } }%
}{
  \textcolor{\coeffectcolor}{ \textcolor{\coeffectcolor}{\gamma} }\! \cdot \! \rho   \vdash_{\mathit{full} }  \ottkw{dist} \, \ottnt{M}  \Downarrow   \ottkw{return} _{\textcolor{\coeffectcolor}{  \textcolor{\coeffectcolor}{1}  } }   \mathbf{clo}( \textcolor{\coeffectcolor}{\gamma}' ,  \rho' ,  \ottnt{M'} )    \mathop{\#} \textcolor{\effectcolor}{ \textcolor{\effectcolor}{\phi}' } }{%
{\ottdrulename{eval\_full\_comp\_dist\_dist}}{}%
}}

\newcommand{\ottdefnevalXXfullXXcomp}[1]{\begin{ottdefnblock}[#1]{$ \mu  \vdash_{\mathit{full} }  \ottnt{M}  \Downarrow  \ottnt{T}  \mathop{\#} \textcolor{\effectcolor}{ \textcolor{\effectcolor}{\phi} } $}{\ottcom{environment-based effect-coeffect semantics for CBPV (large-step)}}
\ottusedrule{\ottdruleevalXXfullXXcompXXabs{}}
\ottusedrule{\ottdruleevalXXfullXXcompXXcpair{}}
\ottusedrule{\ottdruleevalXXfullXXcompXXapp{}}
\ottusedrule{\ottdruleevalXXfullXXcompXXsplit{}}
\ottusedrule{\ottdruleevalXXfullXXcompXXret{}}
\ottusedrule{\ottdruleevalXXfullXXcompXXletin{}}
\ottusedrule{\ottdruleevalXXfullXXcompXXletinXXzero{}}
\ottusedrule{\ottdruleevalXXfullXXcompXXforce{}}
\ottusedrule{\ottdruleevalXXfullXXcompXXfst{}}
\ottusedrule{\ottdruleevalXXfullXXcompXXsnd{}}
\ottusedrule{\ottdruleevalXXfullXXcompXXsequence{}}
\ottusedrule{\ottdruleevalXXfullXXcompXXcasel{}}
\ottusedrule{\ottdruleevalXXfullXXcompXXcaser{}}
\ottusedrule{\ottdruleevalXXfullXXcompXXappXXabsXXzero{}}
\ottusedrule{\ottdruleevalXXfullXXcompXXsplitXXzero{}}
\ottusedrule{\ottdruleevalXXfullXXcompXXretXXzero{}}
\ottusedrule{\ottdruleevalXXfullXXcompXXtick{}}
\ottusedrule{\ottdruleevalXXfullXXcompXXcsub{}}
\ottusedrule{\ottdruleevalXXfullXXcompXXdist{}}
\ottusedrule{\ottdruleevalXXfullXXcompXXdistXXforce{}}
\ottusedrule{\ottdruleevalXXfullXXcompXXdistXXdist{}}
\end{ottdefnblock}}

\newcommand{\ottdefnsJFullOp}{
\ottdefnevalXXfullXXval{}\ottdefnevalXXfullXXcomp{}}

\newcommand{\ottdruleevalXXgfullXXcompXXabs}[1]{\ottdrule[#1]{%
\ottpremise{ \textcolor{\coeffectcolor}{ \textcolor{\coeffectcolor}{q}' }\;  \textcolor{\coeffectcolor}{\mathop{\leq_{\mathit{co} } } } \; \textcolor{\coeffectcolor}{ \textcolor{\coeffectcolor}{q} } }%
}{
  \textcolor{\coeffectcolor}{ \textcolor{\coeffectcolor}{\gamma} }\! \cdot \! \rho   \vdash_{\mathit{gen} }   \lambda  \ottmv{x} ^{\textcolor{\coeffectcolor}{ \textcolor{\coeffectcolor}{q} } }. \ottnt{M}   \Downarrow   \mathbf{clo}(   \textcolor{\coeffectcolor}{ \textcolor{\coeffectcolor}{\gamma} }\! \cdot \! \rho  ,   \lambda  \ottmv{x} ^{\textcolor{\coeffectcolor}{ \textcolor{\coeffectcolor}{q}' } }. \ottnt{M}   )   \mathop{\#} \textcolor{\effectcolor}{  \textcolor{\effectcolor}{\varepsilon}  } }{%
{\ottdrulename{eval\_gfull\_comp\_abs}}{}%
}}

\newcommand{\ottdruleevalXXgfullXXcompXXcpair}[1]{\ottdrule[#1]{%
}{
  \textcolor{\coeffectcolor}{ \textcolor{\coeffectcolor}{\gamma} }\! \cdot \! \rho   \vdash_{\mathit{gen} }   \langle  \ottnt{M_{{\mathrm{1}}}} , \ottnt{M_{{\mathrm{2}}}}  \rangle   \Downarrow   \mathbf{clo}(   \textcolor{\coeffectcolor}{ \textcolor{\coeffectcolor}{\gamma} }\! \cdot \! \rho  ,   \langle  \ottnt{M_{{\mathrm{1}}}} , \ottnt{M_{{\mathrm{2}}}}  \rangle   )   \mathop{\#} \textcolor{\effectcolor}{  \textcolor{\effectcolor}{\varepsilon}  } }{%
{\ottdrulename{eval\_gfull\_comp\_cpair}}{}%
}}

\newcommand{\ottdruleevalXXgfullXXcompXXapp}[1]{\ottdrule[#1]{%
\ottpremise{  \textcolor{\coeffectcolor}{ \textcolor{\coeffectcolor}{\gamma}_{{\mathrm{1}}} }\! \cdot \! \rho   \vdash_{\mathit{gen} }  \ottnt{M}  \Downarrow   \mathbf{clo}(   \textcolor{\coeffectcolor}{ \textcolor{\coeffectcolor}{\gamma}' }\! \cdot \! \rho'  ,   \lambda  \ottmv{x} ^{\textcolor{\coeffectcolor}{ \textcolor{\coeffectcolor}{q} } }. \ottnt{M'}   )   \mathop{\#} \textcolor{\effectcolor}{ \textcolor{\effectcolor}{\phi}_{{\mathrm{1}}} } }%
\ottpremise{  \textcolor{\coeffectcolor}{ \textcolor{\coeffectcolor}{\gamma}_{{\mathrm{2}}} }\! \cdot \! \rho   \vdash_{\mathit{full} }  \ottnt{V}  \Downarrow  \ottnt{W} }%
\ottpremise{   \textcolor{\coeffectcolor}{ \textcolor{\coeffectcolor}{\gamma}' }\! \cdot \! \rho'    \mathop{,} \;   \ottmv{x}   \mapsto ^{\textcolor{\coeffectcolor}{ \textcolor{\coeffectcolor}{q} } }  \ottnt{W}    \vdash_{\mathit{gen} }  \ottnt{M'}  \Downarrow  \ottnt{T}  \mathop{\#} \textcolor{\effectcolor}{ \textcolor{\effectcolor}{\phi}_{{\mathrm{2}}} } }%
\ottpremise{ \textcolor{\coeffectcolor}{ \textcolor{\coeffectcolor}{\gamma} } \equiv \textcolor{\coeffectcolor}{  \textcolor{\coeffectcolor}{ \textcolor{\coeffectcolor}{\gamma}_{{\mathrm{1}}} \ottsym{+}  \textcolor{\coeffectcolor}{ \textcolor{\coeffectcolor}{q} \cdot \textcolor{\coeffectcolor}{\gamma}_{{\mathrm{2}}} }  }  } }%
}{
  \textcolor{\coeffectcolor}{ \textcolor{\coeffectcolor}{\gamma} }\! \cdot \! \rho   \vdash_{\mathit{gen} }  \ottnt{M} \, \ottnt{V}  \Downarrow  \ottnt{T}  \mathop{\#} \textcolor{\effectcolor}{  \textcolor{\effectcolor}{ \textcolor{\effectcolor}{\phi}_{{\mathrm{1}}}  \cdot  \textcolor{\effectcolor}{\phi}_{{\mathrm{2}}} }  } }{%
{\ottdrulename{eval\_gfull\_comp\_app}}{}%
}}

\newcommand{\ottdruleevalXXgfullXXcompXXsplit}[1]{\ottdrule[#1]{%
\ottpremise{  \textcolor{\coeffectcolor}{ \textcolor{\coeffectcolor}{\gamma}_{{\mathrm{1}}} }\! \cdot \! \rho   \vdash_{\mathit{full} }  \ottnt{V}  \Downarrow  \ottsym{(}  \ottnt{W_{{\mathrm{1}}}}  \ottsym{,}  \ottnt{W_{{\mathrm{2}}}}  \ottsym{)} }%
\ottpremise{     \textcolor{\coeffectcolor}{ \textcolor{\coeffectcolor}{\gamma}_{{\mathrm{2}}} }\! \cdot \! \rho    \mathop{,} \;   \ottmv{x_{{\mathrm{1}}}}   \mapsto ^{\textcolor{\coeffectcolor}{ \textcolor{\coeffectcolor}{q} } }  \ottnt{W_{{\mathrm{1}}}}      \mathop{,} \;   \ottmv{x_{{\mathrm{2}}}}   \mapsto ^{\textcolor{\coeffectcolor}{ \textcolor{\coeffectcolor}{q} } }  \ottnt{W_{{\mathrm{2}}}}    \vdash_{\mathit{gen} }  \ottnt{N}  \Downarrow  \ottnt{T}  \mathop{\#} \textcolor{\effectcolor}{ \textcolor{\effectcolor}{\phi} } }%
\ottpremise{ \textcolor{\coeffectcolor}{ \textcolor{\coeffectcolor}{\gamma} } \equiv \textcolor{\coeffectcolor}{  \textcolor{\coeffectcolor}{  \textcolor{\coeffectcolor}{ \textcolor{\coeffectcolor}{q} \cdot \textcolor{\coeffectcolor}{\gamma}_{{\mathrm{1}}} }  \ottsym{+} \textcolor{\coeffectcolor}{\gamma}_{{\mathrm{2}}} }  } }%
}{
  \textcolor{\coeffectcolor}{ \textcolor{\coeffectcolor}{\gamma} }\! \cdot \! \rho   \vdash_{\mathit{gen} }   \ottkw{case}_{\textcolor{\coeffectcolor}{ \textcolor{\coeffectcolor}{q} } } \;  \ottnt{V} \; \ottkw{of}\;( \ottmv{x_{{\mathrm{1}}}} , \ottmv{x_{{\mathrm{2}}}} )\; \rightarrow\;  \ottnt{N}   \Downarrow  \ottnt{T}  \mathop{\#} \textcolor{\effectcolor}{ \textcolor{\effectcolor}{\phi} } }{%
{\ottdrulename{eval\_gfull\_comp\_split}}{}%
}}

\newcommand{\ottdruleevalXXgfullXXcompXXret}[1]{\ottdrule[#1]{%
\ottpremise{  \textcolor{\coeffectcolor}{ \textcolor{\coeffectcolor}{\gamma}' }\! \cdot \! \rho   \vdash_{\mathit{full} }  \ottnt{V}  \Downarrow  \ottnt{W} }%
\ottpremise{ \textcolor{\coeffectcolor}{ \textcolor{\coeffectcolor}{\gamma} } \equiv \textcolor{\coeffectcolor}{  \textcolor{\coeffectcolor}{ \textcolor{\coeffectcolor}{q} \cdot \textcolor{\coeffectcolor}{\gamma}' }  } }%
}{
  \textcolor{\coeffectcolor}{ \textcolor{\coeffectcolor}{\gamma} }\! \cdot \! \rho   \vdash_{\mathit{gen} }   \ottkw{return} _{\textcolor{\coeffectcolor}{ \textcolor{\coeffectcolor}{q} } }\;  \ottnt{V}   \Downarrow   \ottkw{return} _{\textcolor{\coeffectcolor}{ \textcolor{\coeffectcolor}{q} } }  \ottnt{W}   \mathop{\#} \textcolor{\effectcolor}{  \textcolor{\effectcolor}{\varepsilon}  } }{%
{\ottdrulename{eval\_gfull\_comp\_ret}}{}%
}}

\newcommand{\ottdruleevalXXgfullXXcompXXletin}[1]{\ottdrule[#1]{%
\ottpremise{ \textcolor{\coeffectcolor}{ \textcolor{\coeffectcolor}{q}'_{{\mathrm{2}}} }\; = \; \textcolor{\coeffectcolor}{  \textcolor{\coeffectcolor}{q}_{{\mathrm{2}}} \ \|\ \textcolor{\coeffectcolor}{1}  } }%
\ottpremise{  \textcolor{\coeffectcolor}{ \textcolor{\coeffectcolor}{\gamma}_{{\mathrm{1}}} }\! \cdot \! \rho   \vdash_{\mathit{gen} }  \ottnt{M}  \Downarrow   \ottkw{return} _{\textcolor{\coeffectcolor}{ \textcolor{\coeffectcolor}{q}_{{\mathrm{1}}} } }  \ottnt{W}   \mathop{\#} \textcolor{\effectcolor}{ \textcolor{\effectcolor}{\phi}_{{\mathrm{1}}} } }%
\ottpremise{   \textcolor{\coeffectcolor}{ \textcolor{\coeffectcolor}{\gamma}_{{\mathrm{2}}} }\! \cdot \! \rho    \mathop{,} \;   \ottmv{x}   \mapsto ^{\textcolor{\coeffectcolor}{  \textcolor{\coeffectcolor}{ \textcolor{\coeffectcolor}{q}_{{\mathrm{1}}}   \cdot   \textcolor{\coeffectcolor}{q}'_{{\mathrm{2}}} }  } }  \ottnt{W}    \vdash_{\mathit{gen} }  \ottnt{N}  \Downarrow  \ottnt{T}  \mathop{\#} \textcolor{\effectcolor}{ \textcolor{\effectcolor}{\phi}_{{\mathrm{2}}} } }%
\ottpremise{ \textcolor{\coeffectcolor}{ \textcolor{\coeffectcolor}{\gamma} } \equiv \textcolor{\coeffectcolor}{  \textcolor{\coeffectcolor}{  \textcolor{\coeffectcolor}{ \textcolor{\coeffectcolor}{q}'_{{\mathrm{2}}} \cdot \textcolor{\coeffectcolor}{\gamma}_{{\mathrm{1}}} }  \ottsym{+} \textcolor{\coeffectcolor}{\gamma}_{{\mathrm{2}}} }  } }%
}{
  \textcolor{\coeffectcolor}{ \textcolor{\coeffectcolor}{\gamma} }\! \cdot \! \rho   \vdash_{\mathit{gen} }   \ottmv{x}  \leftarrow^{\textcolor{\coeffectcolor}{ \textcolor{\coeffectcolor}{q}_{{\mathrm{2}}} } }  \ottnt{M} \ \ottkw{in}\  \ottnt{N}   \Downarrow  \ottnt{T}  \mathop{\#} \textcolor{\effectcolor}{  \textcolor{\effectcolor}{ \textcolor{\effectcolor}{\phi}_{{\mathrm{1}}}  \cdot  \textcolor{\effectcolor}{\phi}_{{\mathrm{2}}} }  } }{%
{\ottdrulename{eval\_gfull\_comp\_letin}}{}%
}}

\newcommand{\ottdruleevalXXgfullXXcompXXforce}[1]{\ottdrule[#1]{%
\ottpremise{  \textcolor{\coeffectcolor}{ \textcolor{\coeffectcolor}{\gamma} }\! \cdot \! \rho   \vdash_{\mathit{full} }  \ottnt{V}  \Downarrow   \mathbf{clo}( \textcolor{\coeffectcolor}{\gamma}' ,  \rho' ,  \ottnt{M} )  }%
\ottpremise{  \textcolor{\coeffectcolor}{ \textcolor{\coeffectcolor}{\gamma}' }\! \cdot \! \rho'   \vdash_{\mathit{gen} }  \ottnt{M}  \Downarrow  \ottnt{T}  \mathop{\#} \textcolor{\effectcolor}{ \textcolor{\effectcolor}{\phi} } }%
}{
  \textcolor{\coeffectcolor}{ \textcolor{\coeffectcolor}{\gamma} }\! \cdot \! \rho   \vdash_{\mathit{gen} }  \ottnt{V}  \ottsym{!}  \Downarrow  \ottnt{T}  \mathop{\#} \textcolor{\effectcolor}{ \textcolor{\effectcolor}{\phi} } }{%
{\ottdrulename{eval\_gfull\_comp\_force}}{}%
}}

\newcommand{\ottdruleevalXXgfullXXcompXXfst}[1]{\ottdrule[#1]{%
\ottpremise{  \textcolor{\coeffectcolor}{ \textcolor{\coeffectcolor}{\gamma} }\! \cdot \! \rho   \vdash_{\mathit{gen} }  \ottnt{M}  \Downarrow   \mathbf{clo}(   \textcolor{\coeffectcolor}{ \textcolor{\coeffectcolor}{\gamma}' }\! \cdot \! \rho'  ,   \langle  \ottnt{M_{{\mathrm{1}}}} , \ottnt{M_{{\mathrm{2}}}}  \rangle   )   \mathop{\#} \textcolor{\effectcolor}{ \textcolor{\effectcolor}{\phi}_{{\mathrm{1}}} } }%
\ottpremise{  \textcolor{\coeffectcolor}{ \textcolor{\coeffectcolor}{\gamma}' }\! \cdot \! \rho'   \vdash_{\mathit{gen} }  \ottnt{M_{{\mathrm{1}}}}  \Downarrow  \ottnt{T}  \mathop{\#} \textcolor{\effectcolor}{ \textcolor{\effectcolor}{\phi}_{{\mathrm{2}}} } }%
}{
  \textcolor{\coeffectcolor}{ \textcolor{\coeffectcolor}{\gamma} }\! \cdot \! \rho   \vdash_{\mathit{gen} }   \ottnt{M}  . 1   \Downarrow  \ottnt{T}  \mathop{\#} \textcolor{\effectcolor}{  \textcolor{\effectcolor}{ \textcolor{\effectcolor}{\phi}_{{\mathrm{1}}}  \cdot  \textcolor{\effectcolor}{\phi}_{{\mathrm{2}}} }  } }{%
{\ottdrulename{eval\_gfull\_comp\_fst}}{}%
}}

\newcommand{\ottdruleevalXXgfullXXcompXXsnd}[1]{\ottdrule[#1]{%
\ottpremise{  \textcolor{\coeffectcolor}{ \textcolor{\coeffectcolor}{\gamma} }\! \cdot \! \rho   \vdash_{\mathit{gen} }  \ottnt{M}  \Downarrow   \mathbf{clo}(   \textcolor{\coeffectcolor}{ \textcolor{\coeffectcolor}{\gamma}' }\! \cdot \! \rho'  ,   \langle  \ottnt{M_{{\mathrm{1}}}} , \ottnt{M_{{\mathrm{2}}}}  \rangle   )   \mathop{\#} \textcolor{\effectcolor}{ \textcolor{\effectcolor}{\phi}_{{\mathrm{1}}} } }%
\ottpremise{  \textcolor{\coeffectcolor}{ \textcolor{\coeffectcolor}{\gamma}' }\! \cdot \! \rho'   \vdash_{\mathit{gen} }  \ottnt{M_{{\mathrm{2}}}}  \Downarrow  \ottnt{T}  \mathop{\#} \textcolor{\effectcolor}{ \textcolor{\effectcolor}{\phi}_{{\mathrm{2}}} } }%
}{
  \textcolor{\coeffectcolor}{ \textcolor{\coeffectcolor}{\gamma} }\! \cdot \! \rho   \vdash_{\mathit{gen} }   \ottnt{M}  . 2   \Downarrow  \ottnt{T}  \mathop{\#} \textcolor{\effectcolor}{  \textcolor{\effectcolor}{ \textcolor{\effectcolor}{\phi}_{{\mathrm{1}}}  \cdot  \textcolor{\effectcolor}{\phi}_{{\mathrm{2}}} }  } }{%
{\ottdrulename{eval\_gfull\_comp\_snd}}{}%
}}

\newcommand{\ottdruleevalXXgfullXXcompXXseq}[1]{\ottdrule[#1]{%
\ottpremise{  \textcolor{\coeffectcolor}{ \textcolor{\coeffectcolor}{\gamma}_{{\mathrm{1}}} }\! \cdot \! \rho   \vdash_{\mathit{full} }  \ottnt{V}  \Downarrow  \ottsym{()} }%
\ottpremise{  \textcolor{\coeffectcolor}{ \textcolor{\coeffectcolor}{\gamma}_{{\mathrm{2}}} }\! \cdot \! \rho   \vdash_{\mathit{gen} }  \ottnt{N}  \Downarrow  \ottnt{T}  \mathop{\#} \textcolor{\effectcolor}{ \textcolor{\effectcolor}{\phi} } }%
\ottpremise{ \textcolor{\coeffectcolor}{ \textcolor{\coeffectcolor}{\gamma} } \equiv \textcolor{\coeffectcolor}{  \textcolor{\coeffectcolor}{ \textcolor{\coeffectcolor}{\gamma}_{{\mathrm{1}}} \ottsym{+} \textcolor{\coeffectcolor}{\gamma}_{{\mathrm{2}}} }  } }%
}{
  \textcolor{\coeffectcolor}{ \textcolor{\coeffectcolor}{\gamma} }\! \cdot \! \rho   \vdash_{\mathit{gen} }   \ottnt{V}  ;  \ottnt{N}   \Downarrow  \ottnt{T}  \mathop{\#} \textcolor{\effectcolor}{ \textcolor{\effectcolor}{\phi} } }{%
{\ottdrulename{eval\_gfull\_comp\_seq}}{}%
}}

\newcommand{\ottdruleevalXXgfullXXcompXXcasel}[1]{\ottdrule[#1]{%
\ottpremise{  \textcolor{\coeffectcolor}{ \textcolor{\coeffectcolor}{\gamma}_{{\mathrm{1}}} }\! \cdot \! \rho   \vdash_{\mathit{full} }  \ottnt{V}  \Downarrow  \ottkw{inl} \, \ottnt{W} }%
\ottpremise{   \textcolor{\coeffectcolor}{ \textcolor{\coeffectcolor}{\gamma}_{{\mathrm{2}}} }\! \cdot \! \rho    \mathop{,} \;   \ottmv{x_{{\mathrm{1}}}}   \mapsto ^{\textcolor{\coeffectcolor}{ \textcolor{\coeffectcolor}{q} } }  \ottnt{W}    \vdash_{\mathit{gen} }  \ottnt{M_{{\mathrm{1}}}}  \Downarrow  \ottnt{T}  \mathop{\#} \textcolor{\effectcolor}{ \textcolor{\effectcolor}{\phi} } }%
\ottpremise{ \textcolor{\coeffectcolor}{ \textcolor{\coeffectcolor}{\gamma} } \equiv \textcolor{\coeffectcolor}{  \textcolor{\coeffectcolor}{  \textcolor{\coeffectcolor}{ \textcolor{\coeffectcolor}{q} \cdot \textcolor{\coeffectcolor}{\gamma}_{{\mathrm{1}}} }  \ottsym{+} \textcolor{\coeffectcolor}{\gamma}_{{\mathrm{2}}} }  } }%
\ottpremise{ \textcolor{\coeffectcolor}{ \textcolor{\coeffectcolor}{q} }\;  \textcolor{\coeffectcolor}{\mathop{\leq_{\mathit{co} } } } \; \textcolor{\coeffectcolor}{  \textcolor{\coeffectcolor}{1}  } }%
}{
  \textcolor{\coeffectcolor}{ \textcolor{\coeffectcolor}{\gamma} }\! \cdot \! \rho   \vdash_{\mathit{gen} }   \ottkw{case}_{\textcolor{\coeffectcolor}{ \textcolor{\coeffectcolor}{q} } }\;  \ottnt{V} \; \ottkw{of}\; \ottkw{inl} \; \ottmv{x_{{\mathrm{1}}}}  \rightarrow\;  \ottnt{M_{{\mathrm{1}}}}  ;  \ottkw{inr} \; \ottmv{x_{{\mathrm{2}}}}  \rightarrow\;  \ottnt{M_{{\mathrm{2}}}}   \Downarrow  \ottnt{T}  \mathop{\#} \textcolor{\effectcolor}{ \textcolor{\effectcolor}{\phi} } }{%
{\ottdrulename{eval\_gfull\_comp\_casel}}{}%
}}

\newcommand{\ottdruleevalXXgfullXXcompXXcaser}[1]{\ottdrule[#1]{%
\ottpremise{  \textcolor{\coeffectcolor}{ \textcolor{\coeffectcolor}{\gamma}_{{\mathrm{1}}} }\! \cdot \! \rho   \vdash_{\mathit{full} }  \ottnt{V}  \Downarrow  \ottkw{inr} \, \ottnt{W} }%
\ottpremise{   \textcolor{\coeffectcolor}{ \textcolor{\coeffectcolor}{\gamma}_{{\mathrm{2}}} }\! \cdot \! \rho    \mathop{,} \;   \ottmv{x_{{\mathrm{2}}}}   \mapsto ^{\textcolor{\coeffectcolor}{ \textcolor{\coeffectcolor}{q} } }  \ottnt{W}    \vdash_{\mathit{gen} }  \ottnt{M_{{\mathrm{2}}}}  \Downarrow  \ottnt{T}  \mathop{\#} \textcolor{\effectcolor}{ \textcolor{\effectcolor}{\phi} } }%
\ottpremise{ \textcolor{\coeffectcolor}{ \textcolor{\coeffectcolor}{\gamma} } \equiv \textcolor{\coeffectcolor}{  \textcolor{\coeffectcolor}{  \textcolor{\coeffectcolor}{ \textcolor{\coeffectcolor}{q} \cdot \textcolor{\coeffectcolor}{\gamma}_{{\mathrm{1}}} }  \ottsym{+} \textcolor{\coeffectcolor}{\gamma}_{{\mathrm{2}}} }  } }%
\ottpremise{ \textcolor{\coeffectcolor}{ \textcolor{\coeffectcolor}{q} }\;  \textcolor{\coeffectcolor}{\mathop{\leq_{\mathit{co} } } } \; \textcolor{\coeffectcolor}{  \textcolor{\coeffectcolor}{1}  } }%
}{
  \textcolor{\coeffectcolor}{ \textcolor{\coeffectcolor}{\gamma} }\! \cdot \! \rho   \vdash_{\mathit{gen} }   \ottkw{case}_{\textcolor{\coeffectcolor}{ \textcolor{\coeffectcolor}{q} } }\;  \ottnt{V} \; \ottkw{of}\; \ottkw{inl} \; \ottmv{x_{{\mathrm{1}}}}  \rightarrow\;  \ottnt{M_{{\mathrm{1}}}}  ;  \ottkw{inr} \; \ottmv{x_{{\mathrm{2}}}}  \rightarrow\;  \ottnt{M_{{\mathrm{2}}}}   \Downarrow  \ottnt{T}  \mathop{\#} \textcolor{\effectcolor}{ \textcolor{\effectcolor}{\phi} } }{%
{\ottdrulename{eval\_gfull\_comp\_caser}}{}%
}}

\newcommand{\ottdruleevalXXgfullXXcompXXtick}[1]{\ottdrule[#1]{%
}{
  \textcolor{\coeffectcolor}{  \textcolor{\coeffectcolor}{\overline{0} }  }\! \cdot \! \rho   \vdash_{\mathit{gen} }   \textcolor{\effectcolor}{ \ottkw{tick} }   \Downarrow   \ottkw{return} _{\textcolor{\coeffectcolor}{  \textcolor{\coeffectcolor}{1}  } }  \ottsym{()}   \mathop{\#} \textcolor{\effectcolor}{  \textcolor{\effectcolor}{\ottkw{Tick} }  } }{%
{\ottdrulename{eval\_gfull\_comp\_tick}}{}%
}}

\newcommand{\ottdruleevalXXgfullXXcompXXcsub}[1]{\ottdrule[#1]{%
\ottpremise{  \textcolor{\coeffectcolor}{ \textcolor{\coeffectcolor}{\gamma}' }\! \cdot \! \rho   \vdash_{\mathit{gen} }  \ottnt{M}  \Downarrow  \ottnt{T}  \mathop{\#} \textcolor{\effectcolor}{ \textcolor{\effectcolor}{\phi} } }%
\ottpremise{ \textcolor{\coeffectcolor}{ \textcolor{\coeffectcolor}{\gamma} }\; \textcolor{\coeffectcolor}{\mathop{\leq_{\mathit{co} } } } \; \textcolor{\coeffectcolor}{ \textcolor{\coeffectcolor}{\gamma}' } }%
}{
  \textcolor{\coeffectcolor}{ \textcolor{\coeffectcolor}{\gamma} }\! \cdot \! \rho   \vdash_{\mathit{gen} }  \ottnt{M}  \Downarrow  \ottnt{T}  \mathop{\#} \textcolor{\effectcolor}{ \textcolor{\effectcolor}{\phi} } }{%
{\ottdrulename{eval\_gfull\_comp\_csub}}{}%
}}

\newcommand{\ottdruleevalXXgfullXXcompXXletinXXzero}[1]{\ottdrule[#1]{%
\ottpremise{  \textcolor{\coeffectcolor}{ \textcolor{\coeffectcolor}{\gamma}_{{\mathrm{1}}} }\! \cdot \! \rho   \vdash_{\mathit{gen} }  \ottnt{M}  \Downarrow   \ottkw{return} _{\textcolor{\coeffectcolor}{ \textcolor{\coeffectcolor}{q}_{{\mathrm{1}}} } }  \ottnt{W}   \mathop{\#} \textcolor{\effectcolor}{ \textcolor{\effectcolor}{\phi}_{{\mathrm{1}}} } }%
\ottpremise{   \textcolor{\coeffectcolor}{ \textcolor{\coeffectcolor}{\gamma}_{{\mathrm{2}}} }\! \cdot \! \rho    \mathop{,} \;   \ottmv{x}   \mapsto ^{\textcolor{\coeffectcolor}{ \textcolor{\coeffectcolor}{q}_{{\mathrm{1}}} } }  \ottnt{W}    \vdash_{\mathit{gen} }  \ottnt{N}  \Downarrow  \ottnt{T}  \mathop{\#} \textcolor{\effectcolor}{ \textcolor{\effectcolor}{\phi}_{{\mathrm{2}}} } }%
}{
  \textcolor{\coeffectcolor}{  \textcolor{\coeffectcolor}{ \textcolor{\coeffectcolor}{\gamma}_{{\mathrm{1}}} \ottsym{+} \textcolor{\coeffectcolor}{\gamma}_{{\mathrm{2}}} }  }\! \cdot \! \rho   \vdash_{\mathit{gen} }   \ottmv{x}  \leftarrow^{\textcolor{\coeffectcolor}{0} }_{\textcolor{\effectcolor}{\varepsilon} }  \ottnt{M} \ \ottkw{in}\  \ottnt{N}   \Downarrow  \ottnt{T}  \mathop{\#} \textcolor{\effectcolor}{  \textcolor{\effectcolor}{ \textcolor{\effectcolor}{\phi}_{{\mathrm{1}}}  \cdot  \textcolor{\effectcolor}{\phi}_{{\mathrm{2}}} }  } }{%
{\ottdrulename{eval\_gfull\_comp\_letin\_zero}}{}%
}}

\newcommand{\ottdruleevalXXgfullXXcompXXsequence}[1]{\ottdrule[#1]{%
\ottpremise{  \textcolor{\coeffectcolor}{ \textcolor{\coeffectcolor}{\gamma}_{{\mathrm{1}}} }\! \cdot \! \rho   \vdash_{\mathit{full} }  \ottnt{V}  \Downarrow  \ottsym{()} }%
\ottpremise{  \textcolor{\coeffectcolor}{ \textcolor{\coeffectcolor}{\gamma}_{{\mathrm{2}}} }\! \cdot \! \rho   \vdash_{\mathit{gen} }  \ottnt{N}  \Downarrow  \ottnt{T}  \mathop{\#} \textcolor{\effectcolor}{ \textcolor{\effectcolor}{\phi} } }%
\ottpremise{ \textcolor{\coeffectcolor}{ \textcolor{\coeffectcolor}{\gamma} } \equiv \textcolor{\coeffectcolor}{  \textcolor{\coeffectcolor}{ \textcolor{\coeffectcolor}{\gamma}_{{\mathrm{1}}} \ottsym{+} \textcolor{\coeffectcolor}{\gamma}_{{\mathrm{2}}} }  } }%
}{
  \textcolor{\coeffectcolor}{ \textcolor{\coeffectcolor}{\gamma} }\! \cdot \! \rho   \vdash_{\mathit{gen} }   \ottnt{V}  ;  \ottnt{N}   \Downarrow  \ottnt{T}  \mathop{\#} \textcolor{\effectcolor}{ \textcolor{\effectcolor}{\phi} } }{%
{\ottdrulename{eval\_gfull\_comp\_sequence}}{}%
}}

\newcommand{\ottdefnevalXXgfullXXcomp}[1]{\begin{ottdefnblock}[#1]{$ \mu  \vdash_{\mathit{gen} }  \ottnt{M}  \Downarrow  \ottnt{T}  \mathop{\#} \textcolor{\effectcolor}{ \textcolor{\effectcolor}{\phi} } $}{\ottcom{environment-based effect-coeffect semantics for CBPV (large-step)}}
\ottusedrule{\ottdruleevalXXgfullXXcompXXabs{}}
\ottusedrule{\ottdruleevalXXgfullXXcompXXcpair{}}
\ottusedrule{\ottdruleevalXXgfullXXcompXXapp{}}
\ottusedrule{\ottdruleevalXXgfullXXcompXXsplit{}}
\ottusedrule{\ottdruleevalXXgfullXXcompXXret{}}
\ottusedrule{\ottdruleevalXXgfullXXcompXXletin{}}
\ottusedrule{\ottdruleevalXXgfullXXcompXXforce{}}
\ottusedrule{\ottdruleevalXXgfullXXcompXXfst{}}
\ottusedrule{\ottdruleevalXXgfullXXcompXXsnd{}}
\ottusedrule{\ottdruleevalXXgfullXXcompXXseq{}}
\ottusedrule{\ottdruleevalXXgfullXXcompXXcasel{}}
\ottusedrule{\ottdruleevalXXgfullXXcompXXcaser{}}
\ottusedrule{\ottdruleevalXXgfullXXcompXXtick{}}
\ottusedrule{\ottdruleevalXXgfullXXcompXXcsub{}}
\ottusedrule{\ottdruleevalXXgfullXXcompXXletinXXzero{}}
\ottusedrule{\ottdruleevalXXgfullXXcompXXsequence{}}
\end{ottdefnblock}}

\newcommand{\ottdefnsJGFullOp}{
\ottdefnevalXXgfullXXcomp{}}

\newcommand{\ottdefnss}{
\ottdefnsJSmallStepCBV
\ottdefnsJSmallStepCBPV
\ottdefnsJEnv
\ottdefnsJEffEnv
\ottdefnsJInstrEnv
\ottdefnsJSTLC
\ottdefnsJCBPV
\ottdefnsJEff
\ottdefnsJMonEff
\ottdefnsJContextOps
\ottdefnsJCoeff
\ottdefnsJLin
\ottdefnsJCBVFull
\ottdefnsJCBNFull
\ottdefnsJLambdaCBNCoeff
\ottdefnsJLambdaCBVCoeff
\ottdefnsJComLin
\ottdefnsJEffCoeffCBPV
\ottdefnsJFullOp
\ottdefnsJGFullOp
}

\newcommand{\ottall}{\ottmetavars\\[0pt]
\ottgrammar\\[5.0mm]
\ottdefnss}

  \renewottcommands[ott]


\copyrightyear{2024}

\begin{abstract}
  Effect and coeffect tracking integrate many types of compile-time
  analysis, such as cost, liveness, or dataflow, directly into a language's type
  system. In this paper, we investigate the addition of effect and coeffect
  tracking to the type system of call-by-push-value (CBPV), a computational
  model useful in compilation for its isolation of effects and for its ability
  to cleanly express both call-by-name and call-by-value computations.  Our
  main result is \emph{effect-and-coeffect soundness}, which asserts that the
  type system accurately bounds the effects that the program may trigger
  during execution and accurately tracks the demands that the program may make
  on its environment. This result holds for two different dynamic semantics: a
  generic one that can be adapted for different coeffects and one that is
  adapted for reasoning about resource usage. In particular, the second 
  semantics discards the evaluation of unused values and pure computations
  while ensuring that effectful computations are always evaluated, even if 
  their results are not required. Our results have been mechanized using 
  the Coq proof assistant.
\end{abstract}

\maketitle

\bibliographystyle{ACM-Reference-Format}
\citestyle{acmauthoryear}   

\ifextended
\textbf{This paper is an extended version of \citet{torczon:cbpv}. }
\fi

\section{Introduction}
\label{sec:introduction}

Computations interact with the world in which they run. Sometimes the
computation does something the world can observe, known as
an \emph{effect}~\cite{lucassen-gifford}, and sometimes computations demand
something that the world must provide, known as a
\emph{coeffect}~\cite{petricek:2014,Brunel:2014,granule-project}.
For example, running a computation might take time (a clock ticking is an
effect) and might require resources (using input parameters is a coeffect).

Some programming languages track effects and coeffects
statically.  Frank~\cite{mcbride:frank},
Koka~\cite{leijen2014koka}, and the Verse functional logic
language~\cite{verse-manual} do this for effects such as state, exceptions,
divergence, and failure; Linear Haskell~\cite{linear-haskell} does this for a
resource management coeffect, while Agda and Idris 2~\cite{brady:idris2} do
this for a relevancy coeffect. The Effekt
language~\cite{brachthauser2022effects} both tracks effects statically and
uses a limited form of coeffect tracking to ensure that effect handlers are
well-scoped. Finally, the Granule language~\cite{orchard:2019} uses monads and
comonads graded by abstract structures to track various effects and
coeffects in a flexible and expressive system.

We would like to update the type systems of existing languages with effect and
coeffect tracking by annotating their existing type systems. However, in
contrast to systems that use monads and comonads to isolate effectful and
coeffectful code from the rest of the language, we need an approach that is
descriptive and that does not restrict programmers in how they structure their
code. 

Because effectful computation depends on evaluation order, precisely tracking
effects works best in a language that makes its ``ambient monad'' explicit,
such as Moggi's computational lambda calculus~\cite{moggi:computational} and
fine-grained CBV~\cite{levy:fine_grained}. These systems separate inert
``values'' from executable ``computations'' and include ``return'' and
``let'' constructs to sequence evaluation. This ``ambient monad'' is part of the 
structure of the language itself; all computations are monadic. 

Levy's Call-By-Push-Value (CBPV)~\cite{levy:call-by-push-value} is a calculus
that makes both the ambient computational monad \emph{and comonad}
explicit. As above, it separates values from computations and uses ``return''
and ``let'' constructs to track how computations manipulate values. However,
CBPV also includes thunks, which temporarily suspend computations and treat
them as values, for the opposite purpose; as a result all computations are
also comonadic. In CBPV, then, we can annotate these existing structures
directly to track effects and coeffects, instead of adding new features to the language.

CBPV is a low-level language and is appropriate for use as a compiler
intermediate
representation~\cite{garbuzov2018structural,rizkallah:cbpv}.
Its distinction between values and computations allows CBPV to work with
strict and nonstrict language features explicitly, enabling it to model both
call-by-value and call-by-name languages with the same facility.  Adding
effects and coeffects to CBPV would enrich this intermediate representation to
support program optimizations; for example, to justify dead code elimination
for pure code whose coeffect annotations mark it as unused.

The ability of CBPV to model both CBV and CBN also lets us observe how
evaluation order changes the way a program alters and makes demands on the
world.  Levy characterizes the difference between values and computations with
the slogan: ``a value \emph{is}, a computation \emph{does}.''~\cite{levy:call-by-push-value}  Our
interpretation of this slogan is that only computations may contain effectful
subcomponents---values must be pure throughout. Conversely, coeffects describe
the demands a program makes on its inputs, which are always values in CBPV.

CBPV uses separate types for values and computations.  Values have
\emph{positive} types (for which we use the metavariable $\ottnt{A}$), while
computations have \emph{negative} types (for which we use $\ottnt{B}$).  These two
forms are connected via an adjunction: the thunk type $\ottkw{U} \, \ottnt{B}$ suspends a
computation as an inert value, and the type of return $\ottkw{F} \, \ottnt{A}$ creates a
fine-grained structure similar to monadic bind that threads values through
computations.  Due to the structure of the adjunction, the combination
$\ottkw{U} \, \ottsym{(}  \ottkw{F} \, \ottnt{A}  \ottsym{)}$ forms a monad and the combination $\ottkw{F} \, \ottsym{(}  \ottkw{U} \, \ottnt{B}  \ottsym{)}$ forms a comonad
\cite{levy2003adjunction}.

The duality between values and computations gives CBPV its power,
and it is reflected in the structures we use
to statically track effects and coeffects. For
effects, we add effect information $\textcolor{\effectcolor}{\phi}$ to the thunk type $ \ottkw{U}_{\color{\effectcolor}{ \textcolor{\effectcolor}{\phi} } }\;  \ottnt{B} $,
recording the latent effect of suspended computations. Similarly, to track
coeffects, we add coeffect information $\textcolor{\coeffectcolor}{q}$ to the returner type
$ \ottkw{F}_{\color{\coeffectcolor}{ \textcolor{\coeffectcolor}{q} } }\;  \ottnt{A} $, describing the demands subsequent
computation is allowed to make on the returned value.
With this augmentation, we will show that the types $ \ottkw{U}_{\color{\effectcolor}{ \textcolor{\effectcolor}{\phi} } }\;  \ottsym{(}  \ottkw{F} \, \ottnt{A}  \ottsym{)} $ and
$ \ottkw{F}_{\color{\coeffectcolor}{ \textcolor{\coeffectcolor}{q} } }\;  \ottsym{(}  \ottkw{U} \, \ottnt{B}  \ottsym{)} $ can encode the graded monads and comonads associated with
effect and coeffect tracking.

Following this duality, this paper begins with two mirrored halves and then
combines them. The first part (Section~\ref{sec:effects}) extends CBPV with
effect tracking and shows how we can recover the graded monad by grading the
thunk type with latent effects. The second part (Section~\ref{sec:coeffects})
extends CBPV with coeffect tracking and recovers a graded comonad by grading
the returner type with latent coeffects; we also discuss modifications to the
system for resource tracking with coeffects (Section~\ref{sec:resource-usage}).  Finally, we combine the two
systems and explore their interaction (Section~\ref{sec:full}). This paper is best read in color:
effects $\textcolor{\effectcolor}{\phi}$ appear in red and coeffects $\textcolor{\coeffectcolor}{q}$ in blue. Without these
colorful annotations, the type system and semantics are the standard rules of
CBPV.

Along the way, we prove the following results about our extensions.
\begin{itemize}
\item We prove \emph{effect soundness} for our effect-annotated extension of CBPV,
  demonstrating that the type-and-effect system accurately bounds what happens at runtime.
  To do so, we define an environment-based big-step operational semantics for
  CBPV instrumented to precisely track effects during evaluation, and we use a
  logical relation to prove our soundness theorem. (Section~\ref{sec:effect-soundness})
\item We prove that the standard translations from call-by-value (CBV) and
  call-by-name (CBN) lambda calculi to CBPV are \emph{type-and-effect
    preserving}. Starting with a well-typed CBV or (monadic) CBN program, we can produce a
  well-typed CBPV program with the same effects as the source program.
  (Section~\ref{sec:effect-translations})
\item We prove \emph{coeffect soundness} for a coeffect-annotated extension
  of CBPV, demonstrating that the type-and-coeffect system accurately tracks the demands
  a program may make on its inputs. We do so using an
  environment-based big-step operational semantics for CBPV, where
  the environment has been instrumented to track coeffects during evaluation.
  (Section~\ref{sec:coeffect-soundness})
\item We observe that our generic coeffect-tracking operational semantics behavior has 
  strange implications when reasoning about resource usage.
  Therefore, we adapt the rules of our
  semantics so that it does not demand resources for discarded values, providing
  a better model of how the program uses its inputs in this coeffect.
  (Section~\ref{sec:resource-usage})
\item We prove that the standard translations from both CBN and CBV to CBPV
  are \emph{type-and-coeffect} preserving for this updated coeffect system. Starting with a well-typed CBN or
  CBV program, we can produce a well-typed CBPV program with the same coeffects.
  (Section~\ref{sec:coeffect-translation})
\item We combine the `tick' effect and resource tracking coeffect together into the same
  CBPV type system and prove combined versions of the
  results from each: \emph{type-and-effect-and-coeffect soundness}
  and \emph{type-and-effect-and-coeffect preservation} of the standard translations
  from CBV and CBN. We extend this system with a new rule that does not demand 
  resources for unused \emph{computations}, when they are effect-free.
  Finally, we prove that our discarding semantics produces the same result and has the 
  same effects as our general semantics, justifying the soundness of our resource 
  accounting semantics. (Section~\ref{sec:full})
\end{itemize}

We are not the first to extend CBPV with effect tracking and our type system 
is most similar to \citet{kammar-plotkin} and \citet{forster:expressive-user-effects}. 
However, all other definitions and results of this paper are novel. In
particular, we have found little work that explores the interaction between
CBPV and coeffects.  Furthermore, while we are able use the standard
translations to interpret CBV and CBN in CBPV, designing the effect and
coeffect systems so that these translations ``just work'' is a contribution of
this paper. Our approach to effect-and-coeffect soundness also differs from
prior work---we employ a novel environment-based big-step semantics for CBPV
that leads to short and straightforward proofs.

For simplicity, the effect systems in this paper only track clock effects, and
the coeffect systems only count variable usages. As a result, we do not
explore more sophisticated interactions between other forms of effects and
coeffects, such as local and global state~\cite{nanevski:dynamic-binding}, or
between information flow and nondeterminism, or between usage analysis and
errors~\cite{Gaboardi:2016}.

The results of this paper have been formalized in Coq and are
available online\footnote{\pfloc} and archived on
Zenodo~\cite{zenodo:artifact}. \ifextended\else For space, some parts
  of our mechanization have been elided from this paper, but full details are
  available in an extended version~\cite{extended-version}. \fi

\section{Call-by-push-value (CBPV) and effect tracking}
\label{sec:effects}

In this section, we extend the type system of CBPV with effect tracking. Our
modifications to the base system, which are limited to reasoning about
effect annotations $\textcolor{red}{\textcolor{\effectcolor}{\phi}}$, are marked in red.

CBPV syntactically separates terms into \emph{values} $\ottnt{V}$\!, inhabiting positive
types $\ottnt{A}$, and \emph{computations} $\ottnt{M}$, inhabiting negative types $\ottnt{B}$,
as shown by the following \link{effects/CBPV/syntax.v}{ValTy,CompTy,Val,Comp}{grammar.}

\[
\begin{array}{llcl}
\textit{value types}        & \ottnt{A} &::=&  \texttt{unit} \ |\  \ottkw{U}_{\color{\effectcolor}{ \textcolor{\effectcolor}{\phi} } }\;  \ottnt{B} \ |\  \ottnt{A_{{\mathrm{1}}}} \times \ottnt{A_{{\mathrm{2}}}} \ \ifextended|\ \ottnt{A_{{\mathrm{1}}}}  \ottsym{+}  \ottnt{A_{{\mathrm{2}}}} \fi
\\
\textit{computation types}  & \ottnt{B} &::=& \ottnt{A}  \to  \ottnt{B}\ |\ \ottkw{F} \, \ottnt{A}\ |\  \ottnt{B_{{\mathrm{1}}}}   \mathop{\&}   \ottnt{B_{{\mathrm{2}}}}  \\
\\
\textit{values}       & \ottnt{V} &::=& \ottmv{x}\ |\ \ottsym{\{}  \ottnt{M}  \ottsym{\}}\ \ottsym{()}\ |\ 
 |\ \ottsym{(}  \ottnt{V_{{\mathrm{1}}}}  \ottsym{,}  \ottnt{V_{{\mathrm{2}}}}  \ottsym{)}\  \ifextended |\ \ottkw{inl} \, \ottnt{V}\ |\ \ottkw{inr} \, \ottnt{V} \fi \\
\textit{computations} & \ottnt{M} &::=&
 \lambda  \ottmv{x} . \ottnt{M} \ |\ \ottnt{M} \, \ottnt{V}\ |\ \ottnt{V}  \ottsym{!}\ |\  \ottkw{let}\; ( \ottmv{x_{{\mathrm{1}}}} ,  \ottmv{x_{{\mathrm{2}}}} ) =  \ottnt{V} \; \ottkw{in}\;  \ottnt{N} \  \\
&&|&  \langle  \ottnt{M_{{\mathrm{1}}}} , \ottnt{M_{{\mathrm{2}}}}  \rangle \ |\  \ottnt{M}  . 1 \ |\  \ottnt{M}  . 2 \ |\ \ottkw{return} \, \ottnt{V}\ |\ \ottmv{x}  \leftarrow  \ottnt{M} \, \ottkw{in} \, \ottnt{N}\ |\  \textcolor{\effectcolor}{ \ottkw{tick} }   \\
\ifextended &&|&  \ottnt{V}  ;  \ottnt{M} \ |\ \ottkw{case} \, \ottnt{V} \, \ottkw{of} \, \ottkw{inl} \, \ottmv{x_{{\mathrm{1}}}}  \to  \ottnt{M_{{\mathrm{1}}}}  \mathsf{;} \, \ottkw{inr} \, \ottmv{x_{{\mathrm{2}}}}  \to  \ottnt{M_{{\mathrm{2}}}}\ \\ \fi
\end{array}
\]

Values in CBPV mostly correspond to the values found in a call-by-value
typed functional language, \ifextended 
such as unit and positive products and sums of values.
\else
such as unit and positive products of values.
\fi
CBPV values also include suspended computations, called
\emph{thunks} and written $\ottsym{\{}  \ottnt{M}  \ottsym{\}}$.
(Variables always represent values, so they are always declared with value
types in the context.)

Computations in CBPV include abstractions ($ \lambda  \ottmv{x} . \ottnt{M} $), applications
($\ottnt{M} \, \ottnt{V}$), elimination (\emph{forcing}) of thunks ($\ottnt{V}  \ottsym{!}$),
\ifextended
unit elimination ($ \ottnt{V}  ;  \ottnt{M} $), 
positive product elimination ($ \ottkw{let}\; ( \ottmv{x_{{\mathrm{1}}}} ,  \ottmv{x_{{\mathrm{2}}}} ) =  \ottnt{V} \; \ottkw{in}\;  \ottnt{N} $), 
and 
positive sum elimination ($\ottkw{case} \, \ottnt{V} \, \ottkw{of} \, \ottkw{inl} \, \ottmv{x_{{\mathrm{1}}}}  \to  \ottnt{M_{{\mathrm{1}}}}  \mathsf{;} \, \ottkw{inr} \, \ottmv{x_{{\mathrm{2}}}}  \to  \ottnt{M_{{\mathrm{2}}}}$).
\else
and positive product elimination ($ \ottkw{let}\; ( \ottmv{x_{{\mathrm{1}}}} ,  \ottmv{x_{{\mathrm{2}}}} ) =  \ottnt{V} \; \ottkw{in}\;  \ottnt{N} $).
\fi
In addition to positive products,
CBPV also includes negative products, of type $ \ottnt{B_{{\mathrm{1}}}}   \mathop{\&}   \ottnt{B_{{\mathrm{2}}}} $. These are
introduced with a pair of computations $ \langle  \ottnt{M_{{\mathrm{1}}}} , \ottnt{M_{{\mathrm{2}}}}  \rangle $ and eliminated by
projecting either the first or second component, \emph{i.e.} $ \ottnt{M}  . 1 $ or
$ \ottnt{M}  . 2 $.

Values can be threaded through computations. The $\ottkw{return} \, \ottnt{V}$ form
injects a value into a trivial computation. In the 
``letin'' construct, written $\ottmv{x}  \leftarrow  \ottnt{M} \, \ottkw{in} \, \ottnt{N}$,
the first subcomputation must evaluate to the form $\ottkw{return} \, \ottnt{V}$, and
the second computation can then reference $\ottnt{V}$.
An advantage of CBPV is that this bind-like
method of threading values through computations makes it readily extensible with effectful
language features. \citet{levy:call-by-push-value,
  levy2006,levy:siglog-tutorial} demonstrates how to add
nontermination, nondeterminism, errors, I/O, state, and control effects to CBPV.
In each case, Levy extends the language with new computations and modifies the
operational semantics to account for the new features.

For simplicity, we describe a single effect in this paper, the $ \textcolor{\effectcolor}{ \ottkw{tick} } $
computation.  This effect advances a virtual clock in the operational
semantics, simulating the cost of the program.

\subsection{CBPV: Type-and-effect System}

\begin{figure}
  \drules[eff]{$\Gamma  \vdash_{\mathit{eff} }  \ottnt{V}  \ottsym{:}  \ottnt{A}$}{value effect typing}{var,thunk,unit,pair}
\ifextended \[ \drule{eff-inl} \quad \drule{eff-inr} \] \fi
  \drules[eff]{$ \Gamma   \vdash_{\mathit{eff} }   \ottnt{M} \; :^{\textcolor{\effectcolor}{  \textcolor{\effectcolor}{\phi}  } }\;  \ottnt{B} $}{computation effect
    typing}{abs,app,force,ret,letin,split}

\[ \drule[width=3in]{eff-cpair} \drule{eff-fst}\ \drule{eff-snd} \]

\ifextended \[ \drule{eff-sequence} \qquad \drule[width=4in]{eff-case} \] \fi
\[ \drule{eff-tick} \qquad \drule[width=3in]{eff-sub} \]
\caption{CBPV typing and effect tracking}
\label{fig:cbpv-typing}
\label{fig:cbpv-effects}
\Description{The typing rules for CBPV augmented with effect tracking}
\end{figure}

Our \link{effects/CBPV/typing.v}{VWt,CWt}{type-and-effect system} for CBPV is
shown in Figure~\ref{fig:cbpv-effects}.  Under some typing context $\Gamma$,
this system assigns a value type to values ($\Gamma  \vdash_{\mathit{eff} }  \ottnt{V}  \ottsym{:}  \ottnt{A}$) and both a
computation type and effect to computations ($ \Gamma   \vdash_{\mathit{eff} }   \ottnt{M} \; :^{\textcolor{\effectcolor}{  \textcolor{\effectcolor}{\phi}  } }\;  \ottnt{B} $), where
$\textcolor{\effectcolor}{\phi}$ is an upper bound on the effects that could occur during the evaluation
of $\ottnt{M}$. The judgement for values does not need an effect annotation
because values are pure. In \rref{eff-thunk}, the thunk type $ \ottkw{U}_{\color{\effectcolor}{ \textcolor{\effectcolor}{\phi} } }\;  \ottnt{B} $
  records the latent effect of a suspended computation.

Following \citet{katsumata:2014}, our system models effects using an arbitrary
\emph{preordered monoid}. This gives us an identity element $\textcolor{red}{\varepsilon}$,
an associative combining operation $ \textcolor{\effectcolor}{ \textcolor{\effectcolor}{\phi}_{{\mathrm{1}}}  \cdot  \textcolor{\effectcolor}{\phi}_{{\mathrm{2}}} } $, and a preorder relation $ \textcolor{\effectcolor}{\mathop{\leq_{\mathit{eff} } } } $ that
respects the operation. We also include a primitive effect $ \textcolor{\effectcolor}{\ottkw{Tick} } $ produced by the $ \textcolor{\effectcolor}{ \ottkw{tick} } $
computation. However, the only parts of the system that are specific to this effect are the rules
for $ \textcolor{\effectcolor}{ \ottkw{tick} } $, which is our \emph{only} effectful computation. All other rules are presented
generically and are adaptable to other effects and effectful computations (e.g. a $ \textcolor{\effectcolor}{\ottkw{Read} } $ effect
produced by a $ \textcolor{\effectcolor}{ \ottkw{read} } $ computation).

Concretely, we could use the natural number monoid with the usual ordering, 0 as the
identity element $ \textcolor{\effectcolor}{\varepsilon} $, and addition as the combining operation to have our type system perform a cost
analysis. Using 1 as our model of the $ \textcolor{\effectcolor}{\ottkw{Tick} } $ effect, the system would statically bound the number
of $ \textcolor{\effectcolor}{ \ottkw{tick} } $s that are evaluated. For example, the type system would tell us that the computation
$ \langle   \textcolor{\effectcolor}{ \ottkw{tick} }  , \ottmv{y}  \leftarrow   \textcolor{\effectcolor}{ \ottkw{tick} }  \, \ottkw{in} \,  \textcolor{\effectcolor}{ \ottkw{tick} }   \rangle $ advances the clock at most twice. If the first component of
the pair is projected, the type system overapproximates the effect produced during execution.
Note that to track other behaviors with our type system, we need only change our preordered monoid
accordingly (e.g. we could track possible effects with the power set monoid ordered by set inclusion).

\Rref{eff-ret,eff-letin} motivate the choice of a monoid structure.
Returning a value has no effect, so the
effect of $\ottkw{return} \, \ottnt{V}$ should always be $\textcolor{red}{\varepsilon}$.
\Rref{eff-letin} must combine
effects because $\ottmv{x}  \leftarrow  \ottnt{M} \, \ottkw{in} \, \ottnt{N}$ is the only computation in our system
with two subcomputations,
both of which may be effectful. Finally, because return and letin
satisfy identity and associativity
properties as
the building blocks of the CBPV monad,
we need these same properties in our effect structure.

\Rref{eff-sub} allows for imprecision in the type system. That is, an effect
annotation $\textcolor{\effectcolor}{\phi}$ on the type of a program indicates that the program
will have \emph{at most} $\textcolor{\effectcolor}{\phi}$ as its effect; it may have less.  If the
type system determines that the computation will complete within 5 ticks, it
is also sound, but less precise, for it to say that it will complete within 7
ticks.  Choosing the discrete ordering (i.e. using equality for $ \textcolor{\effectcolor}{\mathop{\leq_{\mathit{eff} } } } $) forces the type system to
track effects precisely. Note that to allow the discrete ordering, we do not
assume $ \textcolor{\effectcolor}{  \textcolor{\effectcolor}{\varepsilon}  \;  \textcolor{\effectcolor}{\mathop{\leq_{\mathit{eff} } } } \;  \textcolor{\effectcolor}{\phi}  } $ from the effect structure. In other words, the 
type system does not need $ \textcolor{\effectcolor}{\varepsilon} $ to be the least effect, only an identity element for the combining operation.

This imprecision allows more programs to type check. In a program with
branching, different branches may have different effects.  For example, in
\rref{eff-cpair}, only one side of a computational pair will ever be
evaluated. However, for soundness, both computations must be typed with the same effect
(which may be an overapproximation due to subeffecting).

Unlike in effect systems for the $\lambda$-calculus, the latent effects of
function bodies are not recorded in function types. Instead, they are
propagated to the conclusion of \rref{eff-abs}. This makes sense because
abstractions are not values in CBPV. From an operational sense, they are
computations that pop the argument off the stack before continuing~\cite{levy:call-by-push-value}.

\subsection{Instrumented Operational Semantics and Effect Soundness}

\ifextended
\else
\begin{figure}
\drules[eval-val]{$ \rho  \vdash  \ottnt{V} \ \Downarrow\  \ottnt{W} $}{Value closing}{var,unit,thunk,vpair}
\caption{Operational semantics of CBPV with effect tracking}
\Description{Operational semantics of CBPV with effect tracking}
\label{fig:eval-val}
\end{figure}

\begin{figure}
\drules[eval-eff-comp]{$ \rho  \vdash_{\mathit{eff} }  \ottnt{M}  \Downarrow  \ottnt{T}  \mathop{\#} \textcolor{\effectcolor}{ \textcolor{\effectcolor}{\phi} } $}{Computation rules}{}
\[  \drule{eval-eff-comp-abs}\ \quad \drule[width=3in]{eval-eff-comp-app-abs}\  \]

\[  \drule[width=4in]{eval-eff-comp-force-thunk}\ \drule{eval-eff-comp-return}\   \]

\[  \drule[width=4in]{eval-eff-comp-letin-ret} \]

\[  \drule{eval-eff-comp-split}\ \quad \drule{eval-eff-comp-cpair} \]

\[  \drule{eval-eff-comp-fst}\  \quad \drule{eval-eff-comp-snd} \]

\caption{Operational semantics of CBPV with effect tracking}
\Description{Operational semantics of CBPV with effect tracking}
\label{fig:eff-big-step}
\end{figure}
\fi

We next define a big-step, \emph{environment-based} operational semantics for CBPV.
Here, an \link{effects/CBPV/semantics.v}{env}{environment,} $\rho$, is a mapping from variables to \link{effects/CBPV/semantics.v}{VClos}{\emph{closed values},}
$\ottnt{W}$, and can be thought of as a sequence of delayed substitutions.  Closed values
include closures, \emph{i.e.} suspended computations paired with closing
environments, as well as \ifextended unit, positive products and sums of
closed values.  \else unit and positive products. \fi

\[
\begin{array}{llcl}
\textit{environments}     & \rho & ::= & \emptyset\ |\  \rho   \mathop{,}   \ottmv{x}  \mapsto  \ottnt{W}  \\
\textit{closed values}    & \ottnt{W}   & ::= & \ottsym{()}\ |\  \mathbf{clo}( \rho , \{  \ottnt{M}  \} ) \ |\
 \ottsym{(}  \ottnt{W_{{\mathrm{1}}}}  \ottsym{,}  \ottnt{W_{{\mathrm{2}}}}  \ottsym{)}\ \ifextended |\ \ottkw{inl} \, \ottnt{W}\ |\ \ottkw{inr} \, \ottnt{W} \fi \\
\end{array}
\]

This semantics is new but straightforward. Past presentations of CBPV define its
operational behavior using small-step, big-step, or stack-based semantics, but
all the ones we have found use immediate substitution~\cite{levy:siglog-tutorial}.
We choose an
environment-based big-step semantics for two reasons. First, the big-step
structure corresponds closely to the structure of the type system; there is
only one rule of the operational semantics for each rule of the type
system. Together with the use of environments, this semantics eliminates the
need for substitution lemmas, leading to a remarkably straightforward
soundness proof (Section~\ref{sec:effect-soundness}).  Second, the environment
lets us track the demands that computations make on their inputs in our
coeffect soundness proof (Section~\ref{sec:coeffect-soundness}). For example,
with resource usage, we can include annotations in the environment
that count how many times
the program accesses each variable during computation, mirroring
the annotations in the context in the type system. A substitution-based
semantics does not support this instrumentation.

\ifextended Appendix~\ref{fig:eff-big-step} \else Figure~\ref{fig:eff-big-step} \fi shows the definition of the operational
semantics. This semantics consists of two relations. The \link{effects/CBPV/semantics.v}{EvalVal}{first relation,} written $ \rho  \vdash  \ottnt{V} \ \Downarrow\  \ottnt{W} $, uses the provided environment $\rho$ to ``evaluate'' a
value $\ottnt{V}$ to a closed value $\ottnt{W}$. This operation is essentially a
substitution operation in that it replaces each variable found in the value
with its definition in the environment.

The \link{effects/CBPV/semantics.v}{EvalClos}{second relation,} written $ \rho  \vdash_{\mathit{eff} }  \ottnt{M}  \Downarrow  \ottnt{T}  \mathop{\#} \textcolor{\effectcolor}{ \textcolor{\effectcolor}{\phi} } $, shows how computations
evaluate to \link{effects/CBPV/semantics.v}{CClos}{\emph{closed terminal computations},} $\ottnt{T}$. Closed terminals
are computations that cannot step any further, such as
returned (closed) values and suspended abstractions and pairs.
The effect annotation
$\textcolor{\effectcolor}{\phi}$ on this relation counts the number of ticks that occur during
evaluation of $\ottnt{M}$.
While suspended abstractions and pairs resemble closures, they are not
first class. Instead, they provide a convenient notation describing the
propagation of the environment during evaluation.

\[
\begin{array}{llcl}
\textit{closed terminals} & \ottnt{T}   & ::= &
   \ottkw{return} \, \ottnt{W}\ |\  \mathbf{clo}(  \rho ,   \lambda  \ottmv{x} . \ottnt{M}   ) \ |\  \mathbf{clo}(  \rho ,   \langle  \ottnt{M_{{\mathrm{1}}}} , \ottnt{M_{{\mathrm{2}}}}  \rangle   )    \\
\end{array}
\]

The operational semantics of the $ \textcolor{\effectcolor}{ \ottkw{tick} } $ computation is trivial---it merely
produces a unit value and a single $ \textcolor{\effectcolor}{\ottkw{Tick} } $ effect.  Other computations either
produce no effect (as in \rref{eval-eff-comp-abs}) or
combine the effects of their subcomponents (as in
\rref{eval-eff-comp-app-abs}).  As in the type-and-effect system, the only
rule that is specific to the $ \textcolor{\effectcolor}{\ottkw{Tick} } $ effect is the rule for $ \textcolor{\effectcolor}{ \ottkw{tick} } $.
All other effects in these rules are parameterized over the input monoid.

While the type system allows for imprecision, the operational semantics
precisely tracks the effects of computation.

\subsection{Type and Effect Soundness}
\label{sec:effect-soundness}

We state our effect
soundness theorem as follows: closed, well-typed computations of type $\ottkw{F} \, \ottnt{A}$
return closed values and produce effects that are bounded by the type system.

\begin{theorem}[\link{effects/CBPV/soundness.v}{soundness}{Effect soundness}]
\label{thm:effect-soundness}\ 
If $ \varnothing   \vdash_{\mathit{eff} }   \ottnt{M} \; :^{\textcolor{\effectcolor}{  \textcolor{\effectcolor}{\phi}  } }\;  \ottkw{F} \, \ottnt{A} $ then $ \emptyset  \vdash_{\mathit{eff} }  \ottnt{M}  \Downarrow  \ottkw{return} \, \ottnt{W}  \mathop{\#} \textcolor{\effectcolor}{ \textcolor{\effectcolor}{\phi}_{{\mathrm{1}}} } $ where
  $ \textcolor{\effectcolor}{ \textcolor{\effectcolor}{\phi}_{{\mathrm{1}}} \;  \textcolor{\effectcolor}{\mathop{\leq_{\mathit{eff} } } } \;  \textcolor{\effectcolor}{\phi}  } $.
\end{theorem}
\ifextended The reason that this theorem is limited to type $\ottkw{F} \, \ottnt{A}$ is
  because we do not assume $ \textcolor{\effectcolor}{  \textcolor{\effectcolor}{\varepsilon}  \;  \textcolor{\effectcolor}{\mathop{\leq_{\mathit{eff} } } } \;  \textcolor{\effectcolor}{\phi}  } $ in the preordered monoid. As a
  result, at other types, the soundness theorem is more complex. For a general
  type $\ottnt{B}$, we know that if $\ottnt{M}$ evaluates to some terminal $\ottnt{T}$
  with effect $\textcolor{\effectcolor}{\phi}'$, then there is some $\textcolor{\effectcolor}{\phi}''$ where
  $ \textcolor{\effectcolor}{  \textcolor{\effectcolor}{ \textcolor{\effectcolor}{\phi}'  \cdot  \textcolor{\effectcolor}{\phi}'' }  \;  \textcolor{\effectcolor}{\mathop{\leq_{\mathit{eff} } } } \;  \textcolor{\effectcolor}{\phi}  } $. This extra $\textcolor{\effectcolor}{\phi}''$ is the latent effect
  from the case where $\ottnt{T}$ is a closure.  \fi

The proof is simple and based on the following
logical relation, consisting of three functions defined
mutually over the structure of types: closed
values $ \mathcal{W}\llbracket  \ottnt{A}  \rrbracket $, closed terminal computations $ \mathcal{T}\llbracket  \ottnt{B} \rrbracket^{\textcolor{\effectcolor}{ \textcolor{\effectcolor}{\phi} } } $, and
computations tupled with environments $ \mathcal{M}\llbracket  \ottnt{B} \rrbracket^{\textcolor{\effectcolor}{ \textcolor{\effectcolor}{\phi} } } $.

\begin{definition}[\link{effects/CBPV/semtyping.v}{LRV}{CBPV with Effects: Logical Relation}]
\[
\begin{array}{lcl@{\ }l}
 \mathcal{W}\llbracket   \ottkw{U}_{\color{\effectcolor}{ \textcolor{\effectcolor}{\phi} } }\;  \ottnt{B}   \rrbracket    &=& \{\  \mathbf{clo}( \rho , \{  \ottnt{M}  \} )  &|\  \ottsym{(}  \rho  \ottsym{,}  \ottnt{M}  \ottsym{)} \, \in \,  \mathcal{M}\llbracket  \ottnt{B} \rrbracket^{\textcolor{\effectcolor}{ \textcolor{\effectcolor}{\phi} } } \ \}  \\
 \mathcal{W}\llbracket  \ottkw{unit}  \rrbracket    &=& \{\ ()\ \}                                  \\
 \mathcal{W}\llbracket   \ottnt{A_{{\mathrm{1}}}} \times \ottnt{A_{{\mathrm{2}}}}   \rrbracket    &=& \{\ \ottsym{(}  \ottnt{W_{{\mathrm{1}}}}  \ottsym{,}  \ottnt{W_{{\mathrm{2}}}}  \ottsym{)}  &|\  \ottnt{W_{{\mathrm{1}}}} \, \in \,  \mathcal{W}\llbracket  \ottnt{A_{{\mathrm{1}}}}  \rrbracket  \ \textit{and}\  \ottnt{W_{{\mathrm{2}}}} \, \in \,  \mathcal{W}\llbracket  \ottnt{A_{{\mathrm{2}}}}  \rrbracket  \ \}\\
\ifextended
 \mathcal{W}\llbracket  \ottnt{A_{{\mathrm{1}}}}  \ottsym{+}  \ottnt{A_{{\mathrm{2}}}}  \rrbracket    &=& \{\ \ottkw{inl} \, \ottnt{W}      &|\ \ottnt{W} \, \in \,  \mathcal{W}\llbracket  \ottnt{A_{{\mathrm{1}}}}  \rrbracket \ \}\ \cup
                        \{\ \ottkw{inr} \, \ottnt{W}\ |\ \ottnt{W} \, \in \,  \mathcal{W}\llbracket  \ottnt{A_{{\mathrm{2}}}}  \rrbracket \ \}  \\ \fi
\\
 \mathcal{T}\llbracket  \ottkw{F} \, \ottnt{A} \rrbracket^{\textcolor{\effectcolor}{ \textcolor{\effectcolor}{\phi} } }  &=&  \{\ \ottkw{return} \, \ottnt{W} &|\  \ottnt{W} \, \in \,  \mathcal{W}\llbracket  \ottnt{A}  \rrbracket  \ \textit{and}\   \textcolor{\effectcolor}{ \textcolor{\effectcolor}{\phi} \; \equiv \;   \textcolor{\effectcolor}{\varepsilon}  }  \ \} \\
 \mathcal{T}\llbracket  \ottnt{A}  \to  \ottnt{B} \rrbracket^{\textcolor{\effectcolor}{ \textcolor{\effectcolor}{\phi} } }  &=&  \{\  \mathbf{clo}(  \rho ,   \lambda  \ottmv{x} . \ottnt{M}   )  &|\  \textit{for all}\  \ottnt{W} \, \in \,  \mathcal{W}\llbracket  \ottnt{A}  \rrbracket  ,\  \ottsym{(}  \ottsym{(}   \rho   \mathop{,}   \ottmv{x}  \mapsto  \ottnt{W}   \ottsym{)}  \ottsym{,}  \ottnt{M}  \ottsym{)} \, \in \,  \mathcal{M}\llbracket  \ottnt{B} \rrbracket^{\textcolor{\effectcolor}{ \textcolor{\effectcolor}{\phi} } }  \ \} \\
 \mathcal{T}\llbracket   \ottnt{B_{{\mathrm{1}}}}   \mathop{\&}   \ottnt{B_{{\mathrm{2}}}}  \rrbracket^{\textcolor{\effectcolor}{ \textcolor{\effectcolor}{\phi} } }  &=& \{\  \mathbf{clo}(  \rho ,   \langle  \ottnt{M_{{\mathrm{1}}}} , \ottnt{M_{{\mathrm{2}}}}  \rangle   )    &|\
                          \ottsym{(}  \rho  \ottsym{,}  \ottnt{M_{{\mathrm{1}}}}  \ottsym{)} \, \in \,  \mathcal{M}\llbracket  \ottnt{B_{{\mathrm{1}}}} \rrbracket^{\textcolor{\effectcolor}{ \textcolor{\effectcolor}{\phi} } }  \ \textit{and}\  \ottsym{(}  \rho  \ottsym{,}  \ottnt{M_{{\mathrm{2}}}}  \ottsym{)} \, \in \,  \mathcal{M}\llbracket  \ottnt{B_{{\mathrm{2}}}} \rrbracket^{\textcolor{\effectcolor}{ \textcolor{\effectcolor}{\phi} } }  \ \} \\
\\
 \mathcal{M}\llbracket  \ottnt{B} \rrbracket^{\textcolor{\effectcolor}{ \textcolor{\effectcolor}{\phi} } }  &=& \{\ (\rho , \ottnt{M}) &|\   \rho  \vdash_{\mathit{eff} }  \ottnt{M}  \Downarrow  \ottnt{T}  \mathop{\#} \textcolor{\effectcolor}{ \textcolor{\effectcolor}{\phi}_{{\mathrm{1}}} }  \ \textit{and}\  \ottnt{T} \, \in \,  \mathcal{T}\llbracket  \ottnt{B} \rrbracket^{\textcolor{\effectcolor}{ \textcolor{\effectcolor}{\phi}_{{\mathrm{2}}} } }  
   \mathit{and}\  \textcolor{\effectcolor}{  \textcolor{\effectcolor}{ \textcolor{\effectcolor}{\phi}_{{\mathrm{1}}}  \cdot  \textcolor{\effectcolor}{\phi}_{{\mathrm{2}}} }  \;  \textcolor{\effectcolor}{\mathop{\leq_{\mathit{eff} } } } \;  \textcolor{\effectcolor}{\phi}  } \ \} \\
\end{array}
\]
\end{definition}

We use this relation to define semantic typing for
environments, values, and computations.
\begin{definition}[\link{effects/CBPV/semtyping.v}{SemVWt,SemCWt}{CBPV with Effects: Semantic Typing}]
\[
\begin{array}{lcl@{\ }l}
\Gamma  \vDash  \rho &=& \multicolumn{2}{l}{  \ottmv{x}  \ottsym{:}  \ottnt{A} \, \in \, \Gamma \ \textit{implies}\  \ottmv{x}  \mapsto  \ottnt{W} \, \in \, \rho  \ \textit{and}\  \ottnt{W} \, \in \,  \mathcal{W}\llbracket  \ottnt{A}  \rrbracket  } \\
\Gamma  \vDash_{\mathit{eff} }  \ottnt{V}  \ottsym{:}  \ottnt{A} &=& \multicolumn{2}{l}{  \Gamma  \vDash  \rho \ \textit{implies}\   \rho  \vdash  \ottnt{V} \ \Downarrow\  \ottnt{W}   \ \textit{and}\  \ottnt{W} \, \in \,  \mathcal{W}\llbracket  \ottnt{A}  \rrbracket  } \\
 \Gamma   \vDash_{\mathit{eff} }   \ottnt{M} \; :^{\textcolor{\effectcolor}{  \textcolor{\effectcolor}{\phi}  } }\;  \ottnt{B}  &=& \multicolumn{2}{l}{ \Gamma  \vDash  \rho \ \textit{implies}\  \ottsym{(}  \rho  \ottsym{,}  \ottnt{M}  \ottsym{)} \, \in \,  \mathcal{M}\llbracket  \ottnt{B} \rrbracket^{\textcolor{\effectcolor}{ \textcolor{\effectcolor}{\phi} } }  } \\
\end{array}
\]
\end{definition}

Using these definitions, we can prove semantic typing lemmas corresponding to
each of the syntactic typing rules shown in
Figure~\ref{fig:cbpv-effects}. These proofs require our assumptions about the
monoidal structure of effects: that $ \textcolor{\effectcolor}{\varepsilon} $ is an identity element for the
associative combining operation.  

With these lemmas, we show the fundamental lemma as a straightforward
induction.
\begin{lemma}[\link{effects/CBPV/soundness.v}{fundamental\_lemma}{Fundamental Lemma: effect soundness}]\ 
\begin{enumerate}
\item If $\Gamma  \vdash_{\mathit{eff} }  \ottnt{V}  \ottsym{:}  \ottnt{A}$ then $\Gamma  \vDash_{\mathit{eff} }  \ottnt{V}  \ottsym{:}  \ottnt{A}$.
\item If $ \Gamma   \vdash_{\mathit{eff} }   \ottnt{M} \; :^{\textcolor{\effectcolor}{  \textcolor{\effectcolor}{\phi}  } }\;  \ottnt{B} $ then $ \Gamma   \vDash_{\mathit{eff} }   \ottnt{M} \; :^{\textcolor{\effectcolor}{  \textcolor{\effectcolor}{\phi}  } }\;  \ottnt{B} $.
\end{enumerate}
\end{lemma}

The effect soundness theorem (\ref{thm:effect-soundness}) follows from the second 
clause of this lemma, after instantiating $\Gamma$ with the empty context and $\ottnt{B}$ with $\ottkw{F} \, \ottnt{A}$.
Unfolding the definition of $ \varnothing   \vDash_{\mathit{eff} }   \ottnt{M} \; :^{\textcolor{\effectcolor}{  \textcolor{\effectcolor}{\phi}  } }\;  \ottkw{F} \, \ottnt{A} $ gives us some $\textcolor{red}{\phi_1}$ and $\textcolor{red}{\phi_2}$ such 
that $ \emptyset  \vdash_{\mathit{eff} }  \ottnt{M}  \Downarrow  \ottnt{T}  \mathop{\#} \textcolor{\effectcolor}{ \textcolor{\effectcolor}{\phi}_{{\mathrm{1}}} } $ and $\ottnt{T} \, \in \,  \mathcal{T}\llbracket  \ottkw{F} \, \ottnt{A} \rrbracket^{\textcolor{\effectcolor}{ \textcolor{\effectcolor}{\phi}_{{\mathrm{2}}} } } $ and $ \textcolor{\effectcolor}{  \textcolor{\effectcolor}{ \textcolor{\effectcolor}{\phi}_{{\mathrm{1}}}  \cdot  \textcolor{\effectcolor}{\phi}_{{\mathrm{2}}} }  \;  \textcolor{\effectcolor}{\mathop{\leq_{\mathit{eff} } } } \;  \textcolor{\effectcolor}{\phi}  } $.
Further unfolding definitions means that 
$\ottnt{T}$ must be $\ottkw{return} \, \ottnt{W}$, $\textcolor{red}{\phi_2}$ must be $ \textcolor{\effectcolor}{\varepsilon} $, and thus $ \textcolor{\effectcolor}{ \textcolor{\effectcolor}{\phi}_{{\mathrm{1}}} \;  \textcolor{\effectcolor}{\mathop{\leq_{\mathit{eff} } } } \;  \textcolor{\effectcolor}{\phi}  } $.

\subsection{Type-and-effect Preserving Translations}
\label{sec:effect-translations}

\citet{levy2006} provides translations from call-by-value (CBV) and
call-by-name (CBN) $\lambda$-calculi to CBPV and shows that those translations
preserve types, denotational semantics, and (substitution-based) big-step
operational semantics. We show here that those
translations also preserve effects.

For the CBV translation, we start with a $\lambda$-calculus that has a
simple type-and-effect system, loosely based on
\citet{lucassen-gifford}. However, as few CBN
languages directly include effects, for the CBN
translation we start with a simply-typed $\lambda$-calculus
that encapsulates effects
using a \emph{graded monad}. Furthermore,
we show that we can use this same monad with the CBV translation because
effects are encapsulated.

\subsubsection{CBV Type-and-effect System}

The \link{effect/CBV/typing.v}{Wt}{simple CBV language with effect tracking} in this subsection features the
same $ \textcolor{\effectcolor}{ \ottkw{tick} } $ term and $ \textcolor{\effectcolor}{\ottkw{Tick} } $ effect as before, along with the usual forms of the $\lambda$-calculus.\ifextended\else\footnote{For space, we elide the typing rules for unit and products. These rules are available in the extended version~\cite{extended-version}.} \fi 
\ifextended
\[
\begin{array}{llcl}
\mathit{types} &\tau &::=&  \texttt{unit} \ |\  \tau_{{\mathrm{1}}}  \stackrel{\textcolor{\effectcolor}{ \textcolor{\effectcolor}{\phi} } }{\rightarrow}  \tau_{{\mathrm{2}}} \ |\  \tau_{{\mathrm{1}}}  \otimes  \tau_{{\mathrm{2}}} \ |\ \tau_{{\mathrm{1}}}  \ottsym{+}  \tau_{{\mathrm{2}}} \\
\mathit{terms} &\ottnt{e}   &::=&  \textcolor{\effectcolor}{ \ottkw{tick} } \ |\ \ottmv{x}\ |\  \lambda  \ottmv{x} . \ottnt{e} \ |\ \ottnt{e_{{\mathrm{1}}}} \, \ottnt{e_{{\mathrm{2}}}}\ |\ \ottsym{()}\ |\  \ottnt{e_{{\mathrm{1}}}}  ;  \ottnt{e_{{\mathrm{2}}}}  \\
              & & | & \ottsym{(}  \ottnt{e_{{\mathrm{1}}}}  \ottsym{,}  \ottnt{e_{{\mathrm{2}}}}  \ottsym{)}\ |\  \ottkw{let}\; ( \ottmv{x_{{\mathrm{1}}}} ,  \ottmv{x_{{\mathrm{2}}}} ) =  \ottnt{e_{{\mathrm{1}}}} \; \ottkw{in}\;  \ottnt{e_{{\mathrm{2}}}} \ \\
              & & | & \ottkw{inl} \, \ottnt{e}\ |\ \ottkw{inr} \, \ottnt{e}\ |\  \ottkw{case}\;  \ottnt{e} \; \ottkw{of}\; \ottkw{inl}\;  \ottmv{x_{{\mathrm{1}}}}  \rightarrow  \ottnt{e_{{\mathrm{1}}}}  \ottkw{;} \ottkw{inr}\;  \ottmv{x_{{\mathrm{2}}}}  \rightarrow  \ottnt{e_{{\mathrm{2}}}}  \\
\end{array}
\]
\fi

\ifextended
 \drules[lam-eff]{$ \Gamma   \vdash_{\mathit{eff} }   \ottnt{e}  :^{\textcolor{\effectcolor}{ \textcolor{\effectcolor}{\phi} } }  \tau $}{STLC + effect typing}
    {var,abs,app,unit,sequence,pair,split,inl,inr,case,tick}
\else
\[ \drule[width=1in]{lam-eff-var}\;\drule{lam-eff-abs}\;\drule{lam-eff-app}\;\drule{lam-eff-tick} \]
\fi

Function types, written $ \tau_{{\mathrm{1}}}  \stackrel{\textcolor{\effectcolor}{ \textcolor{\effectcolor}{\phi} } }{\rightarrow}  \tau_{{\mathrm{2}}} $, are annotated with \emph{latent
  effects}, which occur when the function is called.  In the application rule
\rref{lam-eff-app}, this latent effect is
combined with $\textcolor{red}{\textcolor{\effectcolor}{\phi}_{{\mathrm{1}}}}$,
the effects that occur when evaluating the function
$\ottnt{e_{{\mathrm{1}}}}$ to a $\lambda$ expression, and $\textcolor{red}{\textcolor{\effectcolor}{\phi}_{{\mathrm{2}}}}$, the effects that occur when
evaluating the argument to a value.

The \link{effects/CBV/translations.v}{translateType,translateTerm}{CBV type
  and term translations} follow directly from \citet{levy:siglog-tutorial}.
Besides adding a case for the $ \textcolor{\effectcolor}{ \ottkw{tick} } $ expression, the only change that we
make is moving the latent effect from the function type to the thunk type. All
other cases are exactly as in prior work. \ifextended\else Because of this, we
  only show the translation for function types and the $ \textcolor{\effectcolor}{ \ottkw{tick} } $
  expression. \fi

\[
\begin{array}{ll}
\textit{Type translation}\\
 \llbracket {  \tau_{{\mathrm{1}}}  \stackrel{\textcolor{\effectcolor}{ \textcolor{\effectcolor}{\phi} } }{\rightarrow}  \tau_{{\mathrm{2}}}  } \rrbracket_{\textsc{v} } & =  \ottkw{U}_{\color{\effectcolor}{ \textcolor{\effectcolor}{\phi} } }\;  \ottsym{(}   \llbracket { \tau_{{\mathrm{1}}} } \rrbracket_{\textsc{v} }   \to  \ottkw{F} \,  \llbracket { \tau_{{\mathrm{2}}} } \rrbracket_{\textsc{v} }   \ottsym{)}  \\
\ifextended
 \llbracket { \ottkw{unit} } \rrbracket_{\textsc{v} }         & =  \texttt{unit}  \\
 \llbracket {  \tau_{{\mathrm{1}}}  \otimes  \tau_{{\mathrm{2}}}  } \rrbracket_{\textsc{v} } & =    \llbracket { \tau_{{\mathrm{1}}} } \rrbracket_{\textsc{v} }  \times  \llbracket { \tau_{{\mathrm{2}}} } \rrbracket_{\textsc{v} }   \\
 \llbracket { \tau_{{\mathrm{1}}}  \ottsym{+}  \tau_{{\mathrm{2}}} } \rrbracket_{\textsc{v} } & =   \llbracket { \tau_{{\mathrm{1}}} } \rrbracket_{\textsc{v} }   \ottsym{+}   \llbracket { \tau_{{\mathrm{2}}} } \rrbracket_{\textsc{v} }  \\
    &\\
\textit{Context translation}\\
 \llbracket { \varnothing } \rrbracket_{\textsc{v} }  & =  \varnothing  \\
 \llbracket {  \Gamma   \mathop{,}   \ottmv{x}  \ottsym{:}  \tau  } \rrbracket_{\textsc{v} }  & =   \llbracket { \Gamma } \rrbracket_{\textsc{v} }    \mathop{,}   \ottmv{x}  \ottsym{:}   \llbracket { \tau } \rrbracket_{\textsc{v} }   \\ \fi
\ifextended&\\\else\end{array}\qquad\begin{array}{ll}\fi
\textit{Term translation}\\
 \llbracket {  \textcolor{\effectcolor}{ \ottkw{tick} }  } \rrbracket_{\textsc{v} }           & =  \textcolor{\effectcolor}{ \ottkw{tick} }  \\
\ifextended
 \llbracket { \ottmv{x} } \rrbracket_{\textsc{v} }              & = \ottkw{return} \, \ottmv{x} \\
 \llbracket {  \lambda  \ottmv{x} . \ottnt{e}  } \rrbracket_{\textsc{v} }         & = \ottkw{return} \, \ottsym{\{}   \lambda  \ottmv{x} .  \llbracket { \ottnt{e} } \rrbracket_{\textsc{v} }    \ottsym{\}} \\
 \llbracket { \ottnt{e_{{\mathrm{1}}}} \, \ottnt{e_{{\mathrm{2}}}} } \rrbracket_{\textsc{v} }          & = \ottmv{x}  \leftarrow   \llbracket { \ottnt{e_{{\mathrm{1}}}} } \rrbracket_{\textsc{v} }  \, \ottkw{in} \, \ottmv{y}  \leftarrow   \llbracket { \ottnt{e_{{\mathrm{2}}}} } \rrbracket_{\textsc{v} }  \, \ottkw{in} \, \ottmv{x}  \ottsym{!} \, \ottmv{y} \\
 \llbracket { \ottsym{()} } \rrbracket_{\textsc{v} }             & = \ottkw{return} \, \ottsym{()} \\
 \llbracket {  \ottnt{e_{{\mathrm{1}}}}  ;  \ottnt{e_{{\mathrm{2}}}}  } \rrbracket_{\textsc{v} }  & = \ottmv{x}  \leftarrow   \llbracket { \ottnt{e_{{\mathrm{1}}}} } \rrbracket_{\textsc{v} }  \, \ottkw{in} \,  \ottmv{x}  ;   \llbracket { \ottnt{e_{{\mathrm{2}}}} } \rrbracket_{\textsc{v} }   \\
 \llbracket { \ottsym{(}  \ottnt{e_{{\mathrm{1}}}}  \ottsym{,}  \ottnt{e_{{\mathrm{2}}}}  \ottsym{)} } \rrbracket_{\textsc{v} }      & = \ottmv{x}  \leftarrow   \llbracket { \ottnt{e_{{\mathrm{1}}}} } \rrbracket_{\textsc{v} }  \, \ottkw{in} \, \ottmv{y}  \leftarrow   \llbracket { \ottnt{e_{{\mathrm{2}}}} } \rrbracket_{\textsc{v} }  \, \ottkw{in} \, \ottkw{return} \, \ottsym{(}  \ottmv{x}  \ottsym{,}  \ottmv{y}  \ottsym{)} \\
 \llbracket {  \ottkw{let}\; ( \ottmv{x_{{\mathrm{1}}}} ,  \ottmv{x_{{\mathrm{2}}}} ) =  \ottnt{e_{{\mathrm{1}}}} \; \ottkw{in}\;  \ottnt{e_{{\mathrm{2}}}}  } \rrbracket_{\textsc{v} }  & = \ottmv{x}  \leftarrow   \llbracket { \ottnt{e_{{\mathrm{1}}}} } \rrbracket_{\textsc{v} }  \, \ottkw{in} \,  \ottkw{let}\; ( \ottmv{x_{{\mathrm{1}}}} ,  \ottmv{x_{{\mathrm{2}}}} ) =  \ottmv{x} \; \ottkw{in}\;   \llbracket { \ottnt{e_{{\mathrm{2}}}} } \rrbracket_{\textsc{v} }   \\
 \llbracket { \ottkw{inl} \, \ottnt{e} } \rrbracket_{\textsc{v} }         & = \ottmv{x}  \leftarrow   \llbracket { \ottnt{e} } \rrbracket_{\textsc{v} }  \, \ottkw{in} \, \ottkw{return} \, \ottsym{(}  \ottkw{inl} \, \ottmv{x}  \ottsym{)} \\
 \llbracket { \ottkw{inr} \, \ottnt{e} } \rrbracket_{\textsc{v} }         & = \ottmv{x}  \leftarrow   \llbracket { \ottnt{e} } \rrbracket_{\textsc{v} }  \, \ottkw{in} \, \ottkw{return} \, \ottsym{(}  \ottkw{inr} \, \ottmv{x}  \ottsym{)} \\
 \llbracket {  \ottkw{case}\;  \ottnt{e} \; \ottkw{of}\; \ottkw{inl}\;  \ottmv{x_{{\mathrm{1}}}}  \rightarrow  \ottnt{e_{{\mathrm{1}}}}  \ottkw{;} \ottkw{inr}\;  \ottmv{x_{{\mathrm{2}}}}  \rightarrow  \ottnt{e_{{\mathrm{2}}}}  } \rrbracket_{\textsc{v} }  &=
     \ottmv{x}  \leftarrow   \llbracket { \ottnt{e} } \rrbracket_{\textsc{v} }  \, \ottkw{in} \, \ottkw{case} \, \ottmv{x} \, \ottkw{of} \, \ottkw{inl} \, \ottmv{x_{{\mathrm{1}}}}  \to   \llbracket { \ottnt{e_{{\mathrm{1}}}} } \rrbracket_{\textsc{v} }   \mathsf{;} \, \ottkw{inr} \, \ottmv{x_{{\mathrm{2}}}}  \to   \llbracket { \ottnt{e_{{\mathrm{2}}}} } \rrbracket_{\textsc{v} }  \\
\fi
\end{array}
\]

This translation preserves types and effects from the source language.

\begin{lemma}[\link{effects/EffCBV/proofs.v}{translation\_correct}{CBV translation is type correct}]
\label{lem:cbv-trans-effects}
If $ \Gamma   \vdash_{\mathit{eff} }   \ottnt{e}  :^{\textcolor{\effectcolor}{ \textcolor{\effectcolor}{\phi} } }  \tau $ then
  $  \llbracket { \Gamma } \rrbracket_{\textsc{v} }    \vdash_{\mathit{eff} }    \llbracket { \ottnt{e} } \rrbracket_{\textsc{v} }  \; :^{\textcolor{\effectcolor}{  \textcolor{\effectcolor}{\phi}  } }\;  \ottkw{F} \,  \llbracket { \tau } \rrbracket_{\textsc{v} }  $.
\end{lemma}

This result is easy to prove, reassuring us that our
effect system design is correct: we can use CBPV to encode the well-studied
type-and-effect systems developed over the past 40 years.

\subsubsection{Graded Monads}

CBPV is designed to serve as a convenient translation
target for both CBV and CBN languages.
However, in CBN languages, effects are usually\footnote{Instead of graded monads, we could also consider a translation from call-by-name language that does not
encapsulate effects, such as the one defined by~\citet{McDermottMycroft+2018+93+108}.}
tracked using parametric effect
monads, also known as \emph{graded monads}
~\cite{katsumata:2014,smirnov2008graded,wadler-thiemann,orchard:haskell-effects}.
Therefore, here we translate a \link{effects/CBN/typing.v}{Wt}{CBN language
with graded monads} to CBPV.  Our source language for this translation is
the simply-typed $\lambda$-calculus with unit\ifextended, sums\fi \ and products, together with
a graded monadic type $ \ottkw{T}_{\textcolor{\effectcolor}{ \textcolor{\effectcolor}{\phi} } }\;  \tau $, the monadic operations $\ottkw{return}$
and $ \ottkw{bind} $, and the $ \textcolor{\effectcolor}{ \ottkw{tick} } $ operation, with a monadic
type. To account
for imprecision, we include an explicit type coercion, written $\ottkw{coerce} \, \ottnt{e}$ for the graded monad.

\ifextended
  \drules[lam-mon]{$\Gamma  \vdash_{\mathit{mon} }  \ottnt{e}  \ottsym{:}  \tau$}{STLC + graded monad}
{var,abs,app,unit,sequence,pair,split,with,fst,snd,inl,inr,case,coerce,return,bind,tick}
\else
\[  \drule{lam-mon-return}\quad\drule[width=4in]{lam-mon-bind} \]
\[  \drule{lam-mon-tick}\quad\drule[width=3in]{lam-mon-coerce} \]
\fi

Below, we extend Levy's \link{effects/CBN/translation}{translateTerm}{translation of the CBN $\lambda$-calculus} to include the graded
monad. The translation of the core language is as in prior work and all effects
are isolated to the monadic type, so we only show the monadic portion in the figure.

\[
\begin{array}{ll}
\textit{Type translation}\\
\ifextended
 \llbracket { \ottkw{unit} } \rrbracket_{\textsc{n} }            & = \ottkw{F} \, \ottkw{unit} \\
 \llbracket { \tau_{{\mathrm{1}}}  \to  \tau_{{\mathrm{2}}} } \rrbracket_{\textsc{n} }     & = \ottsym{(}  \ottkw{U} \,  \llbracket { \tau_{{\mathrm{1}}} } \rrbracket_{\textsc{n} }   \ottsym{)}  \to   \llbracket { \tau_{{\mathrm{2}}} } \rrbracket_{\textsc{n} }  \\
 \llbracket {  \tau_{{\mathrm{1}}}  \mathop{\&}   \tau_{{\mathrm{2}}}  } \rrbracket_{\textsc{n} }      & =   \llbracket { \tau_{{\mathrm{1}}} } \rrbracket_{\textsc{n} }    \mathop{\&}    \llbracket { \tau_{{\mathrm{2}}} } \rrbracket_{\textsc{n} }   \\
 \llbracket { \tau_{{\mathrm{1}}}  \ottsym{+}  \tau_{{\mathrm{2}}} } \rrbracket_{\textsc{n} }      & = \ottkw{F} \, \ottsym{(}  \ottkw{U} \,  \llbracket { \tau_{{\mathrm{1}}} } \rrbracket_{\textsc{n} }   \ottsym{+}  \ottkw{U} \,  \llbracket { \tau_{{\mathrm{2}}} } \rrbracket_{\textsc{n} }   \ottsym{)} \\  \fi
 \llbracket {  \ottkw{T}_{\textcolor{\effectcolor}{ \textcolor{\effectcolor}{\phi} } }\;  \tau  } \rrbracket_{\textsc{n} }  & = \ottkw{F} \, \ottsym{(}   \ottkw{U}_{\color{\effectcolor}{ \textcolor{\effectcolor}{\phi} } }\;  \ottkw{F} \, \ottsym{(}   \ottkw{U}_{\color{\effectcolor}{  \textcolor{\effectcolor}{\varepsilon}  } }\;   \llbracket { \tau } \rrbracket_{\textsc{n} }    \ottsym{)}   \ottsym{)} \\
&\\
\ifextended
\textit{Context translation}\\
 \llbracket { \varnothing } \rrbracket_{\textsc{n} }  & =  \varnothing  \\
 \llbracket {  \Gamma   \mathop{,}   \ottmv{x}  \ottsym{:}  \tau  } \rrbracket_{\textsc{n} }  & =   \llbracket { \Gamma } \rrbracket_{\textsc{n} }    \mathop{,}   \ottmv{x}  \ottsym{:}   \ottkw{U}_{\color{\effectcolor}{  \textcolor{\effectcolor}{\varepsilon}  } }\;   \llbracket { \tau } \rrbracket_{\textsc{n} }    \\
&\\
\fi
\textit{Term translation}\\
\ifextended
 \llbracket { \ottmv{x} } \rrbracket_{\textsc{n} }              & = \ottmv{x}  \ottsym{!} \\
 \llbracket {  \lambda  \ottmv{x} . \ottnt{e}  } \rrbracket_{\textsc{n} }         & =  \lambda  \ottmv{x} .  \llbracket { \ottnt{e} } \rrbracket_{\textsc{n} }   \\
 \llbracket { \ottnt{e_{{\mathrm{1}}}} \, \ottnt{e_{{\mathrm{2}}}} } \rrbracket_{\textsc{n} }          & =  \llbracket { \ottnt{e_{{\mathrm{1}}}} } \rrbracket_{\textsc{n} }  \, \ottsym{\{}   \llbracket { \ottnt{e_{{\mathrm{2}}}} } \rrbracket_{\textsc{n} }   \ottsym{\}} \\
 \llbracket { \ottsym{()} } \rrbracket_{\textsc{n} }             & = \ottkw{return} \, \ottsym{()} \\
 \llbracket {  \ottnt{e_{{\mathrm{1}}}}  ;  \ottnt{e_{{\mathrm{2}}}}  } \rrbracket_{\textsc{n} }  & = \ottmv{x}  \leftarrow   \llbracket { \ottnt{e_{{\mathrm{1}}}} } \rrbracket_{\textsc{n} }  \, \ottkw{in} \,  \ottmv{x}  ;   \llbracket { \ottnt{e_{{\mathrm{2}}}} } \rrbracket_{\textsc{n} }   \\
 \llbracket {  \langle  \ottnt{e_{{\mathrm{1}}}} , \ottnt{e_{{\mathrm{2}}}}  \rangle  } \rrbracket_{\textsc{n} }      & =  \langle   \llbracket { \ottnt{e_{{\mathrm{1}}}} } \rrbracket_{\textsc{n} }  ,  \llbracket { \ottnt{e_{{\mathrm{2}}}} } \rrbracket_{\textsc{n} }   \rangle  \\
 \llbracket { \ottnt{e}  \ottsym{.}  \ottsym{1} } \rrbracket_{\textsc{n} }         & =   \llbracket { \ottnt{e} } \rrbracket_{\textsc{n} }   . 1  \\
 \llbracket { \ottnt{e}  \ottsym{.}  \ottsym{2} } \rrbracket_{\textsc{n} }         & =   \llbracket { \ottnt{e} } \rrbracket_{\textsc{n} }   . 2  \\
 \llbracket { \ottkw{inl} \, \ottnt{e} } \rrbracket_{\textsc{n} }         & = \ottkw{return} \, \ottkw{inl} \, \ottsym{\{}   \llbracket { \ottnt{e} } \rrbracket_{\textsc{n} }   \ottsym{\}} \\
 \llbracket { \ottkw{inr} \, \ottnt{e} } \rrbracket_{\textsc{n} }         & = \ottkw{return} \, \ottkw{inr} \, \ottsym{\{}   \llbracket { \ottnt{e} } \rrbracket_{\textsc{n} }   \ottsym{\}} \\
 \llbracket {  \ottkw{case}\;  \ottnt{e} \; \ottkw{of}\; \ottkw{inl}\;  \ottmv{x_{{\mathrm{1}}}}  \rightarrow  \ottnt{e_{{\mathrm{1}}}}  \ottkw{;} \ottkw{inr}\;  \ottmv{x_{{\mathrm{2}}}}  \rightarrow  \ottnt{e_{{\mathrm{2}}}}  } \rrbracket_{\textsc{n} }  &=
     \ottmv{x}  \leftarrow   \llbracket { \ottnt{e} } \rrbracket_{\textsc{n} }  \, \ottkw{in} \, \ottkw{case} \, \ottmv{x} \, \ottkw{of} \, \ottkw{inl} \, \ottmv{x_{{\mathrm{1}}}}  \to   \llbracket { \ottnt{e_{{\mathrm{1}}}} } \rrbracket_{\textsc{n} }   \mathsf{;} \, \ottkw{inr} \, \ottmv{x_{{\mathrm{2}}}}  \to   \llbracket { \ottnt{e_{{\mathrm{2}}}} } \rrbracket_{\textsc{n} }  \\
\fi
 \llbracket { \ottkw{return} \, \ottnt{e} } \rrbracket_{\textsc{n} }  & = \ottkw{return} \, \ottsym{\{}  \ottkw{return} \, \ottsym{\{}   \llbracket { \ottnt{e} } \rrbracket_{\textsc{n} }   \ottsym{\}}  \ottsym{\}}\\
 \llbracket { \ottkw{bind} \, \ottmv{x}  \ottsym{=}  \ottnt{e_{{\mathrm{1}}}} \, \ottkw{in} \, \ottnt{e_{{\mathrm{2}}}} } \rrbracket_{\textsc{n} }  & =
  \ottkw{return} \, \ottsym{\{}  \ottmv{x}  \leftarrow  \ottsym{(}  \ottmv{y}  \leftarrow   \llbracket { \ottnt{e_{{\mathrm{1}}}} } \rrbracket_{\textsc{n} }  \, \ottkw{in} \, \ottmv{y}  \ottsym{!}  \ottsym{)} \, \ottkw{in} \, \ottmv{z}  \leftarrow   \llbracket { \ottnt{e_{{\mathrm{2}}}} } \rrbracket_{\textsc{n} }  \, \ottkw{in} \, \ottmv{z}  \ottsym{!}  \ottsym{\}} \\
 \llbracket { \ottkw{coerce} \, \ottnt{e} } \rrbracket_{\textsc{n} }  & = \ottkw{return} \, \ottsym{\{}  \ottmv{x}  \leftarrow   \llbracket { \ottnt{e} } \rrbracket_{\textsc{n} }  \, \ottkw{in} \, \ottmv{x}  \ottsym{!}  \ottsym{\}} \\
 \llbracket {  \textcolor{\effectcolor}{ \ottkw{tick} }  } \rrbracket_{\textsc{n} }  & = \ottkw{return} \, \ottsym{\{}  \ottmv{x}  \leftarrow   \textcolor{\effectcolor}{ \ottkw{tick} }  \, \ottkw{in} \, \ottkw{return} \, \ottsym{\{}  \ottkw{return} \, \ottmv{x}  \ottsym{\}}  \ottsym{\}}  \\
\end{array}
\]

This translation preserves types (with embedded effects) from the source
language.  Note that, because the monadic type marks effectful code, the
translation produces CBPV computations that can be checked with the ``pure''
effect $\textcolor{red}{ \textcolor{\effectcolor}{\varepsilon} }$.

\begin{lemma}[\link{effects/CBN/proofs.v}{translation\_correct}{CBN translation is type correct}]
\label{lem:type-trans-lam-mon-lin}
  If $\Gamma  \vdash_{\mathit{mon} }  \ottnt{e}  \ottsym{:}  \tau$ then
  $  \llbracket { \Gamma } \rrbracket_{\textsc{n} }    \vdash_{\mathit{eff} }    \llbracket { \ottnt{e} } \rrbracket_{\textsc{n} }  \; :^{\textcolor{\effectcolor}{   \textcolor{\effectcolor}{\varepsilon}   } }\;   \llbracket { \tau } \rrbracket_{\textsc{n} }  $.
\end{lemma}

One difficulty of this translation is that the monadic type in the CBPV
adjunction is \textbf{U F}. This type is a value type, and the standard CBN
translation produces terms with computation types. Therefore. to use \textbf{U
  F} as the monad in our CBN translation, we need to bracket it: on the
outside by \textbf{F} to form a computation type, and then on the inside by
\textbf{U} to construct the value type that the monad expects. This bracketing
produces an awkward translation of the monadic operations with doubled
thunking. This awkwardness is due to the presence of the monad in the source
language; it is a separate structure from the ambient monad of the computation
language.

\ifextended

Because the graded monad isolates effects, we can also
evaluate the monadic language using a call-by-value semantics,
reusing the same translation we used for the CBV language with
effects. For the CBV translation, the monadic type is more accessible: the type translation produces value
types, so we don't need the additional bracketing in the translations for $\ottkw{return} \, \ottnt{e}$
and $ \textcolor{\effectcolor}{ \ottkw{tick} } $. The translations for $ \ottkw{bind} $ and $ \ottkw{coerce} $ remain unchanged.

\[
\begin{array}{ll}
\textit{Type translation}\\
 \llbracket {  \ottkw{T}_{\textcolor{\effectcolor}{ \textcolor{\effectcolor}{\phi} } }\;  \tau  } \rrbracket_{\textsc{v} }  & =  \ottkw{U}_{\color{\effectcolor}{ \textcolor{\effectcolor}{\phi} } }\;  \ottkw{F} \,  \llbracket { \tau } \rrbracket_{\textsc{v} }   \\
&\\
\end{array}
\qquad
\begin{array}{ll}
\textit{Term translation}\\
 \llbracket { \ottkw{return} \, \ottnt{e} } \rrbracket_{\textsc{v} }  & = \ottkw{return} \, \ottsym{\{}   \llbracket { \ottnt{e} } \rrbracket_{\textsc{v} }   \ottsym{\}}\\
 \llbracket {  \textcolor{\effectcolor}{ \ottkw{tick} }  } \rrbracket_{\textsc{v} }  & = \ottkw{return} \, \ottsym{\{}  \ottmv{x}  \leftarrow   \textcolor{\effectcolor}{ \ottkw{tick} }  \, \ottkw{in} \, \ottkw{return} \, \ottmv{x}  \ottsym{\}} \\
&\\
\end{array}
\]

\begin{lemma}[\link{effects/CBV/proofs.v}{translation\_correct}{Monadic translation type correctness}]\ 
\label{lem:type-trans-cbv-mon-eff}
    If $\Gamma  \vdash_{\mathit{mon} }  \ottnt{e}  \ottsym{:}  \tau$ then $  \llbracket { \Gamma } \rrbracket_{\textsc{v} }    \vdash_{\mathit{eff} }    \llbracket { \ottnt{e} } \rrbracket_{\textsc{v} }  \; :^{\textcolor{\effectcolor}{   \textcolor{\effectcolor}{\varepsilon}   } }\;  \ottkw{F} \,  \llbracket { \tau } \rrbracket_{\textsc{v} }  $.
\end{lemma}
\fi
\section{CBPV and Coeffects (Version 1: General semantics)}
\label{sec:coeffects}

Next, we construct a parallel extension of CBPV augmented with \emph{coeffect}
tracking. \ifextended Appendix~\ref{fig:cbpv-coeffects} \else
  Figure~\ref{fig:cbpv-coeffects} \fi lists the
\link{general/typing.v}{VWt,CWt}{typing rules,} with coeffect annotations in
blue. Coeffect systems are designed for reasoning about how programs use their
inputs, so we annotate variables at their binding sites and in the context.

Coeffects annotations consist of \emph{grades} $\textcolor{\coeffectcolor}{q}$ taken from a
\emph{preordered semiring}. This structure provides an addition operation $ \textcolor{\coeffectcolor}{ \textcolor{\coeffectcolor}{q}_{{\mathrm{1}}}  +  \textcolor{\coeffectcolor}{q}_{{\mathrm{2}}} } $,
an additive identity element $\textcolor{blue}{ \textcolor{\coeffectcolor}{0} }$, a multiplication
operation $ \textcolor{\coeffectcolor}{ \textcolor{\coeffectcolor}{q}_{{\mathrm{1}}}   \cdot   \textcolor{\coeffectcolor}{q}_{{\mathrm{2}}} } $, a multiplicative identity $\textcolor{blue}{ \textcolor{\coeffectcolor}{1} }$,
and a reflexive and transitive binary relation $ \textcolor{\coeffectcolor}{\mathop{\leq_{\mathit{co} } } } $ that respects
addition and multiplication. (The preorder does not have to be the one defined
by the addition operation.)  The need for a semiring rather than a monoid
arises from the fact that any value may be bound to a variable
that may then be used multiple times, requiring a notion of coeffect multiplication.

\ifextended\else
\begin{figure}
  \drules[coeff]{$ \textcolor{\coeffectcolor}{ \textcolor{\coeffectcolor}{\gamma} }\! \cdot \! \Gamma   \vdash_{\mathit{coeff} }  \ottnt{V}  \ottsym{:}  \ottnt{A}$}{value coeffect
    typing}{}
\[ \drule[width=3in]{coeff-var}\ \drule[width=4in]{coeff-thunk} \ \drule{coeff-unit} \]

\[ \drule[width=4in]{coeff-pair} \quad \drule[width=3in]{coeff-vsub} \]
\ifextended \[ \drule{coeff-inl} \quad \drule{coeff-inr} \] \fi
  \drules[coeff]{$ \textcolor{\coeffectcolor}{ \textcolor{\coeffectcolor}{\gamma} }\! \cdot \! \Gamma   \vdash_{\mathit{coeff} }  \ottnt{M}  \ottsym{:}  \ottnt{B}$}{computation coeffect
    typing}{}
\[ \drule{coeff-abs}\ \drule[width=2in]{coeff-app}\ \drule{coeff-force}\ \]

\[ \drule[width=4in]{coeff-split} \ \drule[width=4in]{coeff-ret}\]

\[ \drule[width=3in]{coeff-letin} \ \drule[width=3in]{coeff-cpair}  \]

\[ \drule{coeff-fst}\ \drule{coeff-snd} \ \drule{coeff-csub} \]

\ifextended \[ \drule{coeff-sequence}\ \drule{coeff-case} \] \fi
\caption{CBPV type system with coeffect tracking}
\label{fig:cbpv-coeffects}
\Description{}
\end{figure}
\fi

Similarly to the previous section, our type system in this section is general
across coeffects and can be specialized via the choice of semiring and
preorder. For example, if we are only concerned with relevance analysis
(i.e. determining which of a functions inputs are relevant to computation)
then we might use a semiring with two elements: $ \textcolor{\coeffectcolor}{0} $ marks inputs that are
known to be unused and $ \textcolor{\coeffectcolor}{1} $ is for elements that may or may not be
needed. Or, in the case of information flow, then we might use a semiring
where $ \textcolor{\coeffectcolor}{0} $ marks secret inputs and $ \textcolor{\coeffectcolor}{1} $ marks public information; only
the latter may influence the result of the computation.

We would also like to use coeffects to track resource usage.  However, as we
discuss in detail below, this general semantics does not provide a satisfying
account of resource usage and requires further refinement in the next
section. Therefore, we first describe the general semantics in terms of the
resource usage coeffect, so that we can prepare for this discussion.

In the case of resource usage, grades bound the \emph{uses} of variables, as
in bounded linear logic, and come from the natural number semiring with the
usual addition and multiplication operators. The additive and multiplicative
identity elements of this semiring mark $ \textcolor{\coeffectcolor}{0} $ and at most $ \textcolor{\coeffectcolor}{1} $ (affine)
use of a variable respectively, and the addition and multiplication semiring
operations calculate the total number of times each variable is used in the
program.

As in many systems for bounded linear logic, $ \textcolor{\coeffectcolor}{ \textcolor{\coeffectcolor}{q}_{{\mathrm{1}}} }\;  \textcolor{\coeffectcolor}{\mathop{\leq_{\mathit{co} } } } \; \textcolor{\coeffectcolor}{ \textcolor{\coeffectcolor}{q}_{{\mathrm{2}}} } $ indicates that
$\textcolor{blue}{\textcolor{\coeffectcolor}{q}_{{\mathrm{1}}}}$ is \emph{less precise} or \emph{less restrictive}
than $\textcolor{blue}{\textcolor{\coeffectcolor}{q}_{{\mathrm{2}}}}$. When counting variable usage, this has the
opposite order from the usual one---we have $ \textcolor{\coeffectcolor}{ \ottsym{3} }\;  \textcolor{\coeffectcolor}{\mathop{\leq_{\mathit{co} } } } \; \textcolor{\coeffectcolor}{ \ottsym{2} } $ because allowing at
most 3 uses is less restrictive than at most 2. With other coeffects, such as
security levels, this ordering has a more intuitive interpretation: a higher
grade corresponds to a higher security level, which is more restrictive than a
low security level.

Like the effect system with subeffecting,
this type system includes a rule for \emph{subcoeffecting}: if
a judgment holds with some annotation $\textcolor{blue}{\textcolor{\coeffectcolor}{q}_{{\mathrm{2}}}}$ on a variable
in the context, then it is also derivable with any $ \textcolor{\coeffectcolor}{ \textcolor{\coeffectcolor}{q}_{{\mathrm{1}}} }\;  \textcolor{\coeffectcolor}{\mathop{\leq_{\mathit{co} } } } \; \textcolor{\coeffectcolor}{ \textcolor{\coeffectcolor}{q}_{{\mathrm{2}}} } $.
For example, we can weaken a judgment
that a computation makes zero ($ \textcolor{\coeffectcolor}{0} $) uses of some variable to observe at most one
use (affine) or any other number. This corresponds to the usual weakening
lemma from typed $\lambda$-calculi.

Again, as in the effect section, including a preorder with the semiring allows for
imprecision,
needed when analyzing branching computations. For example, if one
branch requires $ \textcolor{\coeffectcolor}{1} $ use of a variable $\ottmv{x}$, but the other branch
requires $ \textcolor{\coeffectcolor}{0} $ uses, the system will record that the program must have
the resources to use $\ottmv{x}$ at least once, because
$ \textcolor{\coeffectcolor}{  \textcolor{\coeffectcolor}{1}  }\;  \textcolor{\coeffectcolor}{\mathop{\leq_{\mathit{co} } } } \; \textcolor{\coeffectcolor}{  \textcolor{\coeffectcolor}{0}  } $, in a semiring where $ \textcolor{\coeffectcolor}{1} $ corresponds to affine usage.
This relation is dual to the preorder's role in
the effect system---if one branch ticks once and the other does not tick,
then the system will record at most one tick.
In both cases, replacing the ordering with the
discrete preorder means that the type system must be precise and would
reject both of these examples.

The type system uses a \emph{grade vector} $\textcolor{\coeffectcolor}{\gamma}$, a comma-separated list of
grades, to represent the annotations for the variables in a typing context.
When combined with a typing context $\Gamma$, written $ \textcolor{\coeffectcolor}{ \textcolor{\coeffectcolor}{\gamma} }\! \cdot \! \Gamma $,
the grade vector must have the same length as $\Gamma$.
We extend a combined grade vector and typing context simultaneously with the
notation $  \textcolor{\coeffectcolor}{ \textcolor{\coeffectcolor}{\gamma} }\! \cdot \! \Gamma    \mathop{,}    \ottmv{x}  :^{\textcolor{\coeffectcolor}{ \textcolor{\coeffectcolor}{q} } }  \ottnt{A}  $, equivalent to
$ \textcolor{\coeffectcolor}{ \ottsym{(}   \textcolor{\coeffectcolor}{ \textcolor{\coeffectcolor}{\gamma} \mathop{,} \textcolor{\coeffectcolor}{q} }   \ottsym{)} }\! \cdot \! \ottsym{(}   \Gamma   \mathop{,}   \ottmv{x}  \ottsym{:}  \ottnt{A}   \ottsym{)} $.

The grade vector written $ \textcolor{\coeffectcolor}{\overline{0} } $ contains
only zeros and is used where its length can be inferred from context.
Grade vectors of the same length can be added together pointwise, written
$ \textcolor{\coeffectcolor}{ \textcolor{\coeffectcolor}{\gamma}_{{\mathrm{1}}} \ottsym{+} \textcolor{\coeffectcolor}{\gamma}_{{\mathrm{2}}} } $, and compared pointwise, written $ \textcolor{\coeffectcolor}{ \textcolor{\coeffectcolor}{\gamma}_{{\mathrm{1}}} }\; \textcolor{\coeffectcolor}{\mathop{\leq_{\mathit{co} } } } \; \textcolor{\coeffectcolor}{ \textcolor{\coeffectcolor}{\gamma}_{{\mathrm{2}}} } $.
Grade vectors can also be pointwise scaled, written $ \textcolor{\coeffectcolor}{ \textcolor{\coeffectcolor}{q} \cdot \textcolor{\coeffectcolor}{\gamma} } $.

The basis of this system is \rref{coeff-var}. When
introducing a variable $\ottmv{x}$, the context must grade $\ottmv{x}$ with $ \textcolor{\coeffectcolor}{1} $,
indicating that it is used once.
No other variables in the context should affect
the typing judgement, so they must have grade $ \textcolor{\coeffectcolor}{0} $.
Similarly, the unit value $\ottsym{()}$ can make no demands on the environment, so \rref{coeff-unit} requires
that all variables in the typing context be graded $ \textcolor{\coeffectcolor}{0} $.

In \rref{coeff-thunk} \ifextended , \rref{coeff-inl},
\rref{coeff-inr}, \fi  and \rref{coeff-force}, there is a single subterm that
makes exactly the same demands on its environment as the term in the
conclusion, so we use the same grade vector in the conclusion and the premise.
\ifextended For example, in a sum type, $\ottkw{inl} \, \ottnt{V}$ makes the same demands as
$\ottnt{V}$. \fi

In other rules, the term in the conclusion has multiple subterms, so we
combine the demands made by each. In \rref{coeff-pair},
the subterms both get evaluated and do not directly
interact, so we combine their grade vectors via simple pointwise addition.
Conversely, with negative products, the two subterms must use the same resources,
so we use the same grade
vector in each premise and the conclusion.
Intuitively, this is because we can only ever
project out one subterm from a computation pair (see \rref{coeff-fst} and
\rref{coeff-snd}), so the projected term will make all the same demands on
the environment as the pair.

In \rref{coeff-abs}, we know from the premise that $\ottnt{M}$ will require a
grade of $\textcolor{\coeffectcolor}{q}$ on $\ottmv{x}$, so we store that grade as an annotation on
$\ottmv{x}$ in the term syntax. For flexibility, we allow the annotation in the
type, $\textcolor{\coeffectcolor}{q}'$, to be a less precise approximation of $\textcolor{\coeffectcolor}{q}$. (This
expressiveness is useful for the translation results in the next section.
Note that subcoeffecting is not sufficient as it cannot allow the annotation 
on the $\lambda$ to differ from the annotation on the function type.)
Both the premise and the
conclusion make the same demands on the variables in $\Gamma$, so $\textcolor{\coeffectcolor}{\gamma}$
is otherwise the same in both.

In some rules, we must
combine the grade vectors of subterms using both scaling and addition.
For example, in \rref{coeff-app}, $\textcolor{\coeffectcolor}{\gamma}_{{\mathrm{1}}}$ denotes
the demands the operator $\ottnt{M}$ makes on the environment, and $\textcolor{\coeffectcolor}{\gamma}_{{\mathrm{2}}}$
denotes the demands the argument $\ottnt{V}$ makes.
$\ottnt{M}$ has type $ \ottnt{A} ^{\textcolor{\coeffectcolor}{ \textcolor{\coeffectcolor}{q} } } \rightarrow  \ottnt{B} $, indicating that
when it is reduced to some terminal $ \lambda  \ottmv{x} ^{\textcolor{\coeffectcolor}{ \textcolor{\coeffectcolor}{q}' } }. \ottnt{M'} $, then $\ottnt{M'}$ will
require $\ottmv{x}$ to have a grade of $\textcolor{\coeffectcolor}{q}'$, where $ \textcolor{\coeffectcolor}{ \textcolor{\coeffectcolor}{q} }\;  \textcolor{\coeffectcolor}{\mathop{\leq_{\mathit{co} } } } \; \textcolor{\coeffectcolor}{ \textcolor{\coeffectcolor}{q}' } $.
This means we must scale
$\textcolor{\coeffectcolor}{\gamma}_{{\mathrm{1}}}$ by $\textcolor{\coeffectcolor}{q}$ before adding it to $\textcolor{\coeffectcolor}{\gamma}_{{\mathrm{2}}}$ to calculate
the total demands that $\ottnt{M} \, \ottnt{V}$ makes on its environment.

\ifextended
\Rref{coeff-split} follows a similar pattern. In this rule, we
require a grade of $\textcolor{\coeffectcolor}{q}$ on $\ottmv{x_{{\mathrm{1}}}}$ and $\ottmv{x_{{\mathrm{2}}}}$ in $\ottnt{N}$, so we scale
$\textcolor{\coeffectcolor}{\gamma}_{{\mathrm{1}}}$, or the demands made by $\ottnt{V}$, by $\textcolor{\coeffectcolor}{q}$. (Subcoeffecting
allows us to use the same grade for $\ottmv{x_{{\mathrm{1}}}}$ and $\ottmv{x_{{\mathrm{2}}}}$ even though
the exact demands $\ottnt{N}$ makes on each may be different.)

In \rref{coeff-case}, we additionally require that $ \textcolor{\coeffectcolor}{ \textcolor{\coeffectcolor}{q} }\;  \textcolor{\coeffectcolor}{\mathop{\leq_{\mathit{co} } } } \; \textcolor{\coeffectcolor}{  \textcolor{\coeffectcolor}{1}  } $.
We need to evaluate $\ottnt{V}$ to either $\ottkw{inl} \, \ottnt{V_{{\mathrm{1}}}}$ or $\ottkw{inr} \, \ottnt{V_{{\mathrm{2}}}}$ for some
$\ottnt{V_{{\mathrm{1}}}}$ or $\ottnt{V_{{\mathrm{2}}}}$ in order for this branching to be well-defined. In
our resource usage example, we can interpret this as requiring
at least $ \textcolor{\coeffectcolor}{1} $ copy of $\ottnt{V}$ in order to proceed.
\fi

In the effect system, we annotate the type $ \ottkw{U}_{\color{\effectcolor}{ \textcolor{\effectcolor}{\phi} } }\;  \ottnt{B} $ with the effect of
the suspended computation. In the coeffect system, we dually annotate the returner
type $ \ottkw{F}_{\color{\coeffectcolor}{ \textcolor{\coeffectcolor}{q} } }\;  \ottnt{A} $. In our resource usage example, the $\textcolor{\coeffectcolor}{q}$
indicates that we require enough resources from the environment to
produce $\textcolor{\coeffectcolor}{q}$ copies of a value. For example, $ \ottkw{return} _{\textcolor{\coeffectcolor}{ \ottsym{3} } }\;  \ottnt{V} $ indicates
that we require the resources to create $\ottsym{3}$ copies of $\ottnt{V}$.
Therefore, \rref{coeff-ret} scales the demands needed to create
$\ottnt{V}$ by $\textcolor{\coeffectcolor}{q}$.

In \rref{coeff-letin}, $\ottnt{M}$ has returner type $ \ottkw{F}_{\color{\coeffectcolor}{ \textcolor{\coeffectcolor}{q}_{{\mathrm{1}}} } }\;  \ottnt{A} $, and its
result value has been scaled by $\textcolor{\coeffectcolor}{q}_{{\mathrm{1}}}$. However, the expression includes
another scaling annotation $\textcolor{\coeffectcolor}{q}_{{\mathrm{2}}}$, that allows duplication of the computation
$\ottnt{M}$ itself. If $\textcolor{\coeffectcolor}{\gamma}_{{\mathrm{1}}}$ denotes the demands $\ottnt{M}$ makes on its environment,
$ \textcolor{\coeffectcolor}{ \textcolor{\coeffectcolor}{q}_{{\mathrm{1}}}   \cdot   \textcolor{\coeffectcolor}{q}_{{\mathrm{2}}} } $ denotes the grade $\ottnt{N}$ requires $\ottmv{x}$ to have, and $\textcolor{\coeffectcolor}{\gamma}_{{\mathrm{2}}}$
denotes the demands $\ottnt{N}$ makes from the rest of the environment, then we need
$ \textcolor{\coeffectcolor}{  \textcolor{\coeffectcolor}{ \textcolor{\coeffectcolor}{q}_{{\mathrm{2}}} \cdot \textcolor{\coeffectcolor}{\gamma}_{{\mathrm{1}}} }  \ottsym{+} \textcolor{\coeffectcolor}{\gamma}_{{\mathrm{2}}} } $ to type the entire term.

The scaling annotations in $ \ottkw{return} _{\textcolor{\coeffectcolor}{ \textcolor{\coeffectcolor}{q} } }\;  \ottnt{V} $ and $ \ottmv{x}  \leftarrow^{\textcolor{\coeffectcolor}{ \textcolor{\coeffectcolor}{q} } }  \ottnt{M} \ \ottkw{in}\  \ottnt{N} $ increase
the expressiveness of the language and are required for the translation of a
CBV $\lambda$-calculus to CBPV described in
Section~\ref{cbv-coeffects-translation}. Because CBV is strict, when
translating an application, we must use a let binding to evaluate the
translated argument before applying the translated function to it.
However, the function may require a particular grade $\textcolor{\coeffectcolor}{q}$ on its argument,
so we must be able to scale this computation. Similarly, to translate the graded CBV
comonadic type, we need to be able to duplicate values.

The two subsumption \rref{coeff-vsub,coeff-csub} allow for subcoeffecting.

\subsection{General Instrumented Operational Semantics and Coeffect Soundness}
\label{sec:coeffect-soundness}

\ifextended 
\begin{figure}
\drules[eval-coeff-val]{$  \textcolor{\coeffectcolor}{ \textcolor{\coeffectcolor}{\gamma} }\! \cdot \! \rho   \vdash_{\mathit{coeff} }  \ottnt{V}  \Downarrow  \ottnt{W} $}{Value rules}{}
\[ \drule[width=3in]{eval-coeff-val-var} \]

\[ \drule{eval-coeff-val-thunk}\!\drule[width=3in]{eval-coeff-val-vpair} \]

\[ \drule{eval-coeff-val-inl} \quad \drule{eval-coeff-val-inr} \]

\[ \drule[width=3in]{eval-coeff-val-vsub} \]
\caption{Instrumented operational semantics (values)}
\Description{Rules for the instrumented operational semantics}
\label{fig:eval-coeff-val}
\Description{}
\end{figure}
\begin{figure}
\drules[eval-coeff-comp]{$  \textcolor{\coeffectcolor}{ \textcolor{\coeffectcolor}{\gamma} }\! \cdot \! \rho   \vdash_{\mathit{coeff} }  \ottnt{M}  \Downarrow  \ottnt{T} $}{Computation rules}{}
\[  \drule{eval-coeff-comp-abs}\   \drule{eval-coeff-comp-cpair} \]

\[  \drule[width=3in]{eval-coeff-comp-app-abs}\!\drule{eval-coeff-comp-split}\]

\[  \drule{eval-coeff-comp-return}\ \drule{eval-coeff-comp-letin-ret} \]

\[  \drule[width=4in]{eval-coeff-comp-force-thunk}    \]

\[  \drule{eval-coeff-comp-fst}\      \drule{eval-coeff-comp-snd} \]

\[ \drule[width=5in]{eval-coeff-comp-sequence} \]

\[ \drule[width=5in]{eval-coeff-comp-case-inl}  \]

\[ \drule[width=5in]{eval-coeff-comp-case-inr}  \]

\[ \drule[width=5in]{eval-coeff-comp-csub} \]

\caption{Instrumented operational semantics (computations)}
\label{fig:eval-coeff-comp}
\Description{Rules for the instrumented operational semantics}
\end{figure}
\else 
\begin{figure}
\drules[eval-coeff-val]{$  \textcolor{\coeffectcolor}{ \textcolor{\coeffectcolor}{\gamma} }\! \cdot \! \rho   \vdash_{\mathit{coeff} }  \ottnt{V}  \Downarrow  \ottnt{W} $}{Value rules}{}
\[ \drule[width=3in]{eval-coeff-val-var} \quad \drule{eval-coeff-val-unit} \]

\[ \drule{eval-coeff-val-thunk}\!\drule[width=3in]{eval-coeff-val-vpair} \]

\[ \drule[width=3in]{eval-coeff-val-vsub} \]

\drules[eval-coeff-comp]{$  \textcolor{\coeffectcolor}{ \textcolor{\coeffectcolor}{\gamma} }\! \cdot \! \rho   \vdash_{\mathit{coeff} }  \ottnt{M}  \Downarrow  \ottnt{T} $}{Computation rules}{}
\[  \drule{eval-coeff-comp-abs}\   \drule{eval-coeff-comp-cpair} \]

\[  \drule[width=3in]{eval-coeff-comp-app-abs}\!\!\!\!\!\drule{eval-coeff-comp-split}\]

\[  \drule{eval-coeff-comp-return}\ \drule[width=4in]{eval-coeff-comp-force-thunk} \]

\[ \drule[width=3in]{eval-coeff-comp-letin-ret} \ \drule{eval-coeff-comp-fst}\]

\[       \drule[width=5in]{eval-coeff-comp-snd}  \ \drule[width=5in]{eval-coeff-comp-csub}\]

\caption{Instrumented operational semantics}
\label{fig:eval-coeff-val}
\label{fig:eval-coeff-comp}
\Description{Rules for the instrumented operational semantics}
\end{figure}
\fi

Next, we develop an \link{general/semantics.v}{EvalVal,EvalComp}{instrumented
  operational semantics}, shown in Figure~\ref{fig:eval-coeff-val}\ifextended
  and Figure~\ref{fig:eval-coeff-comp}\fi, that tracks coeffects using an
environment $\rho$, which maps variables to closed values, and a grade
vector $\textcolor{\coeffectcolor}{\gamma}$ of equal length, which implicitly maps variables to their
coeffects.  As before, we extend both a grade vector and corresponding
environment simultaneously with the notation
$  \textcolor{\coeffectcolor}{ \textcolor{\coeffectcolor}{\gamma} }\! \cdot \! \rho    \mathop{,} \;   \ottmv{x}   \mapsto ^{\textcolor{\coeffectcolor}{ \textcolor{\coeffectcolor}{q} } }  \ottnt{W}  $, equivalent to
$ \textcolor{\coeffectcolor}{ \ottsym{(}   \textcolor{\coeffectcolor}{ \textcolor{\coeffectcolor}{\gamma} \mathop{,} \textcolor{\coeffectcolor}{q} }   \ottsym{)} }\! \cdot \! \ottsym{(}   \rho   \mathop{,}   \ottmv{x}  \mapsto  \ottnt{W}   \ottsym{)} $.

We also use $\ottnt{W}$ as a metavariable
for \emph{closed} values and
$\ottnt{T}$ as a metavariable for \emph{closed terminal} computations. However,
closed terminals include coeffects here. They have the form
$ \ottkw{return} _{\textcolor{\coeffectcolor}{ \textcolor{\coeffectcolor}{q} } }  \ottnt{W} $, $ \mathbf{clo}(   \textcolor{\coeffectcolor}{ \textcolor{\coeffectcolor}{\gamma} }\! \cdot \! \rho  ,   \lambda  \ottmv{x} ^{\textcolor{\coeffectcolor}{ \textcolor{\coeffectcolor}{q} } }. \ottnt{M}   ) $, or
$ \mathbf{clo}(   \textcolor{\coeffectcolor}{ \textcolor{\coeffectcolor}{\gamma} }\! \cdot \! \rho  ,   \langle  \ottnt{M_{{\mathrm{1}}}} , \ottnt{M_{{\mathrm{2}}}}  \rangle   ) $, where $ \mathbf{clo}(   \textcolor{\coeffectcolor}{ \textcolor{\coeffectcolor}{\gamma} }\! \cdot \! \rho  ,  \ottnt{M}  ) $ denotes the
\emph{closure} of $\ottnt{M}$ under $ \textcolor{\coeffectcolor}{ \textcolor{\coeffectcolor}{\gamma} }\! \cdot \! \rho $. The grade vector in
the closure indicates the demands on the variables used by $\ottnt{M}$.

Unlike our instrumented operational semantics for effects, which calculates
the exact effect of a computation, this semantics cannot track
coeffects with precision.
For example, suppose we have a term $ \lambda  \ottmv{x} ^{\textcolor{\coeffectcolor}{  \textcolor{\coeffectcolor}{1}  } }. \ottnt{M} $
where $\ottnt{M}$ is a computation that both branches on its argument and uses it in at exactly one branch, such as
$ \ottkw{case}_{\textcolor{\coeffectcolor}{  \textcolor{\coeffectcolor}{1}  } }\;  \ottmv{x} \; \ottkw{of}\; \ottkw{inl} \; \ottmv{x_{{\mathrm{1}}}}  \rightarrow\;  \ottkw{return} \, \ottmv{x}  ;  \ottkw{inr} \; \ottmv{x_{{\mathrm{2}}}}  \rightarrow\;  \ottkw{return} \, \ottkw{inr} \, \ottsym{()} $.
\ifextended\else
\footnote{The system in our extended version~\cite{extended-version} and Coq development includes sums.}
\fi
What should this step to? If provided with an argument of the form $\ottkw{inl} \, \ottmv{y}$,
it should step to $ \mathbf{clo}(   \ottmv{x}   \mapsto ^{\textcolor{\coeffectcolor}{  \textcolor{\coeffectcolor}{1}  } }  \ottkw{inl} \, \ottmv{y}  ,   \lambda  \ottmv{x} ^{\textcolor{\coeffectcolor}{  \textcolor{\coeffectcolor}{1}  } }. \ottnt{M}   ) $.
If provided with an argument of the form $\ottkw{inr} \, \ottmv{y}$, it should step to
$ \mathbf{clo}(   \ottmv{x}   \mapsto ^{\textcolor{\coeffectcolor}{  \textcolor{\coeffectcolor}{0}  } }  \ottkw{inr} \, \ottmv{y}  ,   \lambda  \ottmv{x} ^{\textcolor{\coeffectcolor}{  \textcolor{\coeffectcolor}{1}  } }. \ottnt{M}   ) $. But, if this term is the entire program,
it is not clear what it should step to.
In general, depending on the argument, the body of a function
$ \lambda  \ottmv{x} ^{\textcolor{\coeffectcolor}{ \textcolor{\coeffectcolor}{q} } }. \ottnt{M} $ may require a different exact grade
on $\ottmv{x}$; all we know from the typing judgement is that $\textcolor{\coeffectcolor}{q}$ must be a bound on that usage.
We cannot write a precise rule for evaluating abstractions
to their closed terminal forms, because we do not have access to the argument yet
when doing that evaluation.

Therefore, as in the typing rules, the operational semantics also includes rules
for subcoeffecting, \rref{eval-coeff-val-vsub,eval-coeff-comp-csub}. These rules say
that if we can step a term with grades given by $\textcolor{\coeffectcolor}{\gamma}$ attached to the environment,
then we can step it with $\textcolor{\coeffectcolor}{\gamma}'$ for any $ \textcolor{\coeffectcolor}{ \textcolor{\coeffectcolor}{\gamma}' }\; \textcolor{\coeffectcolor}{\mathop{\leq_{\mathit{co} } } } \; \textcolor{\coeffectcolor}{ \textcolor{\coeffectcolor}{\gamma} } $, i.e., any less
precise accounting.

As in the semantics for CBPV without coeffects, we define ``evaluation'' of
values using the given environment (see
Figure~\ref{fig:eval-coeff-val}).
These rules mirror the typing rules: \rref{eval-coeff-val-var} requires the
evaluating variable to have $ \textcolor{\coeffectcolor}{1} $ as its corresponding grade and all other variables
to have $ \textcolor{\coeffectcolor}{0} $;
\rref{eval-coeff-val-unit}
requires that every variable be graded with 0;
\rref{eval-coeff-val-thunk} simply includes the grade vector in the closure along
with the environment, and \rref{eval-coeff-val-vpair}
sums the grades needed to evaluate subterms to their
closures.


Figure~\ref{fig:eval-coeff-comp} also shows
computations.
\Rref{eval-coeff-comp-abs,eval-coeff-comp-force-thunk,eval-coeff-comp-cpair,eval-coeff-comp-fst,eval-coeff-comp-snd} are largely the same as before,
just with the inclusion of grade vectors along with environments.
\ifextended
\Rref{eval-coeff-comp-sequence} simply sums the vectors required to evaluate each
subterm.
The sum type elimination rules scale the demands made by the term being
eliminated by $\textcolor{\coeffectcolor}{q}$ before adding them to the demands needed to evaluate
the rest of the computation, as in the typing rules. They also require
that $ \textcolor{\coeffectcolor}{ \textcolor{\coeffectcolor}{q} }\;  \textcolor{\coeffectcolor}{\mathop{\leq_{\mathit{co} } } } \; \textcolor{\coeffectcolor}{  \textcolor{\coeffectcolor}{1}  } $ for the branching behavior to be well-defined, as
in the typing rules. Intuitively, in a resource counting context, if we
have 0 copies of a value, we should not be able to use it to determine
which branch to take. \fi

In \rref{eval-coeff-comp-return}, we scale the grade needed to evaluate the
subterm to its closure by $\textcolor{\coeffectcolor}{q}$. In the elimination
\rref{eval-coeff-comp-app-abs,eval-coeff-comp-letin-ret,eval-coeff-comp-split},
if we are eliminating a value $\ottnt{V}$ and binding it to a variable
$\ottmv{x}$ with a grade $\textcolor{\coeffectcolor}{q}$ for use in some computation $\ottnt{M}$,
we must scale the grade vector needed to evaluate $\ottnt{V}$
by $\textcolor{\coeffectcolor}{q}$ before adding it to the grade vector needed to continue with $\ottnt{M}$, as
in the typing rules.

We prove a coeffect soundness theorem stating that if a term is well-typed
with some grade vector $\textcolor{\coeffectcolor}{\gamma}$, then given $\textcolor{\coeffectcolor}{\gamma}$ and some
environment $\rho$ that provides values of the correct type for all free
variables, it can evaluate to a terminal. Because both values and computations
make demands on their inputs, we state this property for both.
We formalize the requirement on $\rho$ as $\Gamma  \vDash  \rho$ in our logical
relation below, and this theorem follows immediately from the fundamental lemma.

\begin{theorem}[\link{general/soundness.v}{soundness}{Coeffect soundness}]
\label{thm:general-soundness}
    Let $\Gamma$ be a context and $\rho$ an environment mapping all variables
    in the domain of $\Gamma$ to closed values of the expected type, such 
that $\Gamma  \vDash  \rho$. Then:
\begin{enumerate}
\item If $ \textcolor{\coeffectcolor}{ \textcolor{\coeffectcolor}{\gamma} }\! \cdot \! \Gamma   \vdash_{\mathit{coeff} }  \ottnt{V}  \ottsym{:}  \ottnt{A}$ then
  $  \textcolor{\coeffectcolor}{ \textcolor{\coeffectcolor}{\gamma} }\! \cdot \! \rho   \vdash_{\mathit{coeff} }  \ottnt{V}  \Downarrow  \ottnt{W} $ for some closed value $\ottnt{W}$.
\item If $ \textcolor{\coeffectcolor}{ \textcolor{\coeffectcolor}{\gamma} }\! \cdot \! \Gamma   \vdash_{\mathit{coeff} }  \ottnt{M}  \ottsym{:}  \ottnt{B}$ then
  $  \textcolor{\coeffectcolor}{ \textcolor{\coeffectcolor}{\gamma} }\! \cdot \! \rho   \vdash_{\mathit{coeff} }  \ottnt{M}  \Downarrow  \ottnt{T} $ for some closed terminal
  computation $\ottnt{T}$.
\end{enumerate}
\end{theorem}

The proof of coeffect soundness is similar to the proof
of effect soundness, and requires a similar logical relation.

\begin{definition}[\link{general/semtyping.v}{LRV,LRC}{CBPV with General Coeffects: Logical Relation}]
\[
\begin{array}{lcl@{\ }l}
 \mathcal{W}\llbracket  \ottkw{U} \, \ottnt{B}  \rrbracket    &=& \{\  \mathbf{clo}(  \textcolor{\coeffectcolor}{ \textcolor{\coeffectcolor}{\gamma} }\! \cdot \! \rho  , \{  \ottnt{M}  \} )  &|\
                           \ottsym{(}   \textcolor{\coeffectcolor}{ \textcolor{\coeffectcolor}{\gamma} }\! \cdot \! \rho   \ottsym{,}  \ottnt{M}  \ottsym{)} \, \in \,  \mathcal{M}\llbracket  \ottnt{B}  \rrbracket \ \} \\
 \mathcal{W}\llbracket  \ottkw{unit}  \rrbracket    &=& \{\ ()\ \}                                  \\
 \mathcal{W}\llbracket   \ottnt{A_{{\mathrm{1}}}} \times \ottnt{A_{{\mathrm{2}}}}   \rrbracket  &=& \{\ \ottsym{(}  \ottnt{W_{{\mathrm{1}}}}  \ottsym{,}  \ottnt{W_{{\mathrm{2}}}}  \ottsym{)} &|\
                              \ottnt{W_{{\mathrm{1}}}} \, \in \,  \mathcal{W}\llbracket  \ottnt{A_{{\mathrm{1}}}}  \rrbracket  \ \textit{and}\  \ottnt{W_{{\mathrm{2}}}} \, \in \,  \mathcal{W}\llbracket  \ottnt{A_{{\mathrm{2}}}}  \rrbracket   \}\\
\ifextended
 \mathcal{W}\llbracket  \ottnt{A_{{\mathrm{1}}}}  \ottsym{+}  \ottnt{A_{{\mathrm{2}}}}  \rrbracket  &=& \{\ \ottkw{inl} \, \ottnt{W}     &|\
                             \ottnt{W} \, \in \,  \mathcal{W}\llbracket  \ottnt{A_{{\mathrm{1}}}}  \rrbracket  \}\ \cup
                         \{\ \ottkw{inr} \, \ottnt{W} |\ \ottnt{W} \, \in \,  \mathcal{W}\llbracket  \ottnt{A_{{\mathrm{2}}}}  \rrbracket  \}  \\ \fi
\\
 \mathcal{T}\llbracket   \ottkw{F}_{\color{\coeffectcolor}{ \textcolor{\coeffectcolor}{q} } }\;  \ottnt{A}   \rrbracket     &=&  \{\  \ottkw{return} _{\textcolor{\coeffectcolor}{ \textcolor{\coeffectcolor}{q} } }  \ottnt{W}   &|\
                               \ottnt{W} \, \in \,  \mathcal{W}\llbracket  \ottnt{A}  \rrbracket  \ \} \\
 \mathcal{T}\llbracket   \ottnt{A} ^{\textcolor{\coeffectcolor}{ \textcolor{\coeffectcolor}{q} } } \rightarrow  \ottnt{B}   \rrbracket  &=&  \{\  \mathbf{clo}(   \textcolor{\coeffectcolor}{ \textcolor{\coeffectcolor}{\gamma} }\! \cdot \! \rho  ,   \lambda  \ottmv{x} ^{\textcolor{\coeffectcolor}{ \textcolor{\coeffectcolor}{q} } }. \ottnt{M}   )  &|\
                \textit{for all}\ \ottnt{W} \, \in \,  \mathcal{W}\llbracket  \ottnt{A}  \rrbracket ,
                \ottsym{(}  \ottsym{(}    \textcolor{\coeffectcolor}{ \textcolor{\coeffectcolor}{\gamma} }\! \cdot \! \rho    \mathop{,} \;   \ottmv{x}   \mapsto ^{\textcolor{\coeffectcolor}{ \textcolor{\coeffectcolor}{q} } }  \ottnt{W}    \ottsym{)}  \ottsym{,}  \ottnt{M}  \ottsym{)} \, \in \,  \mathcal{M}\llbracket  \ottnt{B}  \rrbracket \ \} \\
 \mathcal{T}\llbracket   \ottnt{B_{{\mathrm{1}}}}   \mathop{\&}   \ottnt{B_{{\mathrm{2}}}}   \rrbracket   &=& \{\  \mathbf{clo}(   \textcolor{\coeffectcolor}{ \textcolor{\coeffectcolor}{\gamma} }\! \cdot \! \rho  ,   \langle  \ottnt{M_{{\mathrm{1}}}} , \ottnt{M_{{\mathrm{2}}}}  \rangle   )    &|\
                  \ottsym{(}   \textcolor{\coeffectcolor}{ \textcolor{\coeffectcolor}{\gamma} }\! \cdot \! \rho   \ottsym{,}  \ottnt{M_{{\mathrm{1}}}}  \ottsym{)} \, \in \,  \mathcal{M}\llbracket  \ottnt{B_{{\mathrm{1}}}}  \rrbracket  \ \textit{and}\  \ottsym{(}   \textcolor{\coeffectcolor}{ \textcolor{\coeffectcolor}{\gamma} }\! \cdot \! \rho   \ottsym{,}  \ottnt{M_{{\mathrm{2}}}}  \ottsym{)} \, \in \,  \mathcal{M}\llbracket  \ottnt{B_{{\mathrm{2}}}}  \rrbracket  \ \} \\
\\
\textit{Closures}\\
 \mathcal{V}\llbracket  \ottnt{A}  \rrbracket  &=& \{\ ( \textcolor{\coeffectcolor}{ \textcolor{\coeffectcolor}{\gamma} }\! \cdot \! \rho  , \ottnt{V}) &|\
                       \textcolor{\coeffectcolor}{ \textcolor{\coeffectcolor}{\gamma} }\! \cdot \! \rho   \vdash_{\mathit{coeff} }  \ottnt{V}  \Downarrow  \ottnt{W}  \ \textit{and}\  \ottnt{W} \, \in \,  \mathcal{W}\llbracket  \ottnt{A}  \rrbracket  \ \} \\
 \mathcal{M}\llbracket  \ottnt{B}  \rrbracket  &=& \{\ ( \textcolor{\coeffectcolor}{ \textcolor{\coeffectcolor}{\gamma} }\! \cdot \! \rho  , \ottnt{M}) &|\
                       \textcolor{\coeffectcolor}{ \textcolor{\coeffectcolor}{\gamma} }\! \cdot \! \rho   \vdash_{\mathit{coeff} }  \ottnt{M}  \Downarrow  \ottnt{T}  \ \textit{and}\  \ottnt{T} \, \in \,  \mathcal{T}\llbracket  \ottnt{B}  \rrbracket  \ \} \\
\\
\end{array}
\]
\end{definition}

\begin{definition}[\link{general/semtyping.v}{SemVWt,SemCWt}{CBPV with General Coeffects: Semantic Typing}]
\[
\begin{array}{lcl@{\ }l}
\Gamma  \vDash  \rho &=&
  \multicolumn{2}{l}{\ottmv{x}  \ottsym{:}  \ottnt{A} \, \in \, \Gamma
  \ implies \ \textit{exists} \ \ottnt{W}, \  \ottmv{x}  \mapsto  \ottnt{W} \, \in \, \rho \ \textit{and}\  \ottnt{W} \, \in \,  \mathcal{W}\llbracket  \ottnt{A}  \rrbracket   }  \\
 \textcolor{\coeffectcolor}{ \textcolor{\coeffectcolor}{\gamma} }\! \cdot \! \Gamma   \vDash_{\mathit{coeff} }  \ottnt{V}  \ottsym{:}  \ottnt{A} &=& \multicolumn{2}{l}{ \textit{for all} \ \rho,
  \Gamma  \vDash  \rho \ implies \ \textit{exists} \ \ottnt{W}, \    \textcolor{\coeffectcolor}{ \textcolor{\coeffectcolor}{\gamma} }\! \cdot \! \rho   \vdash_{\mathit{coeff} }  \ottnt{V}  \Downarrow  \ottnt{W}  \ \textit{and}\  \ottnt{W} \, \in \,  \mathcal{W}\llbracket  \ottnt{A}  \rrbracket  } \\
 \textcolor{\coeffectcolor}{ \textcolor{\coeffectcolor}{\gamma} }\! \cdot \! \Gamma   \vDash_{\mathit{coeff} }  \ottnt{M}  \ottsym{:}  \ottnt{B} &=& \multicolumn{2}{l}{ \textit{for all} \ \rho,
  \Gamma  \vDash  \rho \ implies \ \ottsym{(}   \textcolor{\coeffectcolor}{ \textcolor{\coeffectcolor}{\gamma} }\! \cdot \! \rho   \ottsym{,}  \ottnt{M}  \ottsym{)} \, \in \,  \mathcal{M}\llbracket  \ottnt{B}  \rrbracket } \\
\end{array}
\]
\end{definition}

We can now state the fundamental lemma, which derives the
soundness theorem as a corollary.

\begin{theorem}[\link{general/soundness.v}{fundamental\_lemma}{Fundamental lemma: coeffect soundness}]
For all $\textcolor{\coeffectcolor}{\gamma}$, $\Gamma$, if $ \textcolor{\coeffectcolor}{ \textcolor{\coeffectcolor}{\gamma} }\! \cdot \! \Gamma   \vdash_{\mathit{coeff} }  \ottnt{V}  \ottsym{:}  \ottnt{A}$ then
$ \textcolor{\coeffectcolor}{ \textcolor{\coeffectcolor}{\gamma} }\! \cdot \! \Gamma   \vDash_{\mathit{coeff} }  \ottnt{V}  \ottsym{:}  \ottnt{A}$,
and
for all $\textcolor{\coeffectcolor}{\gamma}$, $\Gamma$, if $ \textcolor{\coeffectcolor}{ \textcolor{\coeffectcolor}{\gamma} }\! \cdot \! \Gamma   \vdash_{\mathit{coeff} }  \ottnt{M}  \ottsym{:}  \ottnt{B}$ then
$ \textcolor{\coeffectcolor}{ \textcolor{\coeffectcolor}{\gamma} }\! \cdot \! \Gamma   \vDash_{\mathit{coeff} }  \ottnt{M}  \ottsym{:}  \ottnt{B}$.
\end{theorem}

We can show \ref{thm:general-soundness} by unfolding the definitions of 
$ \textcolor{\coeffectcolor}{ \textcolor{\coeffectcolor}{\gamma} }\! \cdot \! \Gamma   \vDash_{\mathit{coeff} }  \ottnt{V}  \ottsym{:}  \ottnt{A}$ and $ \textcolor{\coeffectcolor}{ \textcolor{\coeffectcolor}{\gamma} }\! \cdot \! \Gamma   \vDash_{\mathit{coeff} }  \ottnt{M}  \ottsym{:}  \ottnt{B}$, which 
give us the desired evaluations.

\subsection{A Strange Semantics?}

The operational semantics and soundness proof in this section 
work for any instantiation of the coeffect semiring. However, this semantics has
strange implications for the resource usage coeffect. Here, the soundness
theorem should say that if $  \textcolor{\coeffectcolor}{ \textcolor{\coeffectcolor}{\gamma} }\! \cdot \! \rho   \vdash_{\mathit{coeff} }  \ottnt{M}  \Downarrow  \ottnt{T} $, then the evaluation
of $\ottnt{M}$ used its variables at most the number of times indicated by
$\textcolor{\coeffectcolor}{\gamma}$. If $\textcolor{\coeffectcolor}{\gamma}$ says that a variable $\ottmv{x}$ has grade $ \textcolor{\coeffectcolor}{0} $, then
there should never be a use of \rref{eval-coeff-val-var} with the variable $\ottmv{x}$.

But, on closer examination of the operational semantics, this is not exactly
what this soundness theorem implies. Consider the following example:
\[  \ottmv{x}  :^{\textcolor{\coeffectcolor}{  \textcolor{\coeffectcolor}{0}  } }  \ottkw{U} \, \ottsym{(}  \ottkw{F} \, \ottkw{unit}  \ottsym{)}   \vdash_{\mathit{coeff} }   \ottmv{z_{{\mathrm{1}}}}  \leftarrow^{\textcolor{\coeffectcolor}{  \textcolor{\coeffectcolor}{0}  } }  \ottmv{x}  \ottsym{!} \ \ottkw{in}\   \ottkw{return} _{\textcolor{\coeffectcolor}{  \textcolor{\coeffectcolor}{1}  } }\;  \ottsym{()}    \ottsym{:}   \ottkw{F}_{\color{\coeffectcolor}{  \textcolor{\coeffectcolor}{1}  } }\;  \ottkw{unit}  \]
$\ottmv{x}  \ottsym{!}$ does not contribute to the final result,
and the resources used in its evaluation are accordingly multiplied by $ \textcolor{\coeffectcolor}{0} $ when
we calculate the grade for $\ottmv{x}$ in the context. However, our
semantics evaluates $\ottmv{x}$ once here using \rref{eval-coeff-val-var},
violating the principle we described above.

More generally, we encounter this issue with any rule in the operational semantics
that scales resources based on some annotation in the terms. For example, in
\rref{eval-coeff-comp-app-abs}, the resources used by the evaluation of the
argument $\textcolor{\coeffectcolor}{\gamma}_{{\mathrm{2}}}$ are scaled by $\textcolor{\coeffectcolor}{q}$, the grade on the function
argument. The total resources of the application $\textcolor{\coeffectcolor}{\gamma}$ must equal
this scaled vector plus
$\textcolor{\coeffectcolor}{\gamma}_{{\mathrm{1}}}$, the resources used to evaluate the
function -- \emph{i.e.}, we must have $ \textcolor{\coeffectcolor}{ \textcolor{\coeffectcolor}{\gamma} } \equiv \textcolor{\coeffectcolor}{  \textcolor{\coeffectcolor}{ \textcolor{\coeffectcolor}{\gamma}_{{\mathrm{1}}} \ottsym{+}  \textcolor{\coeffectcolor}{ \textcolor{\coeffectcolor}{q} \cdot \textcolor{\coeffectcolor}{\gamma}_{{\mathrm{2}}} }  }  } $. What if $\textcolor{\coeffectcolor}{q}$ is $ \textcolor{\coeffectcolor}{0} $?  The resources
needed to compute the argument are then not accounted for in
$\textcolor{\coeffectcolor}{\gamma}$. This suggests that we should not evaluate the argument at all in
this case, so we need to adjust our operational semantics.

\section{CBPV and coeffects (Version 2: Resource  Tracking)}
\label{sec:resource-usage}

\begin{figure}
\drules[lin]{$ \textcolor{\coeffectcolor}{ \textcolor{\coeffectcolor}{\gamma} }\! \cdot \! \Gamma   \vdash_{\mathit{lin} }  \ottnt{M}  \ottsym{:}  \ottnt{B}$}{Modified typing rule}{}
\[ \drule[width=8in]{lin-letin}\]
\drules[eval-lin-comp]{$  \textcolor{\coeffectcolor}{ \textcolor{\coeffectcolor}{\gamma} }\! \cdot \! \rho   \vdash_{\mathit{lin} }  \ottnt{M}  \Downarrow  \ottnt{T} $}{New and modified computation rules}{}
\[ \drule[width=5in]{eval-lin-comp-app-abs-zero} \]

\[ \drule{eval-lin-comp-ret-zero} \ \ \drule[width=3in]{eval-lin-comp-split-zero} \]

\[ \drule[width=8in]{eval-lin-comp-letin-ret} \]
\caption{Typing rules and instrumented operational semantics for resource tracking}
\label{fig:eval-lin-comp-zero}
\label{fig:cbpv-lin}
\Description{Modified computation rules for the instrumented operational semantics}
\end{figure}

In this section, we discuss how, with a few additional axioms, we can modify
our instrumented operational semantics and type system to produce a better
model for resource tracking. Our goal is to ensure that we never evaluate
values and computations without including their resource usage in the final count. The
modifications that we discuss here are summarized in
Figure~\ref{fig:eval-lin-comp-zero}. We
use the judgements $ \textcolor{\coeffectcolor}{ \textcolor{\coeffectcolor}{\gamma} }\! \cdot \! \Gamma   \vdash_{\mathit{lin} }  \ottnt{M}  \ottsym{:}  \ottnt{B}$ and $ \textcolor{\coeffectcolor}{ \textcolor{\coeffectcolor}{\gamma} }\! \cdot \! \Gamma   \vdash_{\mathit{lin} }  \ottnt{V}  \ottsym{:}  \ottnt{A}$
to refer to the \link{resource/CBPV/typing.v}{}{modified typing rules} of this section and
$  \textcolor{\coeffectcolor}{ \textcolor{\coeffectcolor}{\gamma} }\! \cdot \! \rho   \vdash_{\mathit{lin} }  \ottnt{M}  \Downarrow  \ottnt{T} $ to refer to the \link{resource/CBPV/semantics.v}{}{modified operational semantics},
highlighting the connection between resource usage coeffects and bounded linear logic.

First, we axiomatize that the semiring is nontrivial. If $ \textcolor{\coeffectcolor}{  \textcolor{\coeffectcolor}{1}  }\; = \; \textcolor{\coeffectcolor}{  \textcolor{\coeffectcolor}{0}  } $, resource
tracking via grades is meaningless, and our general semantics degenerates to
standard CBPV. Second, we require
that if $ \textcolor{\coeffectcolor}{  \textcolor{\coeffectcolor}{0}  }\;  \textcolor{\coeffectcolor}{\mathop{\leq_{\mathit{co} } } } \; \textcolor{\coeffectcolor}{  \textcolor{\coeffectcolor}{ \textcolor{\coeffectcolor}{q}_{{\mathrm{1}}}  +  \textcolor{\coeffectcolor}{q}_{{\mathrm{2}}} }  } $, then $ \textcolor{\coeffectcolor}{ \textcolor{\coeffectcolor}{q}_{{\mathrm{1}}} }\; = \; \textcolor{\coeffectcolor}{  \textcolor{\coeffectcolor}{0}  } $ and $ \textcolor{\coeffectcolor}{ \textcolor{\coeffectcolor}{q}_{{\mathrm{2}}} }\; = \; \textcolor{\coeffectcolor}{  \textcolor{\coeffectcolor}{0}  } $. If either
subterm in a value pair requires nonzero resources, we should not be
able to evaluate the pair with no resources. Finally, for similar reasons, we
require that there be no nonzero zero divisors in the semiring, \emph{i.e.},
if $ \textcolor{\coeffectcolor}{  \textcolor{\coeffectcolor}{0}  }\; = \; \textcolor{\coeffectcolor}{  \textcolor{\coeffectcolor}{ \textcolor{\coeffectcolor}{q}_{{\mathrm{1}}}   \cdot   \textcolor{\coeffectcolor}{q}_{{\mathrm{2}}} }  } $, then $ \textcolor{\coeffectcolor}{ \textcolor{\coeffectcolor}{q}_{{\mathrm{1}}} }\; = \; \textcolor{\coeffectcolor}{  \textcolor{\coeffectcolor}{0}  } $ or $ \textcolor{\coeffectcolor}{ \textcolor{\coeffectcolor}{q}_{{\mathrm{2}}} }\; = \; \textcolor{\coeffectcolor}{  \textcolor{\coeffectcolor}{0}  } $.
Semirings that satisfy these additional constraints include natural numbers, with their usual or discrete orderings, or the $\{ 0 , 1 , \omega \}$ semiring that tracks whether inputs are unused, only used linearly, or with any usage. Note that 
$ \textcolor{\coeffectcolor}{1} $ is incomparable to $ \textcolor{\coeffectcolor}{0} $ in this semiring.

In this system, the $ \textcolor{\coeffectcolor}{0} $ grade denotes that the corresponding variable is
inaccessible, i.e., used $ \textcolor{\coeffectcolor}{0} $ times, so anywhere we eliminate a value and bind it to an
inaccessible variable (or return a value with grade $ \textcolor{\coeffectcolor}{0} $), we require special
treatment.
\Rref{eval-coeff-comp-app-abs,eval-coeff-comp-return,eval-coeff-comp-split} all have this property,
so we modify these rules to require that the relevant grade
be nonzero.
We also add new rules that apply when the grade
\emph{is} zero. These rules, shown in
Figure~\ref{fig:eval-lin-comp-zero}, discard the unused value $\ottnt{V}$ without evaluating it
and use a new, untyped, closed value $ \mathop{\lightning} $ in place of the result of evaluating $\ottnt{V}$.
Because values are pure, discarding an unused value does not alter any effects of the program.

However, \rref{eval-coeff-comp-letin-ret} requires special consideration. Unlike in the
rules above, which discard values, this rule discards a \emph{computation} -- but
because that computation could be effectful, this could change the semantics in unintended ways.
Following related work~\cite{gavazzo:quantitative,dal-lago:relational-theory},
we reconcile this by adding a notion of $ \textcolor{\coeffectcolor}{q} \ \|\ \textcolor{\coeffectcolor}{1} $,
which is equivalent to $\textcolor{\coeffectcolor}{q}$ when $\textcolor{\coeffectcolor}{q}$ is nonzero and $ \textcolor{\coeffectcolor}{1} $ otherwise.
We continue to allow the syntax of the term itself to contain any $\textcolor{\coeffectcolor}{q}_{{\mathrm{2}}}$, but the rest of the typing rule
refers to $ \textcolor{\coeffectcolor}{q}_{{\mathrm{2}}} \ \|\ \textcolor{\coeffectcolor}{1} $ instead.
(All other typing rules stay the same.)
The evaluation rule, \rref{eval-lin-comp-letin-ret}, follows the same pattern (see Figure~\ref{fig:eval-lin-comp-zero}).
Note that this modified evaluation rule introduces a new source of imprecision: we may
consume resources to evaluate code without ever using its result, making our final resource
accounting more of an overapproximation.

With these modifications, we update our logical relation with
a special case for zero resources below.
For brevity we show only the changes.

\begin{definition}[\link{resource/CBPV/semtyping.v}{LRV,LRC}{CBPV with Resource Coeffects: Logical Relation}]
\[
\begin{array}{lcl@{\ }l}
\textit{Closed graded values}\\
 \mathcal{W}_{\textcolor{\coeffectcolor}{  \textcolor{\coeffectcolor}{0}  } }\llbracket  \ottnt{A}  \rrbracket    &=& \{\  \mathop{\lightning}  \ \} \\
 \mathcal{W}_{\textcolor{\coeffectcolor}{ \textcolor{\coeffectcolor}{q} } }\llbracket  \ottnt{A}  \rrbracket    &=&  \mathcal{W}\llbracket  \ottnt{A}  \rrbracket  \mbox{ when $\textcolor{\coeffectcolor}{q} \neq 0$} \\
\\

\textit{Closed terminals}\\
 \mathcal{T}\llbracket   \ottkw{F}_{\color{\coeffectcolor}{ \textcolor{\coeffectcolor}{q} } }\;  \ottnt{A}   \rrbracket     &=&  \{\  \ottkw{return} _{\textcolor{\coeffectcolor}{ \textcolor{\coeffectcolor}{q} } }  \ottnt{W}   &|\
                               \ottnt{W} \, \in \,  \mathcal{W}_{\textcolor{\coeffectcolor}{ \textcolor{\coeffectcolor}{q} } }\llbracket  \ottnt{A}  \rrbracket  \ \} \\
 \mathcal{T}\llbracket   \ottnt{A} ^{\textcolor{\coeffectcolor}{ \textcolor{\coeffectcolor}{q} } } \rightarrow  \ottnt{B}   \rrbracket  &=&  \{\  \mathbf{clo}(   \textcolor{\coeffectcolor}{ \textcolor{\coeffectcolor}{\gamma} }\! \cdot \! \rho  ,   \lambda  \ottmv{x} ^{\textcolor{\coeffectcolor}{ \textcolor{\coeffectcolor}{q} } }. \ottnt{M}   )  &|\
                \mathit{forall} \ \ottnt{W} \, \in \,  \mathcal{W}_{\textcolor{\coeffectcolor}{ \textcolor{\coeffectcolor}{q} } }\llbracket  \ottnt{A}  \rrbracket , \\
&&&\qquad                \ottsym{(}  \ottsym{(}    \textcolor{\coeffectcolor}{ \textcolor{\coeffectcolor}{\gamma} }\! \cdot \! \rho    \mathop{,} \;   \ottmv{x}   \mapsto ^{\textcolor{\coeffectcolor}{ \textcolor{\coeffectcolor}{q} } }  \ottnt{W}    \ottsym{)}  \ottsym{,}  \ottnt{M}  \ottsym{)} \, \in \,  \mathcal{M}\llbracket  \ottnt{B}  \rrbracket \ \} \\
\end{array}
\]
\end{definition}

Furthermore, we update our semantic typing relation for environments to also
include a special case for zero; in this case the environment need not have a closed value for that variable.
(The remaining definitions do not change other than to use the
resource accounting operational semantics. In particular, $ \mathcal{V}\llbracket  \ottnt{A}  \rrbracket $ still requires
the resulting closed value to be in $ \mathcal{W}\llbracket  \ottnt{A}  \rrbracket $. )

\begin{definition}[\link{resource/CBPV/semtyping.v}{SemVWt,SemCWt}{CBPV with Resource Coeffects: Semantic Typing}]
\[
\begin{array}{lcl@{\ }l}
\textcolor{\coeffectcolor}{\gamma}  \cdot  \Gamma  \vDash  \rho &=&
  \multicolumn{2}{l}{ \ottmv{x}  :^{\textcolor{\coeffectcolor}{ \textcolor{\coeffectcolor}{q} } }  \ottnt{A}  \, \in \,  \textcolor{\coeffectcolor}{ \textcolor{\coeffectcolor}{\gamma} }\! \cdot \! \Gamma 
  \ \mathit{implies} \   \textcolor{\coeffectcolor}{ \textcolor{\coeffectcolor}{q} }\; = \; \textcolor{\coeffectcolor}{  \textcolor{\coeffectcolor}{0}  }  \ \textit{or}\   (  \ottmv{x}  \mapsto  \ottnt{W} \, \in \, \rho \ \textit{and}\  \ottnt{W} \, \in \,  \mathcal{W}\llbracket  \ottnt{A}  \rrbracket   )     }  \\
\end{array}
\]
\end{definition}

With these updates, we again prove the fundamental theorem. As in the previous
section, if we unfold the definitions above, this theorem gives us exactly the
\link{resource/CBPV/soundness.v}{soundness}{soundness theorem} we would like.

\begin{theorem}[\link{resource/CBPV/soundness.v}{fundamental\_lemma}{Fundamental lemma: resource soundness}]
For all $\textcolor{\coeffectcolor}{\gamma}$, $\Gamma$, if $ \textcolor{\coeffectcolor}{ \textcolor{\coeffectcolor}{\gamma} }\! \cdot \! \Gamma   \vdash_{\mathit{lin} }  \ottnt{V}  \ottsym{:}  \ottnt{A}$, then
$\textcolor{\coeffectcolor}{\gamma}  \cdot  \Gamma  \vDash_{\mathit{lin} }  \ottnt{V}  \ottsym{:}  \ottnt{A}$,
and
for all $\textcolor{\coeffectcolor}{\gamma}$, $\Gamma$, if $ \textcolor{\coeffectcolor}{ \textcolor{\coeffectcolor}{\gamma} }\! \cdot \! \Gamma   \vdash_{\mathit{lin} }  \ottnt{M}  \ottsym{:}  \ottnt{B}$, then
$\textcolor{\coeffectcolor}{\gamma}  \cdot  \Gamma  \vDash_{\mathit{lin} }  \ottnt{M}  \ottsym{:}  \ottnt{B}$.
\end{theorem}

We can also use this theorem to reason about unused variables. For example,
suppose we type check some computation $\ottnt{M}$ in the context of an
inaccessible variable $x$. Instantiating the theorem above with this context
assures us that evaluation succeeds even when variables are mapped to
$ \mathop{\lightning} $ in the environment.

\begin{corollary}[Inaccessible variable example]
  For all $\ottnt{M}$ and $\ottnt{B}$, if $ \ottmv{x}  :^{\textcolor{\coeffectcolor}{  \textcolor{\coeffectcolor}{0}  } }  \ottnt{A}   \vdash_{\mathit{lin} }  \ottnt{M}  \ottsym{:}  \ottnt{B}$, then 
  there exists some $\ottnt{T}$, such that $  \ottmv{x}   \mapsto ^{\textcolor{\coeffectcolor}{  \textcolor{\coeffectcolor}{0}  } }  \mathop{\lightning}   \vdash_{\mathit{lin} }  \ottnt{M}  \Downarrow  \ottnt{T} $.
\end{corollary}

Because the operational semantics does not include any rules for evaluating
$ \mathop{\lightning} $, we can conclude that $ \textcolor{\coeffectcolor}{0} $-marked variables are never used by
the operational semantics. Furthermore, there are no assumptions about the
structure of $ \mathop{\lightning} $ values, so we can discard them during computation.

\subsection{Translation Soundness}
\label{sec:coeffect-translation}

As with effects, we explore the translation of coeffect-aware CBN and CBV
$\lambda$-calculi to CBPV. As in our CBPV extension with coeffects, the source
type systems are parameterized by a preordered semiring structure of coeffects
and combine the typing context with $\textcolor{\coeffectcolor}{\gamma}$, a vector of coeffect
annotations that describe the demands on each variable.

\begin{figure}
\[ \drule[width=4in]{cbncoeff-app} \drule{cbncoeff-box}\; \]
\[ \drule[width=4in]{cbncoeff-unbox} \]
\caption{CBN with coeffect tracking}
\label{fig:cbncoeff}
\Description{The typing rules for CBN augmented with coeffect tracking}
\end{figure}

\begin{figure}
\[ \drule[width=5in]{cbvcoeff-app} \]
\[ \drule{cbvcoeff-box}\;\drule[width=4in]{cbvcoeff-unbox} \]
\caption{CBV with coeffect tracking}
\label{fig:cbvcoeff}
\Description{The typing rules for CBV augmented with coeffect tracking}
\end{figure}

The type-and-coeffect system that we consider as the starting point of our CBN
translation is adapted from the simple type system of
\citet{weirich:graded-haskell} and is similar to the system developed by
\citet{abel:icfp2020}. The differences between this source language and the
related work are minor. The design of our CBV language is inspired by
\citet{dal-lago:relational-theory}. To make the comparison clear, we present
it as a standard CBV lambda calculus instead of fine-grained CBV.  Other
changes to the language include the introduction of subcoeffecting, allowing
functions to take $\textcolor{\coeffectcolor}{q}$ copies of their argument instead of one (and
annotating applications with $\textcolor{\coeffectcolor}{q}$), and replacing $ \textcolor{\coeffectcolor}{q}  \wedge \textcolor{\coeffectcolor}{1} $ with
$ \textcolor{\coeffectcolor}{q} \ \|\ \textcolor{\coeffectcolor}{1} $ to force the evaluation of subterms. (We choose $ \textcolor{\coeffectcolor}{q} \ \|\ \textcolor{\coeffectcolor}{1} $
over $ \textcolor{\coeffectcolor}{q}  \wedge \textcolor{\coeffectcolor}{1} $ to avoid requiring the existence
of $ \textcolor{\coeffectcolor}{q}  \wedge \textcolor{\coeffectcolor}{1} $ as an axiom of the semiring. The difference is minor.)

The rules for the CBN version of the system appear in
\link{resource/CBN/typing.v}{Wt}{Figure~\ref{fig:cbncoeff}}; the rules for the CBV version are in
\link{resource/CBV/typing.v}{Wt}{Figure~\ref{fig:cbvcoeff}}. Most of the rules
parallel those of the corresponding terms in CBPV; for brevity, then, we show
only the rules for application, boxing and unboxing here.  The two languages
differ in the application rule. In CBV, we annotate applications with the number
of times the function uses its argument. Because the argument will always be
evaluated once in CBV, if $\textcolor{\coeffectcolor}{q}$ is zero, we force it to be one.

These latter two terms introduce and
eliminate the modal type $ \square_{\textcolor{\coeffectcolor}{ \textcolor{\coeffectcolor}{q} } }\;  \tau $. The introduction form requires
grade $\textcolor{\coeffectcolor}{q}$ on its argument. When we unbox the argument, the second
subterm has access to it with grade $ \textcolor{\coeffectcolor}{ \textcolor{\coeffectcolor}{q}_{{\mathrm{1}}}   \cdot   \textcolor{\coeffectcolor}{q}_{{\mathrm{2}}} } $.  The $\textcolor{\coeffectcolor}{q}_{{\mathrm{1}}}$ comes from
when the box was created, and the $\textcolor{\coeffectcolor}{q}_{{\mathrm{2}}}$ comes from the unboxing term, as in
let bindings in CBPV. In CBV, we use $\ottkw{letin}$ in
both translations, so we include $ \textcolor{\coeffectcolor}{q} \ \|\ \textcolor{\coeffectcolor}{1} $ in both rules in an
analogous way to its use in \rref{coeff-letin}.
In CBN, we use $\ottkw{letin}$ in the translation of $\ottkw{unbox}$ but
not $\ottkw{box}$, so we can drop $ \textcolor{\coeffectcolor}{q} \ \|\ \textcolor{\coeffectcolor}{1} $ from the typing rule for box.
This imprecision makes sense in the source languages for the same reason it makes
sense in CBPV: because we are combining effects and coeffects,
we sometimes need to evaluate subterms for their effects even if the results of those
subterms are never used.

\subsubsection{Call-by-name Translation}

We first consider a
call-by-name translation to CBPV. For brevity, we show just the translation
of function and box types on the left below and the translation of applications and
the $\ottkw{box}$ and $\ottkw{unbox}$ terms on the right.

\[
\begin{array}{ll}
 \llbracket {  \tau_{{\mathrm{1}}} ^{\textcolor{\coeffectcolor}{ \textcolor{\coeffectcolor}{q} } } \rightarrow  \tau_{{\mathrm{2}}}  } \rrbracket_{\textsc{n} }  & =  \ottsym{(}  \ottkw{U} \,  \llbracket { \tau_{{\mathrm{1}}} } \rrbracket_{\textsc{n} }   \ottsym{)} ^{\textcolor{\coeffectcolor}{ \textcolor{\coeffectcolor}{q} } } \rightarrow   \llbracket { \tau_{{\mathrm{2}}} } \rrbracket_{\textsc{n} }    \\
 \llbracket {  \square_{\textcolor{\coeffectcolor}{ \textcolor{\coeffectcolor}{q} } }\;  \tau  } \rrbracket_{\textsc{n} }  & =  \ottkw{F}_{\color{\coeffectcolor}{ \textcolor{\coeffectcolor}{q} } }\;  \ottsym{(}  \ottkw{U} \,  \llbracket { \tau } \rrbracket_{\textsc{n} }   \ottsym{)}  \\
\end{array}
\quad
\begin{array}{ll}
 \llbracket { \ottnt{e_{{\mathrm{1}}}} \, \ottnt{e_{{\mathrm{2}}}} } \rrbracket_{\textsc{n} }                            & =  \llbracket { \ottnt{e_{{\mathrm{1}}}} } \rrbracket_{\textsc{n} }  \, \ottsym{\{}   \llbracket { \ottnt{e_{{\mathrm{2}}}} } \rrbracket_{\textsc{n} }   \ottsym{\}} \\
 \llbracket {  \ottkw{box} _{\textcolor{\coeffectcolor}{ \textcolor{\coeffectcolor}{q} } }\  \ottnt{e}  } \rrbracket_{\textsc{n} }               & =  \ottkw{return} _{\textcolor{\coeffectcolor}{ \textcolor{\coeffectcolor}{q} } }\;  \ottsym{\{}   \llbracket { \ottnt{e} } \rrbracket_{\textsc{n} }   \ottsym{\}}  \\
 \llbracket {  \ottkw{unbox} _{\textcolor{\coeffectcolor}{ \textcolor{\coeffectcolor}{q} } }\  \ottmv{x}  =  \ottnt{e_{{\mathrm{1}}}} \  \ottkw{in} \  \ottnt{e_{{\mathrm{2}}}}  } \rrbracket_{\textsc{n} }             & =  \ottmv{x}  \leftarrow^{\textcolor{\coeffectcolor}{ \textcolor{\coeffectcolor}{q} } }   \llbracket { \ottnt{e_{{\mathrm{1}}}} } \rrbracket_{\textsc{n} }  \ \ottkw{in}\   \llbracket { \ottnt{e_{{\mathrm{2}}}} } \rrbracket_{\textsc{n} }   \\
\end{array}
\]

In this translation, the coeffect on the $\lambda$-calculus function
type translates directly
to the coeffect on the CBPV function type. Furthermore,
the modal type $ \square_{\textcolor{\coeffectcolor}{ \textcolor{\coeffectcolor}{q} } }\;  \tau $ is a graded comonad, so it can be translated to
the comonad in CBPV, adding the grade to the returner type.

The CBN translation of $\lambda$ terms is as usual. However, the translation of the
box introduction and elimination forms follows from the definition of the CBPV comonadic type.
To create a box, we return the thunked translation of the expression. To eliminate a box,
we use $\ottkw{letin}$ to move the thunk to the environment.

\subsubsection{Call-by-value Translation} \label{cbv-coeffects-translation}

Next, we define a corresponding CBV translation to CBPV. For brevity, we again
show only the translation of function and graded modal types and of applications and the
$\ottkw{box}$ and $\ottkw{unbox}$ terms.

\[
\begin{array}{ll}
 \llbracket {  \tau_{{\mathrm{1}}} ^{\textcolor{\coeffectcolor}{ \textcolor{\coeffectcolor}{q} } } \rightarrow  \tau_{{\mathrm{2}}}  } \rrbracket_{\textsc{v} }  & = \ottkw{U} \, \ottsym{(}    \llbracket { \tau_{{\mathrm{1}}} } \rrbracket_{\textsc{v} }  ^{\textcolor{\coeffectcolor}{ \textcolor{\coeffectcolor}{q} } } \rightarrow   \ottkw{F}_{\color{\coeffectcolor}{  \textcolor{\coeffectcolor}{1}  } }\;   \llbracket { \tau_{{\mathrm{2}}} } \rrbracket_{\textsc{v} }     \ottsym{)} \\
 \llbracket {  \square_{\textcolor{\coeffectcolor}{ \textcolor{\coeffectcolor}{q} } }\;  \tau  } \rrbracket_{\textsc{v} }  & = \ottkw{U} \, \ottsym{(}   \ottkw{F}_{\color{\coeffectcolor}{ \textcolor{\coeffectcolor}{q} } }\;   \llbracket { \tau } \rrbracket_{\textsc{v} }    \ottsym{)} \\
\\
 \llbracket {  \ottnt{e_{{\mathrm{1}}}} ^{ \textcolor{\coeffectcolor}{q} }  \ottnt{e_{{\mathrm{2}}}}  } \rrbracket_{\textsc{v} }    & =  \ottmv{x}  \leftarrow^{\textcolor{\coeffectcolor}{  \textcolor{\coeffectcolor}{1}  } }   \llbracket { \ottnt{e_{{\mathrm{1}}}} } \rrbracket_{\textsc{v} }  \ \ottkw{in}\   \ottmv{y}  \leftarrow^{\textcolor{\coeffectcolor}{ \textcolor{\coeffectcolor}{q} } }   \llbracket { \ottnt{e_{{\mathrm{2}}}} } \rrbracket_{\textsc{v} }  \ \ottkw{in}\  \ottmv{x}  \ottsym{!}   \, \ottmv{y} \\
 \llbracket {  \ottkw{box} _{\textcolor{\coeffectcolor}{ \textcolor{\coeffectcolor}{q} } }\  \ottnt{e}  } \rrbracket_{\textsc{v} }    & =  \ottmv{x}  \leftarrow^{\textcolor{\coeffectcolor}{ \textcolor{\coeffectcolor}{q} } }   \llbracket { \ottnt{e} } \rrbracket_{\textsc{v} }  \ \ottkw{in}\   \ottkw{return} _{\textcolor{\coeffectcolor}{  \textcolor{\coeffectcolor}{1}  } }\;  \ottsym{\{}   \ottkw{return} _{\textcolor{\coeffectcolor}{  \textcolor{\coeffectcolor}{q} \ \|\ \textcolor{\coeffectcolor}{1}  } }\;  \ottmv{x}   \ottsym{\}}   \\
 \llbracket {  \ottkw{unbox} _{\textcolor{\coeffectcolor}{ \textcolor{\coeffectcolor}{q} } }\  \ottmv{x}  =  \ottnt{e_{{\mathrm{1}}}} \  \ottkw{in} \  \ottnt{e_{{\mathrm{2}}}}  } \rrbracket_{\textsc{v} }  & =  \ottmv{y}  \leftarrow^{\textcolor{\coeffectcolor}{ \textcolor{\coeffectcolor}{q} } }   \llbracket { \ottnt{e_{{\mathrm{1}}}} } \rrbracket_{\textsc{v} }  \ \ottkw{in}\   \ottmv{x}  \leftarrow^{\textcolor{\coeffectcolor}{ \textcolor{\coeffectcolor}{q} } }  \ottmv{y}  \ottsym{!} \ \ottkw{in}\   \llbracket { \ottnt{e_{{\mathrm{2}}}} } \rrbracket_{\textsc{v} }    \\
\end{array}
\]

As above, we propagate the coeffect from the $\lambda$-calculus
function type directly to the CBPV function type. Similarly, we propagate
the grade in the modal type to the inner returner type and let binding in CBPV.

For applications, we use let bindings to access the translations of the function and
the argument. The argument is not thunked in translation, so it is strict, but the
function is thunked in translation, so we must force it before applying it.
$\ottkw{box}$ is also strict in CBV, so its translation first
evaluates its argument. The rest of the translation follows its type
definition. In CBPV, the computation \mbox{$ \ottmv{x}  \leftarrow^{\textcolor{\coeffectcolor}{  \textcolor{\coeffectcolor}{1}  } }  \ottnt{M} \ \ottkw{in}\   \ottkw{return} _{\textcolor{\coeffectcolor}{  \textcolor{\coeffectcolor}{1}  } }\;  \ottmv{x}  $} is
equivalent to $\ottnt{M}$, but the computation $ \ottmv{x}  \leftarrow^{\textcolor{\coeffectcolor}{ \textcolor{\coeffectcolor}{q} } }  \ottnt{M} \ \ottkw{in}\   \ottkw{return} _{\textcolor{\coeffectcolor}{ \textcolor{\coeffectcolor}{q} } }\;  \ottmv{x}  $
corresponds to duplicating $\ottnt{M}$ $\textcolor{\coeffectcolor}{q}$ times in a resource usage coeffect.
This propagation of the grade is exactly the feature that we need to translate
the $\ottkw{box}$ term. Like the CBN translation of the modal type, in the CBV
translation, the comonadic type is difficult to access. In this translation,
$\ottkw{box}$ must include an extra thunk that is forced in the translation of
the $\ottkw{unbox}$ term, giving us access to the comonadic type \textbf{F U}.
We must also use the annotation capability of
$\ottkw{letin}$ (twice) to mirror the annotation in the source language. The
correctness proofs for both the \link{resource/CBN/proofs.v}{translation\_correct}{CBN} and 
\link{resource/CBV/proofs.v}{translation\_correct}{CBV} translations follow from the
corresponding proofs of the combined system (taking the trivial effect).

\section{Combined Effects and Coeffects}
\label{sec:full}

Next, we present a system that tracks both effects and coeffects, by
combining the effect system of Section~\ref{sec:effects}
with the resource usage system of Section~\ref{sec:resource-usage}, and adding one
new rule.

\begin{definition}[\link{full/CBPV/typing.v}{VWt,CWt}{Combined type system}]
The judgements $ \textcolor{\coeffectcolor}{ \textcolor{\coeffectcolor}{\gamma} }\! \cdot \! \Gamma   \vdash_{\mathit{full} }  \ottnt{V}  \ottsym{:}  \ottnt{A}$ and 
$  \textcolor{\coeffectcolor}{ \textcolor{\coeffectcolor}{\gamma} }\! \cdot \! \Gamma    \vdash_{\mathit{full} }   \ottnt{M}  :^{\textcolor{\effectcolor}{ \textcolor{\effectcolor}{\phi} } }  \ottnt{B} $ refer to the CBPV type system with
effect annotations from Figure~\ref{fig:cbpv-effects} and coeffect annotations
(resource tracking version) from Figure~\ref{fig:cbpv-lin}.%
\ifextended
 The full definition is available in Appendix~\ref{app:cbpv-full-typing}.
\else
\footnote{For space, we do not include the entire combined system here. The full
rules of this system are available in the extended version of this
paper~\cite{extended-version} and in the Coq development.}
\fi
\end{definition}

This type system is a straightforward
combination of the systems presented earlier.
For example, the typing \rref{full-letin} combines \rref{eff-letin}
with \rref{lin-letin} and includes both the grade vector
$ \textcolor{\coeffectcolor}{  \textcolor{\coeffectcolor}{ \textcolor{\coeffectcolor}{q}'_{{\mathrm{2}}} \cdot \textcolor{\coeffectcolor}{\gamma}_{{\mathrm{1}}} }  \ottsym{+} \textcolor{\coeffectcolor}{\gamma}_{{\mathrm{2}}} } $ and the effect $ \textcolor{\effectcolor}{ \textcolor{\effectcolor}{\phi}_{{\mathrm{1}}}  \cdot  \textcolor{\effectcolor}{\phi}_{{\mathrm{2}}} } $ for the
computation.

\[ \drule[width=8in]{full-letin} \]

Similarly, we augment our instrumented operational semantics to track both
effects and coeffects.
\begin{definition}[\link{full/CBPV/semantics.v}{EvalVal,EvalComp}{Combined Resource Semantics}]
\label{def:full}%
\ifextended
\footnote{The full definition is available in Appendix~\ref{app:cbpv-full-semantics}}
\fi
  The judgements $  \textcolor{\coeffectcolor}{ \textcolor{\coeffectcolor}{\gamma} }\! \cdot \! \rho   \vdash_{\mathit{full} }  \ottnt{V}  \Downarrow  \ottnt{W} $ and
  $  \textcolor{\coeffectcolor}{ \textcolor{\coeffectcolor}{\gamma} }\! \cdot \! \rho   \vdash_{\mathit{full} }  \ottnt{M}  \Downarrow  \ottnt{T}  \mathop{\#} \textcolor{\effectcolor}{ \textcolor{\effectcolor}{\phi} } $ refer to the CBPV operational semantics
  with effect annotations from
  Figures~\ref{fig:eval-val} and \ref{fig:eff-big-step} and coeffect
  annotations from Figure~\ref{fig:eval-coeff-comp}, with updates for resource tracking from
  Figure~\ref{fig:eval-lin-comp-zero}.
\end{definition}

\noindent
For example, the $\ottkw{letin}$ evaluation rule computes the
instrumented grade vector and effect and requires that the computation
$\ottnt{M}$ be evaluated at least once, as in \rref{eval-lin-comp-letin-ret}.

\[ \drule[width=8in]{eval-full-comp-letin} \]

We can use this operational semantics to show both effect and coeffect soundness
of the combined type system. However, before we do so, we make one more extension
to the language.

\paragraph{Skipping Unused Discardable Computations}
\label{sec:letzero}

In Section~\ref{sec:resource-usage}, we developed several ``zero'' rules
for discarding unused \emph{values}. But, unused \emph{computations} could not be
discarded, because they may have effects. However, in this system, we can identify
unused, pure computations, and add a new syntactic form, written $ \ottmv{x}  \leftarrow^{\textcolor{\coeffectcolor}{0} }_{\textcolor{\effectcolor}{\varepsilon} }  \ottnt{M} \ \ottkw{in}\  \ottnt{N} $,
indicating that they can be discarded. The typing rule (below left) requires
that $\ottnt{M}$ be effect free and its result unused in $\ottnt{N}$.
\[ \drule[width=5in]{full-letin-zero}  \drule[width=5in]{eval-full-comp-letin-zero}  \]

\noindent
Furthermore, the operational semantics of this new expression form (above
right) does not evaluate $\ottnt{M}$. Instead it uses the junk value $ \mathop{\lightning} $
for the result of this computation.

To see this rule in action, consider the CBV translation of an expression
$y_2\ (y_1\ x)$, where $y_2$ is a constant function and $y_1$ is pure. In this
case, the type system can observe that $\ottmv{x}$ does not contribute to the
final result when the application of $y_1$ to $x$ is marked as discardable.

\[     \ottmv{x}  :^{\textcolor{\coeffectcolor}{  \textcolor{\coeffectcolor}{0}  } }  \ottnt{A}    \mathop{,}    \ottmv{y_{{\mathrm{1}}}}  :^{\textcolor{\coeffectcolor}{  \textcolor{\coeffectcolor}{0}  } }   \ottkw{U}_{\color{\effectcolor}{  \textcolor{\effectcolor}{\varepsilon}  } }\;  \ottsym{(}   \ottnt{A} ^{\textcolor{\coeffectcolor}{  \textcolor{\coeffectcolor}{1}  } } \rightarrow  \ottkw{F} \, \ottnt{A}   \ottsym{)}      \mathop{,}    \ottmv{y_{{\mathrm{2}}}}  :^{\textcolor{\coeffectcolor}{  \textcolor{\coeffectcolor}{1}  } }   \ottkw{U}_{\color{\effectcolor}{ \textcolor{\effectcolor}{\phi} } }\;  \ottsym{(}   \ottnt{A} ^{\textcolor{\coeffectcolor}{  \textcolor{\coeffectcolor}{0}  } } \rightarrow  \ottnt{B}   \ottsym{)}      \vdash_{\mathit{full} }    \ottmv{z}  \leftarrow^{\textcolor{\coeffectcolor}{0} }_{\textcolor{\effectcolor}{\varepsilon} }  \ottmv{y_{{\mathrm{1}}}}  \ottsym{!} \, \ottmv{x} \ \ottkw{in}\  \ottmv{y_{{\mathrm{2}}}}  \ottsym{!}  \, \ottmv{z}  :^{\textcolor{\effectcolor}{ \textcolor{\effectcolor}{\phi} } }  \ottnt{B}  \]

\paragraph{Soundness Proof for Discardable Computations}

We next show that discarding unused values and unused pure computations
does not change the evaluation behavior of computations. To do so,
we need the following properties that state that $\textcolor{red}{\varepsilon}$ is
the minimum element of the effect preorder.
\footnote{These properties hold for $ \textcolor{\effectcolor}{ \ottkw{tick} } $ effects, but we have not used them
before now.}
\begin{definition}[Min identity]
%
(1) For all $\textcolor{\effectcolor}{\phi}$, $ \textcolor{\effectcolor}{  \textcolor{\effectcolor}{\varepsilon}  \;  \textcolor{\effectcolor}{\mathop{\leq_{\mathit{eff} } } } \;  \textcolor{\effectcolor}{\phi}  } $
(2) For all $\textcolor{\effectcolor}{\phi}$, $  \textcolor{\effectcolor}{ \textcolor{\effectcolor}{\phi} \;  \textcolor{\effectcolor}{\mathop{\leq_{\mathit{eff} } } } \;   \textcolor{\effectcolor}{\varepsilon}   }  \ \textit{implies}\   \textcolor{\effectcolor}{ \textcolor{\effectcolor}{\phi} }\; = \; \textcolor{\effectcolor}{  \textcolor{\effectcolor}{\varepsilon}  }  $.
\end{definition}

To prove that discarding is sound, we establish a relation between our
combined resource semantics and one that does not discard terms.

\begin{definition}[\link{full/CBPV/junk.v}{G.EvalComp}{Combined nondiscarding semantics}]
\label{def:gfull}
The judgement
$  \textcolor{\coeffectcolor}{ \textcolor{\coeffectcolor}{\gamma} }\! \cdot \! \rho   \vdash_{\mathit{gen} }  \ottnt{M}  \Downarrow  \ottnt{T}  \mathop{\#} \textcolor{\effectcolor}{ \textcolor{\effectcolor}{\phi} } $ refers to the operational semantics
that is the combination of CBPV with effect annotations from Figure
\ref{fig:eff-big-step} and coeffect annotations from
Figure~\ref{fig:eval-coeff-comp}, with the modified \rref{full-letin} that
always evaluates computations. This semantics does not include rules that
discard values or computations and uses \rref{eval-full-comp-letin} to evaluate
the new $\ottkw{letin}$ expression.
\end{definition}

Our simulation lemma states that for closed boolean-valued
computations\ifextended\else\footnote{The system in our extended version~\cite{extended-version} and Coq development
  includes sums, necessary to implement the boolean type.}\fi, evaluating with either the nondiscarding
semantics (Definition~\ref{def:gfull}) or with the resource
semantics (Definition~\ref{def:full}) produces the same result and the same
effect.

\begin{lemma}[\link{full/CBPV/junk.v}{resource\_simulation}{Resource Simulation}]
If $  \textcolor{\coeffectcolor}{ \emptyset }\! \cdot \! \varnothing    \vdash_{\mathit{full} }   \ottnt{M}  :^{\textcolor{\effectcolor}{ \textcolor{\effectcolor}{\phi} } }   \ottkw{F}_{\color{\coeffectcolor}{  \textcolor{\coeffectcolor}{1}  } }\;  \ottsym{(}  \ottkw{unit}  \ottsym{+}  \ottkw{unit}  \ottsym{)}  $ then
either
\begin{enumerate}
\item $  \textcolor{\coeffectcolor}{ \emptyset }\! \cdot \! \emptyset   \vdash_{\mathit{gen} }  \ottnt{M}  \Downarrow   \ottkw{return} _{\textcolor{\coeffectcolor}{  \textcolor{\coeffectcolor}{1}  } }  \ottsym{(}  \ottkw{inl} \, \ottsym{()}  \ottsym{)}   \mathop{\#} \textcolor{\effectcolor}{ \textcolor{\effectcolor}{\phi}_{{\mathrm{1}}} } $ and
    $  \textcolor{\coeffectcolor}{ \emptyset }\! \cdot \! \emptyset   \vdash_{\mathit{full} }  \ottnt{M}  \Downarrow   \ottkw{return} _{\textcolor{\coeffectcolor}{  \textcolor{\coeffectcolor}{1}  } }  \ottsym{(}  \ottkw{inl} \, \ottsym{()}  \ottsym{)}   \mathop{\#} \textcolor{\effectcolor}{ \textcolor{\effectcolor}{\phi}_{{\mathrm{1}}} } $ or
\item $  \textcolor{\coeffectcolor}{ \emptyset }\! \cdot \! \emptyset   \vdash_{\mathit{gen} }  \ottnt{M}  \Downarrow   \ottkw{return} _{\textcolor{\coeffectcolor}{  \textcolor{\coeffectcolor}{1}  } }  \ottsym{(}  \ottkw{inr} \, \ottsym{()}  \ottsym{)}   \mathop{\#} \textcolor{\effectcolor}{ \textcolor{\effectcolor}{\phi}_{{\mathrm{1}}} } $ and
    $  \textcolor{\coeffectcolor}{ \emptyset }\! \cdot \! \emptyset   \vdash_{\mathit{full} }  \ottnt{M}  \Downarrow   \ottkw{return} _{\textcolor{\coeffectcolor}{  \textcolor{\coeffectcolor}{1}  } }  \ottsym{(}  \ottkw{inr} \, \ottsym{()}  \ottsym{)}   \mathop{\#} \textcolor{\effectcolor}{ \textcolor{\effectcolor}{\phi}_{{\mathrm{1}}} } $.
\end{enumerate}
\end{lemma}

This simulation lemma is a corollary of a much more general result---the
fundamental lemma for a binary logical relation between computations that are
evaluated with the two different semantics.
\ifextended This relation appears in
  Appendix~\ref{sec:binary-rel} and 
  \link{full/CBPV/junk.v}{LRM,LRV,LRC,SemVWt,SemCWt}{in the Coq development}.
\else This relation,
  shown below, is mutually defined with relations between closed values and
  closed terminals (not shown, but available in the extended version~\cite{extended-version} and 
  \link{full/CBPV/junk.v}{LRM,LRV,LRC,SemVWt,SemCWt}{in the Coq development}).

\[
\begin{array}{lcl@{\ }l}
 \mathcal{M}\llbracket  \ottnt{B} \rrbracket^{\textcolor{\effectcolor}{ \textcolor{\effectcolor}{\phi} } }  &=& \{\ ( \textcolor{\coeffectcolor}{ \textcolor{\coeffectcolor}{\gamma} }\! \cdot \! \rho_{{\mathrm{1}}}  , \ottnt{M_{{\mathrm{1}}}},  \textcolor{\coeffectcolor}{ \textcolor{\coeffectcolor}{\gamma} }\! \cdot \! \rho_{{\mathrm{2}}}  ,  \ottnt{M_{{\mathrm{2}}}} )\
|\    \textcolor{\coeffectcolor}{ \textcolor{\coeffectcolor}{\gamma} }\! \cdot \! \rho_{{\mathrm{1}}}   \vdash_{\mathit{gen} }  \ottnt{M_{{\mathrm{1}}}}  \Downarrow  \ottnt{T_{{\mathrm{1}}}}  \mathop{\#} \textcolor{\effectcolor}{ \textcolor{\effectcolor}{\phi}_{{\mathrm{1}}} }  \ \textit{and}\    \textcolor{\coeffectcolor}{ \textcolor{\coeffectcolor}{\gamma} }\! \cdot \! \rho_{{\mathrm{2}}}   \vdash_{\mathit{full} }  \ottnt{M_{{\mathrm{2}}}}  \Downarrow  \ottnt{T_{{\mathrm{2}}}}  \mathop{\#} \textcolor{\effectcolor}{ \textcolor{\effectcolor}{\phi}_{{\mathrm{1}}} }   \\
&&\quad                 \mathit{and}\  \ottsym{(}  \ottnt{T_{{\mathrm{1}}}}  \ottsym{,}  \ottnt{T_{{\mathrm{1}}}}  \ottsym{)} \, \in \,  \mathcal{T}\llbracket  \ottnt{B} \rrbracket^{\textcolor{\effectcolor}{ \textcolor{\effectcolor}{\phi}_{{\mathrm{2}}} } }  \ \textit{and}\  \ottsym{(}  \ottnt{T_{{\mathrm{1}}}}  \ottsym{,}  \ottnt{T_{{\mathrm{2}}}}  \ottsym{)} \, \in \,  \mathcal{T}\llbracket  \ottnt{B} \rrbracket^{\textcolor{\effectcolor}{ \textcolor{\effectcolor}{\phi}_{{\mathrm{2}}} } }  \
                 \mathit{and}\  \textcolor{\effectcolor}{  \textcolor{\effectcolor}{ \textcolor{\effectcolor}{\phi}_{{\mathrm{1}}}  \cdot  \textcolor{\effectcolor}{\phi}_{{\mathrm{2}}} }  \;  \textcolor{\effectcolor}{\mathop{\leq_{\mathit{eff} } } } \;  \textcolor{\effectcolor}{\phi}  } \ \} \\
\end{array}
\]

\fi

\noindent
Using this relation, we define a binary version of the semantic typing
relation. Two environments $\rho_{{\mathrm{1}}}$ and $\rho_{{\mathrm{2}}}$ are related when the
closed values in the first environment are related to themselves, and, if the
usage is nonzero, the closed value in the second environment is related to the
first. The first condition ensures that we know something about closed values
in the first relation even when the corresponding value in the second relation
has been discarded in the resource semantics.

\begin{definition}[\link{full/CBPV/junk.v}{SemVWt,SemCWt}{Semantic double typing}]
\label{fig:binary-typing}
\[
\begin{array}{l@{\ }c@{\ }ll}
 \textcolor{\coeffectcolor}{ \textcolor{\coeffectcolor}{\gamma} }  \cdot   \Gamma   \vDash   \rho_{{\mathrm{1}}}  \sim  \rho_{{\mathrm{2}}}  &=&
   \ottmv{x}  :^{\textcolor{\coeffectcolor}{ \textcolor{\coeffectcolor}{q} } }  \ottnt{A}  \, \in \,  \textcolor{\coeffectcolor}{ \textcolor{\coeffectcolor}{\gamma} }\! \cdot \! \Gamma  \ \mathit{implies}\
   \ottmv{x}  \mapsto  \ottnt{W_{{\mathrm{1}}}} \, \in \, \rho_{{\mathrm{1}}} \ \textit{and}\  \ottsym{(}  \ottnt{W_{{\mathrm{1}}}}  \ottsym{,}  \ottnt{W_{{\mathrm{1}}}}  \ottsym{)} \, \in \,  \mathcal{W}\llbracket  \ottnt{A}  \rrbracket   \ \mathit{and} \  \\
&&  (   \textcolor{\coeffectcolor}{ \textcolor{\coeffectcolor}{q} }\; = \; \textcolor{\coeffectcolor}{  \textcolor{\coeffectcolor}{0}  }  \ \textit{or}\   (  \ottmv{x}  \mapsto  \ottnt{W_{{\mathrm{2}}}} \, \in \, \rho_{{\mathrm{2}}} \ \textit{and}\  \ottsym{(}  \ottnt{W_{{\mathrm{1}}}}  \ottsym{,}  \ottnt{W_{{\mathrm{2}}}}  \ottsym{)} \, \in \,  \mathcal{W}\llbracket  \ottnt{A}  \rrbracket   )   )     \\
  \textcolor{\coeffectcolor}{ \textcolor{\coeffectcolor}{\gamma} }  \cdot   \Gamma   \vDash_{\mathit{full} }   \ottnt{V_{{\mathrm{1}}}}  \sim  \ottnt{V_{{\mathrm{2}}}} \; :  \ottnt{A}  &=&
  \mathit{forall} \ \rho_{{\mathrm{1}}}, \rho_{{\mathrm{2}}},  \textcolor{\coeffectcolor}{ \textcolor{\coeffectcolor}{\gamma} }  \cdot   \Gamma   \vDash   \rho_{{\mathrm{1}}}  \sim  \rho_{{\mathrm{2}}}  \  \\
&& \mathit{implies}\    \textcolor{\coeffectcolor}{ \textcolor{\coeffectcolor}{\gamma} }\! \cdot \! \rho_{{\mathrm{1}}}   \vdash_{\mathit{full} }  \ottnt{V_{{\mathrm{1}}}}  \Downarrow  \ottnt{W_{{\mathrm{1}}}}  \ \textit{and}\    \textcolor{\coeffectcolor}{ \textcolor{\coeffectcolor}{\gamma} }\! \cdot \! \rho_{{\mathrm{2}}}   \vdash_{\mathit{full} }  \ottnt{V_{{\mathrm{1}}}}  \Downarrow  \ottnt{W_{{\mathrm{1}}}}   \\
&& \mathit{and}\  \ottsym{(}  \ottnt{W_{{\mathrm{1}}}}  \ottsym{,}  \ottnt{W_{{\mathrm{1}}}}  \ottsym{)} \, \in \,  \mathcal{W}\llbracket  \ottnt{A}  \rrbracket  \ \textit{and}\  \ottsym{(}  \ottnt{W_{{\mathrm{1}}}}  \ottsym{,}  \ottnt{W_{{\mathrm{2}}}}  \ottsym{)} \, \in \,  \mathcal{W}\llbracket  \ottnt{A}  \rrbracket  \ \\
   \textcolor{\coeffectcolor}{ \textcolor{\coeffectcolor}{\gamma} }  \cdot   \Gamma   \vDash_{\mathit{full} }   \ottnt{M_{{\mathrm{1}}}}  \approx  \ottnt{M_{{\mathrm{2}}}} \; :^{\textcolor{\effectcolor}{  \textcolor{\effectcolor}{\phi}  } }\;  \ottnt{B}  &=&
  \mathit{forall} \ \rho_{{\mathrm{1}}}, \rho_{{\mathrm{2}}},  \textcolor{\coeffectcolor}{ \textcolor{\coeffectcolor}{\gamma} }  \cdot   \Gamma   \vDash   \rho_{{\mathrm{1}}}  \sim  \rho_{{\mathrm{2}}}  \   \mathit{implies} \
   ( \textcolor{\coeffectcolor}{\gamma} \cdot   \rho_{{\mathrm{1}}} ,  \ottnt{M_{{\mathrm{1}}}} ,  \textcolor{\coeffectcolor}{\gamma} \cdot \rho_{{\mathrm{2}}} ,  \ottnt{M_{{\mathrm{2}}}}  )  \in    \mathcal{M}\llbracket  \ottnt{B} \rrbracket^{\textcolor{\effectcolor}{ \textcolor{\effectcolor}{\phi} } }    \\
\end{array}
\]
\end{definition}

The fundamental theorem shows that this binary relation is reflexive. 
\begin{theorem}[\link{full/CBPV/junk.v}{fundamental}{Fundamental lemma: simulation}]\ 
\begin{enumerate}
\item For all $\textcolor{\coeffectcolor}{\gamma}$, $\Gamma$, if $ \textcolor{\coeffectcolor}{ \textcolor{\coeffectcolor}{\gamma} }\! \cdot \! \Gamma   \vdash_{\mathit{full} }  \ottnt{V}  \ottsym{:}  \ottnt{A}$, then
$ \textcolor{\coeffectcolor}{ \textcolor{\coeffectcolor}{\gamma} }  \cdot   \Gamma   \vDash_{\mathit{full} }   \ottnt{V}  \sim  \ottnt{V} \; :  \ottnt{A} $,
and
\item for all $\textcolor{\coeffectcolor}{\gamma}$, $\Gamma$, if $  \textcolor{\coeffectcolor}{ \textcolor{\coeffectcolor}{\gamma} }\! \cdot \! \Gamma    \vdash_{\mathit{full} }   \ottnt{M}  :^{\textcolor{\effectcolor}{ \textcolor{\effectcolor}{\phi} } }  \ottnt{B} $, then
$ \textcolor{\coeffectcolor}{ \textcolor{\coeffectcolor}{\gamma} }  \cdot   \Gamma   \vDash_{\mathit{full} }   \ottnt{M}  \approx  \ottnt{M} \; :^{\textcolor{\effectcolor}{  \textcolor{\effectcolor}{\phi}  } }\;  \ottnt{B} $.
\end{enumerate}
\end{theorem}
This fundamental lemma combines and generalizes prior results of this
paper. In particular, it shows the effect-and-coeffect soundness of the
combined type system with respect to both the nondiscarding and resource
accounting semantics--the effects and coeffects of the evaluation are bounded
by the type system.  For clarity, we also separately show effect-and-coeffect
soundness of the combined type system
\link{full/CBPV/soundness.v}{soundness}{in the Coq development}.

\paragraph{CBN and CBV Translations}
Finally, we have defined CBN and CBV with combined effects and coeffects
and have proved the soundness of translations to the combined CPBV type system.

\noindent
\begin{theorem}[\link{full/CBN/proofs.v}{translation\_correct}{CBN} and 
\link{full/CBV/proofs.v}{translation\_correct}{CBV} translation correctness]\ 
\begin{enumerate}
\item For all $\textcolor{\coeffectcolor}{\gamma}$, $\Gamma$, $\ottnt{e}$, $\tau$,
  if $  \textcolor{\coeffectcolor}{ \textcolor{\coeffectcolor}{\gamma} }\! \cdot \! \Gamma    \vdash_{\mathit{cbncoeff} }   \ottnt{e}  : \tau $, then
    $  \textcolor{\coeffectcolor}{ \textcolor{\coeffectcolor}{\gamma} }\! \cdot \!  \llbracket { \Gamma } \rrbracket_{\textsc{n} }     \vdash_{\mathit{full} }    \llbracket { \ottnt{e} } \rrbracket_{\textsc{n} }   :^{\textcolor{\effectcolor}{  \textcolor{\effectcolor}{\varepsilon}  } }   \llbracket { \tau } \rrbracket_{\textsc{n} }  $,
and
\item For all $\textcolor{\coeffectcolor}{\gamma}$, $\Gamma$, $\ottnt{e}$, $\tau$,
  if $  \textcolor{\coeffectcolor}{ \textcolor{\coeffectcolor}{\gamma} }\! \cdot \! \Gamma    \vdash_{\mathit{cbvcoeff} }   \ottnt{e}  : \tau $, then
    $  \textcolor{\coeffectcolor}{ \textcolor{\coeffectcolor}{\gamma} }\! \cdot \!  \llbracket { \Gamma } \rrbracket_{\textsc{v} }     \vdash_{\mathit{full} }    \llbracket { \ottnt{e} } \rrbracket_{\textsc{v} }   :^{\textcolor{\effectcolor}{  \textcolor{\effectcolor}{\varepsilon}  } }   \ottkw{F}_{\color{\coeffectcolor}{  \textcolor{\coeffectcolor}{1}  } }\;   \llbracket { \tau } \rrbracket_{\textsc{v} }   $,
\end{enumerate}
\end{theorem}

Like \ref{lem:cbv-trans-effects}, these proofs follow by
simple induction, so we omit them here; however, they can be found in the
Coq development.

\section{Related Work}
\label{sec:related-work}

Call-by-push-value (CBPV) was originally developed by
\citet{levy:call-by-push-value}. \citet{forster:cbpv} mechanized proofs of its
metatheoretic properties
and translation soundness and
inspired our mechanized proofs.  Current applications of CBPV include
modeling compiler intermediate
languages~\cite{rizkallah:cbpv,new:cbpv-stal},
understanding the role that polarity plays in bidirectional
typing~\cite{dunfield-krishnaswami} and
subtyping~\cite{lakhani:polarized-subtyping}, and incorporating effects into
dependent type
theories~\cite{pedrot:reasonably-exceptional-type-theory,pedrot:fire-triangle}.


\paragraph{CBPV and Effects}

Call-by-value languages with effect tracking go back to
FX~\cite{lucassen-gifford}.  \citet{wadler-thiemann} showed the connection
between graded monads and effects by translating the effect system of
\citet{talpin-jouvelot} to a language that isolates effects using graded
monads.  Our monadic effect language is inspired by this paper, generalized
following~\citet{katsumata:2014}.  In this paper, our translation is the
reverse of Wadler and Thiemann, mapping a language with graded monads to an
effect-style extension of CBPV. Like us, \citet{rajani:cost-analysis} use a
logical relation to show the soundness of their monadic cost analysis.

Although CBPV has often been used to model the semantics of effects, its type
system has only rarely been extended with effect tracking.  The type system
that we present in Section~\ref{sec:effects} is most similar to MAM
(multi-adjunctive metalanguage) from \citet{forster:expressive-user-effects},
which builds on \citet{kammar-plotkin} and \citet{Kammar2013}. Forster et
al. use MAM to compare the relative expressiveness of effect handlers, monadic
reflection and delimited control. The differences between our system and MAM
are in the abstract structure of effects: MAM does not use a
preordered monoid to track effects. Instead, in each extension effects are
interpreted differently. For effect handlers, effects are a set of operations
specified by some effect signature; for monadic reflection, effects are monad
stacks; for delimited control, effects are a stack of computation types.

\citet{wuttke:ms-thesis} defines a cost-annotated version of CBPV by
annotating the thunk type in CBPV with a bound $[a < I]$ that limits the number
of times that thunks can be forced. This work includes both call-by-value
and call-by-name translations from cost-annotated PCF terms to cost-annotated
CBPV. For expressiveness, the system includes subtyping and indexed
types.

Some extensions of CBPV annotate effects on $\ottkw{F} \, \ottnt{A}$ instead of $\ottkw{U} \, \ottnt{B}$. These
systems isolate effects so that they need not be tracked by the typing judgement.
Extended Call-by-Push-Value
(ECBPV)~\cite{extended-cbpv} adds call-by-need evaluation to CBPV and
layers an effect system to augment equational reasoning. This system uses an
operation $\langle\textcolor{\effectcolor}{\phi}\rangle B$ to extend the effect annotation to other
computation types, combining effects in returner types and pushing effects to
the result type of functions and inside with-products.
\citet{rioux:computation-focusing} tracks divergence. In this system, the sequencing operation requires that the
annotation on the returner type be less than or equal to any annotation on the
result of the continuation.


\paragraph{Coeffects}

Type systems that track coeffects were introduced by
\citet{Brunel:2014,Ghica:2014,petricek:2014} and developed by
\citet{orchard:2019,granule-project,abel:icfp2020}. Early applications were
for bounded linearity; but these systems have also been used for
tracking information flow in differential privacy~\cite{Reed:2010}, dynamic binding~\cite{nanevski:dynamic-binding}
and have also been applied for resource usage in Haskell~\cite{linear-haskell} and
irrelevance in dependently-typed
languages~\cite{atkey:2018,weirich:graded-haskell,abel:2023}.
\citet{petricek:2014} give a number of additional examples, including dataflow (the
number of past values needed in a stream processing language) and data
liveness (whether references to a variable are still needed).

As in our work, all prior semantics that ``count'' uses of variables are
imprecise and allow execution to waste resources.  \citet{abel:2023} and
\citet{weirich:graded-haskell} use a heap-based operational semantics to show
coeffect soundness for a language with a small-step, call-by-name semantics,
but do not consider the interactions with effects.  \citet{bianchini:oopsla24}
proves resource soundness for a fine-grained call-by-value language using a
big-step semantics. Their language includes a nontermination effect through
recursive functions and recursive types. Their soundness proof is based on a
heap-based semantics, which must simultaneously evaluate $\textcolor{\coeffectcolor}{q}$ copies of an
expression. In contrast, because our environment-based semantics can separate
the resource usage of a subexpression from the rest of the computation, our
semantics uses multiplication instead of multi-usage.  For consistency with
effects, several rules of their type system require that the number of copies
of the produced value to be nonzero, similar to our use of $ \textcolor{\coeffectcolor}{q} \ \|\ \textcolor{\coeffectcolor}{1} $.

\citet{dal-lago:relational-theory} also explore the addition of effects and
coeffects to a fine-grained call-by-value language. They also force the
$\ottkw{letin}$ term to count the coeffects of the computation at least once,
through the use of $ \textcolor{\coeffectcolor}{q}  \wedge \textcolor{\coeffectcolor}{1} $. (This rule is derived from
\citet{gavazzo:quantitative}.) Unlike our work, Dal Lago and Gavazzo give a
denotational semantics based on a monadic evaluation function and do not track
resource usage. Their main result is a definition of a program relation in the
presence of effects and coeffects. Their approach is to refine a standard
logical relation with \emph{relators} and \emph{corelators} that capture the
interaction of effects and coeffects with the language semantics. This
approach is more general than ours, which is tied to a specific effect and
coeffect.

\paragraph{CBPV and Linearity} Our extension of CBPV with coeffect typing is
novel and inspired by the duality with effects. The most related systems are
those involving linearity in the context of low-level or compiler intermediate
languages.  \citet{schopp2015computation} develops a low-level language,
similar to CBPV, that includes linear operations in its type system.  The
\emph{enriched effect
  calculus}~\cite{egger:enriched-effect-calculus,egger:enriched-effect-calculus-journal}
extends a type theory for computational effects, with primitives from linear
logic.  \citet{ahmed:L3} augment a variant of typed assembly language with
linear types. \citet{jang_et_al:LIPIcs.FSCD.2024.15} develop a natural
deduction formulation of adjoint logic (which is similar to CBPV) and use its
structure to combine linear, affine, strict and intuitionistic logics in a
uniform setting.


\paragraph{Interactions Between Effects and Coeffects}
\scw{TODO: Add discussion of Effekt language?}

Several systems describe interactions between effects and coeffects.
\citet{nanevski:dynamic-binding} uses comonads to guard the usage of local
state (dynamic binding) and monads to guard the usage of global state. In each
case the type system tracks the set of locations can be safely read and
updated.  In future work, we would like to extend this work with state effects
and local effect handlers so that we can track this interaction using
annotations on thunk and returner types, instead of encapsulated within monad
and comonadic structures.

\citet{Gaboardi:2016} present a combined calculus featuring effects and
coeffects.  Unlike this work, their lambda calculus isolates effects and
coeffects using graded monadic and comonadic modal types. A key feature of
their system are ``graded distributive laws'', that permit interactions
between the monad and comonad. The exact interactions are mediated by operations
determined by the particular effects and coeffects being modeled. For example,
we could distribute a term of type $ \square_{\textcolor{\coeffectcolor}{ \ottsym{3} } }\;  \ottsym{(}   \ottkw{T}_{\textcolor{\effectcolor}{  \textcolor{\effectcolor}{2_{\mathit{eff} } }  } }\;  \tau   \ottsym{)} $
into a term of type $ \ottkw{T}_{\textcolor{\effectcolor}{  \textcolor{\effectcolor}{6_{\mathit{eff} } }  } }\;  \ottsym{(}   \square_{\textcolor{\coeffectcolor}{ \ottsym{3} } }\;  \tau   \ottsym{)} $. That is, it could turn
3 copies of a monad which ticks twice and returns a term of type $\tau$
into a monad which ticks 6 times then returns 3 copies of the term.

In future work, we hope to add distributivity to this system. Unlike the
distributive property described above, in this context the distributive laws
need not change the structure of the computation. Instead, we would like it to
redistribute grades on types in the form $ \ottkw{F}_{\color{\coeffectcolor}{ \textcolor{\coeffectcolor}{q}_{{\mathrm{1}}} } }\;   \ottkw{U}_{\color{\effectcolor}{ \textcolor{\effectcolor}{\phi} } }\;   \ottkw{F}_{\color{\coeffectcolor}{ \textcolor{\coeffectcolor}{q}_{{\mathrm{2}}} } }\;  \ottnt{A}   $ or
$ \ottkw{U}_{\color{\effectcolor}{ \textcolor{\effectcolor}{\phi}_{{\mathrm{1}}} } }\;   \ottkw{F}_{\color{\coeffectcolor}{ \textcolor{\coeffectcolor}{q} } }\;   \ottkw{U}_{\color{\effectcolor}{ \textcolor{\effectcolor}{\phi}_{{\mathrm{2}}} } }\;  \ottnt{B}   $. However, we have yet to determine what sorts of
rearrangement are sound in this context.


\section{Conclusion and Future work}
\label{sec:future-work}
\label{sec:conclusion}

In this paper we have annotated the ambient monad and comonad of CBPV to
statically track effects and coeffects. We have presented these extensions
separately to provide a gentle introduction, before developing a combined
calculus that tracks both simultaneously. We have identified semantic subtleties
in resource tracking and have developed an alternative semantics that better
describes our understanding of this coeffect. We have proven soundness for all
versions of our type system, identifying the required assumptions of the
effect and coeffect algebras. To make sure that our designs are expressive, we
have shown the standard translations from call-by-value and call-by-name
lambda calculi into call-by-push-value preserve tick and resource tracking
with our system.  By exploring both effects and coeffects together, we were
also able to observe similarities between these dual notions, and, more
importantly, identify their differences.

However, this work is only the starting point for investigation in this space.
The natural next step is to go beyond a single effect (tick) and single
coeffect (resource usage) to develop a more general structure for extensions
of CBPV, perhaps based on algebraic effects~\cite{plotkin:2008} or effect
signatures~\cite{katsumata:2014}. This structure would allow us to verify that
our rules stay general in the presence of other effects, such as
nontermination and state, or other coeffects, such as information-flow
tracking and differential privacy.

We can also extend this work by adding language features that interact with
effect and coeffect tracking, such as polymorphism, indexed or dependent
types, and quantification over effects and coeffects. Subtyping would captures
the idea that the type $ \ottkw{U}_{\color{\effectcolor}{ \textcolor{\effectcolor}{\phi}_{{\mathrm{1}}} } }\;  \ottnt{B} $ is a subtype of $ \ottkw{U}_{\color{\effectcolor}{ \textcolor{\effectcolor}{\phi}_{{\mathrm{2}}} } }\;  \ottnt{B} $ when
$ \textcolor{\effectcolor}{ \textcolor{\effectcolor}{\phi}_{{\mathrm{1}}} \;  \textcolor{\effectcolor}{\mathop{\leq_{\mathit{eff} } } } \;  \textcolor{\effectcolor}{\phi}_{{\mathrm{2}}}  } $, and that the type $ \ottkw{F}_{\color{\coeffectcolor}{ \textcolor{\coeffectcolor}{q}_{{\mathrm{1}}} } }\;  \ottnt{A} $ is a subtype of
$ \ottkw{F}_{\color{\coeffectcolor}{ \textcolor{\coeffectcolor}{q}_{{\mathrm{2}}} } }\;  \ottnt{A} $ when $ \textcolor{\coeffectcolor}{ \textcolor{\coeffectcolor}{q}_{{\mathrm{2}}} }\;  \textcolor{\coeffectcolor}{\mathop{\leq_{\mathit{co} } } } \; \textcolor{\coeffectcolor}{ \textcolor{\coeffectcolor}{q}_{{\mathrm{1}}} } $.  Finally, we would like to explore the
practical concerns of this system in more depth, focusing on how users or
compilers might make effective use of the statically tracked information.

\ifanonymous
\else
\begin{acks}
  Thanks to Dominic Orchard, Richard Eisenberg and Kevin Diggs for comments
  and suggestions. Yiyun Liu assisted with the initial setup of our Coq
  proofs, building on a prior Autosubst development of CBPV in
  Coq~\cite{forster:cbpv}.  This work was supported by the
  \grantsponsor{}{National Science Foundation}{} under grants
  \grantnum{}{CCF-2006535}, \grantnum{}{CNS-2244494}, and
  \grantnum{}{CCF-2327738}.

\end{acks}
\fi

\bibliography{paper}

\ifextended
\newpage
\appendix

\section{CBPV with efffects: operational semantics}
Coq definitions in \texttt{effects/CBPV/semantics.v:EvalVal,EvalComp}.
\label{fig:eval-val}
\label{fig:eff-big-step}

\drules[eval-val]{$ \rho  \vdash  \ottnt{V} \ \Downarrow\  \ottnt{W} $}{Value closing}{var,unit,thunk,vpair}
\[ \drule{eval-val-inl} \qquad \drule{eval-val-inr}
\]

\drules[eval-eff-comp]{$ \rho  \vdash_{\mathit{eff} }  \ottnt{M}  \Downarrow  \ottnt{T}  \mathop{\#} \textcolor{\effectcolor}{ \textcolor{\effectcolor}{\phi} } $}{Computation rules}{}
\[  \drule{eval-eff-comp-abs}\ \quad \drule[width=3in]{eval-eff-comp-app-abs}\  \]

\[  \drule[width=4in]{eval-eff-comp-force-thunk}\ \drule{eval-eff-comp-return}\   \]

\[  \drule[width=4in]{eval-eff-comp-letin-ret} \]

\[  \drule{eval-eff-comp-split}\ \quad \drule{eval-eff-comp-cpair} \]

\[  \drule{eval-eff-comp-fst}\  \quad \drule{eval-eff-comp-snd} \]
\[ \drule{eval-eff-comp-sequence} \quad \drule{eval-eff-comp-case-inl} \]

\[ \drule{eval-eff-comp-case-inr} \]

\[  \drule{eval-eff-comp-tick} \]

\section{CBPV with coeffects: typing}
Coq definitions in \texttt{general/typing.v:VWt,CWt}.
\label{fig:cbpv-coeffects}

\drules[coeff]{$ \textcolor{\coeffectcolor}{ \textcolor{\coeffectcolor}{\gamma} }\! \cdot \! \Gamma   \vdash_{\mathit{coeff} }  \ottnt{V}  \ottsym{:}  \ottnt{A}$}{value coeffect typing}{}
\[ \drule[width=3in]{coeff-var}\ \drule[width=4in]{coeff-thunk} \ \drule{coeff-unit} \]

\[ \drule[width=4in]{coeff-pair} \quad \drule[width=3in]{coeff-vsub} \]

\[ \drule{coeff-inl} \quad \drule{coeff-inr} \]
  \drules[coeff]{$ \textcolor{\coeffectcolor}{ \textcolor{\coeffectcolor}{\gamma} }\! \cdot \! \Gamma   \vdash_{\mathit{coeff} }  \ottnt{M}  \ottsym{:}  \ottnt{B}$}{computation coeffect
    typing}{}
\[ \drule{coeff-abs}\ \drule[width=2in]{coeff-app}\ \drule{coeff-force}\ \]

\[ \drule[width=4in]{coeff-split} \ \drule[width=4in]{coeff-ret}\]

\[ \drule[width=3in]{coeff-letin} \ \drule[width=3in]{coeff-cpair}  \]

\[ \drule{coeff-fst}\ \drule{coeff-snd} \ \drule{coeff-csub} \]

\[ \drule{coeff-sequence}\ \drule{coeff-case} \]

\section{Combined Effects and Coeffects}

\subsection{CBPV with Effect and Coeffect Resource Typing}
Coq definitions in \texttt{full/CBPV/typing.v:VWt,CWt}.
\label{app:cbpv-full-typing}
\drules[full]{$ \textcolor{\coeffectcolor}{ \textcolor{\coeffectcolor}{\gamma} }\! \cdot \! \Gamma   \vdash_{\mathit{full} }  \ottnt{V}  \ottsym{:}  \ottnt{A}$}{value typing}{var,thunk,unit,pair,inl,inr,vsub}
\drules[full]{$  \textcolor{\coeffectcolor}{ \textcolor{\coeffectcolor}{\gamma} }\! \cdot \! \Gamma    \vdash_{\mathit{full} }   \ottnt{M}  :^{\textcolor{\effectcolor}{ \textcolor{\effectcolor}{\phi} } }  \ottnt{B} $}{computation typing}{abs,app,force,ret,letin,split,cunit,fst,snd,cpair,sequence,case,tick,csub,letin-zero}

\subsection{CBPV with Effect and Coeffect Resource Instrumented Semantics}
Coq definitions in \texttt{full/CBPV/semantics.v:EvalComp,EvalVal}.
\label{app:cbpv-full-semantics}
\drules[eval-full-val]{$  \textcolor{\coeffectcolor}{ \textcolor{\coeffectcolor}{\gamma} }\! \cdot \! \rho   \vdash_{\mathit{full} }  \ottnt{V}  \Downarrow  \ottnt{W} $}{value rules}{var,thunk,unit,vpair,inl,inr,vsub}
\drules[eval-full-comp]{$  \textcolor{\coeffectcolor}{ \textcolor{\coeffectcolor}{\gamma} }\! \cdot \! \rho   \vdash_{\mathit{full} }  \ottnt{M}  \Downarrow  \ottnt{T}  \mathop{\#} \textcolor{\effectcolor}{ \textcolor{\effectcolor}{\phi} } $}{computation rules}{abs,app,force,ret,letin,split,fst,snd,cpair,sequence,casel,caser,app-abs-zero,ret-zero,split-zero,tick,csub,letin-zero}

\subsection{CBPV with Effect and Coeffect Nondiscarding Instrumented Semantics}
Coq definitions in \texttt{full/CBPV/junk.v:G.EvalComp,G.EvalVal}.
\label{app:cbpv-gfull-semantics}
\drules[eval-gfull-comp]{$  \textcolor{\coeffectcolor}{ \textcolor{\coeffectcolor}{\gamma} }\! \cdot \! \rho   \vdash_{\mathit{gen} }  \ottnt{M}  \Downarrow  \ottnt{T}  \mathop{\#} \textcolor{\effectcolor}{ \textcolor{\effectcolor}{\phi} } $}{computation rules}{abs,app,force,ret,letin,split,fst,snd,cpair,sequence,casel,caser,tick,csub,letin-zero}

\subsection{Binary Logical Relation: General and Resource Tracking Semantics}
\label{sec:binary-rel}
Coq definitions in \texttt{full/CBPV/junk.v:LRV,LRC,LRM}.
\[
\begin{array}{lcl@{\ }l}
\multicolumn{3}{l}{\textit{Related closed values}} \\
 \mathcal{W}\llbracket   \ottkw{U}_{\color{\effectcolor}{ \textcolor{\effectcolor}{\phi} } }\;  \ottnt{B}   \rrbracket  &=& \{\  ( \mathbf{clo}(   \textcolor{\coeffectcolor}{ \textcolor{\coeffectcolor}{\gamma} }\! \cdot \! \rho  ,  \ottnt{M}  )  ,  \mathbf{clo}(   \textcolor{\coeffectcolor}{ \textcolor{\coeffectcolor}{\gamma} }\! \cdot \! \rho'  ,  \ottnt{M'}  )  )
\ |\       ( \textcolor{\coeffectcolor}{\gamma} \cdot   \rho ,  \ottnt{M} ,  \textcolor{\coeffectcolor}{\gamma} \cdot \rho' ,  \ottnt{M'}  )  \in    \mathcal{M}\llbracket  \ottnt{B} \rrbracket^{\textcolor{\effectcolor}{ \textcolor{\effectcolor}{\phi} } }  \ \} \\
 \mathcal{W}\llbracket  \ottkw{unit}  \rrbracket    &=& \{\ ()\ , () \}                               \\
 \mathcal{W}\llbracket   \ottnt{A_{{\mathrm{1}}}} \times \ottnt{A_{{\mathrm{2}}}}   \rrbracket  &=& \{\ \ottsym{(}  \ottsym{(}  \ottnt{W_{{\mathrm{1}}}}  \ottsym{,}  \ottnt{W_{{\mathrm{2}}}}  \ottsym{)}  \ottsym{,}  \ottsym{(}  \ottnt{W'_{{\mathrm{1}}}}  \ottsym{,}  \ottnt{W'_{{\mathrm{2}}}}  \ottsym{)}  \ottsym{)}\
|\                              \ottsym{(}  \ottnt{W_{{\mathrm{1}}}}  \ottsym{,}  \ottnt{W'_{{\mathrm{1}}}}  \ottsym{)} \, \in \,  \mathcal{W}\llbracket  \ottnt{A_{{\mathrm{1}}}}  \rrbracket  \ \textit{and}\  \ottsym{(}  \ottnt{W_{{\mathrm{2}}}}  \ottsym{,}  \ottnt{W'_{{\mathrm{2}}}}  \ottsym{)} \, \in \,  \mathcal{W}\llbracket  \ottnt{A_{{\mathrm{2}}}}  \rrbracket  \ \}\\
 \mathcal{W}\llbracket  \ottnt{A_{{\mathrm{1}}}}  \ottsym{+}  \ottnt{A_{{\mathrm{2}}}}  \rrbracket  &=& \{\ \ottsym{(}  \ottkw{inl} \, \ottnt{W_{{\mathrm{1}}}}  \ottsym{,}  \ottkw{inl} \, \ottnt{W'_{{\mathrm{1}}}}  \ottsym{)}\ |\ \ottsym{(}  \ottnt{W_{{\mathrm{1}}}}  \ottsym{,}  \ottnt{W'_{{\mathrm{1}}}}  \ottsym{)} \, \in \,  \mathcal{W}\llbracket  \ottnt{A_{{\mathrm{1}}}}  \rrbracket  \}\ \cup \\
                     & & \{\ \ottsym{(}  \ottkw{inr} \, \ottnt{W_{{\mathrm{2}}}}  \ottsym{,}  \ottkw{inr} \, \ottnt{W'_{{\mathrm{2}}}}  \ottsym{)}\ |\ \ottsym{(}  \ottnt{W_{{\mathrm{2}}}}  \ottsym{,}  \ottnt{W'_{{\mathrm{2}}}}  \ottsym{)} \, \in \,  \mathcal{W}\llbracket  \ottnt{A_{{\mathrm{2}}}}  \rrbracket  \}  \\
\multicolumn{3}{l}{\textit{Related closed terminals}} \\
 \mathcal{T}\llbracket   \ottkw{F}_{\color{\coeffectcolor}{  \textcolor{\coeffectcolor}{0}  } }\;  \ottnt{A}  \rrbracket^{\textcolor{\effectcolor}{  \textcolor{\effectcolor}{\varepsilon}  } }     &=&  \{\ (  \ottkw{return} _{\textcolor{\coeffectcolor}{ \textcolor{\coeffectcolor}{q} } }  \ottnt{W}  ,  \ottkw{return} _{\textcolor{\coeffectcolor}{ \textcolor{\coeffectcolor}{q} } }  \ottnt{W'}  )\ |\
                                (  \ottnt{W'}  \ottsym{=}  \ottnt{W} \ \textit{or}\  \ottnt{W'}  \ottsym{=}  \mathop{\lightning}  ) \
\mathit{and}\ \ottsym{(}  \ottnt{W}  \ottsym{,}  \ottnt{W'}  \ottsym{)} \, \in \,  \mathcal{W}\llbracket  \ottnt{A}  \rrbracket  \} \\
 \mathcal{T}\llbracket   \ottkw{F}_{\color{\coeffectcolor}{ \textcolor{\coeffectcolor}{q} } }\;  \ottnt{A}  \rrbracket^{\textcolor{\effectcolor}{ \textcolor{\effectcolor}{\phi} } }     &=&  \{\ (  \ottkw{return} _{\textcolor{\coeffectcolor}{ \textcolor{\coeffectcolor}{q} } }  \ottnt{W}  ,  \ottkw{return} _{\textcolor{\coeffectcolor}{ \textcolor{\coeffectcolor}{q} } }  \ottnt{W'}  )\ |\
   \textcolor{\coeffectcolor}{ \textcolor{\coeffectcolor}{q} }  \neq  \textcolor{\coeffectcolor}{  \textcolor{\coeffectcolor}{0}  }  \ \textit{and}\  \ottsym{(}  \ottnt{W}  \ottsym{,}  \ottnt{W}  \ottsym{)} \, \in \,  \mathcal{W}\llbracket  \ottnt{A}  \rrbracket   \\
&&\qquad \mathit{and}\ \ottsym{(}  \ottnt{W}  \ottsym{,}  \ottnt{W'}  \ottsym{)} \, \in \,  \mathcal{W}\llbracket  \ottnt{A}  \rrbracket  \ \} \\
 \mathcal{T}\llbracket   \ottnt{A} ^{\textcolor{\coeffectcolor}{  \textcolor{\coeffectcolor}{0}  } } \rightarrow  \ottnt{B}  \rrbracket^{\textcolor{\effectcolor}{ \textcolor{\effectcolor}{\phi} } }  &=&  \{\
    ( \mathbf{clo}(   \textcolor{\coeffectcolor}{ \textcolor{\coeffectcolor}{\gamma} }\! \cdot \! \rho  ,   \lambda  \ottmv{x} ^{\textcolor{\coeffectcolor}{  \textcolor{\coeffectcolor}{0}  } }. \ottnt{M}   ) ,  \mathbf{clo}(   \textcolor{\coeffectcolor}{ \textcolor{\coeffectcolor}{\gamma} }\! \cdot \! \rho'  ,   \lambda  \ottmv{x} ^{\textcolor{\coeffectcolor}{  \textcolor{\coeffectcolor}{0}  } }. \ottnt{M'}   ) )\
|\  \mathit{forall}\ \ottnt{W}, \ottsym{(}  \ottnt{W}  \ottsym{,}  \ottnt{W}  \ottsym{)} \, \in \,  \mathcal{W}\llbracket  \ottnt{A}  \rrbracket  \\
&&\qquad\ \mathit{implies}\   ( \ottsym{(}   \textcolor{\coeffectcolor}{ \textcolor{\coeffectcolor}{\gamma} \mathop{,}  \textcolor{\coeffectcolor}{0}  }   \ottsym{)} \cdot   \ottsym{(}   \rho   \mathop{,}   \ottmv{x}  \mapsto  \ottnt{W}   \ottsym{)} ,  \ottnt{M} ,  \ottsym{(}   \textcolor{\coeffectcolor}{ \textcolor{\coeffectcolor}{\gamma} \mathop{,}  \textcolor{\coeffectcolor}{0}  }   \ottsym{)} \cdot \ottsym{(}   \rho'   \mathop{,}   \ottmv{x}  \mapsto  \mathop{\lightning}   \ottsym{)} ,  \ottnt{M'}  )  \in    \mathcal{M}\llbracket  \ottnt{B} \rrbracket^{\textcolor{\effectcolor}{ \textcolor{\effectcolor}{\phi} } }  \ \} \\
 \mathcal{T}\llbracket   \ottnt{A} ^{\textcolor{\coeffectcolor}{ \textcolor{\coeffectcolor}{q} } } \rightarrow  \ottnt{B}  \rrbracket^{\textcolor{\effectcolor}{ \textcolor{\effectcolor}{\phi} } }  &=&  \{\
    ( \mathbf{clo}(   \textcolor{\coeffectcolor}{ \textcolor{\coeffectcolor}{\gamma} }\! \cdot \! \rho  ,   \lambda  \ottmv{x} ^{\textcolor{\coeffectcolor}{ \textcolor{\coeffectcolor}{q} } }. \ottnt{M}   ) ,  \mathbf{clo}(   \textcolor{\coeffectcolor}{ \textcolor{\coeffectcolor}{\gamma} }\! \cdot \! \rho'  ,   \lambda  \ottmv{x} ^{\textcolor{\coeffectcolor}{ \textcolor{\coeffectcolor}{q} } }. \ottnt{M'}   ) )
|\    \textcolor{\coeffectcolor}{ \textcolor{\coeffectcolor}{q} }  \neq  \textcolor{\coeffectcolor}{  \textcolor{\coeffectcolor}{0}  } \ \\
&&\qquad \mathit{forall}\ \ottnt{W}\ottnt{W'}, \ \ottsym{(}  \ottnt{W}  \ottsym{,}  \ottnt{W}  \ottsym{)} \, \in \,  \mathcal{W}\llbracket  \ottnt{A}  \rrbracket \ \mathit{and}\ \ottsym{(}  \ottnt{W}  \ottsym{,}  \ottnt{W'}  \ottsym{)} \, \in \,  \mathcal{W}\llbracket  \ottnt{A}  \rrbracket \  \\
&&\qquad \mathit{implies}\  ( \ottsym{(}   \textcolor{\coeffectcolor}{ \textcolor{\coeffectcolor}{\gamma} \mathop{,}  \textcolor{\coeffectcolor}{0}  }   \ottsym{)} \cdot   \ottsym{(}   \rho   \mathop{,}   \ottmv{x}  \mapsto  \ottnt{W}   \ottsym{)} ,  \ottnt{M} ,  \ottsym{(}   \textcolor{\coeffectcolor}{ \textcolor{\coeffectcolor}{\gamma} \mathop{,}  \textcolor{\coeffectcolor}{0}  }   \ottsym{)} \cdot \ottsym{(}   \rho'   \mathop{,}   \ottmv{x}  \mapsto  \mathop{\lightning}   \ottsym{)} ,  \ottnt{M'}  )  \in    \mathcal{M}\llbracket  \ottnt{B} \rrbracket^{\textcolor{\effectcolor}{ \textcolor{\effectcolor}{\phi} } }   \\
 \mathcal{T}\llbracket   \ottnt{B_{{\mathrm{1}}}}   \mathop{\&}   \ottnt{B_{{\mathrm{2}}}}  \rrbracket^{\textcolor{\effectcolor}{ \textcolor{\effectcolor}{\phi} } }  &=& \{\  \mathbf{clo}(   \textcolor{\coeffectcolor}{ \textcolor{\coeffectcolor}{\gamma} }\! \cdot \! \rho  ,   \langle  \ottnt{M_{{\mathrm{1}}}} , \ottnt{M_{{\mathrm{2}}}}  \rangle   )  ,
    \mathbf{clo}(   \textcolor{\coeffectcolor}{ \textcolor{\coeffectcolor}{\gamma} }\! \cdot \! \rho  ,   \langle  \ottnt{M'_{{\mathrm{1}}}} , \ottnt{M'_{{\mathrm{2}}}}  \rangle   ) \ |\
    (  ( \textcolor{\coeffectcolor}{\gamma} \cdot   \rho ,  \ottnt{M_{{\mathrm{1}}}} ,  \textcolor{\coeffectcolor}{\gamma} \cdot \rho ,  \ottnt{M'_{{\mathrm{1}}}}  )  \in    \mathcal{M}\llbracket  \ottnt{B_{{\mathrm{1}}}} \rrbracket^{\textcolor{\effectcolor}{ \textcolor{\effectcolor}{\phi} } }   )  \\
&&\qquad \mathit{and}  (  ( \textcolor{\coeffectcolor}{\gamma} \cdot   \rho ,  \ottnt{M_{{\mathrm{2}}}} ,  \textcolor{\coeffectcolor}{\gamma} \cdot \rho ,  \ottnt{M'_{{\mathrm{2}}}}  )  \in    \mathcal{M}\llbracket  \ottnt{B_{{\mathrm{2}}}} \rrbracket^{\textcolor{\effectcolor}{ \textcolor{\effectcolor}{\phi} } }   )  \} \\
\multicolumn{3}{l}{\textit{Related computations}}\\
 \mathcal{M}\llbracket  \ottnt{B} \rrbracket^{\textcolor{\effectcolor}{ \textcolor{\effectcolor}{\phi} } }  &=& \{\ ( \textcolor{\coeffectcolor}{ \textcolor{\coeffectcolor}{\gamma} }\! \cdot \! \rho_{{\mathrm{1}}}  , \ottnt{M_{{\mathrm{1}}}},  \textcolor{\coeffectcolor}{ \textcolor{\coeffectcolor}{\gamma} }\! \cdot \! \rho_{{\mathrm{2}}}  ,  \ottnt{M_{{\mathrm{2}}}} )
|\    \textcolor{\coeffectcolor}{ \textcolor{\coeffectcolor}{\gamma} }\! \cdot \! \rho_{{\mathrm{1}}}   \vdash_{\mathit{gen} }  \ottnt{M_{{\mathrm{1}}}}  \Downarrow  \ottnt{T_{{\mathrm{1}}}}  \mathop{\#} \textcolor{\effectcolor}{ \textcolor{\effectcolor}{\phi}_{{\mathrm{1}}} }  \ \textit{and}\    \textcolor{\coeffectcolor}{ \textcolor{\coeffectcolor}{\gamma} }\! \cdot \! \rho_{{\mathrm{2}}}   \vdash_{\mathit{full} }  \ottnt{M_{{\mathrm{2}}}}  \Downarrow  \ottnt{T_{{\mathrm{2}}}}  \mathop{\#} \textcolor{\effectcolor}{ \textcolor{\effectcolor}{\phi}_{{\mathrm{1}}} }   \\
&&\quad                 \mathit{and}\  \ottsym{(}  \ottnt{T_{{\mathrm{1}}}}  \ottsym{,}  \ottnt{T_{{\mathrm{1}}}}  \ottsym{)} \, \in \,  \mathcal{T}\llbracket  \ottnt{B} \rrbracket^{\textcolor{\effectcolor}{ \textcolor{\effectcolor}{\phi}_{{\mathrm{2}}} } }  \ \textit{and}\  \ottsym{(}  \ottnt{T_{{\mathrm{1}}}}  \ottsym{,}  \ottnt{T_{{\mathrm{2}}}}  \ottsym{)} \, \in \,  \mathcal{T}\llbracket  \ottnt{B} \rrbracket^{\textcolor{\effectcolor}{ \textcolor{\effectcolor}{\phi}_{{\mathrm{2}}} } }  \
                 \mathit{and}\  \textcolor{\effectcolor}{ \textcolor{\effectcolor}{\phi} }\; = \; \textcolor{\effectcolor}{  \textcolor{\effectcolor}{ \textcolor{\effectcolor}{\phi}_{{\mathrm{1}}}  \cdot  \textcolor{\effectcolor}{\phi}_{{\mathrm{2}}} }  } \ \} \\
\end{array}
\]

\ifproducts

\section{More products: shared values and disjoint computations}
\label{sec:products}

CBPV is an example of a polarized type system. Value types
are positive and computation types are negative. As a result, CBPV includes
both positive and negative products. The former are values, and correspond to
tuples in a call-by-value programming language. Their introduction form is
strict and their elimination is most naturally defined via pattern
matching. Negative products are computations --- they can contain effectful
computation as subcomponents --- and are most naturally eliminated via projection.

Polarization is also found in linear type systems, which are related to
coeffects.  In such systems, there are also two forms of products:
multiplicative products (also called tensor products) are positive and
additive products (also called ``with'' products) are negative.
Multiplicative products are formed from disjoint resources and must be
eliminated via pattern matching so that those resources are not discarded. In
contrast, additive products are formed from shared resources and must be
eliminated via projection so that resources are not duplicated.

By design, the polarity in CBPV allows it to model the duality between
call-by-value and call-by-name semantics. With the addition of coeffect
tracking, we can use this same structure to observe the duality between shared
and disjoint demands on resources.
However, these two forms of polarity do not have to align. We've seen that we
can have positive products that are values and have disjoint resources, and
negative products that are computations and share resources.  In this section,
we explore the other two forms of products: shared value products and disjoint
computational products. These two additions are not the same: the former adds
expressiveness, whereas the latter can already be simulated by existing
features.

\paragraph{Shared Value Products}
For shared value products, we introduce the syntax $ \ottnt{A_{{\mathrm{1}}}} \mathop{\&} \ottnt{A_{{\mathrm{2}}}} $ for this new
type, and new values $ \langle  \ottnt{V_{{\mathrm{1}}}}  ,  \ottnt{V_{{\mathrm{2}}}}  \rangle $, $\ottnt{V}  \ottsym{.}  \ottsym{1}$, and $\ottnt{V}  \ottsym{.}  \ottsym{2}$. The typing rules
for are as follows:
\[ \drule{coeff-vwith}\ \drule{coeff-vfst}\ \drule{coeff-vsnd} \]

Unlike $ \ottnt{A_{{\mathrm{1}}}} \times \ottnt{A_{{\mathrm{2}}}} $, the existing ``tensor'' product of the value language,
the components of these tuples share resources. As a result, the elimination
form must be a projection operation.  However, unlike $ \ottnt{B_{{\mathrm{1}}}}   \mathop{\&}   \ottnt{B_{{\mathrm{2}}}} $, the
components of this type are values. Because of this strictness, it is sound,
although unusual, to include the elimination form for this type in the value
language. These projections must be effect-free because they can only apply to
values.

For the operational semantics, we must introduce a new form of closure
for this type, written $ \mathbf{clo}(  \textcolor{\coeffectcolor}{ \textcolor{\coeffectcolor}{\gamma}' }\! \cdot \! \rho'  , \langle  \ottnt{V_{{\mathrm{1}}}}  ,  \ottnt{V_{{\mathrm{2}}}}  \rangle ) $, that stores
both parts of the pair with a saved environment. This environment can then be
used when evaluating the projected value.

\renewcommand{\ottdruleevalXXcoeffXXvalXXvwith}[1]{\ottdrule[#1]{%
\ottpremise{  }%
}{
  \textcolor{blue}{ \textcolor{blue}{\gamma} }\! \cdot \! \rho   \vdash_{\mathit{coeff} }   \langle  \ottnt{V_{{\mathrm{1}}}}  ,  \ottnt{V_{{\mathrm{2}}}}  \rangle \Downarrow  \\ \qquad    \mathbf{clo}(  \textcolor{blue}{ \textcolor{blue}{\gamma} }\! \cdot \! \rho  , \langle  \ottnt{V_{{\mathrm{1}}}}  ,  \ottnt{V_{{\mathrm{2}}}}  \rangle )  }{%
{\ottdrulename{eval\_coeff\_val\_vwith}}{}%
}}

\[ \drule{eval-coeff-val-vwith}\!\!\drule{eval-coeff-val-vfst}\!\!\drule{eval-coeff-val-vsnd} \]

\paragraph{Disjoint computational products}

We can also consider another form of product: nonstrict products that are
composed from distinct resources. These products are more familiar: both
\citet{abel:icfp2020} and \citet{weirich:graded-haskell} include this product
in their CBN languages.

To formalize this extension, we use the syntax $ \ottnt{B_{{\mathrm{1}}}}  \times  \ottnt{B_{{\mathrm{2}}}} $ to denote the
type of these products, introduce them as pairs of computations
and eliminate them via pattern matching.
\[ \drule[width=3in]{coeff-ctensor}\ \drule[width=3in]{coeff-csplit} \]


This type is isomorphic to the computation type $\ottkw{F} \, \ottsym{(}   \ottkw{U} \, \ottnt{B_{{\mathrm{1}}}} \times \ottkw{U} \, \ottnt{B_{{\mathrm{2}}}}   \ottsym{)}$. Indeed,
we use thunks in its typing rules and operational semantics. For example,
the variables inserted in the context by pattern matching must have value
types, so we thunk them.  Also, the evaluation rule for tensors requires a new terminal form
$\ottsym{(}  \ottnt{W_{{\mathrm{1}}}}  \ottsym{,}  \ottnt{W_{{\mathrm{2}}}}  \ottsym{)}$, where each component is a closed value thunk for each
subterm. This thunk must divide up the available resources between the two
subcomputations. When this pair is eliminated, the two thunks are added
to the environments.
\[ \drule{eval-coeff-comp-ctensor} \]
\[ \drule[width=6in]{eval-coeff-comp-csplit} \]

Together, the four product types provide insight into the nature of computation with
resources---that there are fundamental choices in how to allocate resources
between subcomputations and how their results may be used.

\fi

\fi

%
%
%
%
%

\end{document}